%% file: main.tex
\PassOptionsToPackage{dvipsnames,svgnames}{xcolor}
\documentclass[a4paper,onecolumn,10pt,accepted=2026-04-28]{quantumarticle}
\pdfoutput=1

\usepackage[utf8]{inputenc}
\usepackage[english]{babel}
\usepackage[T1]{fontenc}

\usepackage{amsmath,amssymb,amsthm,mathtools}
\usepackage{graphicx}
\usepackage{tikz}
\usepackage{tikz-cd}
\usetikzlibrary{decorations.markings}

\usepackage{braket}
\usepackage{bm}
\usepackage{dsfont}
\usepackage{manfnt}
\usepackage{mathrsfs}
\usepackage{multirow}
\usepackage{caption}
\usepackage{subcaption}
\usepackage{placeins}
\usepackage{soul}
\usepackage{enumitem}
\usepackage[most]{tcolorbox}
\usepackage{verbatim}
\usepackage{etoolbox}

\usepackage[numbers]{natbib}
\usepackage{hyperref}

\setlist[itemize]{itemsep=6pt, topsep=6pt}
\setlist[enumerate]{itemsep=6pt, topsep=6pt}

\newcommand{\hi}[1]{#1}
\newcommand{\str}[1]{}

\def\ketbra#1#2{{\vert#1\rangle\!\langle#2\vert}}

\newtheoremstyle{thmstyleone}% 〈name〉
  {8pt minus2pt plus2pt}% 〈space above〉
  {8pt minus2pt plus2pt}% 〈space below〉
  {\itshape}% 〈body font〉
  {}% 〈indent amount〉
  {\bfseries}% 〈theorem head font〉
  {.}% 〈punctuation after theorem head〉
  { }% 〈space after theorem head〉
  {}% 〈theorem head spec〉

\theoremstyle{thmstyleone}
\newtheorem{theorem}{Theorem}[section]
\newtheorem{lemma}{Lemma}[section]
\newtheorem{definition}{Definition}[section]
\newtheorem{corollary}{Corollary}[section]
\newtheorem{proposition}{Proposition}[section]
\newtheorem{fact}{Fact}
\newtheorem{algorithm}{Algorithm}
\newtheorem{remark}{Remark}

\newtheorem{conjecture}{Conjecture}

\numberwithin{equation}{section}

\newcommand{\nc}{\newcommand}
\nc{\rnc}{\renewcommand}

\nc{\rc}[1]{{\color{red}{#1}}}
\nc{\bc}[1]{{\color{blue}{#1}}}

\nc{\be}{\begin{equation}}
\nc{\ee}{\end{equation}}
\nc{\df}{definition}

\nc{\MD}[1][1]{}
\nc{\cmt}[1][1]{}
\nc{\SB}[1]{}
\nc{\ethan}[1]{}
\nc{\gacs}[1]{}

\nc{\eq}[1]{\eqref{eq:#1}}

\rnc{\L}{\left}
\nc{\R}{\right}
\nc{\ot}{\otimes}

\nc\cA{\mathcal{A}}
\nc\cB{\mathcal{B}}
\nc\cC{\mathcal{C}}
\nc{\cD}{{\mathcal{D}}}
\nc\cE{\mathcal{E}}
\nc\cF{\mathcal{F}}
\nc\cG{\mathcal{G}}
\nc\cH{\mathcal{H}}
\nc\cI{\mathcal{I}}
\nc\cJ{\mathcal{J}}
\nc\cK{\mathcal{K}}
\nc\cL{\mathcal{L}}
\nc\cM{\mathcal{M}}
\nc\cN{\mathcal{N}}
\nc\cO{\mathcal{O}}
\nc\cP{\mathcal{P}}
\nc\cQ{\mathcal{Q}}
\nc\cR{\mathcal{R}}
\nc\cS{\mathcal{S}}
\nc\cT{\mathcal{T}}
\nc\cU{\mathcal{U}}
\nc\cV{\mathcal{V}}
\nc\cW{\mathcal{W}}
\nc\cX{\mathcal{X}}
\nc\cY{\mathcal{Y}}
\nc\cZ{\mathcal{Z}}

\let\emptyset\varnothing

\nc{\numT}{k}
\nc{\finT}{m}
\nc{\EC}{EC}
\nc{\nil}{\ell}

\nc\trel{t_{\rm rel}}
\nc\supp{\mathrm{supp}}

\nc{\maj}{\mathsf{Maj}}
\nc{\recmaj}{\mathsf{Maj}_r}

\newtcolorbox{yellowsection}{
  breakable,
  enhanced,
  colback=yellow!10,
  colframe=yellow!100,
  boxrule=4pt,
  arc=0pt,
  left=0mm,right=0mm,top=1mm,bottom=1mm
}

\begin{document}

\title{A local automaton for the 2D toric code}

\author[1]{Shankar Balasubramanian}
\author[2,3]{Margarita Davydova}
\author[4]{Ethan Lake}

\affil[1]{Center for Theoretical Physics, Massachusetts Institute of Technology, Cambridge, MA 02139}
\affil[2]{Department of Physics, Massachusetts Institute of Technology, Cambridge, MA 02139}
\affil[3]{Walter Burke Institute for Theoretical Physics and Institute for Quantum Information and Matter, California Institute of Technology, Pasadena, CA 91125}
\affil[4]{Department of Physics, University of California Berkeley, Berkeley, CA 94720}

\maketitle

\begin{abstract}
We construct a local decoder for the 2D toric code using ideas from the hierarchical classical cellular automata of Tsirelson and G\'acs. Our decoder is a circuit of strictly local quantum operations preserving a logical state for exponential time in the presence of circuit-level noise without the need for non-local classical computation or communication.  Our construction is not translation invariant in spacetime, but can be made time-translation invariant in 3D with stacks of 2D toric codes. This solves the open problem of constructing a local topological quantum memory below four dimensions.
\end{abstract}

\setcounter{tocdepth}{2}
\tableofcontents

\input{1.introduction}

\input{2.tsirelson.tex}

\input{3.TCconstruction.tex}
\input{4.1tsirelson.tex}

\input{5.TCproof.tex}
\input{6.3D.tex}

\FloatBarrier

\section*{Acknowledgments}
We are grateful to Peter G\'acs for numerous insightful discussions and for providing many comments that significantly improved this paper.  We thank Daniel Gottesman, Alexei Kitaev, John Preskill, and Ali Lavasani for useful discussions.  SB was supported by the National Science Foundation Graduate Research
Fellowship under Grant No.~1745302, as well as PHY-1818914 and PHY-2325080. MD was supported by the Walter Burke Institute for Theoretical Physics at Caltech. EL was supported by a Miller research fellowship.

\appendix
\input{AppendixA0.tex}
\input{AppendixA.tex}
\input{AppendixB.tex}
\FloatBarrier

\bibliographystyle{quantum}
\bibliography{ref}

\end{document}

%% file: 1.introduction.tex
\section{Introduction}

Fault tolerant quantum computation is a necessary ingredient for building a scalable quantum computer. 
The first proof of fault tolerance with a constant threshold, due to Aharonov and Ben-Or~\cite{aharonov1996faulttolerantquantumcomputation,Aharonov1999Jun} as well as Knill, Laflamme, and Zurek~\cite{knill1998resilient}, involved encoding a quantum circuit with a concatenated code.  %
Later, topological codes like the surface and toric code~\cite{Kitaev1997Jul,Dennis_2002} were introduced for storing quantum information, and have since remained \hi{one of the leading practical schemes} thanks to high thresholds, low-weight stabilizers, and simple prescriptions for fault-tolerant gates. More recently, qLDPC codes with superior encoding rates and distances were proposed~\cite{breuckmann2021quantum} and provide promising alternatives for architectures operating on some amount of geometric non-locality.

However, none of the existing approaches, apart from circuit concatenation~\cite{Gottesman_2000}, are currently known to be \emph{truly} scalable in three or fewer spatial dimensions while remaining fault tolerant. For true scalability, locality and speed limitations in classical computation and communication need to be included.  For example, in an architecture that supports non-local connections, the amount of time needed to implement these operations (such as bringing atoms together to implement a gate in cold atom/trapped ion architectures) may not be able to compete with the noise rate in the limit of infinite system size.

Thus, the interplay between fault tolerance and locality has remained a fundamental question that has not been addressed in full generality.  This paper focuses on understanding whether a \hi{\emph{topological}} quantum memory can operate in a geometrically local way \hi{in three dimensions or below} and preserve encoded information for a long time. While we approach this question from a purely theoretical point of view, its resolution has potentially important practical implications, since backlogs due to the need for global decoding~\cite{RevModPhys.87.307} has increasingly become a problem in current quantum devices.

One of the existing end-to-end fault-tolerant quantum computation schemes is due to Aharonov and Ben-Or~\cite{aharonov1996faulttolerantquantumcomputation,Aharonov1999Jun}, who use concatenated codes~\cite{knill1996concatenatedquantumcodes,aharonov1996faulttolerantquantumcomputation}. \hi{Their construction becomes fully local when classical computation responsible for error correction is treated as a part of the circuit to be concatenated}. Many of the details are fleshed out by Gottesman~\cite{Gottesman_2000}, and additional details regarding classical locality are discussed in his upcoming book Ref.~\cite{GottesmanBook}. \hi{A crucial feature of the concatenated codes prescription is that it relies on procedures like fault-tolerant measurement, which requires a qubit overhead whose function is to store intermediate information needed to infer measurement outcomes. These qubits must then be incorporated in the next level of circuit concatenation, resulting in a rather complicated local fault-tolerance scheme.}

It could be possible to have a model where \hi{quantum information is encoded in a topological code.  A paradigmatic example of a topological code is the Kitaev's surface code}~\cite{Kitaev1997Jul}.  
In four dimensions and higher, some topological codes are passive self-correcting quantum memories (such as the 4D toric code~\cite{Dennis_2002,alicki2008thermalstabilitytopologicalqubit}). 
Such a memory can encode logical information at a non-zero temperature and usually admits a local fault tolerant scheme (also known as a local decoder) via a set of explicitly local rules which preserves the encoded logical qubit for a time diverging in the limit of large system size even in the presence of noise. 
% an exponentially long time even in the presence of noise.  
However, local decoders are only known for topological codes that possess exclusively membrane-like excitations, which so far is only known to be possible in an unphysical number of spatial dimensions (four and higher).  While topological memories in lower spatial dimensions are generally believed to not be finite-temperature quantum memories, this \emph{does not} rule out the existence of a local decoder for such memories.

A reason to believe that low-dimensional topological codes can possess local decoders is due to results in the field of reliable local classical computation. In the past, several schemes of local rules have been proposed to stabilize a classical bit of information in the presence of noise. The problem of reliable classical computation from unreliable components was first introduced by von Neumann~\cite{von1956probabilistic}, who proposed a solution based on self-simulating automata (also referred to as a universal self-replicating machine)~\cite{von1966theory}, although this construction was non-local. 

The most famous example of a \emph{local} scheme that realizes a reliable classical memory is Toom's rule \cite{toom1980stable} which preserves the classical bit encoded in a 2D classical repetition code. This rule has been adapted to construct local quantum decoders for the 4D toric code~\cite{ahn2004extending,breuckmann2016local} since the errors that create islands of membranes can be shrunk similarly to the error domains in the 2D repetition code. 
However, in the 2D toric code, errors can be revealed only through their point-like defects, which is more similar to the 1D classical repetition code.  Unfortunately, the ``entropy versus energy'' argument in statistical mechanics forbids the existence of a finite-temperature ordered phase in 1D.
However, these arguments do not rule out the possibility of reliably storing a bit of information using non-thermalizing dynamics, which do not obey detailed balance and thus do not equilibrate to a Gibbs state.
An example of such a 1D dynamical system was shown by G\'acs in a beautiful paper~\cite{gacs1983reliable} and \cite{Gacs_2001}\footnote{See also the newer and improved version at
\url{https://arxiv.org/abs/math/0003117}.}, which relied on hierarchical (renormalization group-like) methods to eliminate islands of errors. G\'acs's construction has other striking properties, such as being translation invariant, being stable to asynchrony of local clocks (and thus implementable in continuous time), 
and possessing a constant encoding rate. However, the details of G\'acs's construction are rather complicated and the local state space requires, as a rough estimate, between $2^{24}$ and $2^{400}$ states~\cite{Gray2001Apr}. 
It remains an outstanding open challenge to propose a quantum scheme that is similar to G\'acs's, and operates in a physical number of spatial dimensions (three or less). The only other example of a fault-tolerant classical memory in one dimension that is known to us is the lesser-known Tsirelson's automaton~\cite{Cirel'son2006Aug} that was proposed prior to G\'acs's construction and was the inspiration, along with the G\'acs automaton, for the quantum fault tolerance theorem using concatenated codes~\cite{Aharonov1999Jun}. Tsirelson's automaton, which features heavily in our work, preserves one bit of information in a repetition code for a long time, but strongly breaks translation symmetry in space and time. 

From the physics perspective, having a \hi{new} example of a local fault-tolerant quantum memory would have interesting implications for non-equilibrium quantum phases of matter.  Geometric locality is an important assumption for physical systems. Unlike zero-temperature gapped phases of matter in low dimensions, where phase equivalence has been relatively well understood and a conjectured classification has been put forward, very little is known about low dimensional stable phases of matter in the non-equilibrium setting (circuit or Lindbladian dynamics), such as general no-go theorems, notions of phase equivalence, and a conjectured classification. G\'acs automaton, in particular, should constitute a stable phase under any valid classification, but it has not been studied from this perspective because of its complexity; Tsirelson's automaton has not been either. A fully local quantum fault-tolerant memory (that remains local even without needing to reliably store classical bits on the side) \hi{is} another especially interesting example of a stable non-equilibrium quantum phase of matter.

In this paper, we use Tsirelson's automaton~\cite{Cirel'son2006Aug}, which provides a simpler starting point than G\'acs's automaton, and generalize it to construct a local ``quantum automaton'' for the 2D toric code. 
Operationally, our decoder consists of small local components that measure checks of the toric code in a local region at each time step and apply local feedback\footnote{Note that measurements perform the role of ``cooling'' the system, since it is known that unitary/reversible operations alone do not suffice for reliably preserving a quantum memory \cite{benor2013quantumrefrigerator}. To avoid measurements entirely, one could replace local measurement-and-feedback framework with local unitary gates and ancilla reset.}. The measurement outcomes are discarded a constant number of timesteps after having been collected, and thus our construction does not require the help of a reliable classical computer to process or store the measurement outcomes. These local operations are composed in such a way that the same action is simulated (via ``self-simulation'') at a growing hierarchy of length scales, thus, implementing error correction at all scales.
As we discuss in \hi{``Comparison with earlier literature''}, Subsec.~\ref{sec:comparison}, we are not aware of another quantum construction in fewer than four dimensions where this is the case. 

Our 2D decoder lacks translation invariance in both space and time, i.e. the entire hierarchy of operations in the circuit is ``hardwired''. 
However, we show how to get rid of the hardwiring by making our construction \emph{time-translation invariant} in 3D using a stack of 2D toric codes. This constitutes the first example of a local fault tolerant scheme for a topological memory in below four spatial dimensions.  Moving forward, it should be possible to eventually turn the 3D construction into a fault-tolerant and universal \emph{local quantum computer} based on topological codes, as we discuss in the ``Outlook'' section \ref{sec:outlook} below.

\subsection{Overview of main results and structure of the paper}

We will provide an informal but self-contained discussion of the main results in the paper and reference appropriate sections where these results are proved. 

\subsubsection*{Starting point: one-dimensional automaton}

In Sec.~\ref{sec:tsirelson}, we discuss Tsirelson's classical automaton, which stores one bit of information encoded in a repetition code and is built using a few elementary gadgets and a recursive substitution rule. We briefly review Tsirelson's original results in Subsec.~\ref{ss:tsirelsonOG}. We then show that, with a minor modification, we can interpret this automaton as repeated circuit concatenation, see Subsec.~\ref{ss:modification_of_tsirelson}.
This modern approach makes it easy to see why Tsirelson's automaton is fault-tolerant, and also allows one to borrow the methods from quantum concatenated codes~\cite{Aliferis2005Apr,gottesman2009introductionquantumerrorcorrection}. We provide a very concise proof that Tsirelson's automaton is fault-tolerant in Subsec.~\ref{subsec:proof1}.

The elementary gadgets of Tsirelson's automaton are:
\begin{equation} 
 \includegraphics[scale=0.28]{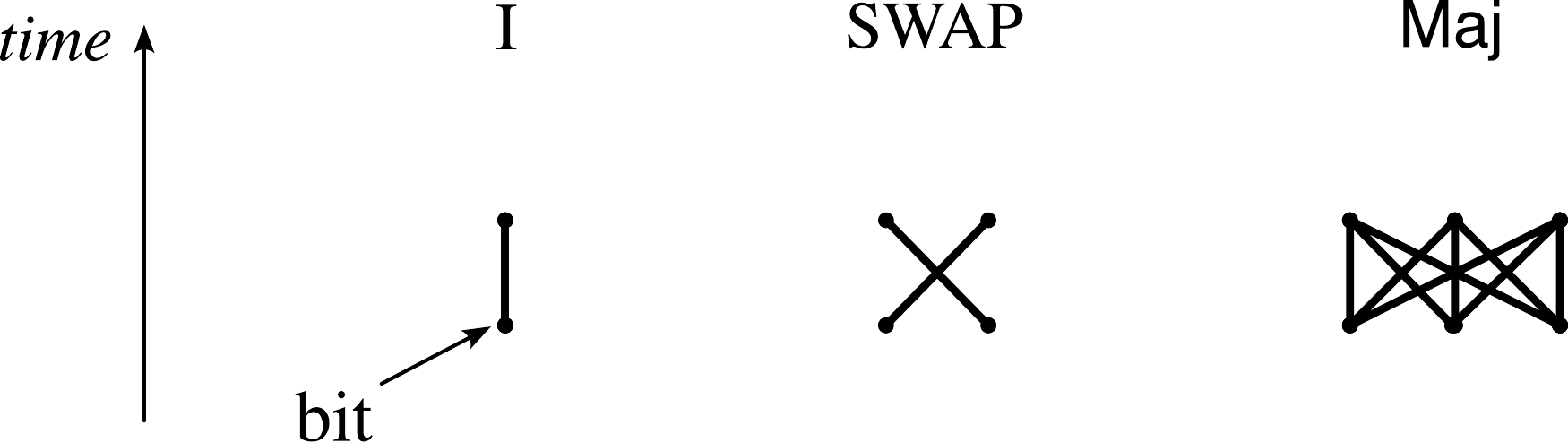}
 \end{equation}
where the left gadget acts as an identity, the middle gadget swaps the state of two bits, and the right gadget performs a majority vote $\maj$ on a 3-bit block and copies the output three times. All three operations are treated as gates; the third operation is treated as a gate but also as an error-correction ($EC$) gadget. 

One then constructs a concatenated circuit repeatedly concatenating the circuit with a 3-bit repetition code and replacing gates in the circuit with ``fault-tolerant'' versions  $FT(\cdot)$, shown below: 
\begin{equation} 
 \includegraphics[scale=0.2]{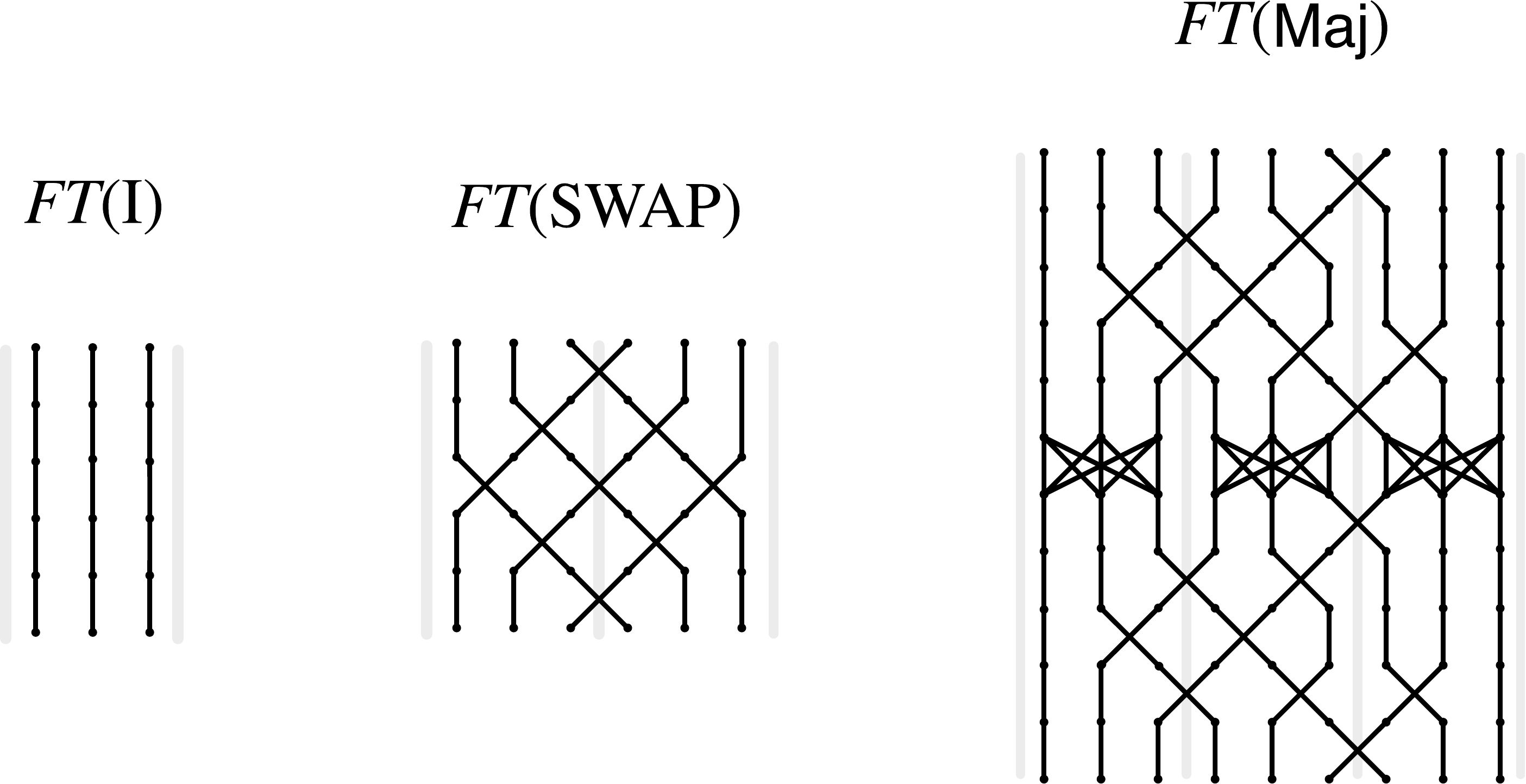}
 \end{equation}
The logical action of $FT(G)$ on 3-bit code blocks acts as an encoded version of $G$ and can tolerate any single qubit fault within it. A level-$n$ fault-tolerant simulation of any circuit is then defined in the standard way:
\begin{definition}[Fault tolerant simulation; informal] \label{def:FT_original-intro}
A fault-tolerant simulation of a circuit of elementary gadgets $\cC$, denoted $FT(\cC)$, is a circuit wherein each gate $G$ is replaced with its fault-tolerant version $FT(G)$, and a layer of $EC$ gadgets is inserted between layers of fault-tolerant gates.
A repeated fault tolerant simulation of a circuit $\cC$, denoted $FT^n(\cC)$, is the circuit formed from the recursion $FT^n(\cC) = FT(FT^{n-1}(\cC))$.
\end{definition}
Tsirelson's automaton is defined as the circuit consisting of $T$ applications of $FT^n(I)$. Denoting the system size to be $L = 3^n$, this automaton is fault-tolerant under the bit-flip noise:
\begin{theorem}[Informal, Tsirelson~\cite{Cirel'son2006Aug}]\label{thm:FT_TS-intro}
Let $s(t)$ denote a bit string describing the value of each bit after $T$ repetitions of Tsirelson's automaton on $L$ bits.  At a low enough rate of (bit flip) noise and starting from the initial state where $\forall i, \, s_i(0) = 1$, we have
\begin{equation}
\mathbb{P}\left(D(s(T)) = 1\right) \geq 1 - T \cdot \exp\left(-\alpha L^{\beta}\right)
\end{equation}
where $D(\cdot): \{0,1\}^L \to \{0,1\}$ is a decoder that implements a hierarchical majority vote to determine the logical state, with $\alpha$ and $\beta = \log_3 2$ positive constants.
\end{theorem}
\noindent We prove this statement in Subsec.~\ref{subsec:proof1}, and extend the result to a more general noise model. We estimate the threshold numerically in Subsec.~\ref{ss:tsirelson_numerics}. 

\subsubsection*{Repetition code as a stabilizer code}

Tsirelson's automaton is a promising starting point because it is a local circuit preserving the codestate of the classical 1D repetition code for superpolynomial in $L$ duration of time. The 1D repetition code has an important similarity to 2D topological codes in that its error syndrome can be viewed as a set of points that have to be matched for decoding.  
However, in the quantum case, one cannot use an automaton that directly reads off the state of each qubit before operating on them. Instead, we need to use a logical information-non-revealing approach called ``measurement and feedback model''. 
\begin{definition}[Informal, stabilizer quantum automaton]
A (noisy) stabilizer quantum automaton for a stabilizer code is a sequence of local operations that, at each timestep: (a) measures stabilizer generators of the code, thereby obtaining syndromes, (b) applies local Pauli feedback operators which only depend on a local ($O(1)$) neighborhood of syndromes, (c) discards syndromes and applies local stochastic noise.
\end{definition}
We show that Tsirelson's automaton can be formulated in this way by recasting the classical 1D repetition code as a stabilizer code with stabilizer group generated by $Z_i Z_{i+1}$. Placing qubits on edges of a 1D lattice, the presence of a nontrivial syndrome at a vertex is denoted below:
\begin{equation} 
 \includegraphics[width=0.35 \textwidth]{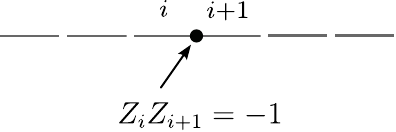}
 \end{equation}
We can rewrite elementary gadgets so that they operate by measuring syndrome and applying feedback:
\begin{itemize}
    \item $\cI_0$: identity gate;
    \item $\cT_0$: ``splits the error''.  If a nontrivial syndrome is present at vertex $v$, this gadget applies a Pauli-$X$ operator to two qubits adjacent to the vertex:
\begin{equation} 
 \includegraphics[width=0.34 \textwidth]{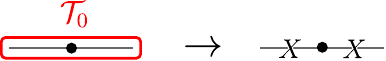}
 \end{equation}
 One can check that in the computational basis, this operation is identical to SWAP. 
    \item $\cM_0$: a ``matching'' operation, which applies $X$ to an edge if both endpoints contain nontrivial syndrome (thereby merging the two syndrome points):
\begin{equation} 
 \includegraphics[width=0.42 \textwidth]{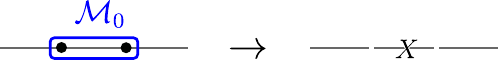}
 \end{equation}
\end{itemize}
\noindent The measurement-and-feedback version of the earlier $\maj$ (error correction) gadget will now be a composite gadget that we call $\cR_0$ (``recovery'' operation), that is decomposed in 5 steps:
\begin{equation*} 
 \includegraphics[width=1 \textwidth]{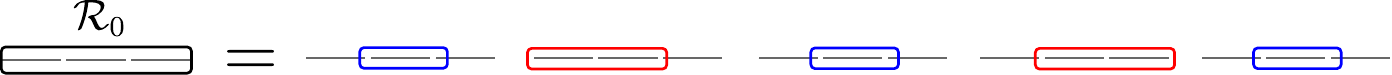}
 \end{equation*}
We can now define a new version of Tsirelson's automaton where the elementary gadgets $I$, SWAP and $\maj$ are replaced with $\cI_0$, $\cT_0$, and $\cM_0$.  We also write down fault tolerant simulations $FT(\cI_0)$, $FT(\cT_0)$, and $FT(\cM_0)$. This version of the 1D construction is explained in detail in Subsec.~\ref{modified_tsirelson2}.  We also prove its fault-tolerance in Sec.~\ref{sec:tsirelson-proof-2}, wherein we also establish some of the new proof techniques that will be needed for our toric code automaton. 

A useful intuition about the operation of the automaton can be gained by considering the following example. Consider an input with some bits flipped (shown in red); and act on it with a noiseless implementation of gadgets $\cR_0$ that cover the lattice (which is roughly equivalent to $FT(\cI_0)$), and its higher-level simulations, $FT(\cR_0)$ and $FT^2(\cR_0)$. This is shown below: 
\begin{equation*} 
 \includegraphics[width=1 \textwidth]{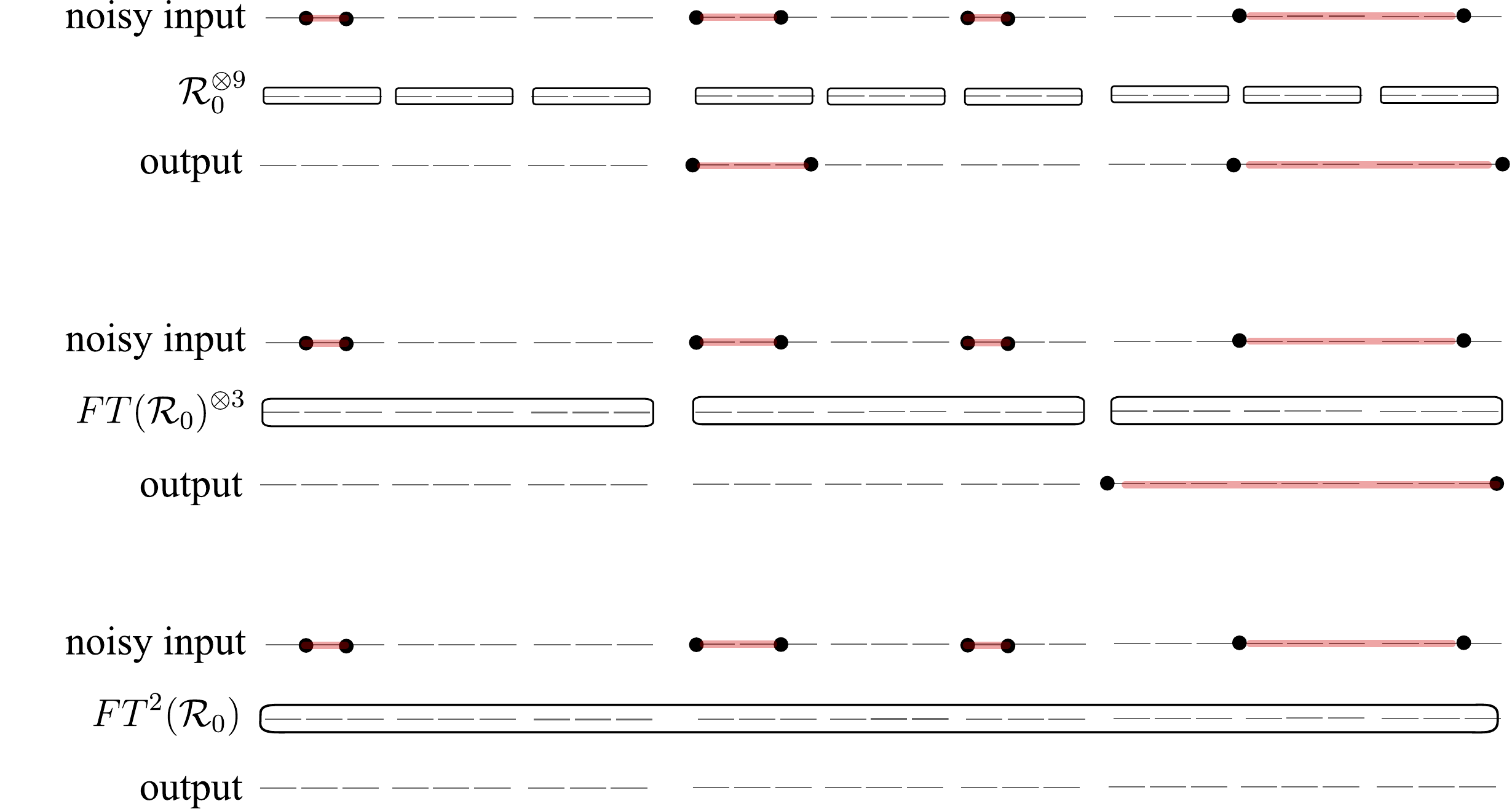}
 \end{equation*}
In each case, the output consists of blocks of codewords of a concatenated repetition code.  In the top figure, it is a 3-bit repetition code, in the middle figure a $9$-bit repetition code, and in the bottom figure a $27$-bit repetition code. As a consequence, the syndrome in the output state is pinned to the closest syndrome configuration of a superlattice corresponding to the concatenation level. We call this a \emph{coarse-graining} action that occurs at all levels of the construction.  The reader may draw some parallels with the real-space renormalization methods.

\subsubsection*{Quantum automaton: summary of assumptions}

Let us start this section by clarifying several key assumptions that we make throughout the paper:
\begin{itemize}
    \item [(1)] \textit{We assume a $p$-bounded  gadget error model (a local stochastic noise model), where an arbitrary error channel can act in the support of each elementary gadget during its operation.  }
\end{itemize}
See Def.~\ref{def:p-bounded-classical} for the classical automaton and Def.~\ref{def:p-bounded-gadget-model} for the quantum one. This error model only requires that the probability of a particular configuration of errors is suppressed exponentially in the size of the configuration. 
It also automatically incorporates the effects of measurement errors, because in a local automaton, the effect of a measurement error is equivalent to failing the output of an elementary gadget.
In addition, since the toric code is a CSS code, we can split error correction into separate phase-flip and bit-flip correcting processes. Therefore, when explaining how error correction works, we focus on the bit-flip (Pauli-$X$) part of noise.  The phase-flip correction works analogously, except on the dual lattice. Both operations are then combined to run in parallel, resulting in the full quantum automaton.  

Another assumption is:
\begin{itemize}
    \item [(2)] \textit{The qubits are located on the links of a square lattice,  the bit-flip noise syndrome points live on the vertices of the lattice, and the face-flip noise syndrome points live on the faces. }
\end{itemize}
which we believe is inessential, and that there should exist a generalization of our automaton for any self-similar cellulation. Generalizing our construction to this case, as well as to more general cellulations and codes, is an interesting problem for future research. Finally, we make the following assumption:
\begin{itemize}
    \item [(3)] \textit{In all dimensions, our decoders are assumed to be run in discrete time, with all sites of the automaton updating in unison at each time step. }
\end{itemize}
This is an important assumption, and generalizing our construction to continuous time/local clocks will require ideas from G\'acs's automaton; we discuss this in the Outlook, Subsec.~\ref{sec:outlook}.  Nevertheless, we believe it should be possible to achieve this.

Our construction in two dimensions will be a hardwired circuit, similarly to Tsirelson's automaton in 1D. We then show how to obtain a time-translation invariant (but not spatially translation invariant) automaton by introducing one extra spatial dimension.

\subsubsection*{Hardwired rule: 2D toric code automaton}
 
Although the toric code cannot be obtained from code concatenation, there is still a simple way of understanding how the fault-tolerant simulation operates. Similarly to Tsirelson's automaton action on syndrome points, our automaton can be understood to dynamically coarse-grain the error configuration.  We assign superlattices in the square-lattice toric code that correspond to each level of the simulation hierarchy.  We define a level-$k$ superlattice to be a square lattice whose primitive cells have size $3^k \times 3^k$, and level-$k$ coarse-grained error syndromes to be syndromes restricted to live on vertices of a level-$k$ superlattice:
\begin{equation*} 
 \includegraphics[width=0.6 \textwidth]{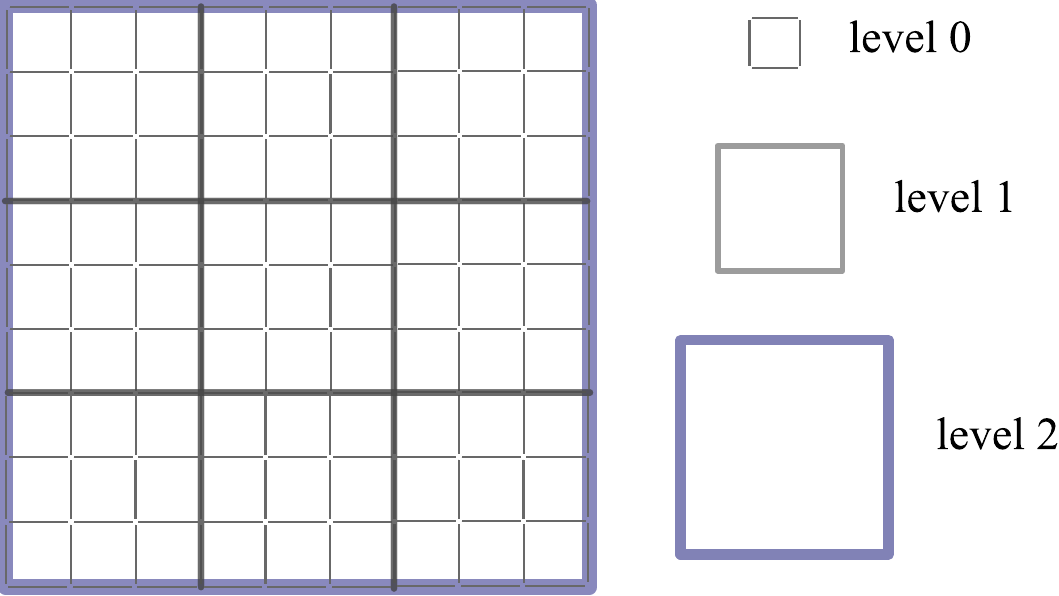}
 \end{equation*}
A suitable dynamics for coarse-graining in the toric code is one where errors continuously evolve to the closest coarse-grained configuration at each scale, while performing error correction at lower scales, as illustrated below. At the smallest $3 \times 3$ scale, we will have a generalized $\cR_0$ gadget whose output must be coarse-grained at level 1, i.e.:
\begin{equation} \label{eq:R0-intro}
 \includegraphics[width=0.9 \textwidth]{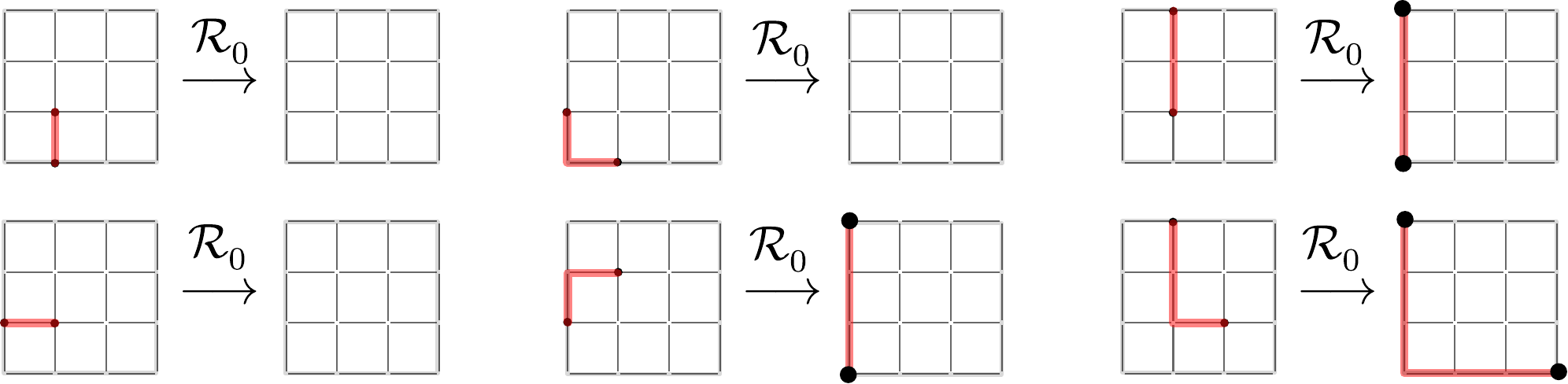}
 \end{equation}
where the black dots correspond to locations of vertex syndrome points (violation of vertex stabilizers of the toric code, of the $Z^{\otimes 4}$ type), and the edges highlighted red have been acted upon by $X$-type noise. Then, at all levels, we expect to observe the following behavior:
\begin{equation*} 
 \includegraphics[width=1 \textwidth]{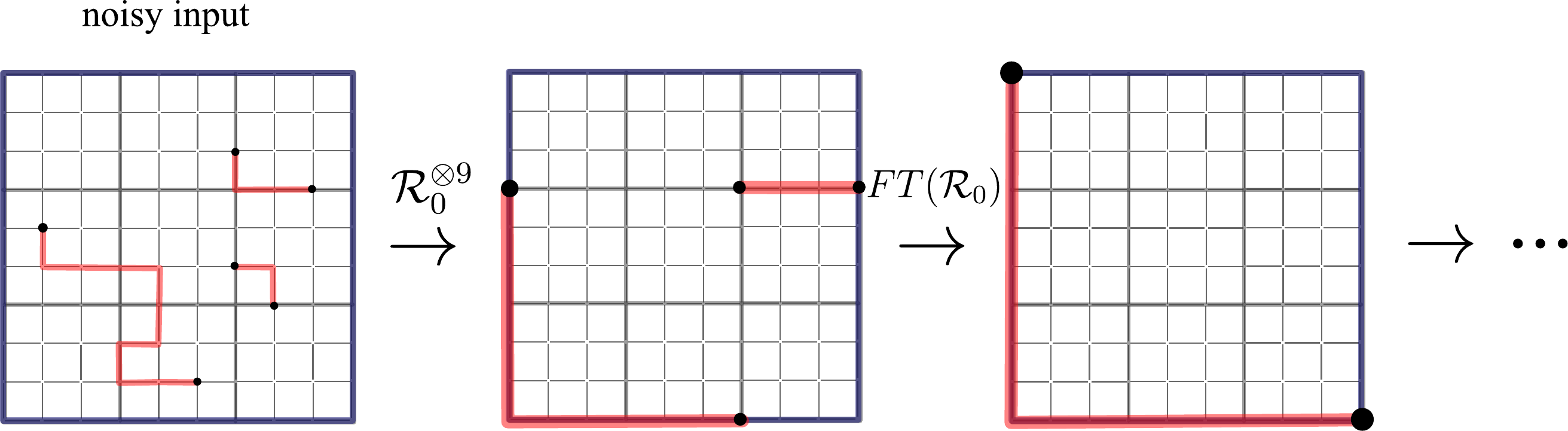}
 \end{equation*}
At the largest scale, with probability close to 1, an initial error pattern must eventually be cleaned up, assuming it is sampled from a $p$-bounded noise model for sufficiently small $p$.

Inspired by this intuition, we can now formulate a simple set of criteria according to which a candidate for a fault-tolerant toric code automaton should operate: 
\begin{enumerate}
    \item [(1)] We require an ``error correction'' gadget $\cR_0$ defined on a $3\times 3$ patch of cells and built from elementary gadgets that involve local measurement and feedback. This gadget has to be able to clean up small enough errors and otherwise output a configuration coarse-grained at level 1. 
    \item [(2)] To simulate this operation at larger scales, we must design fault-tolerant gadgets $FT(\cdot)$ for each elementary gadget, defined as circuits built out of elementary operations. Their action should be a scaled-up version of that of the respective elementary gadget. To control error spread, we also require them to satisfy additional properties.
    \item[(3)] The automaton is obtained by iterated fault-tolerant simulation of a given circuit of elementary operations $FT^n(\cC)$. 
\end{enumerate}
These are quite generic and are also required of any concatenated circuit-based construction (such as Tsirelson's automaton). 
The main elements of our construction (which satisfy (1-3) above and exhibit the desired coarse-graining dynamics) are below:
\begin{itemize}
    \item $\cI_0$: elementary identity gadget, acting on a given qubit.
    \item $\cT^{h/v}_0$:  horizontal/vertical syndrome ``splitting''. On a two-dimensional square lattice, we now have two versions of some types of gates for two possible orientations. If a nontrivial syndrome is present at vertex $v$, this gadget applies a Pauli-$X$ operator to two qubits horizontally/vertically adjacent to the vertex:
    \begin{equation*} 
    \includegraphics[width=0.4 \textwidth]{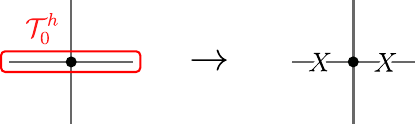}
    \end{equation*}
    \item $\cM^{h/v}_0$: horizontal/vertical ``matching'' operation, which applies $X$ to an edge if both endpoints contain nontrivial syndrome (thereby merging the two syndrome points on a horizontal/vertical link):
    \begin{equation*} 
    \includegraphics[width=0.4 \textwidth]{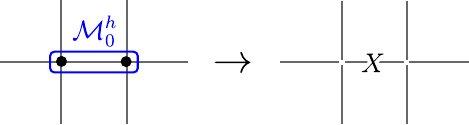}
    \end{equation*}
    \item $\cR_0$: error correction (``recovery'') gadget. We first define auxiliary vertical/horizontal gadgets $R_0^{v/h}$ pictorially as: 
     \begin{equation}
     \includegraphics[width = 0.75 \textwidth]{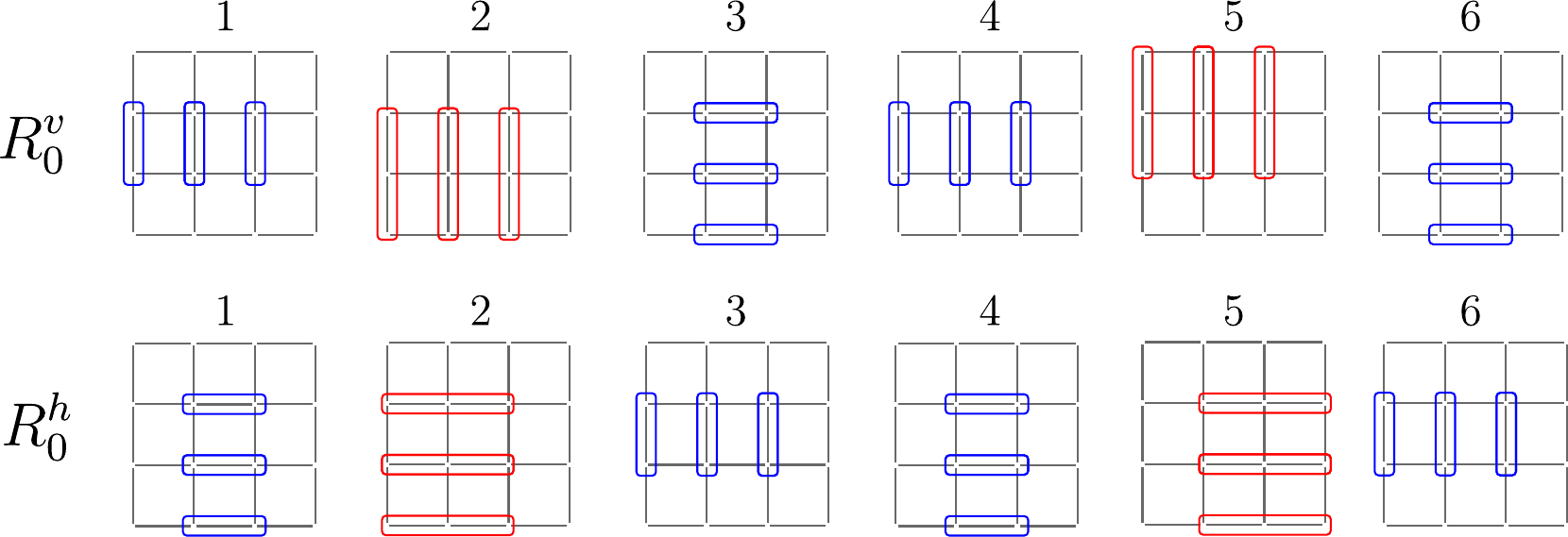}
     \end{equation}
     In the pictures above, the sequence of gates in $R_0^{v/h}$ acts in six steps displayed from left to right.  In terms of these gadgets, $R_0$ can be expressed as the following 48-step operation
    \begin{equation} 
        R_0 = (R_0^h \circ R_0^h \circ R_0^v \circ R_0^v)^2.
    \end{equation}
    This gadget has the action that we depicted in Eq.~\ref{eq:R0-intro}.
 \item We also construct the $FT(\cdot)$ version of each elementary gadget. This is done in a very similar way to the one-dimensional case. We do not summarize this here but direct the interested reader to Sec.~\ref{sec:toriccode}. 
 \item Our choice of gadgets has the following property: syndromes within and on boundaries of gadgets do not have any effect on syndromes outside the gadget. This property is imposed because in 2D, neighboring gadgets can interact on a one-dimensional boundary, and we would like to minimize the spread of correlations. 
\end{itemize}
Thus, we can consider a circuit $\cC$ which is expressed through elementary operations $\cI_0$, $\cT_0^{h/v}$ and $\cM_0^{h/v}$, and the error correction gadget $\cR_0$. We define fault-tolerant simulation $FT^k(\cC)$ similarly to Def.~\ref{def:FT_original-intro}. The toric code automaton is then a repeated application of circuit $FT^n(\cI_0)$ $T$ times, where $n = \log_3 L$, and the circuit operates on a $L \times L$ torus.

We prove fault tolerance of our construction in Sec.~\ref{sec:FT_proof}.  We find that to achieve error suppression from a $p$-bounded error model to an $O(p^2)$-bounded error model, we need to perform a constant number $k$ levels of fault-tolerant simulation. Every time we perform $k$ levels of simulation, we say that we implement a single step of an \emph{outer} simulation. We call the simulation that is performed within these $k$ levels the \emph{inner} simulation. To make proofs easier, we slightly modify how we perform the inner simulation in comparison to the outer one, which we explain in detail in Sec.~\ref{sec:toriccode}. However, the main idea explained above does not change.

For the proof of fault tolerance, we introduce a modification of the exRecs method, which is necessary due to the fact that neighboring gadgets have overlapping support and that the toric code cannot be viewed as a concatenated code\footnote{Colloquially speaking, our approach can be viewed as a hybrid between the exRecs method and G\'acs's sparsity-based method~\cite{gacs1983reliable,gacs_slides}.}.  
We then introduce the notion of gadget nilpotence (see Def.~\ref{def:tsirelson_nilpotence}), which determines how many levels $\ell$ of simulation one must go up for a single elementary gadget failure to not affect the operation of the simulated gadget.
With the help of an exhaustive computer-aided search (the full code is available at~\cite{code}), we find that 
gadget nilpotence must hold at $\ell = 6$, which is likely a loose upper bound. 
We then show that the error model changes from $p$-bounded noise to $O(p^2)$-bounded noise every time we go up some constant number $k$ levels up in the hierarchy, i.e. every time we perform $k$ iterations of the (inner) fault-tolerant simulation procedure, where the number $k$ is a constant related to the nilpotence level. For the toric code automaton, a generous analytical estimate is $k \leq 34$. 
However, in reality, a much smaller number of levels would suffice (our tentative numerical results suggest that for an i.i.d. Pauli $X$ and $Z$ qubit error model, one could have $k=2$, and accounting for measurement errors might lead to $k = 3$). The construction and proof in the paper are intended to be a proof of principle result, and optimizing the proof (and the performance of the automaton) is an interesting future direction. 

The main theorem that we prove can be summarized as follows:
\begin{theorem}[Informal, fault tolerance of toric code automaton]
Construct the iterated fault tolerant simulation of it  $FT^{\log_3 L }(\mathcal{\cI_0})$.  Assuming a $p$-bounded gadget error model,  there exist positive constants $A, \gamma$ such that the logical error rate is $(Ap^\gamma)^{L^{\left \lfloor(\log_3 2)/k\right \rfloor}}$-bounded where $k \leq 34$. 
 \end{theorem}
\noindent From this suppression, we then show that the logical failure probability remains suppressed on the order of $\exp(-\alpha L^\beta)$ for some positive constants $\alpha, \beta$ for an exponential amount of time in $L^\beta$ so long as the noise rate is below $O(A^{-1/\gamma})$. We remark that during this time, the state evolved under our automaton has the following unusual property: if it is fault-tolerantly terminated by destructive on-site Pauli measurements, the logical state can be correctly recovered from these measurements \emph{without} needing the spacetime history of past measurements.  This is strikingly different from logical state readout with the usual methods (which, in addition, would require storing a spacetime measurement record for time proportional to the distance of the code). Thus, the logical state remains \emph{efficiently recoverable} even after an exponentially long time. In addition, we show that a toric code logical state can be fault-tolerantly initialized in $\mathrm{poly}(L)$ time.

Numerical simulations of the noisy operation of our construction are shown in Sec.~\ref{sec:num-sim-and-conj}; these results are tentative because the spacetime volume of the circuits grows exponentially with the level of the hierarchy, and we did not optimize our construction for performance. Nevertheless, already after few iterations of the simulation, one can see the suppression of the effective logical error rate as one increases the hierarchy level.   

 \subsubsection*{Time-translation invariant construction in 3D}

Finally, in Sec.~\ref{sec:3D} we make our construction \emph{time-translation invariant} by going to three spatial dimensions.
% , which is done in . 
For this, we consider $2T$ layers, each hosting a different toric code state, stacked in three dimensions with periodic boundary conditions. The even-numbered layers (with index $2t$) are equipped with idle gadgets that act at each timestep, while the odd-numbered layers (with index $2t+1$) apply the gadget action of a 2D toric code automaton corresponding to timestep $t$. After applying either the gadget or an idle, we then transversally swap the pair of layers $2t$ and $2t+1$ (and transversally swap $2t+1$ and $2t+2$ during the next timestep). This is shown in Fig.~\ref{fig:3dfig}. The result is that the $O(L)$  toric code states move through the third dimension while effectively experiencing the action of the 2D toric code automaton, but with a time-independent operation applied at each site. The reason this construction is fault-tolerant is that the additional SWAP gates are transversal, and a failure of such a gate, followed by measurements of stabilizers of each toric code, can be absorbed into a $p$-bounded gadget error model.  This construction is capable of encoding $O(T)$ logical qubits, which matches the encoding rate of the 2D toric code.

\subsection{Comparison with earlier literature} \label{sec:comparison}

We now provide a brief comparison between the decoder constructed in this work and the earlier literature. To this end, we first summarize below the most relevant aspects of our construction:  
\begin{enumerate}
	\item The results of stabilizer measurements are used immediately and locally to perform feedback, and are then discarded, i.e. measurement results are communicated only an $O(1)$ distance in both space and time. 
    \item In 2D, our decoder is ``hardwired''. This means that each site of the automaton has access to a reliable read-only classical memory (``instruction set'') on each site with $O(\log L)$ bits.
    In 3D, no such hardwiring is required. 
    \item In all dimensions, our decoders are assumed to be run in discrete time, with all sites of the automaton updating in unison at each time step. 
\end{enumerate}
For the 2D toric code, we emphasize that the $O(\log L)$ bits per site needed to store the instruction set for the gates can be computed in advance once and for all: no new computation needs to be done during the working of automaton, which after initialization proceeds in an entirely parallelized way (as any truly scalable decoding scheme must). \hi{Our first example is what is achievable using concatenated circuits:}

\begin{itemize}
    \item \hi{The approach of concatenated codes by Aharonov \emph{et al}.}~\cite{aharonov1996faulttolerantquantumcomputation,Aharonov1999Jun,Gottesman_2000} \hi{implies the existence a ``hardwired'' fault-tolerant quantum computer already in 1D. By using the trick explained at the end of the previous subsection, this gives a time- (but not space-) translation-invariant scheme in 2D.} 
\end{itemize}

To the best of our knowledge, all existing approaches in the literature \hi{that use topological codes} require either classical computation that uses an additional $O(\log L)$ bits per qubit, or makes use of classical feedback acting non-locally in time over a time scale at least as long as $O(\log L)$.\footnote{The ubiquity of $O(\log L)$ factors in these constructions usually stems from the fact that in a $p$-bounded noise model, the largest size errors one needs to correct in order to ensure a threshold are of size $O(\log L)$.}  We now briefly review earlier work on decoders for topological memories, highlighting the most relevant proposals, and focusing on codes in three or fewer spatial dimensions. 

\begin{itemize}
    \item In Ref.~\cite{kubica2019cellular}, a local cellular automaton decoder was proposed for the 3D toric code only for the membrane (flux) sector of the code and not the charge sector. Thus, this example does not constitute a decoder for a full quantum memory.
\end{itemize}

\noindent In the 2D toric code (or in the charge sector of the 3D toric code), the standard algorithm used for decoding is minimum weight perfect matching (MWPM). This achieves a threshold quite close to optimal~\cite{Dennis_2002}, but requires nonlocal information about all syndromes as input, and has a worst-case time complexity scaling polynomially in $L$, since it is not fully parallelizable.  Decoding schemes that attempt to parallelize MWPM~\cite{skoric2023parallel} are more local from a practical standpoint but still require classical communication that leads to a $\text{poly}(L)$ slowdown.

Several other constructions combine renormalization group ideas and matching~\cite{Harrington2004,duclos2010fast,bravyi2011analytic}, and are more relevant for the purposes of comparison with our construction:

\begin{itemize}
    \item Of particular note is the decoder introduced by Harrington in Ref.~\cite{Harrington2004}, which was also analyzed in detail in Ref.~\cite{breuckmann2016local}.  This is a renormalization group-based decoder based on Gray's simplified model of the G\'acs automaton~\cite{Gray2001Apr} which operates assuming a constant speed of classical communication, but requires an unbounded density of classical bits per qubit and assumes \emph{noiseless} classical computation. In particular the construction utilizes $O(\log L)$ classical bits per qubit that are used to make decisions and employ feedback over a time window of $O(\log L)$\footnote{The fact that Harrington's decoder requires access to a read/write memory---as opposed to our read-only memory---also implies that unlike our model, Harrington's does not appear to admit a straightforward time-translation-invariant extension in three dimensions (for reasons discussed in Sec.~\ref{sec:3D}). }. Harrington's decoder is also no longer fault tolerant if the values stored in these classical bits experience errors. It was speculated that some of these issues could be overcome by employing more sophisticated ideas from G\'acs's construction, but no explicit proposal was put forward. 
    Furthermore, the argument in Ref.~\cite{Harrington2004} does not appear to rigorously prove fault tolerance in the presence of transient faults (a proof of a simpler version which makes the additional assumption of instantaneous classical communication was given in Ref.~\cite{Dauphinais_2017}). Numerical studies appear to indicate a threshold~\cite{Harrington2004,breuckmann2016local} assuming the classical part of the decoder is noiseless.
    
    \item Other decoders based on renormalization group ideas, such as those of Refs.~\cite{duclos2010fast,bravyi2011analytic}, operate non-locally in space on a scale of size $O(\log L)$. 

    \item  A rather different line of approach is found in ``field-based decoders''. These decoders operate according to a dynamics that uses a local field updated according to a Poisson equation which, when the field is updated quickly enough, causes long-range attractive interactions between syndromes, thus preventing them from moving far apart and causing a logical error. This tactic was first suggested in Ref.~\cite{Dennis_2002} and was studied numerically in \cite{herold2015cellular}. In this construction, a 2D (or 3D) cellular automaton is used to induce the attractive interaction by mimicking a discretized ``gravitational'' field. In order for the interaction to be long-ranged enough to produce sufficient suppression of logical errors, the processing speed of this automaton was shown to necessarily diverge polylogarithmically in $L$. 
\end{itemize}

We note that we do not know how to modify any of these constructions to be fault tolerant, local, and \emph{time-translation invariant} even if we allow to go up by one dimension.  This feature further distinguishes our construction from the existing ones.

\subsection{Outlook and discussion} \label{sec:outlook}

Our work opens up many new avenues for local decoders in quantum codes and for studying the fundamental interplay between fault tolerance and locality in quantum systems.  Below, we highlight several classes of interesting questions.

\begin{enumerate}
\item [(1)] \emph{Further generalizations of the toric code automaton.}
\end{enumerate}
A simple starting point would be to generalize our construction to the surface code on manifolds with different cellulations (for example, the triangular lattice) and to open boundary conditions. 

It also appears that there can be generalizations to subsystem topological codes~\cite{bravyi2013subsystemsurfacecodesthreequbit} as well as Floquet codes~\cite{hastings2021dynamically}, which would be interesting as the local decoder would not be simply measuring commuting checks of a stabilizer code; likely the ideas in Ref.~\cite{bauer2024topological} could help construct these generalizations. 
Another generalization would be to other Abelian topological codes, and more interestingly, to non-Abelian topological codes.  In the latter case, a good starting point would be non-Abelian topological models based on solvable groups, whereby repeatedly fusing a collection of anyons allows one to eventually fuse to the vacuum. A local decoder (even a hardwired one) for such a non-Abelian topological code would have the same advantage over the Harrington decoder-inspired construction in~\cite{Dauphinais_2017} as our construction has to Harrington's decoder.

An intriguing possibility would be a local decoder (perhaps even a translation invariant one) for Haah's cubic code~\cite{Haah_2011}, where logical operators have a fractal-like structure.   Additionally, it might be possible to design a simpler rule for the 2D toric code on a hyperbolic manifold (a very simple rule for the 4D hyperbolic toric code was presented in Ref.~\cite{hastings2013decoding}).

\begin{enumerate}
\item  [(2)] \emph{An automaton for practical local error correction.} 
\end{enumerate} 
It remains to estimate the threshold of our automaton under various types of realistic error models. Under the i.i.d. on-site Pauli noise of strength $p$, the threshold noise strength $p_c$ appears to be rather small, with the numerics of Sec.~\ref{ss:tc_numerics} suggesting\footnote{The fact that our automaton is naturally defined only on systems of linear sizes that are powers of 3 makes it difficult to pinpoint $p_c$ using only the modest system sizes considered in Sec.~\ref{ss:tc_numerics}.} $p_c \lesssim 10^{-4}$. 

Our construction was designed for the purpose of proving fault tolerance, not necessarily to optimize $p_c$. Thus, in the future it will be important to explore how large one can make $p_c$.  In addition, there is a big space of possible schemes that should have similar properties to ours but could have much better performance, especially at smaller scales. For example, one could construct error correction procedures on a slightly larger unit cell which could provide larger suppression of the error rate and a simulation scheme that achieves error suppression for every level of the hierarchy (rather than every $k$ levels). In addition, it is of practical benefit to find constructions that work on system sizes that are not powers of a given integer as well as to adapt the construction to different lattices and architectures. 

\begin{enumerate}
\item  [(3)] \emph{A quantum version of G\'acs's automaton.} 
\end{enumerate}
The construction of a local quantum automaton generalizing G\'acs's classical automaton is an open problem. Any fault-tolerant local quantum automaton in three dimensions or less that possesses even some of the following features on top of our construction would constitute an important advance toward solving this problem: (i) translation invariance in both space and time; (ii) constant encoding rate;  (iii) operating in continuous time; (iv) a local model for universal quantum computation. Each of these points is an important open challenge.

We are optimistic about (ii) being readily achievable if spacetime translation invariance is allowed to be violated.  Some progress towards (ii) and (iv) was made in Ref.~\cite{yamasaki2024time} using code concatenation (a local and constant rate reliable classical memory is required for the operation of this proposal but one was not constructed).

Our construction might possibly be extended to realize (i) and (iv) in the future. The spacetime structure of the rules of our automaton resembles self-similar tiling which could be possible to generate in a local and translation-invariant way~\cite{tsirelsonunpublished}. However, the bits needed to generate this structure must also be reliable, which poses a significant challenge.

To make progress towards (iv), a universal local quantum computer, one can use the 3D version of our automaton, which can be modified to allow for the implementation of constant-depth local unitary gates. Achieving the real part of Clifford gates this way is relatively simple; for complex Clifford gates, one would need to modify our construction to work in the presence of boundaries. For a non-Clifford gate, which would complete a universal set of gates, the most natural approach would be to use a two-dimensional spacetime protocol~\cite{Brown2020universal,upcomingDavydova} implementing a logical non-Clifford gate on copies of toric codes. It may be possible to achieve local error correction in these protocols by constructing generalizations to $D_4$ topological order in the presence of domain walls and boundaries. This would be an interesting future direction.

As for (iii), i.e. a continuous time version of our construction, it could be possible to use Toom's rule (which itself does not rely on synchronicity) to synchronize neighboring clocks in 2D and thus fight possible errors that arise due to asynchrony~\cite{10.1109/SFCS.1991.185379}.  In addition, some interesting attempts that were not based on hierarchical methods were presented in 3D in Ref.~\cite{cook2008self} which used some older techniques in Ref.~\cite{berman1988investigations}.  

\begin{enumerate}
\item  [(4)] \emph{Constraints on local and scalable fault tolerance.}  
\end{enumerate}
One interesting question would be whether there are fundamental bounds we can place on quantum systems that are locally fault tolerant.  For example, the quantum extension of G\'acs construction would yield a system undergoing local dynamics
which encodes a constant rate of quantum information fault tolerantly. This would get around the Bravyi-Poulin-Terhal (BPT) bound~\cite{Bravyi_2010} for local stabilizer models by virtue of the stabilizers in such a construction {\it not} being local. The construction in Ref.~\cite{yamasaki2024time}, if implemented locally, would overcome this bound by the same argument. Thus, we know that spatial locality likely does not place any fundamental limit on the encoding rate. However, this still remains to be rigorously shown. 

It is also an interesting question whether all low-dimensional and local schemes of fault tolerance necessarily require a hierarchical construction.  One early proposal for a non-hierarchical classical one-dimensional memory was the G\'acs-Kurdyumov-Levin automaton~\cite{gacs1978one}, but this was shown to lack a threshold in the presence of generic noise~\cite{park1997ergodicity}.  Finally, if we want to develop a physical model for local fault tolerance, we would want such a system to be governed by some kind of universal Linbladian which robustly encodes an input program and fault tolerantly performs a computation.  Understanding how to move towards such a result would (we believe) constitute an important advance in the theory of quantum fault tolerance.

\begin{enumerate}
\item  [(5)] \emph{Consequences for non-equilibrium physics.} 
\end{enumerate}
Our results also have interesting implications for quantum many-body physics, where geometric locality is a natural physical constraint.  Our construction provides an example of a completely stable open quantum system below four dimensions, and we are currently unaware of any other examples that are stable to {\it arbitrary} (but weak) noise.
It will be interesting to see how this example fits into a more general classification of stable phases in open quantum systems, the definition and understanding of which is still in its infancy (see e.g.~\cite{rakovszky2023defining}). 

An important comment is that the dynamics implemented by our circuit is fundamentally non-equilibrium in character, in that it does not satisfy detailed balance, and does not have a thermal steady state. This is an example of how (both classical and quantum) systems that break detailed balance can exhibit ordered phases in lower dimensions than systems undergoing thermalizing dynamics, which obey detailed balance.  Classically, we are aware of an example obeying detailed balance which robustly stabilizes a single bit of information in 3D (the classical Ising gauge theory, see Ref.~\cite{PhysRevB.99.094103}) which can be decoded using ideas in Ref.~\cite{serna2024worldsheet}.  It would be interesting to prove a no-go theorem in lower dimensions or to understand how exactly detailed balance must be broken to lead to ordered phases.  In the quantum case, it would interesting if one could construct a self-correcting quantum memory below four dimensions, either by constructing a 3D topological stabilizer code, or by studying random Hamiltonians with non-commuting terms.

From the perspective of non-equilibrium statistical mechanics, one interesting question is the nature of the phase transition that occurs in Tsirelson's automaton when the noise strength is tuned across the threshold value $p_c$. Since there is no ordering in 1D (locally-interacting) models of equilibrium statistical mechanics, such a transition in the non-equilibrium setting is rather exotic.  With symmetric noise, this transition is associated with the spontaneous breaking of the $\mathbb{Z}_2$ spin-flip symmetry in the system's non-equilibrium steady state. Numerics finds this transition to be continuous, likely corresponding to a new critical point that will be explored in future work.

%% file: 2.tsirelson.tex
\section{Tsirelson's automaton and proof of its fault tolerance}\label{sec:tsirelson}

%\cmt{When have time: change $n$ to $N$ when it denotes the full system size, i.e. $L = 3^N$.}

In this section, we review the construction of a fault-tolerant one-dimensional automaton by Tsirelson~\cite{Cirel'son2006Aug} that will inspire our toric code decoder. We then modify Tsirelson's original construction slightly, which allows us to provide a fault tolerance proof using the extended rectangles method introduced by Aliferis, Gottesman, and Preskill~\cite{Aliferis2005Apr, gottesman2009introductionquantumerrorcorrection}. This proof is much simpler than the original approach taken by Tsirelson and provides a framework for the rest of the paper. The local automaton that Tsirelson proposed operates on a one-dimensional line of interacting bits and can reliably store a single logical bit for an exponential amount of time in a power $\alpha > 0$ of the system size. Though the construction is rather simple, it lacks translation-invariance in both space and time, and operates in discrete time.

We now introduce a number of basic definitions which will be used in the remainder of this work. 
In this section, we will let $\cL$ denote a finite 1D lattice of bits, with one bit living at each site of $\cL$. 
\begin{definition}[State]
A classical state is a configuration of bit values $s \in \{0,1\}^{|\cL|}$, where $|\cL|$ is the total number of bits.   We will write $s(t) = (s_1 (t),s_2(t),...,s_{|\cL|}(t) )$, where each $s_i(t) \in \{0,1\}$.  A $t$-history of classical states $H_{t,s(0)}$ is a set of states $\{s(0), s(1), \cdots, s(t)\}$ up to time $t$.
\end{definition}

\begin{definition}[Classical noise]\label{def:classnoise}
A classical noise model is given by the conditional PDF $p(\cdot, t | H_{t-1, s(0)})$, which is a distribution over spatial locations $\bm{x}$ at fixed time $t$ conditioned on the $t$-history of states $H_{t-1,s(0)} = (s(0), s(1), s(2), \cdots, s(t-1))$.  The noise model operates by \hi{assuming that the state $s(t)$ is noisy at locations $\bm{x}_1, \bm{x}_2, \cdots , \bm{x}_n$ with probability distribution} $p(\bm{x}_1, \bm{x}_2, \cdots , \bm{x}_n, t | H_{t-1, s(0)})$.

An i.i.d. noise model is one where $p(\cdot, t | H_{t-1,s(0)}) = p(\cdot, t)$ and \hi{$p(\bm{x}_1, \bm{x}_2, \cdots ,\bm{x}_n, t) =\prod_{i=1}^n q_i$, where $q_i$ is the probability of flipping the bit at $\bm{x}_i$}.
\end{definition}

Throughout the paper, we will refer to the set of spacetime points where the noise occurs as $E$, and will say that $E$ constitutes a particular ``noise realization''.

We will use the word ``gadget'' to refer to the operation of some rule with noise. 
\begin{definition}[Gadget]
A classical gadget over a spatial region $\cR \subset \cL$  is a
map $G : \{0,1\}^{|\cR|}\times E \to \{0,1\}^{|\cR|}$,  which locally rewrites the values of bits in region $\cR$ according to a predetermined rule, subject to noise flipping bit values in $E$. The region $\cR$ will be referred to as the spatial support of gadget $G$.
\end{definition}

\begin{definition}[Composition of gadgets] \label{def:composition}
    Gadgets can be composed one after another in a sequence.  The gadget $G_2 \circ G_1: \{0,1\}^{|\cR|}\times E \to \{0,1\}^{|\cR|}$ is a map, where $\cR$ is the union of the spatial supports of $G_1$ and $G_2$, such that $G_1$ is applied to the input, and is then followed by $G_2$.  Note that the ordering of a sequence of gadgets in all expressions will always be read from right to left (in figures, however, time will run either bottom to top or left to right).  
    
    If two gadgets $G_1$ and $G_2$ occur at the same time and their actions on an input can be parallelized, then we will write their combined action as $G_1 \otimes G_2$.
\end{definition}

\begin{definition} [Automaton]
    A classical automaton $\mathcal{A}_C$ is a circuit of (classical) gadgets, each having spatial support of size $O(1)$, where the union of spatial supports of the gadgets cover the lattice $\cL$.  The gadgets are generally assumed to be noisy.  As a depth-$t$ circuit of gadgets, a classical automaton (subject to a particular noise realization) maps an input classical configuration $s(0)$ to an output configuration $s(t)$.

\end{definition}

\begin{definition} [Damage]
    Consider a classical automaton $\cA_C$ of depth $t$ with noise realization $E$ and another automaton $\widetilde \cA_C$ that is a noiseless version of $\cA_C$. For a given input $s(0)$ and time $0 < t_0 \leq t$, let the states obtained by running these automata up to time $t_0$ be $s(t_0)$ and $\widetilde s(t_0)$, respectively. The damage at time $t_0$ with respect to noise realization $E$ is defined to be the set $\{\bm{x} | \ s(\bm{x}, t_0) \oplus \widetilde s(\bm x, t_0) = 1 \}$.
\end{definition}

We will refer to an automaton as \emph{hardwired} if it applies gadgets in a non time-translation invariant way:
\begin{definition}[Hardwiring]
We say that an automaton $\mathcal{A}$ is {\it hardwired} if the elementary gadgets $G \in \mathcal{A}$ of the automaton are applied in a time-dependent manner. The instruction set specifying this time dependence is always assumed to not be subject to noise. 
\end{definition}

 In this paper, we will work with $p$-bounded noise models (also called local stochastic noise models), which are more general than i.i.d. noise models: 
\begin{definition} [$p$-bounded classical noise] \label{def:p-bounded-classical}
    For an automaton of depth $t$, let a classical noise realization $E$ be a random subset of the set of spacetime coordinates $\cL\times [0,t]$ \hi{(since time is discrete, it is understood that $[0,t] \equiv [0,t] \cap \mathbb{Z}$)}. The noise model is said to be $p$-bounded if for any $A \subseteq \mathcal L$, we have
\be
\mathbb P [ A \subseteq E ] \leq p^{|A|}.
\ee
\end{definition}
In the paper, we will be using several $p$-bounded noise models. One of them will act on individual (qu)bits, in which case $\cL$ will be the spacetime lattice of (qu)bits, and the other one will be a gadget noise model, where each gadget can experience an arbitrary error, in which case $\cL$ will denote the set of spacetime coordinates of all gadgets. The use of $p$-bounded noise (rather than i.i.d. noise) is essential because the operation of the automata introduce noise correlations that cannot be captured with an i.i.d. noise model. 

\subsection{Tsirelson's automaton} \label{ss:tsirelsonOG}

A non-hardwired fault tolerant 1D classical automaton was found by G\'acs in the seminal work~\cite{gacs1983reliable, Gacs_2001}, which refuted the `positive rates' conjecture in statistical mechanics and dynamical systems~\cite{Liggett2005}. While G\'acs' automaton is extremely complicated, the task of fault tolerantly storing a single bit in 1D turns out to be significantly simplified if one allows for hardwiring. This was first noticed by Tsirelson \cite{Cirel'son2006Aug}, 
who constructed a hardwired automaton defined by a set of `substitution' rules. His construction is designed to be hierarchically self-simulating, in a way we will make precise below. 

\begin{definition}[Elementary gadgets]\label{def:tsirelsonelem}
The elementary gadgets of Tsirelson's automaton consist of an `idle' gadget $X_0$, a `swap' gadget $Y_0$, and an `error correction' gadget $Z_0$. In the absence of noise, their actions are defined as follows: 
\begin{equation}
     X_0(s_1) = s_1, \hspace{0.75cm}  Y_0(s_1, s_2) = (s_2, s_1)
\end{equation}
and
\begin{equation}
Z_0(s_1, s_2, s_3) = (\maj(s_1, s_2, s_3), \maj(s_1, s_2, s_3), \maj(s_1, s_2, s_3)),
\end{equation}
and are pictorially denoted below:
\begin{equation} \label{ts:bare}
 \includegraphics[scale=0.27]{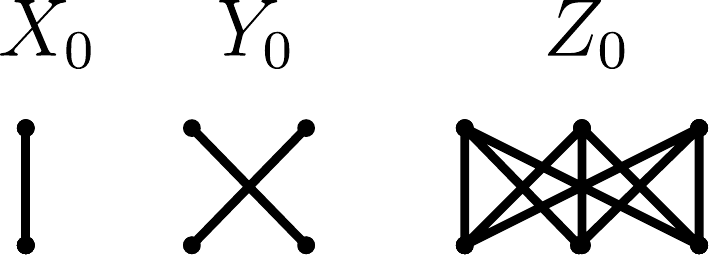}
 \end{equation}
\end{definition}
Thus $X_0$ acts as the identity, $Y_0$ swaps its inputs, and $Z_0$ performs a majority vote of its inputs and broadcasts three copies of the result. 

Tsirelson then defined \emph{higher-level} gadgets $(X_n, Y_n, Z_n)$, which are built from elementary gadgets and behave the same way as the elementary gadgets but at a larger `scale'.  The higher-level gadgets are defined recursively through the following substitution rules:
\begin{definition}[Simulated gadgets]\label{def:tsirelsonaut}
The simulated gadgets of Tsirelson's automaton at level $n$, denoted by $X_n$, $Y_n$, and $Z_n$, are defined through $X_{n-1}$, $Y_{n-1}$, and $Z_{n-1}$ via the substitution rules (where time progresses from bottom to top) 
\begin{equation}
\includegraphics[width = 0.65\textwidth]{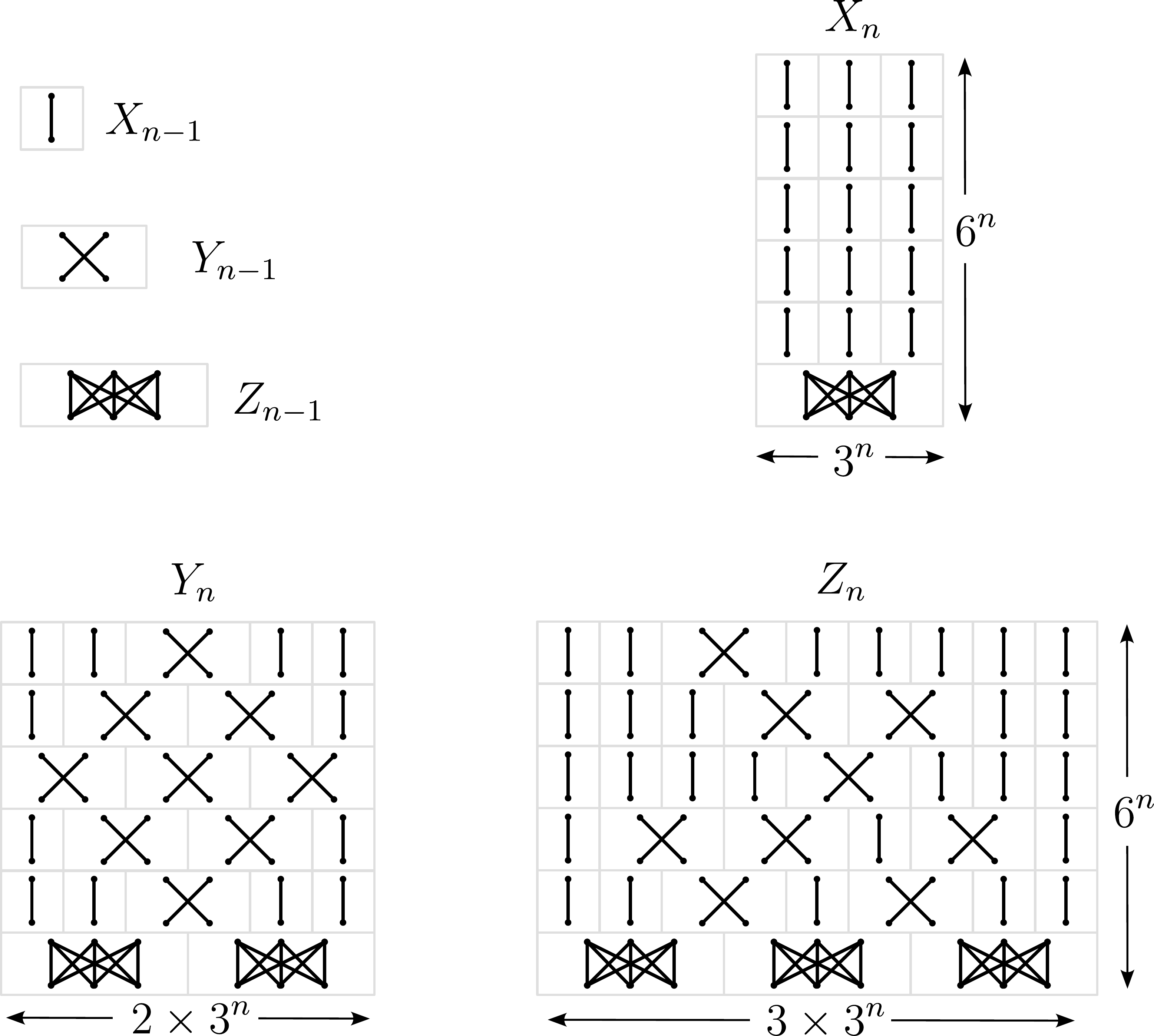}
\end{equation}
These substitution rules uniquely determine $(X_n, Y_n, Z_n)$ in terms of the elementary gadgets $(X_0,Y_0,Z_0)$.  We will call $n$ the \emph{level} of the respective gadget.
\end{definition}
The pattern of $Y_{n-1}$ swaps appearing in the substitution rule for $Z_n$ is designed to redistribute the bits so that the majority vote at higher scales can be implemented in a local way

 Finally, we can define Tsirelson's automaton:
\begin{definition}[Tsirelson's automaton]
Given a 1D line of $L$ bits, denoting $n = \log_3 L$ (or $n = \log_3 2 L$, $n = \log_3 L + 1$, respectively), Tsirelson's automaton for $X$, $Y$ or $Z$ gadget consists of applying the simulated gadget $X_n$, $T_n$ or $Z_n$ repeatedly in time.
\end{definition}
The depth of the circuit implementing Tsirelson's automaton when $X_n$ is repeated $T$ times is  $L^{\log 6/\log 3} \cdot T$. It is the same for the automaton corresponding to $Y_n$ or $Z_n$. The bit that Tsirelson's automaton stores is encoded using the recursive majority function (which is the function computed by a noiseless $Z_n$ gadget).

\begin{definition}[Decoder for Tsirelson's automaton]
The decoder for Tsirelson's automaton is a noiseless circuit $D_n : \{0,1\}^{3^n} \rightarrow \{0,1\}$ which performs recursive majority vote on bit strings of length $3^n$. Explicitly, $D_n$ acts on bit strings via 
\begin{equation} \label{eq:idealdecoderTS}
\begin{split}
    &D_1(s_1,s_2,s_3) = \maj(s_1,s_2,s_3) \\
    &D_k(s_1,...,s_{3^k}) = D_1(D_{k-1}(s_1,...,s_{3^{k-1}}),D_{k-1}(s_{3^{k-1}+1},...,s_{2\times3^{k-1}}),D_{k-1}(s_{2\times3^{k-1}+1}, ...,s_{3^k})) 
    \end{split}
\end{equation}
\end{definition}

Assuming a local $p$-bounded error model where the noise is supported on bits (this noise model \hi{is defined in Prop.}~\ref{prop:wire_error_model}), Tsirelson~\cite{Cirel'son2006Aug} proved a stronger version of the following statement:

\begin{figure}[!htbp]
    \centering
    \includegraphics[width=0.65\linewidth]{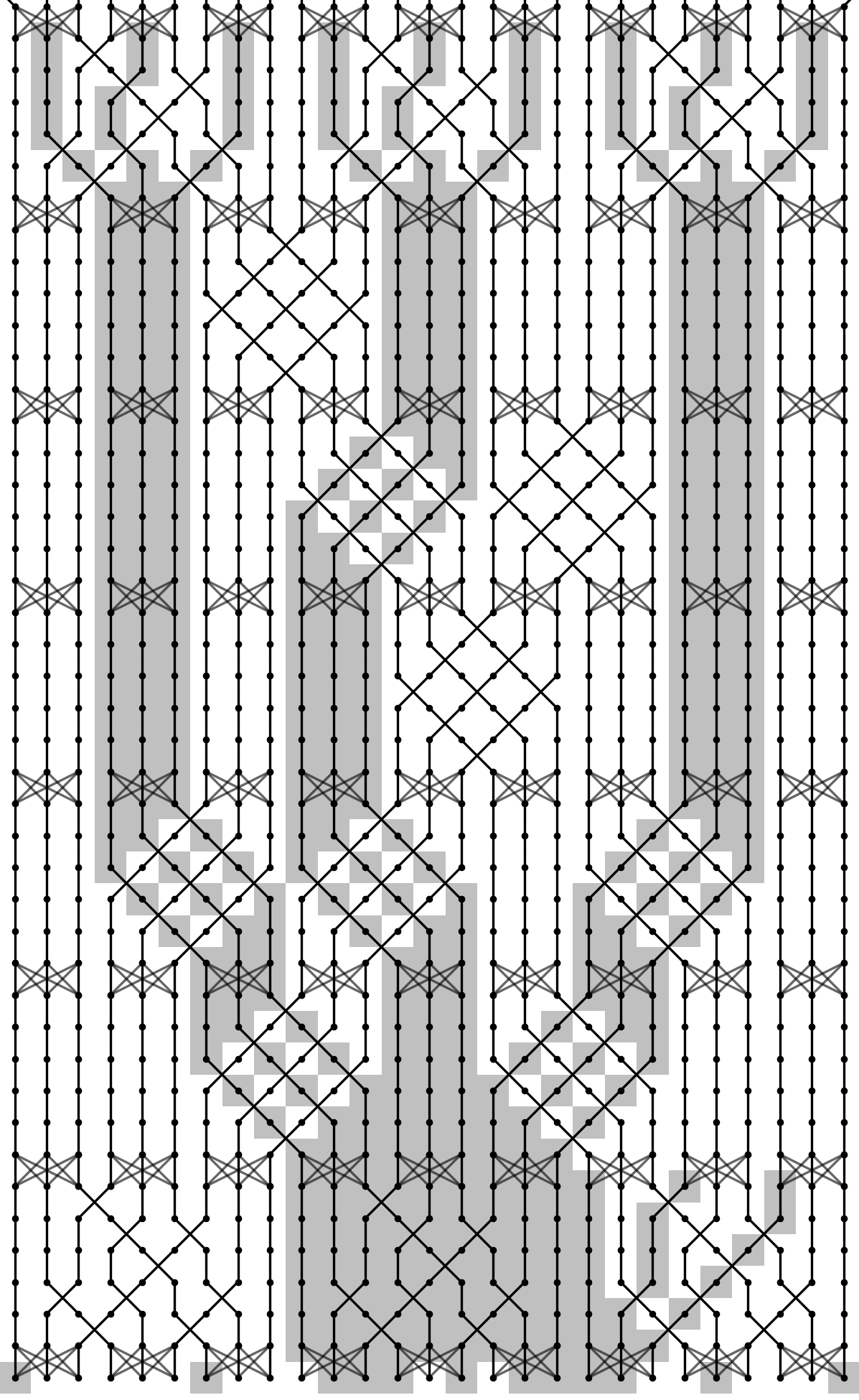}
    \caption{The evolution of an input state with randomly chosen damage under a noise-free level-2 error correction gadget $Z_2$. As in other figures, time goes from bottom to top. White squares represent $0$ and and gray squares represent $1$.}
    \label{fig:clean_operation}
\end{figure}

\begin{theorem}[Tsirelson~\cite{Cirel'son2006Aug}]\label{thm:FT_TS}
Denote $s_i(t)$ to be the value of the $i$th bit after $T$ repetitions of Tsirelson's automaton applied to the initial state $s_i(0) = 1 \, \forall \, i$. Then under a $p$-bounded noise model, provided $p$ is smaller than some value $p_c$  
\begin{equation}
\mathbb{P}\left(D(s(T)) = 1\right) \geq 1 - T \cdot \exp\left(-\alpha L^{\beta}\right)
\end{equation}
for positive constants $\alpha$ and $\beta$.
\end{theorem}
 
By symmetry, this also holds if the initial state is $s_i(0) = 0$. This implies that Tsirelson's automaton can remember a bit of information for infinite time in the limit $L\rightarrow \infty$.

This theorem was originally proved in Ref.~\cite{Cirel'son2006Aug} by explicitly studying how the joint distribution of errors evolves under the automaton.  It is likely that this method does not readily generalize for other constructions.
% , which is rather fine-tuned to the construction. 
Thus, of the goals of this section is to provide a simpler, more modern, and generalizable proof that makes use of the extended rectangles formalism, which was also originally used to provide a simpler proof of fault tolerance for the concatenated code scheme~\cite{Aliferis2005Apr}.  
%MD: note to self: review
Before proceeding with the proof, let us provide some informal intuition for why Tsirelson's construction works. 

First, we need a more general discussion of stable cellular automaton dynamics.  In order to stabilize a certain logical state, such an automaton must be able to correct domains of errors created by the noise.  For a generic configuration of noise, the largest spatial domain of errors observed at any given time can scale with the system size $L$ as $\sim \log L$, which diverges as $L\rightarrow\infty$. In two (or more) dimensions, there exist ways of locally correcting these large domains quickly enough. One way is to use Toom's rule~\cite{toom1980stable}, which shrinks domains in a time proportional to their linear size.  The reason this can be done locally is because boundaries of error domains form loops, and the curvature of these loops can be determined locally.  

However, in one dimension, an analogous rule is not possible given that boundaries of error domains are point-like (this is also the reason why there are no ordered phases in 1D classical locally-interacting systems in equilibrium).  Thus, Tsirelson's automaton must operate based on a different rule: rather than trying to shrink domains of errors, it instead operates using a ``divide and conquer'' strategy. 
The only error domains eliminated directly by the automaton are those of length 1 (which are cleaned up by local majority votes). 
Larger domains, instead of being directly shrunk, are instead fragmented into smaller domains which are then isolated from each other. These domains recursively continue to be fragmented and isolated until they become small enough to be directly corrected (see Fig.~\ref{fig:clean_operation}). 

Consider first how the gadgets operate at low levels. Every operation $X_1$ consists of a $Z_0$ followed by three columns of $X_0$. This means that $X_1$ takes three bits as an input, computes their majority vote (i.e. error corrects), and then broadcasts this value to all three outputs. $Y_1$ starts by applying a $Z_0$ to two blocks of three bits, and then swapping both blocks. This is equivalent to implementing a fault tolerant swap operation on a pair of encoded bits. Finally, $Z_1$ starts with applying $Z_0$ to three blocks of qubits and redistributing the output via swaps so that each output block contains one bit from each input block. As an explicit example, consider the action of a clean $Z_1$ on the input state `$000 \, 101 \,000$'. Applying majority votes ($Z_0$ gadgets) turns this into `$000\, 111 \, 000$', and the redistribution performed by the pattern of swaps then turns this into `$010 \, 010 \, 010$'. The action of $Z_1$ thus fragments the minority domain of errors into three pieces, which are distributed among each of the three output blocks. Following this with a final layer of local majority votes ($Z_0$ gadgets) completely eliminates the minority domain. 

The elimination of errors at larger scales occurs in a self-similar fashion, repeating the same action at each scale recursively (`hierarchically').  An illustration of how an input error configuration gets self-similarly cleaned up at level 2 is shown in Fig.~\ref{fig:clean_operation}.

\subsection{A modification of Tsirelson's automaton}\label{ss:modification_of_tsirelson}
We will now introduce a minor modification of Tsirelson's construction, which is readily amenable to the extended rectangles method~\cite{Aliferis2005Apr, gottesman2009introductionquantumerrorcorrection} for proving fault tolerance. This modification also allows us to define the notion of a \emph{fault tolerant simulation} and interpret the action of the automaton as repeated fault tolerant simulation of a 3-bit repetition code.

Each gadget in the modified automaton is analogous to a `gate' from quantum computation, with higher-level gadgets being defined in a way resembling transversal application of lower-level gadgets: 
\begin{definition}[Modified elementary and simulated gadgets]\label{def:tsirelsongadgets}
Define the modified elementary gadgets $\cX_0 = X_0$, $\cY_0 = Y_0$, and $\cZ_0 = Z_0$ to be the same as in Tsirelson's original construction.  Define the (modified) gadgets $\cX_1$, $\cY_1$, and $\cZ_1$ to be
\begin{equation} \label{ts:2}
\includegraphics[width = 0.55 \textwidth]{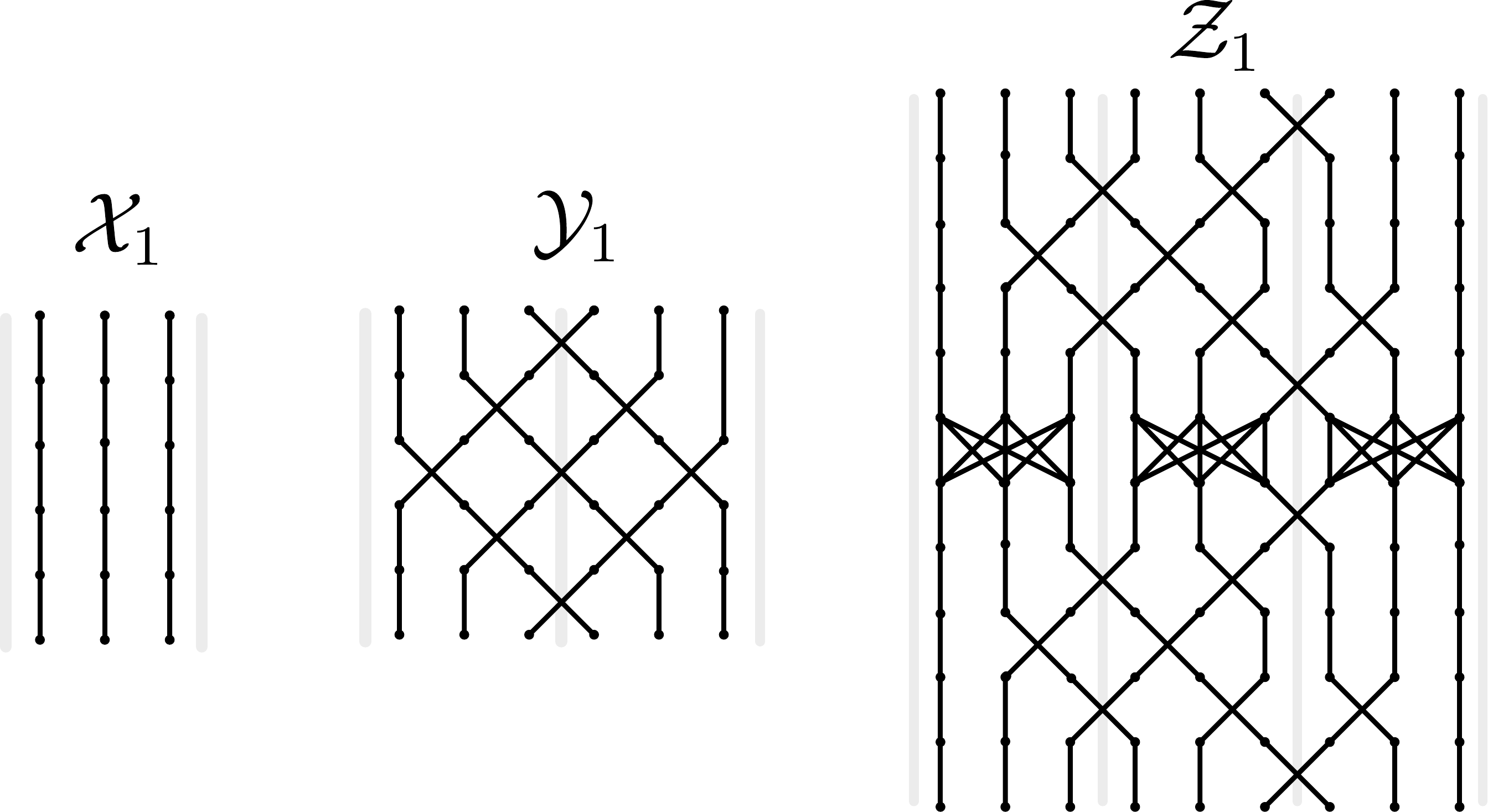}
\end{equation}
We call $\cX$, $\cY$, and $\cZ$-type gadgets `gate gadgets'. 
\end{definition}

In the new notation \hi{(note the difference in the typeface)}, $\cX_0$, $\cY_0$ and $\cZ_0$ are now viewed as 1-, 2-, and 3-bit gates, respectively. 
The new $\cX_1$ and $\cY_1$ are the same as in Tsirelson's original construction, except that they lack the first layer of $Z_0$ gadgets. 
$\cZ_1$ is the same as Tsirelson's gadget repeated twice, with the first layer of $Z_0$ gadgets removed.  The reason why we define $\cZ_1$ differently from Tsirelson's $Z_1$ is that, as previously mentioned, $\cZ_1$ is designed to operate analogously to transversal gates. Namely, the level-1 majority vote is implemented by performing majority votes between encoded inputs in a bitwise manner (i.e. the first majority vote acts on the first bit of each of the encoded inputs). To do this completely locally, we first use a sequence of swaps to bring the desired bits together, apply the majority vote, and send the bits back to their original locations. The transversal nature of this operation ensures that an error confined to a single block of 3 bits does not affect the logical output.

We note that the circuits for $\cX_1$ and $\cY_1$ gadgets have depth 6, while that for $\cZ_1$ has depth 11. Nevertheless, as we will see, an automaton compiled from these gadgets will have well-defined depth despite the gadgets having different depth, because $\cZ$-type gadgets never occur simultaneously with $\cX$ or $\cY$-type gadgets. If this were not the case, one can always add extra idle steps to equalize the depths. 

Nest, we define the notion of a fault tolerant simulation of a circuit $\cC$.
\begin{definition}[Fault tolerant simulation] \label{def:FT_original}
Define $\EC_1 = \cZ_0$. A fault tolerant simulation of a circuit $\cC$ comprised of $\cX_0$, $\cY_0$, and $\cZ_0$ gadgets, denoted $FT(\cC)$, is a circuit wherein (a) each bit in $\cC$ is replaced by a block of 3 bits, (b) each $\cX_0$ is replaced with $FT(\cX_0) = EC_1 \circ \cX_1 \circ EC_1$, (c) each $\cY_0$ is replaced with $FT(\cY_0) = EC_1 \circ \cY_1 \circ EC_1$, and (d) each $\cZ_0$ is replaced with $FT(\cZ_0) = EC_1 \circ \cZ_1 \circ EC_1$, and finally (e) one of the two $\EC_1$ gadgets is removed between each pair of consecutive gate gadgets.
\end{definition}

\begin{definition}[Repeated fault tolerant simulation] \label{def:FT_original_repeated}
A repeated fault tolerant simulation of a circuit $\cC$, denoted $FT^n(\cC)$, is a new circuit formed from the iteration $FT^n(\cC) = FT(FT^{n-1}(\cC))$, with $FT^0(\cC) = \cC$ and $FT^1(\cC) = FT(\cC)$.
\end{definition}
We emphasize that the adjective in the term ``fault tolerant simulation'' is used to indicate that the simulation's aim is produce a circuit with same logical action as the original one that can tolerate certain faults (for example, sparse single-qubit or single-gadget faults), but the fault tolerance in the conventional meaning of this term is only achieved upon repeated simulation and is a property that we subsequently prove.

We also introduce an element that we call the `ideal decoder'. The ideal decoder is a noise-free gadget which performs ideal error correction.  It will play an important role in the proofs.
\begin{definition}[ideal decoder]
An ideal decoder on $3^k$-bit state $s=(s_1,...,s_{3^k})$ is defined to be $D_k$ from Eq.~\eqref{eq:idealdecoderTS}. $D_1$ will be schematically shown as
\begin{equation} \label{tsdec:}
\includegraphics[width = 0.06 \textwidth]{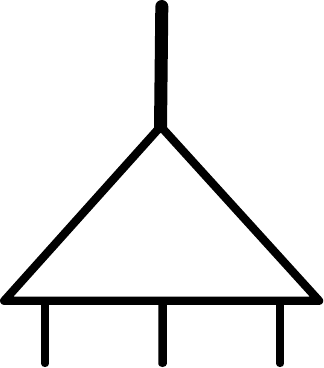}
\end{equation}
\end{definition}

For a noiseless input state where each three bits is viewed as a 3-bit encoded noiseless input, the following graphical identity holds for the ideal decoder acting after noise-free simulated gadgets, which illustrates the self-similar nature of the construction: 
\begin{equation} \label{ts:X1}
 \includegraphics[width = 0.75 \textwidth]{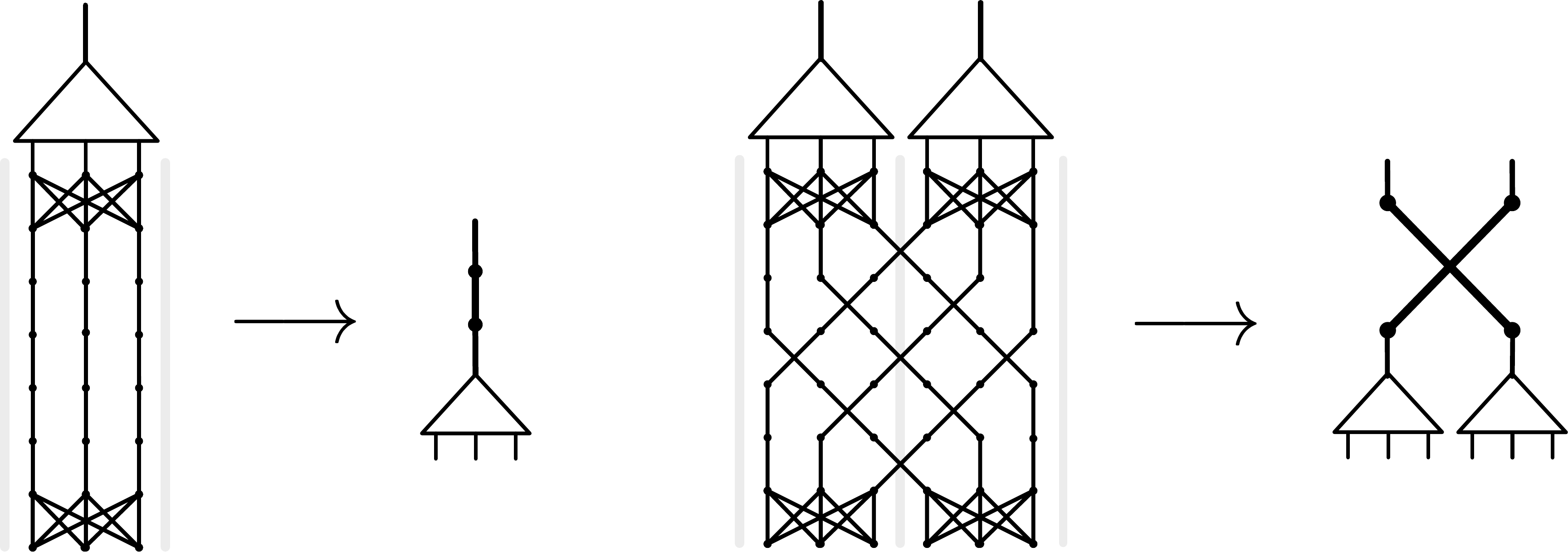}
 \end{equation}
 \begin{equation} \label{ts:Z1}
 \includegraphics[width = 0.45 \textwidth]{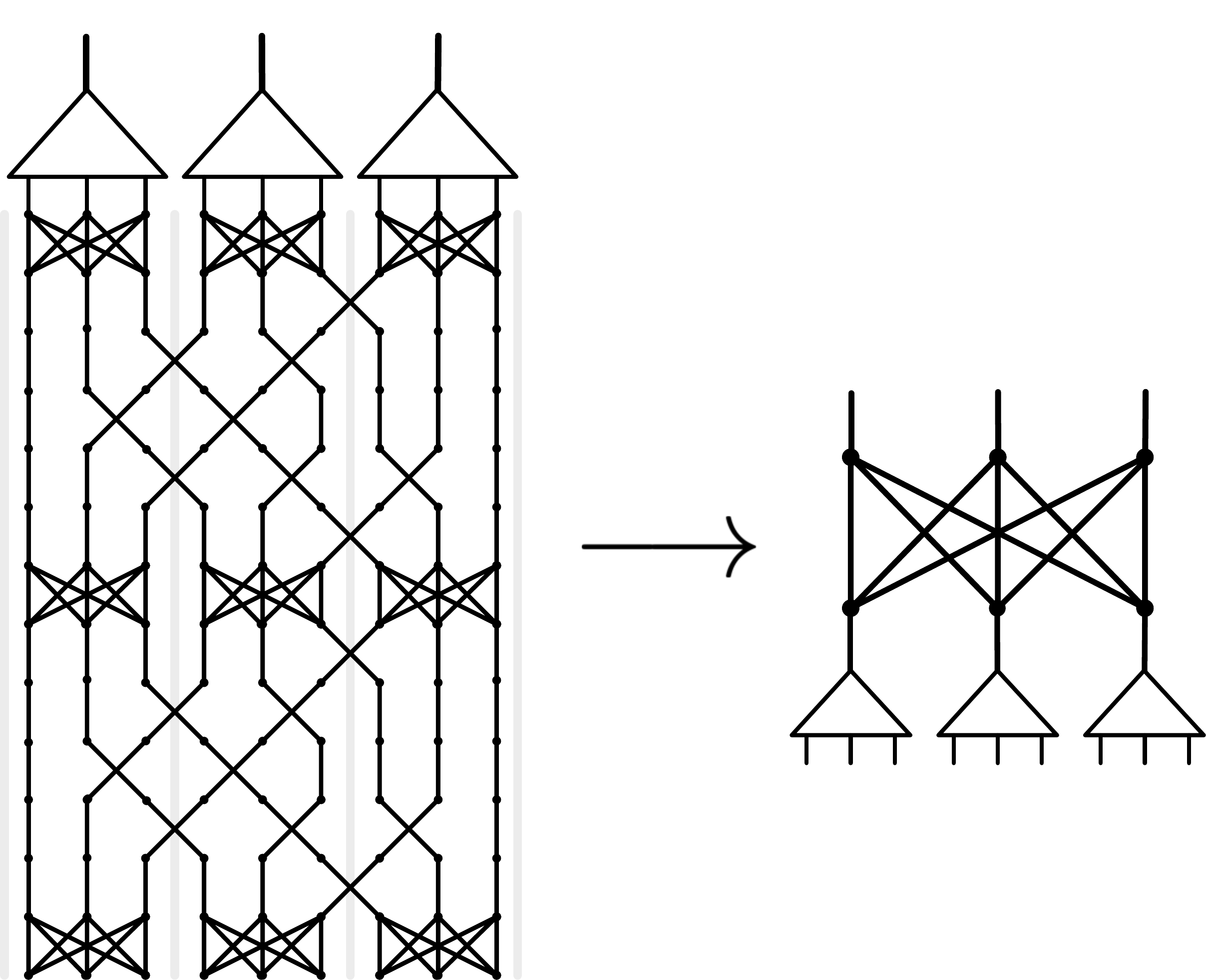}
 \end{equation}
 \begin{figure} [!htbp]
 \includegraphics[width=\textwidth]{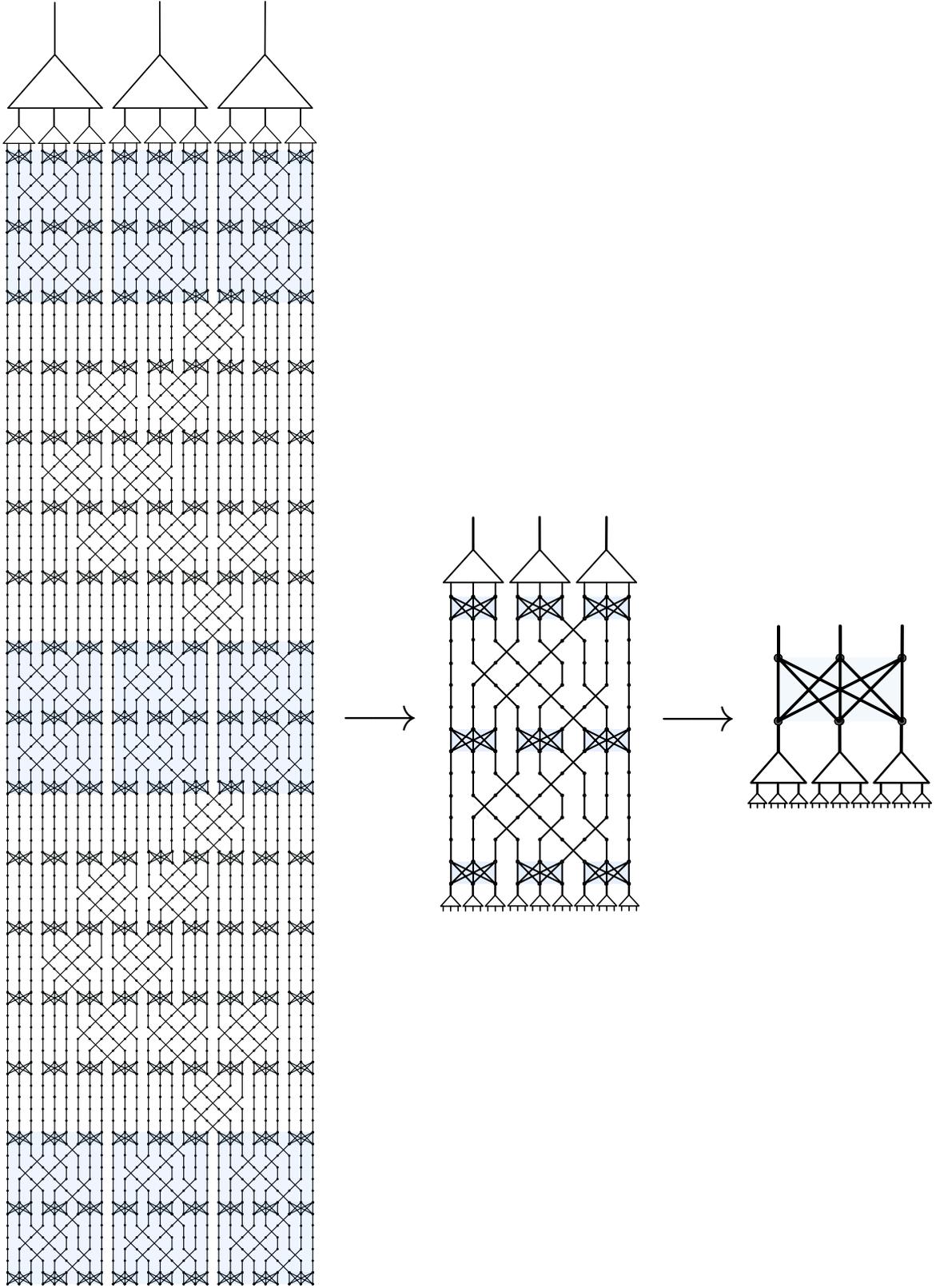}
         \caption{A circuit for $FT^2(\cZ_0)$ followed by two rounds of hierarchical application of the ideal decoder.  Pushing each layer of decoders gives the sequence of circuits $FT^2(\cZ_0) \overset{D_1^{\otimes 9}}{\longrightarrow} FT(Z_0)_{\text{eff} } \overset{D_1^{\otimes 3}}{\longrightarrow}  (\cZ_0)_{\text{eff}'}$. For convenience, the $\cZ$ gadgets before and after pulling ideal decoder through are highlighted in blue. }
         \label{fig:ts:Z2}
 \end{figure}

To further illustrate how self-simulation works, let us show how a hierarchical application of the ideal decoder can reduce the circuit level by level. One application of an ideal decoder reduces $FT^n(\cC)$ to $FT^{n-1}(\cC)$.  An explicit example of this is shown in Fig.~\ref{fig:ts:Z2}. In this example, the input circuit is $FT^2(Z_0)$; upon pushing a layer of ideal decoders through the circuit, we obtain $FT(Z_0)$.  Pushing another layer of ideal decoders through the circuit yields $Z_0$. Thus, pushing layers of the ideal decoder through repeatedly reveals which circuit is being simulated at every level.

In our subsequent proof, we replace Tsirelson's original $X_n$, $Y_n$, and $Z_n$ with $FT^{n}(\cX_0)$, $FT^{n}(\cY_0)$ and $FT^{n}(\cZ_0)$, respectively, which we call the modified Tsirelson automaton:
\begin{definition}[Modified Tsirelson's automaton]
Given a 1D line of $L$ bits, denoting (or $n = \log_3 2 L$, $n = \log_3 L + 1$, respectively), the modified Tsirelson's automaton for $X$, $Y$ or $Z$ gadget consists of applying the simulated gadget $FT^n(X_0^T)$, $FT^n(Y_0^T)$ or $FT^n(Z_0^T)$ where $T$ is the number of repetitions. 
\end{definition}
This modification is rather minor, amounting to repeating each $Z_k$ twice as compared to Tsirelson's original construction. 
However, as we show in the next subsection, one can prove fault tolerance by directly using the exRecs method for quantum concatenated codes.

\subsection{Proof of fault tolerance} \label{subsec:proof1}

We will now prove that the modified Tsirelson's automaton is fault tolerant using the extended rectangles (exRecs) method. Because in this section we use the exRecs method in its unchanged form, we only provide a brief summary of the most used concepts in App.~\ref{app:exRec}.
We encourage the reader to revisit the literature~\cite{Aliferis2005Apr, gottesman2009introductionquantumerrorcorrection} for a more comprehensive introduction\footnote{A 
comprehensive introduction can also be found in notes written by Gottesman that will eventually be published in the form of a book~\cite{GottesmanBook}.}.  In particular, we will be extensively using the following concepts from this method: Gate and EC properties, good and correct exRecs, bad exRecs and truncation, ideal and $\ast$-decoders, as well as the techniques used to prove Theorem~\ref{thm-1level}. 
When proving fault tolerance for the toric code automaton we will need to modify the exRecs method, for which the appropriate definitions will be provided when needed.

Let us start by considering the same error model as in Tsirelson's original paper and show that exRec properties hold for our modified gadgets. Let us refer to the states of the bits at spacetime locations corresponding to the inputs and outputs of each elementary gadget as {\it wires} (gadgets have both incoming and outcoming wires). In the wire error model, errors occur only in the wires between consecutive gadgets.
 \begin{proposition} [Properties of modified Tsirelson's gadgets under wire error model] \label{prop:wire_error_model}
 Assume a $p$-bounded error model (the ``wire error model'') where
 \begin{itemize}
     \item Any wire can fail (including the input and output wires in $\cX_0$, $\cY_0$ and $\cZ_0$ gadgets);
      \item The action of all level-0 gadgets is perfectly reliable.
 \end{itemize} 
 Then, the gate gadgets $\cX_1$, $\cY_1$ and $\cZ_1$ satisfy Gate A and Gate B properties with $t = 1$, and the EC gadget $Z_0$ satisfy EC A and EC B  properties with $t = 1$ (these properties are defined in Def. \ref{def:gateEc-OG}). 
 \end{proposition}
\begin{proof}
    All the properties can be seen by inspecting respective gadgets.  One can also see that if we did not choose to change the action of the original $Z_1$, the property Gate B would not hold for it; an explicit example of this can be seen by considering the gate B property with an error-free $Z_1$ acting on the input state $111\, 100\, 000$.
\end{proof}

Under this noise model, one can think about errors as occurring between the application reliable (noiseless) elementary gadgets.
We will now prove fault tolerance of Tsirelson's automaton under the wire error model:  
  \begin{theorem} [FT of modified Tsirelson's automaton with wire error model]
  \label{th:FT_TS_1}
 Consider a classical circuit $\cC$ defined on 1 bit consisting of $T$ idle operations, i.e. $\cC = \cX_0^T$. Construct the repeated fault-tolerant simulation $FT^n(\cC)$.  Assuming the error model in $FT^n(\cC)$ to be the wire error model from Prop.~\ref{prop:wire_error_model}, then it achieves fault-tolerant simulation of $\cC$ with a $(Ap)^{2^{n}}$-bounded wire error model for some constant $A$. The probability that the initial logical state $s(0)$ in the simulated circuit $\cC$ will get corrupted in time $T$ is bounded by
 \begin{equation}
 \mathbb{P}(D_n(s(T)) \neq D_n(s(0))) \leq T \cdot (Ap)^{2^{n}}
 \end{equation}
 \end{theorem}
 \begin{proof}

Consider the situation where an exRec for either $\cX_1$, $\cY_1$, or $\cZ_1$ is bad. For a $\cX_1$-type exRec, the minimal event causing it to be bad has probability $O(p^2)$. For a $\cY_1$-type exRec, the minimal event also has probability $O(p^2)$. However, while for $\cY_1$ the minimal event failing any of the simulated wires (i.e. the encoded output state on one of the triples of outgoing wires) has a probability $O(p^2)$, the minimal event failing both simulated outputs has a probability $O(p^4)$. Thus, pulling a $\ast$-decoder through a bad exRec containing $\cY_1$ corresponds to a simulation of $\cY_0$ with an $A' p^2$-bounded error model. 

Let us similarly determine the error model resulting from pulling a $\ast$-decoder through a $\cZ_1$-type bad exRec. By direct verification, we note that the failure of a single simulated wire (in either the input or output) occurs with probability $O(p^2)$. To check consistency with a $A'p^2$-bounded wire error model, we need to show that two simulated wires fail with probability $O(p^4)$.  However, naively it seems like there are events where two wire failures can cause two simulated wire failures.  For example, consider the scenario where the encoded input is $111 \, 000 \, 000$; then an error with probability $p^2$ can cause an output (after the ideal decoder) of $111 \, 111 \, 111$. The output in the absence of the error would be $000 \, 000 \, 000$. Thus, the outputs in the presence and absence of the error differ on all three encoded bits. However, this error is equivalent to simply failing an additional input wire.  This is because, such an input failure, which occurs with probability $p^2$ would cause the input $111 \, 110 \, 000$, which exits the gadget as $111 \, 111 \, 111$. This is still consistent with an $A'p^2$-bounded error model. Further examination shows that an error with probability at most $O(p^4)$ can cause two simulated wires to fail, even taking into account the kinds of equivalences described above. 
Thus, the error model experienced by the simulated $\cZ_0$ is majorized by a $A' p^2$-bounded wire error model. 

Now that we know that the wire failures in each bad exRec obey an $A'p^2$-bounded error model, we need to show that the error model for the entire simulated circuit $FT^n(\cC)$ is $A p^2$-bounded for some constant $A$.  Upon pulling a layer of $\ast$-decoders through, if a simulated wire has an error, it must have been a part of a bad extended rectangle.  For a given set $W$ of locations on wires, call $E_W$ an event where there is a failure in each location $w$ in $W$. Suppose that $E_W$ was caused by a set $\{ R_w\}$ of exRecs being bad. There are many possible choices for such a scenario because different noise realizations lead to different truncations (determined by Algorithm~\ref{alg:OG} in App.~\ref{app:exRec}, which determines the assignment of good and bad exRecs), and thus, to the different coverings of the circuit with exRecs. Therefore, we can sum over all such scenarios:
\begin{align}
 \mathbb{P}\left(E_W\right) = \sum_{\{R_w\}}\mathbb{P}\left(E_W \text{ occurred due to badness of } R_w \right)
 \end{align}
All bad exRecs can be chosen to be disjoint by the truncation procedure (see App.~\ref{app:exRec} for a more detailed discussion). Because the noise model is $A' p^2$-bounded  for each bad exRec, we can upper bound the sum of probabilities on the right-hand side as 
\begin{align}
 \mathbb{P}\left(E_W\right) \leq \sum_{\{R_w\}} A'^{|W|} p^{2|W|} \leq (2^3 A' p^{2})^{|W|},
 \end{align}
where in the second inequality we used the fact that the possible number of truncations of any exRec is at most $2^k$ with $k \leq 3$ because Tsirelson's construction requires at most 3-bit gates.  Therefore, the error model for the entire circuit is $A p^2$-bounded for $A = 8 A'$.

Having proved that one level of simulation results in error suppression, we iterate the $FT$ simulation $n$ times (sequentially pulling decoders past the circuit at higher and higher levels). Upon doing so, the new error model is a $p_k$-bounded wire error model after $k$ iterations, with
 \begin{equation}
 p_k \leq A p_{k-1}^2;\hspace{0.3cm} p_0 = p
 \end{equation}
 Solving this recurrence gives $A p_n \leq (A p)^{2^{n}}$. Since the error model for $\cC$ is $(A p)^{2^{n}}$-bounded (removing the extra factor of $A^{-1}$ since $A > 1$), the probability of a logical failure upon union bounding over $T$ repetitions of the simulated $\cX_0$ gate completes the proof.
 \end{proof}

Above, we proved fault tolerance only for the $FT^n(\cX_0^T)$ automaton; however, a proof for the $\cY_0$ and $\cZ_0$-generated automata is completely analogous. 
The result above implies a threshold value of $p$ for the modified Tsirelson automaton which is lower bounded by $A^{-1}$.
Tsirelson's Theorem~\ref{thm:FT_TS} for the modified automaton is reproduced by setting $n = O(\log_3 L)$. 

Next, let us now assume a more realistic noise model where both wires and gadgets can fail, and the noise is $p$-bounded (when a gadget fails, it is allowed to produce an arbitrary output, and $p$-boundedness means that if $M$ is a set of gadgets that fail in a certain noise realization, then that noise realization occurs with probability $\leq p^{|M|}$). We call this error model a \emph{gadget error model}. Then we can show the following: 

  \begin{corollary} [FT of modified Tsirelson automaton with gadget error model]
  Consider the same circuit $\cC = \cX_0^T$ as in Theorem~\ref{th:FT_TS_1}. Assuming a $p$-bounded gadget error model for the circuit $FT^n(\cC)$, for $n \ge 1$, the simulated circuit $\cC$ experiences a $(Ap)^{2^{n-1}}$-bounded wire error model for some constant $A$.
 \end{corollary}
 \begin{proof}
 We will show that for any circuit $\cC'$, a $FT(\cC')$ with gadget failure model simulates $\cC'$ with wire error model instead. By expressing  $FT^n(\cC) = FT\left ( FT^{n-1}(\cC) \right)$, and identifying $\cC'= FT^{n-1}(\cC)$, we see that the $n$ levels of noisy-gadget simulation of the original circuit are analogous to $n-1$ levels of noisy-wire simulation of the original circuits instead, for which we then can use Prop.~\ref{prop:wire_error_model} and Theorem~\ref{th:FT_TS_1}.

A single simulation $FT(\cC)$ operates effectively with $t = 0$ (where the definition of $t$ is provided in App.~\ref{app:exRec}). Therefore, any exRec containing even one error is bad. The gate and EC conditions hold only when there are no errors at all. Consider a bad exRec containing one of the gadgets. For $\cX_1$, the minimal event has simply $O(p)$ probability. For $\cY_1$, by inspection, one sees that the bad exRec induces an error model in the simulated $Y_0$ that is $A' p$-bounded on the wires. For exRec containing $\cZ_1$, by a similar consideration to Theorem~\ref{th:FT_TS_1}, we get $A' p$-bounded error model on the wires. Thus, we have shown that $FT(\cC')$  simulates $\cC'$ with the $A'p$ wire error model instead, which concludes the proof. 
 \end{proof}

Finally, we remark that Tsirelson in fact proved a stronger result for his automaton in~\cite{Cirel'son2006Aug}, namely, that the state input into the automaton is preserved at each site $i$ with some constant probability (which he bounds by $3/4$ from below).  A similar result could likely be derived from our Thm.~\ref{th:FT_TS_1} by considering an (inverted) discrete lightcone of gadgets extending to the past from the point $i$  at a given time $t$, similarly to Tsirelson's original proof.

 \subsection{Numerics for Tsirelson's automaton} \label{ss:tsirelson_numerics}

In this subsection we provide a numerical estimate of the threshold in Tsirelson's original automaton by simulating its performace with Monte Carlo numerics.\footnote{The numerics to follow are run for Tsirelson's original automaton, rather than our modified version, which requires larger circuit depth, making it less convenient for numerical study. } 
Our numerics will employ the ``wire error model'' of Prop.~\ref{prop:wire_error_model} which we model at the circuit level as an i.i.d. noise of strength $p$. Explicitly, at a given time $t$ of the circuit evolution, we 1) noiselessly apply the gates applied by the automaton between times $t$ and $t+1$, and then 2) apply noise: with probability $p$, each spin is replaced by a Bernoulli random variable with mean $(1+\eta)/2$, with the noise being biased towards the \hi{$s_i = 0$ ($s_i = 1$)} logical state if $\eta >0$ ($\eta <0$). See \cite{code} for code generating the data in the plots to follow.

\subsubsection{Decoding fidelity}

\begin{figure}[!htbp]
	\centering
	\begin{subfigure}[b]{0.49\textwidth}
		\centering
		\includegraphics[width=\textwidth]{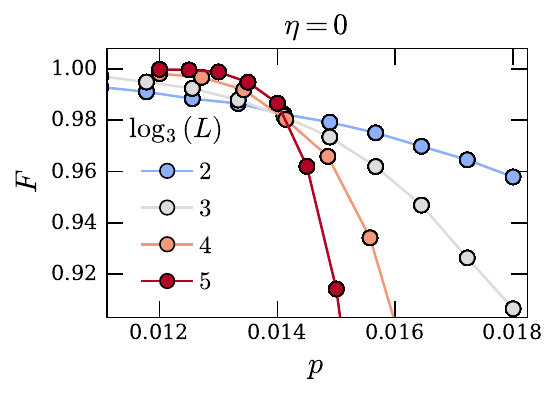} 
	\end{subfigure}
	\hfill
	\begin{subfigure}[b]{0.49\textwidth}
		\centering
		\includegraphics[width=\textwidth]{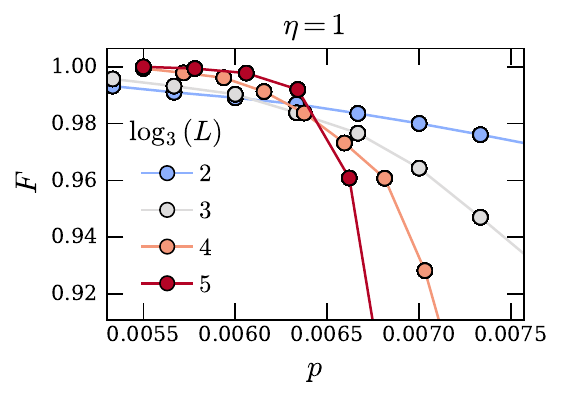} 
	\end{subfigure}
	\caption{Logical fidelities in Tsirelson's automaton after $T=50$ repetitions of the error-correcting gadget $X_n$ subjected to iid wire noise, shown for unbiased noise ($\eta=0$, left) and maximally biased noise ($\eta=1$, right). Changing the value of $T$ does not appreciably shift the threshold point, but does modify the value of $F$ at which the crossing occurs. The initial state in the maximally biased case is chosen to be the logical state disfavored by the bias of the noise. }
	\label{fig:Ft}
\end{figure}

The threshold value $p_c$ is easiest to establish by initializing the system in a logical state and computing the expected time over which the system retains the memory of its initial conditions. 
To this end, we define the logical fidelity $F(T)$ as the probability for ideal decoding to succeed after $T$ repetitions of $X_n$: 
\be F(T) =\min_{s(0)} \mathbb{P}(D(s(T)) = D(s(0))). \ee 
Here the minimum over initial states is needed in the case of biased noise, where it selects out the logical state with the shortest lifetime. 

A threshold can be deduced from the $L$ dependence of $F$ measured after a fixed ($L$-independent) number of applications of $X_n$, where $n = \log_3(L)$. As shown in Fig.~\ref{fig:Ft}, for unbiased noise we find a transition at $p_c|_{\eta = 0} \approx 0.0144$, while for maximally biased noise the threshold is approximately halved: $p_c|_{\eta = 1} \approx 0.0065$. In an i.i.d. gadget failure model with rate $p$, we observe a similar threshold of $p_c|_{\rm gadget} \approx 0.008$.

\subsubsection{Relaxation time}

\begin{figure}[!htbp]
	\centering
\includegraphics[width=.5\textwidth]{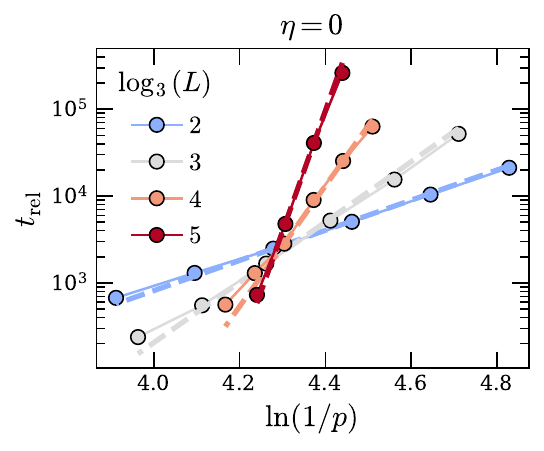}
	\caption{Relaxation times near threshold in Tsirelson's automaton for unbiased noise. The dashed lines are fits to the scaling given in \eqref{wiretrel}. }
	\label{fig:trel}
\end{figure}

To evaluate the lifetime of the memory at noise rates below the threshold, we define the relaxation time $t_{\rm rel}$ as the average time, measured in a number of applications of $X_n$ rather than in circuit depth (the time measured in circuit depth will be $T$ times an exponential in $n$, which will be $\text{poly}(L,T)$), at which ideal decoding fails:
\be \label{tsirelson_trel_def} t_{\rm rel} = \mathbb{E}[\min_{s(0)} \min\{ T \, : \, D(s(T)) \neq s(0)\}], \ee 
where $\mathbb{E}$ is taken over noise realizations. Below threshold, the relaxation time scales as $t_{\rm rel} = (p_c/p)^{L^\lambda}$
for constant $\lambda > 0$. For Tsirelson's automaton with a wire error model our fault tolerance theorem gives $\lambda = \ln(2) / \ln(3)$, producing the scaling behavior
\begin{equation} \label{wiretrel} 
t_{\rm rel} \sim p^{-2^{\log_3(L)}}.
\end{equation}
Fig.~\ref{fig:trel} shows that this functional form is indeed obeyed quite well when $p<p_c$. 

% \begin{comment}
\subsection{Generalization to other codes}

We will briefly discuss a generalization of the modified Tsirelson's automaton by replacing the majority vote procedure with error correction of a different code.  Let $C$ be a (small) classical linear code with parameters $[n,k,d]$, and define the operation $EC(s_1, s_2, \cdots, s_n)$ to be an error correction operation that maps the input bit string $s_1, s_2, \cdots, s_n$ to the closest codeword specified by the parity check matrix $H$ of the code.  As in Tsirelson's construction, define $\cX_0$ to be an idle operation, $\cY_0$ to be a swap of two bits, and $\cZ_0 = EC(s_1, s_2, \cdots, s_n)$ which acts on blocks of $n$ bits. By performing a similar fault tolerant simulation as before, we will obtain a concatenation scheme for the code $C$. This will result in a hierarchical decoder for the code with \hi{the parity check matrix} 
\begin{equation}
   H' =  \begin{pmatrix}
H \otimes I \otimes I... \otimes I   \\
 L \otimes H \otimes I ... \otimes I  \\
  L \otimes L \otimes H ... \otimes I \\
... \\
 L \otimes L \otimes L ... \otimes H  \\
\end{pmatrix}
\end{equation} 
where $\otimes$ denotes the tensor product for parity check matrices \hi{and $L$ is the logical generator matrix}. This corresponds to a classical tensor-product code (not to be confused with tensor products of chain complexes frequently used in quantum codes) with parameters $[n^\ell, k^\ell, d^\ell]$ if the simulation is iterated $\ell$ times.

Let us construct the fault-tolerant gadgets. We define $\EC_1 = \cZ_0$.  $FT(\cX_0)$ is defined on $n$ bits and is a circuit of idles for depth $2n-1$.  $FT(\cY_0)$ is defined on $2n$ bits and is a pattern of swaps and idles designed to function in a transversal way, i.e. the operation results in a bitwise swap of the first block of $n$ bits with the second.  Similarly to Tsirelson's construction, it will be a `diamond-shaped' pattern of swaps whose depth is $2n-1$.  

Finally, the gadget $FT(\cZ_0)$ acts on $n$ blocks of $n$ bits. The first part of this gadget brings the $j$-th bit of block $s$ to the $s$-th location in $j$-th block for all $j$ and $s$ using a pattern of swaps and idles\footnote{More precisely, let the $n^2$ bits be arranged on a line with the bit positions indexed by an integer $j \in \{0,\dots,n^2-1\}$. Writing $j = j_1 n + j_2$, with $j_{1,2} \in \{0,\dots,n\}$, the circuit of swaps and idles sends the bit at location $j_1n + j_2$ to the location $j_2n + j_1$. The specific way in which this circuit is implemented is not important, but one such `greedy' implementation yields a depth $(2n^2 - 4n + 3 + (-1)^{n+1})/2$.}. The gadget then applies $\cZ_0$ on each resulting $n$-bit block. Finally, the circuit performs the same sequence of swaps and idles as in the first step. This way, it implements $\cZ_0$ transversally (bitwise) between each block. The fault tolerant simulation procedure then follows from Def.~\ref{def:FT_original} and the discussion afterward.

\hi{As mentioned at the beginning of Sec}.~\ref{ss:modification_of_tsirelson}, Tsirelson's automaton corresponds to choosing $C$ to be the 3-bit repetition code. Since in general the information will be encoded in a $[n^\ell, k^\ell, d^\ell]$ code, by choosing the code $C$ that encodes more than one logical bit, i.e. $k>1$, we obtain a one-dimensional classical memory encoding a power-law number of bits.

The proof of fault tolerance using the exRecs method goes through essentially unchanged regardless of $C$\footnote{
We remark that this approach also straightforwardly generalizes to the case when we concatenate error correction circuits for several classical codes $C_i = \{ C_1, C_2, \dots\} $. For this, one would need to introduce a family of error correction gadgets $\cZ_0^{(i)}$. It may be possible that the fault tolerance property could still hold if one takes a family of codes with growing parameters such that the resulting tensor-product code has a constant rate, similarly to Refs.~\cite{Gacs_2001} and \cite{yamasaki2024time}. We leave this question as well as its generalization to the quantum setting to future work.}, and one can prove that the probability of incorrectly decoding a logical bit is
\be \mathbb{P}(D(s(T)) \neq D(s(0))) \leq T \cdot (A p)^{(t+1)^l},\ee 
where $A$ is a constant, $t= \lfloor (d-1)/2 \rfloor$, and the $\ast$-decoder $D$ is defined by a hierarchical application of $\EC$ in complete analogy to \eqref{eq:idealdecoderTS}.

%% file: 3.TCconstruction.tex
\section{The toric code automaton}\label{sec:toriccode}

In this section, we introduce the toric code automaton. We start by defining the measurement and feedback model that will be used in the rest of the paper.   We show that a general $p$-bounded error model can be reduced to studying a $p$-bounded gadget error model for Pauli-$X$ and $Z$ errors separately. We refer to this error model as a decoupled error model.  For the rest of the paper, we will be studying the automaton under the decoupled error model.

After a brief high-level discussion of the construction, we show how to recast Tsirelson's automaton in terms of the measurement and feedback model and use this perspective to construct a basis of elementary gadgets for the toric code automaton.  Finally, we describe how the toric code automaton can be assembled from these elementary gadgets.

\subsection{Measurement and feedback automaton and decoupled error model}~\label{subsec:TC}

Here, we will informally present the most important definitions and results from Appendix~\ref{app:noisedetails}, which contain all of the details and proofs justifying the noise model that we study in the rest of the paper.

The main model for the quantum automaton that we study is the measurement and feedback model, whereby at each time step we measure all stabilizer generators, apply a feedback operator depending on syndromes, and apply a quantum noise channel. 

\begin{definition}[Measurement and feedback model; informal]

A measurement and feedback model is a quantum dynamical system that takes a stabilizer group $\mathcal{S}$ and at time step $t$: 
\begin{itemize}
\item [(a)] Measures all stabilizer generators, obtaining syndromes $\sigma$,
\item [(b)] Applies the feedback operator $\cF_t(\sigma)$ conditioned on the syndrome measurements, 
\item [(c)] Discards the syndrome measurement outcomes, and
\item [(d)] Applies noise channel $\mathcal{E}_t$.
\end{itemize}
\end{definition}

The measurement and feedback model does not yet reference locality.  For this, we define the notion of a stabilizer quantum gadget, which is an operation that consists of a $O(1)$-sized region of stabilizer measurements followed by feedback that only depends on measurement outcomes in this region, followed by noise (see Defs.~\ref{def:elementary-gadget} and ~\ref{def:stab-quant-gadg} for the formal definition).  We can then define a stabilizer quantum automaton by chaining together the action of many local stabilizer quantum gadgets.

\begin{definition}[Stabilizer quantum automaton; informal] 
A stabilizer quantum automaton for a stabilizer code $\mathcal{S} = \langle g_i \rangle$ and a given noise model is a depth-$t$ circuit of local stabilizer quantum gadgets associated with $\mathcal{S}$ that evolves an initial density matrix $\rho(0)$ to a density matrix $\rho(t)$.
\end{definition}
We will simply call the stabilizer quantum automaton an ``automaton''.  In App.~\ref{app:noisedetails}, we argue that due to stabilizer measurements being performed at every timestep, the effect of general noise can be reduced to studying a simpler model where the noise applies a Pauli operator $\cO_t$ at timestep $t$ and the feedback operator $\cF_t(\sigma)$ is also a Pauli operator.  Furthermore, for noise that can be described as a $p$-bounded qubit or $p$-bounded gadget error model, the effective Pauli error model after this simplification can be described by what we call a $p$-bounded \emph{decoupled} gadget error model.  We provide definitions of $p$-bounded gadget error model and the $p$-bounded decoupled gadget error model below:

\begin{definition} [$p$-bounded gadget error model; informal]\label{def:p-bounded-gadget-model}
A gadget error model is $p$-bounded if, given noise realization with a set of noise locations $E$, an associated set of failed gadgets $F_E$, and any subset of gadget locations $A_g$, it satisfies
    \be \label{eq:p-bounded-gadget}
    \mathbb P [  A_g  \subseteq F_E] \leq p^{|A_g|}.
    \ee
\end{definition}

\begin{definition}[$p$-bounded decoupled gadget error model; informal] \label{def:approximate-decoupled-gadget-error-model2}

Consider a noise realization with noise locations $E$ and noise operator history $\{\cO_1, \cO_2, \cdots, \cO_T \}$ with  $\cO_t = \cO_t^{(X)} \cO_t^{(Z)}$.  We partition $E = E_X \cup E_Z$ so that $E_X$ and $E_Z$ denote sets of the spacetime locations of qubits marked for being possibly acted upon by $X$ and $Z$ operators, respectively, and $\bigcup_t \supp (\cO_t^{(X)}) \subseteq E_X$ and $\bigcup_t \supp (\cO_t^{(Z)}) \subseteq E_Z$. 

Specify sets of failed elementary gadgets  $F_{E_X} = \{ F_{1}^{(X)}, F_2^{(X)}, \dots, F_T^{(X)}\}$ and $F_{E_Z} = \{ F_{1}^{(Z)}, F_2^{(Z)}, \dots, F_T^{(Z)}\}$ such that each individual gadget $f^{(X/Z)} \in F_t^{(X/Z)}$ shares support with $E_{X/Z}$, and $E_{X/Z}$ is contained in the union of supports of the gadgets in $F_{E_{X/Z}}$.  

For each time step $t$, we can choose a decomposition $\cO_t^{(X/Z)} = \prod_i \cO_{t,i}^{(X/Z)}$ where each $\cO_{t,i}^{(X/Z)}$ is contained in the support of $f_{t,i}^{(X/Z)} \in F_t^{(X/Z)}$.  The corresponding noise model is a $p$-bounded decoupled gadget noise model if, for any subsets of spacetime locations of $X$-type and $Z$-type elementary gadgets $A_X$ and $A_Z$,

\begin{equation}
\mathbb{P}[(A_X \subseteq F_{E_X}) \wedge (A_Z \subseteq F_{E_Z})] \leq p^{|A_X| + |A_Z|}.
\end{equation}
\end{definition}

As a technical note, the $p$-bounded gadget error model and $p$-bounded decoupled gadget error model should also include the adjective ``approximate error model'', see Def.~\ref{prop:apprnoise}.  An approximate error model will have a noise history (see Lemma~\ref{lemma:noise-collapse} for a definition) that never applies a logical operator at any fixed time step.  From standard percolation arguments, we can prove that an approximate error model is close in total variation distance to the non-approximate error model.  Later on in the text, we may use the adjective ``approximate error model'' to reinforce this technical subtlety.  

The gadgets of the toric code automaton will be chosen to be $X$-type and $Z$-type, see Def.~\ref{def:xzgadgets}.  For $X$($Z$)-type gadgets, one makes $Z$($X$)-stabilizer measurements in some $O(1)$-sized region, and performs $X$($Z$) Pauli feedback based on these measurement outcomes.  Moving forward, we will be heavily relying on (a quantum) notion of damage:
\begin{definition}[Damage] \label{def:damage}
   Call $\ket{\phi}$ the initial logical state of a stabilizer automaton with stabilizer group $\cS$. Consider a state $\ket {\psi} = \cO \ket {\phi}$  where $\cO$ is a Pauli operator. We call any operator $\cO'$ such that $\cO' \cO \in \cS$ the damage of the state $\ket {\psi}$.
\end{definition}
Since we only need to study the toric code automaton under a decoupled noise model, the following fact holds:
\begin{fact} \label{fact:XZ-decoupling}
    Run the toric code automaton with $X$ and $Z$-type gadgets  experiencing a decoupled gadget noise model for $T$ time steps.  Sample a noise realization with a history of noise operators $\{\cO_1, \cO_2, \cdots, \cO_T\}$, and assume that the initial state is a toric code state of the form $\cO_{\mathrm{in}}^{(X)} \cO_{\mathrm{in}}^{(Z)} \ket{\psi}$ where $\cO_{in}^{(X/Z)}$ are the input damage operators and $\ket{\psi}$ is a syndrome-free initial logical state of the toric code.
    
    Then, the total operator applied by the automaton (i.e. due to both its feedback and the noise) under this noise realization can be determined by running the $X$ and $Z$-type automata whose outputs give the operators $\cO_{\mathrm{tot},T}^{(X)}$ and $\cO_{\mathrm{tot},T}^{(Z)}$, and multiplying these operators together.
\end{fact}
This reduces the analyzes to studying two separate $X$ and $Z$ automata, which significantly simplifies the ensuing analysis.

\subsubsection{Classical gadget maps}

While we defined gadgets of an automaton to be quantum operations that measure stabilizers and apply feedback on the system, let us introduce a slightly different way of describing their action which we will extensively use throughout the paper. We focus on an $X$-automaton operating under a $p$-bounded decoupled gadget error model. We can map it to a classical problem corresponding to the dynamics of $X$-error syndromes on the lattice.  The following definition and proposition makes this precise:

\begin{definition}[Classical gadget map] \label{def:classical-action}
Consider a state $\rho_t$ with definite syndrome $\sigma$ experiencing a noise realization $E$ with syndrome $\varepsilon$. Define $\rho_{t+1}'$ to be the state after applying the noisy elementary gadget $G$. We note that, under a decoupled gadget error model, the state $\rho_{t+1}'$ will have definite syndrome, which we denote $\sigma'$. Then, we define $G[\sigma, \varepsilon]$ to be a classical map acting on $\sigma$ that produces $\sigma'$. 

For any pair of neighboring gadgets $G_1$ and  $G_2$, their joint operation is determined by running them simultaneously and following the prescription above.  For a non-elementary gadget $G$, we define the associated classical action $G[\sigma, \varepsilon]$ to be one that composes the classical actions of its constituent elementary gadgets.  

We will also use the notation $G[\sigma] = G[\sigma, \varnothing]$ to refer to the noiseless gadget operation on a state with syndrome $\sigma$. 
\end{definition}

We will use both the operator and the syndrome picture to analyze the properties of the toric code automaton and its gadgets. For the syndrome picture, we will utilize the classical maps associated with gadgets.

\subsection{Basic idea behind the construction}

Let us now discuss the high-level idea behind our construction. Unlike in a classical automaton, it is impossible to directly measure the state of an individual qubit in a quantum code without affecting the stored logical state. Instead the automaton can only measure stabilizers, thereby revealing information about the errors, but not the logical state of the code. Therefore, the measurement and feedback model is the most natural model which will not destructively measure the underlying quantum state. 

In fact, one can recast Tsirelson's automaton as a stabilizer measurement and feedback model, which we do explicitly in Subsec.~\ref{modified_tsirelson2}.  
The corresponding stabilizer code is simply the classical repetition code.  
%Doing this allows us to gain intuition about how Tsirelson's automaton operates in terms of error syndromes. 
We then introduce the concepts of a \emph{coarse-grained} lattice for each level that the automaton operates at, and we will see that the $\EC$ gadget either cleans up an error if it is small enough or otherwise coarse-grains it to live on a lattice of a larger scale. This process then repeats self-similarly.  Assuming a low noise rate, repeating this map at higher and higher scales cleans up error configurations. To show that the probability of a logical error is small, we can appeal to probabilistic properties of $p$-bounded noise at low noise rates by  using methods similar to G\'acs' sparsity bound~\cite{gacs1983reliable,gacs_slides}. In particular, with high probability, an error configuration sampled from a $p$-bounded noise model can be hierarchically grouped into clusters, where each level-$\ell$ cluster has a characteristic size growing with $\ell$, and different clusters are sufficiently separated from one another. A fault-tolerant automaton should be able to handle large clusters at a given level even when new smaller clusters of errors are dynamically appearing.  To do this, it must eliminate smaller error clusters at a faster rate than larger ones.

Notice that these ideas are very general; they are not specific to Tsirelson's automaton and are independent of the spacetime dimension. Therefore, we can use them as a framework for building an automaton for the toric code, wherein error configurations in $X$ and $Z$ sectors simultaneously undergo a similar coarse-graining action described above. 

We will construct the $X$-type automaton, and our goal will be to consider the $X$-type noise and protect the state against the $X$-type logical errors. As discussed before, the $Z$-type automaton will be defined analogously except on a dual lattice and will have the same properties. To explain how our toric code automaton functions, we begin by defining the primitive gadgets that appear in its construction. One naive starting point for building such an automaton would be to independently run appropriate analogues of the 1D Tsirelson automaton along each row and column of the lattice. This approach fails to be fault tolerant for reasons explained in Sec.~\ref{ss:twod_maj_vote}. We have not been able to find a way to implement a fault tolerant toric code automaton with (the measurement and feedback versions of the) 1D swap and majority vote gates. 

Our construction instead departs from this naive idea in two places. The first place is by decomposing the 1D majority vote gate used in the measurement and feedback version of the Tsirelson's automaton into the syndrome-operating analogue of the `swap' gate and a new type of `match' gate, which pairs up nontrivial syndromes on neighboring sites.  The second place is in defining an $\EC$ gadget which acts on 3$\times$3 cell groups of the toric code lattice.  Our $\EC$ gadget is constructed by running a version of Tsirelson's $Z_0$ gadget vertically and horizontally, and coupling the two operations with match gates inserted at certain times. This $\EC$ gadget will be designed to either eliminate errors in its support or to send syndromes to the corners of the $3 \times 3$ cell. We will then describe how to make this construction hierarchical, using two procedures we call the inner and outer simulations.

\subsection{Stabilizer approach to Tsirelson's automaton}
\label{modified_tsirelson2}

Tsirelson's automaton performs error correction of a repetition code:
\begin{definition}[1D repetition code]
A repetition code on a one-dimensional line of $L$ qubits is a stabilizer code, with the stabilizer group generated by $A_{i+1/2} = Z_{i} Z_{i+1}$, where $i$ labels an edge and $i+1/2$ labels an adjacent vertex.
\end{definition}

In the spirit of quantum stabilizer codes, we will also formally use the term `qubit' instead of `bit', even though the repetition code is classical.

In the measurement and feedback approach, the SWAP gate is analogous to a new elementary gate gadget we will refer to as $\cT_0(v)$. We will use the notation $G(\bm{x})$ where $\bm x$ is either a spatial or a spacetime coordinate associated with the gadget (depending on the context), which should not be confused with the usage of square brackets to indicate the classical gadget action as in $G[\sigma, \varepsilon]$ from Def.~\ref{def:classical-action}. 
This gate gadget is defined as follows:
\begin{definition}[$\cT_0$ for modified Tsirelson's construction]
Let $v+1/2$ denote the edge to the right from the vertex labeled $v$ on the 1D lattice. The action of the $\cT_0(v)$ gadget is:
 \begin{equation}
\cT_0(v) = \begin{cases}
 I &  \text{ if } Z_{v-1/2} Z_{v+1/2} = 1\\
X_{v-1/2} X_{v+1/2}  &  \text{ if }  Z_{v-1/2} Z_{v+1/2}  = -1
 \end{cases}
 \end{equation}
Namely, it (i) measures the $Z_{v-1/2} Z_{v+1/2} $ check, (ii) applies $X_{v-1/2} X_{v+1/2}$ if the measurement outcome was $-1$, and (iii) does nothing otherwise. Given the $\mathbb F_2$-valued syndromes, we can alternatively define the associated classical action $\cT_0(v)[\sigma,\varnothing]$ in the syndrome picture as
 \begin{equation}
 (\sigma_{v-1}, \sigma_{v}, \sigma_{v+1}) \overset{\cT_0(v)[\sigma,\varnothing]}{\longrightarrow} \begin{cases}
 (\sigma_{v-1}, \sigma_{v}, \sigma_{v+1}), &  \text{ if } 
 \sigma_{v} = 0\\
 (\sigma_{v-1} \oplus 1, \sigma_{v}, \sigma_{v+1} \oplus 1), &  \text{ if }  \sigma_{v} = 1
 \end{cases}
 \end{equation}
\end{definition}
One can directly verify that in terms of bits, this action is equivalent to the action of the `swap' gate $Y_0(v)$. We schematically depict the nontrivial action of this gate gadget as:
 \begin{equation}
 \includegraphics[width =0.33\textwidth]{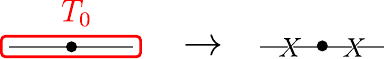}
 \end{equation}
 where the black dot denotes the presence of a nontrivial syndrome.

Next, we discuss the majority vote gadget. Unlike in Tsirelson's construction, rather than using SWAP and majority as a basis, we will use a different (and equivalent) basis of gates. As mentioned above, this is done because it is not natural to perform error correction in the 2D toric code using only 1D majority vote gates. Instead, we will perform error correction in 2D by starting with a grid of intersecting 1D Tsirelson automata, decomposing each 1D majority vote into a sequence of swap and the new `match' gates, and then adding match gates coupling different rows / columns during the action of each 1D majority vote. How exactly this is done will be explained in Sec.~\ref{ss:twod_maj_vote}. First however, we will define the match gate. 

A match gate $\cM_0(e)$ centered on the qubit labeled $e$ is defined as follows:
\begin{definition}[$\cM_0$ for modified Tsirelson's construction]
The elementary $\cM_0(e)$ gadget is defined as
 \begin{equation}
\cM_0(e) = \begin{cases}
 X_{e} &  \text{ if } Z_{e-1} Z_{e} = -1 \text{ and } Z_{e} Z_{e+1} = -1\\
 I & \text{otherwise}
 \end{cases}
 \end{equation}
 i.e. if two nontrivial syndromes are present from measuring checks sharing qubit $e$, the syndromes are matched by flipping the value of the qubit shared by both of the checks. In terms of syndromes, the associated classical action $\cM_0(e)[\sigma, \varnothing]$ is defined as:
 \begin{equation}
 (\sigma_{e-1/2}, \sigma_{e+1/2}) \overset{\cM_0(e)[\sigma, \varnothing]}{\longrightarrow} \begin{cases}
 (0,0) &  \text{ if }  \sigma_{e-1/2} = \sigma_{e+1/2} = 1\\
 (\sigma_{e-1/2}, \sigma_{e+1/2}) & \text{otherwise}
 \end{cases}
 \end{equation}
\end{definition}
We schematically depict the nontrivial action of this gate as:
 \begin{equation} 
 \label{eq:mt_maj_decomp}
 \includegraphics[width =0.45\textwidth]{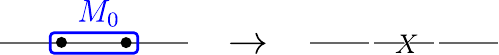}
 \end{equation}

Note that instead of showing the action of $\cT_0$ and $\cM_0$ in spacetime like was done earlier, we just show one time step at a time. This new notation will be useful for the toric code case where the circuits are two-dimensional and a direction illustration of their spacetime action is not feasible.

\begin{definition}[Support of a gadget: 1D repetition code]
Define the location of a syndrome in a 1D repetition code to be the vertex $e+1/2$ for which the syndrome is the measurement outcome of stabilizer generator $Z_e Z_{e+1}$.  The support of a gadget $G$, denoted $\mathrm{supp}(G)$, is the set of spacetime locations that the gadget nontrivially depends on.
\end{definition}

\begin{fact}
The majority gate $Z_0$ can be written in terms of $\cT_0$ and $\cM_0$ as
 \begin{equation}\label{eqn:maj}
 \mathrm{Maj}(e-1,e,e+1) \equiv Z_0 (e-1,e,e+1) =\cM_0(e) \,\cT_0({e+1/2}) \, \cM_0(e) \,\cT_0({e-1/2}) \, \cM_0(e)
 \end{equation}
\end{fact}
Schematically, this looks like 
 \begin{equation} \label{Z0-action}
 \includegraphics[width =0.95\textwidth]{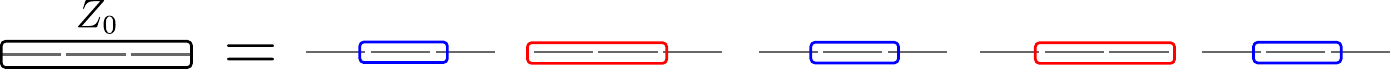}
 \end{equation}

Finally, we denote the elementary `idle' gate gadget as $\cI_0$. We emphasize that the action of the $\cI_0$, $\cT_0$ and $Z_0$ on product states is \emph{identical} to that of $X_0$, $Y_0$ and $Z_0$, respectively.  Thus, using the measurement and feedback paradigm, we can write the modified Tsirelson's automaton in terms of the new gadgets $\cI_0$, $\cT_0$, and $\cM_0$ (with $Z_0$ replaced with the decomposition shown above) with otherwise the same substitution rules introduced in Sec.~\ref{sec:tsirelson}.

In the following, we will also use the following interpretation of the majority vote gadget. 
This gadget removes any single-bit error (a pair of domain walls separated in space by distance 1) within its support. If there is only a single domain wall inside its support, the gadget will move the domain wall to its nearest end. We call the latter action ``\textit{coarse-graining} the error configuration'' (a concept that we discuss more formally shortly). This accounts for all possible input configurations because there are only two possible locations for domain walls within the support of $Z_0$.

This almost completes the new fault-tolerant prescription for expressing a circuit $\cC$ as defined in Def.~\ref{def:FT_original}.  The definitions of $\cI_1$ and $\cT_1$ as well as $\EC$ remain the same as those of $\cX_1$ and $\cY_1$ when interpreted as a substitution rule. We define $FT(\cM_0) = \EC_1 \cM_1 \EC_1$, where $\cM_1$ is defined as:
 \begin{equation}\label{eq:M1-ts}
 \begin{centering}
      \includegraphics[width = 0.23 \textwidth]{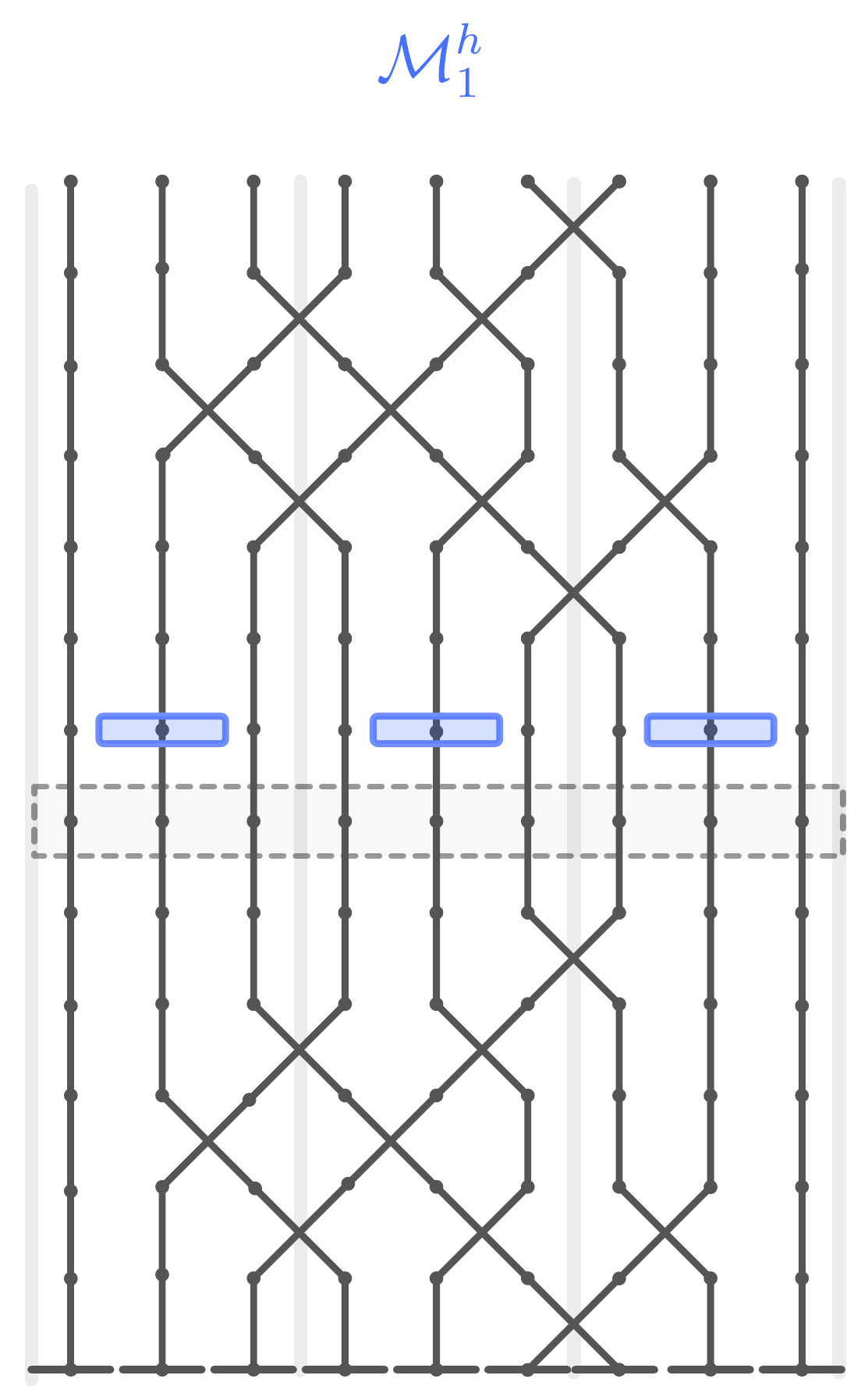}
 \end{centering}
 \end{equation}
with each blue rectangle in the middle of the circuit denoting an $\cM_0$ gadget.
It operates similarly to the modified $\cZ_1$, except the sixth step is now an application of $\cM_0^{\otimes 3}$ rather than $Z_0^{\otimes 3}$. Similarly to $\cZ$-type gadgets in the original construction, $\cM$-type gadgets only occur simultaneously with the gadgets of the same type, and are never combined with $\cI$ or $\cT$ gadgets. Therefore, we do not need to require that the $\cM$ gadgets have the same depth as $\cI$ and $\cT$. 
% ,equalize the depth of their respective circuits.

Finally, the fault-tolerant simulation is defined in the same way as in Sec.~\ref{sec:tsirelson}.

\subsection{Basic definitions for the toric code automaton}

As mentioned previously, we will discuss only the bit-flip ($X$) error correcting part of the automaton in the rest of the paper since the phase-flip error correction has an identical implementation after passing to the dual lattice and exchanging $X \leftrightarrow Z$. 
Therefore, the only stabilizer measurements featured here will be the $A_v$-stabilizer measurements, and thus, $\cL_0$ will be from now on the set of all vertices on the lattice.  We will use the notation $v = (v_x,v_y)$ for the coordinates of the vertex $v$, and denote qubits adjacent to the vertices by half-integer coordinates, such as $(v_x+1/2,v_y)$ for a qubit on a horizontal edge adjacent connecting $(v_x,v_y)$ with $(v_x+1,v_y)$. 
% on the left. 
We will also use $e$ to denote qubit locations when convenient.

We start by introducing the basis of gadgets for our automaton, which are very similar to the $\cT$ and $\cM$ gadgets that we introduced for the stabilizer-based treatment of Tsirelson's automaton. 
The difference is that the $\cT_0$ and $\cM_0$ gadgets now live in two dimensions, and we define a horizontal (`$h$') and vertical (`$v$') action for them. They are designed so that, expressed through their action on syndromes, they have the same action as their one-dimensional counterparts.

\begin{definition}[Elementary $\cT$ gadget for the toric code]
    The action of the $\cT^h_0(w)$ gadget, where $w \in \cL_0$ is defined as
 \begin{equation}
 \cT^h_0(w) = \begin{cases}
 I & \text{ if } A_w = 1\\
 X_{w_x-1/2,w_y} X_{w_x+1/2,w_y} & \text{ if } A_w = -1
 \end{cases}
 \end{equation}
Namely, it (i) measures the $A_w$ stabilizer and (ii) applies a horizontal $X_{w_x-1/2,w_y} X_{w_x+1/2,w_y}$ operator if the measurement outcome was $-1$ and does nothing otherwise. In the syndrome picture, we can define the associated classical action $\cT_0^h(w)[\sigma, \varnothing]$ as
 \begin{equation}
 (\sigma_{w_x-1,w_y}, \sigma_{w_x,w_y}, \sigma_{w_x+1,w_y}) \overset{\cT_0^h(w)[\sigma, \varnothing]}{\longrightarrow}  (\sigma_{w_x-1,w_y} \oplus \sigma_{w_x,w_y}, \sigma_{w_x,w_y}, \sigma_{w_x+1,w_y} \oplus \sigma_{w_x,w_y})
 \end{equation}
\end{definition}
$\cT_0^v$ is defined as a $\pi/2$-rotated version of $\cT_0^h$, so that $x$ and $y$ are exchanged in the above definition;
% , except with $90^\circ$ rotation or $x \leftrightarrow y$.  
the `$h$' (horizontal) or `$v$' (vertical) superscript accordingly stands for the orientation of the feedback action. Intuitively, the action of $\cT_0^h(v)$ broadcasts information about the syndrome present at location $v$ to its left and right neighboring vertices (and for $\cT^v_0(v)$, to the the top and bottom vertices neighboring $v$).

We now define the ``match'' gate for the toric code.
\begin{definition}[Elementary $M$ gadget for the toric code]
    The $\cM_0^h(e)$ gadget is defined as
 \begin{equation}
\cM_0^h(e) = \begin{cases}
 X_e & \text{ if }  A_{e_x-1/2,e_y} = -1 \text{ and }A_{e_x+1/2,e_y} = -1\\
 I & \text{otherwise}
 \end{cases}
 \end{equation}
 i.e. if two nontrivial syndromes are present from measuring the two checks sharing qubit $e$, the syndromes are matched by applying $X$ to the qubit. In terms of syndromes, the associated classical action $\cM_0^h(e)[\sigma, \varnothing]$ is defined as:
 \begin{equation}
 (\sigma_{e_x-1/2,e_y}, \sigma_{e_x+1/2,e_y}) \overset{\cM_0^h(e)[\sigma, \varnothing]}{\longrightarrow} \begin{cases}
 (0,0) & \text{ if }  \sigma_{e_x-1/2,e_y} = \sigma_{e_x+1/2,e_y} = 1\\
 (\sigma_{e_x-1/2,e_y}, \sigma_{e_x+1/2,e_y}) & \text{otherwise}
 \end{cases}
 \end{equation}
\end{definition}
  Pictorially, the gadgets $\cT_0^h$ and $\cM_0^h$ have the following nontrivial action:
 \begin{equation}\label{eq:M0T0}
 \begin{centering}
      \includegraphics[width = 0.42 \textwidth]{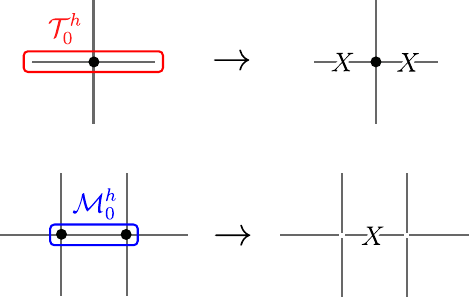}
 \end{centering}
 \end{equation}
$\cM_0^v$ is defined as a $\pi/2$-rotated version of $\cM^h_0$. 
One can see that these gadgets are defined very similarly to the case of the repetition code, except that stabilizers $A_v$ that are measured also have support on adjacent vertical qubits. 

We note that these new gadgets, although mimicking certain properties of Tsirelson's gadgets from the perspective of syndromes, no longer have the same interpretation in terms of their action on error configurations.  In particular, while Tsirelson's $\cT_0$ operation is equivalent to swapping the states of two bits, $\cT_0^{h/v}$ in the toric code no longer has this interpretation: this is because the syndrome no longer detects the difference between two qubits, but rather whether the parity of the states of four neighboring qubits is even or odd.  Consider an example of $\cT_0^h$ acting on an input configuration containing a horizontal error string:

 \begin{equation}
 \begin{centering}
      \includegraphics[width = 0.6 \textwidth]{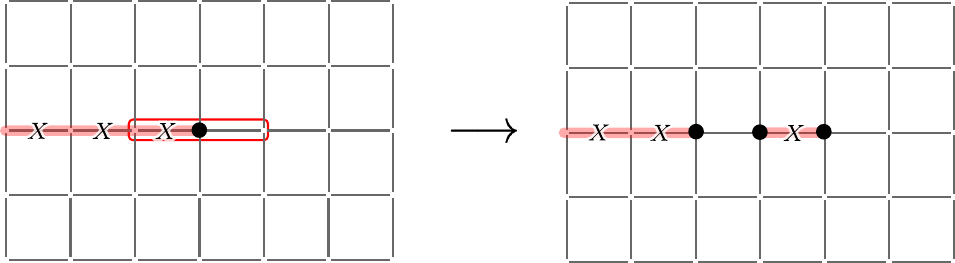}
 \end{centering}
 \end{equation}
In fact, the action shown above is equivalent to the  outcome of the syndrome-operating version of $Y_0$ in Tsirelson's automaton because the error is restricted to the same 1D line as the gadget's feedback. 

However, consider an example with a different input operator corresponding to the same syndrome configuration: 
 \begin{equation}
 \begin{centering}
      \includegraphics[width = 0.6 \textwidth]{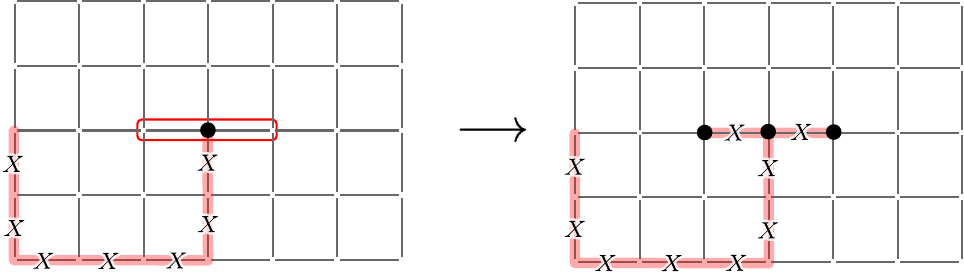}
 \end{centering}
 \end{equation}
In this case, the effect of $T_0^h$ creates an additional perpendicular string. 
Note that the two input error configurations in the examples above are equivalent since they are related to one another by application of a stabilizer.  This example illustrates that, because the error strings acting on the codestates of the toric code can be deformed by stabilizers without affecting the state, it is more meaningful to analyze the evolution of syndromes under the automaton than the evolution of operators.  

\subsection{The construction}\label{ss:construction_of_tc_gadgets}

To build a hierarchical automaton for the toric code, we need to introduce a hierarchy of scales for the toric code first. 
In the earlier sections, the discussion was  similar in spirit to concatenated codes, wherein a $3^k$ qubit repetition code stabilized by Tsirelson's automaton could be viewed as a $k$-fold concatenation of 3-qubit repetition code with itself. 
This was necessary in order to apply the exRecs method for proving fault tolerance of the construction.
However, there is no obvious analogous point of view of the toric code as the toric code cannot be viewed as code concatenation of a number of small codes.
Therefore, our hierarchy will be based on the spatial structure of error configurations instead. 

\begin{definition}
Consider the toric code defined on a square lattice as in Def.~\ref{def:toric-code}. We define the following hierarchy of cells:
\begin{itemize}
    \item Plaquettes of the lattice are called 0-cells and denoted $C_{0,\bm{p}}$, where $\bm{p} = (p_x,p_y)$ is the coordinate of plaquette $p$;
    \item Colonies of neighboring $3\times 3$ groups of 0-cells, such that the 0-cell $C_{0,(3p_x + 1), (3p_y+1)}$ is located at the center are called, 1-cells and denoted $C_{1,\bm{p}}$; 
    \item Colonies of neighboring $3 \times 3$ groups of $(k-1)$-cells, such that the  $(k-1)$-cell $C_{k-1,(3p_x + 1), (3p_y+1)}$ is located at the center, are called $k$-cells and denoted $C_{k,\bm{p}}$.
\end{itemize}
We also define a $k$-link to be the set of locations on a side of a $k$-cell. The collection of vertices at the corners of all 0-cells of the lattice is denoted $\Sigma_0$, and the collection of vertices at the corners of all $k$-cells of the lattice is denoted $\Sigma_k$. We will refer to the lattice formed by vertices in $\Sigma_k$ and the collection of all $k$-links as the $k$-frame.
\end{definition}

\begin{definition}[Coarse-graining]
A coarse-grained state at level $k$ is a state whose nontrivial syndromes lie entirely in $\Sigma_k$.  We also call such a state a ``level-$k$ encoded state''.  
A gadget is said to have a coarse-graining action at level $k$ if a noise-free version of it maps a state with nontrivial syndromes in $\Sigma_0$ to a state with nontrivial syndromes entirely in $\Sigma_k$.  
\end{definition}

\subsubsection{A two-dimensional ``majority vote''} \label{ss:twod_maj_vote}

The toric code is a two-dimensional code; however, since error syndromes are point-like, they bear some resemblance to the syndromes of the repetition code, which are also point-like. As such, we would like to construct an error correction gadget that will play a role analogous to that of the error correction gadget $EC_1 \equiv \cZ_0$ featured in Tsirelson's automaton. In Tsirelson's construction, this was achieved by majority vote, but for quantum codes like the toric code, there is no natural `majority vote' operation. %in the toric code.  
We instead need a gadget whose action mimics the action of the majority vote from the perspective of its action on error syndromes.  The desired gadget, which we call $EC_1 \equiv R_0$ ($R$ for `recovery'), is defined through its action on a 1-cell. In particular, we require that it erase single errors that are within its support, defined to be the 1-cell it operates on, together with its boundaries. Otherwise, the remaining syndrome is coarse-grained to the next level. 
% Described in this way, 
This $R_0$ is thus in some sense a natural two-dimensional generalization of $\cZ_0$. 

Let us now construct the $R_0$ gadget with the desired action as a circuit of $\cI_0$, $\cT_0^{h/v}$ and $\cM_0^{h/v}$  gate gadgets. There is no unique choice for $R_0$, but we endeavored to make a choice which is simple, and for which proving fault tolerance is not too complicated.
\begin{definition}[$R_0$]\label{def:R0}
Define the auxiliary gadgets $R_0^{v/h}$ pictorially as
 \begin{equation} \label{R0-definition}
 \includegraphics[width = 0.7 \textwidth]{figs/R0-def}
 \end{equation}
 In the pictures above, the gadget is shown together with its qubit support, and the sequence of gates in $R_0^{v/h}$ acts in six steps displayed from left to right. 
 
In terms of these gadgets, $R_0$ is the following 48-step operation
\begin{equation} \label{R0-actual-definition}
    R_0 = (R_0^h \circ R_0^h \circ R_0^v \circ R_0^v)^2 
\end{equation}

\end{definition}

\begin{remark} \label{remark:doubling}
The alternative definition
\begin{equation} \label{R0-short-definition}
    R_0' = R_0^h \circ R_0^h \circ R_0^v \circ R_0^v 
\end{equation}
which is simply half of $R_0$, would also be a suitable $\EC$ gadget. However, using the doubled version of $R'$ greatly simplifies the proof of the $\EC$ A condition for the toric code automaton (Prop.~\ref{proposition:gate-ec-prime-TC}). This is the only place where the difference between $R'$ and $R$ will be important: apart from the proof of the $\EC$ A condition, 
% (Prop.~\ref{proposition:gate-ec-prime-TC})
all properties shown in this paper for $R$-type gadgets also hold for $R'$-type gadgets. 
\end{remark}

Because of the overlap of the support of neighboring gadgets, for composite gadgets, we will adopt a new definition of their action that reflects the fact that the gadgets always tile the entire lattice:

\begin{definition}[Action of a non-elementary gadget] \label{def:action-tiling}
    Consider the \hi{automaton at time step $t$, which can be represented as a set of gadgets spatially covering the entire lattice}.  The action of gadget $G$ at time $t$ is the operation of the entire \hi{automaton at time $t$} restricted to the support of $G$.
\end{definition}
This definition differs from simply performing the operation of the elementary gadgets comprising the non-elementary gadget; in particular, we also include the effect of the neighboring gadgets on the boundaries of the non-elementary gadget.  For example, the operation of a \emph{single} composite gadget (say $R_0$) is normally defined so that no level-0 gadgets act on its north and east boundaries. This is done such that we can tile the entire lattice with gadgets in a well-defined way.  However, when analyzing properties of gadgets (which will be done extensively in Sec.~\ref{sec:FT_proof}), we always assume that these gadgets are part of a layer of gadgets \hi{spatially covering} the lattice, and thus, we cannot ignore the level-0 gadgets in the composite gadget's north and east boundaries.  Thus, we will always be assuming operation of level-0 gadgets on all boundaries of higher-level gadgets when analyzing their properties.

The action of $R_0$ cleans up a single 0-link input error and has a level-1 coarse-graining action otherwise. Let us illustrate this with a few input-output configurations (computed assuming no errors occur during the action of $R_0$):
 \begin{equation} \label{R0-action}
 \includegraphics[width = 0.85 \textwidth]{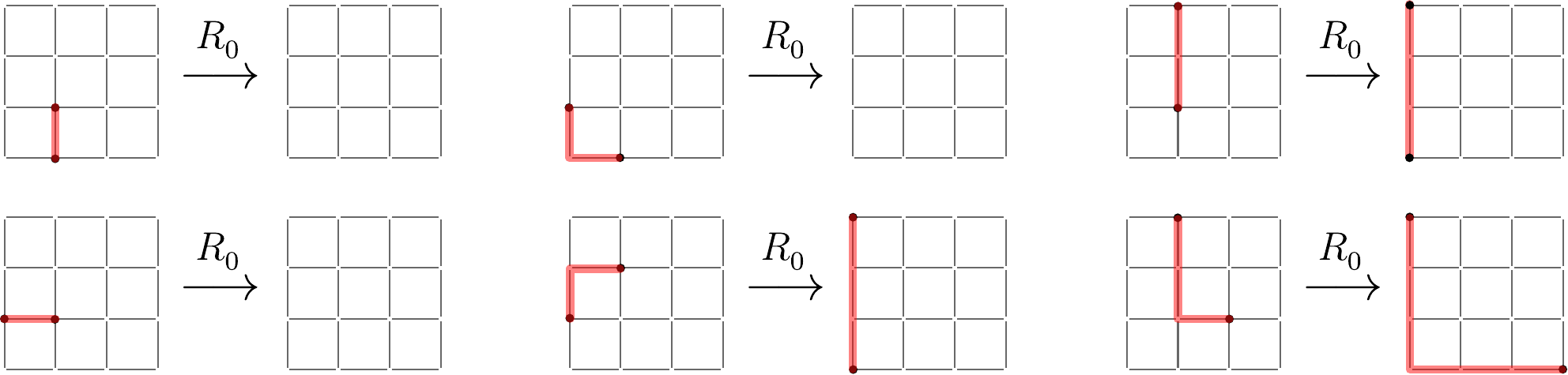}
 \end{equation}
To summarize, any single input 0-link error is removed, as shown in two examples in the first column. The second column shows two examples of `corner errors' (this is an error caused by two $X$ errors that form a right angle so that the syndromes lie on diagonally opposing corners of a 0-cell). These are either completely erased or turn into errors on 1-links. Finally, the third column shows two examples of larger sized input errors which get coarse-grained to the nearest sites in $\Sigma_1$.

Let us provide an informal explanation of how the $R_0$ gadget operates. In our explanation we will instead consider the shorter gadget $R_0'$ as opposed to $R_0$, ignoring the artificial doubling following Remark~\ref{remark:doubling}.  If we ignore steps 3 and 6,  $R_0^v$ is a set of disjoint vertical majority votes on three columns. A direct application of such a one-dimensional majority vote would not behave well in two dimensions: it would turn some of the corner errors into corner errors with syndromes on $\Sigma_1$. This is undesirable because a self-similar gadget operating like this would indefinitely grow a single corner error.  Steps 3 and 6---rounds of $\cM_0$ gadgets oriented perpendicularly to the preceding rounds of $\cT_0$---are introduced to treat specifically this issue.

As an explicit example, let us consider a version of $R_0$ where the additional rounds of $\cM_0$ are skipped, and run this on a corner-shaped input error:
 \begin{equation*} 
 \includegraphics[width = 0.85 \textwidth]{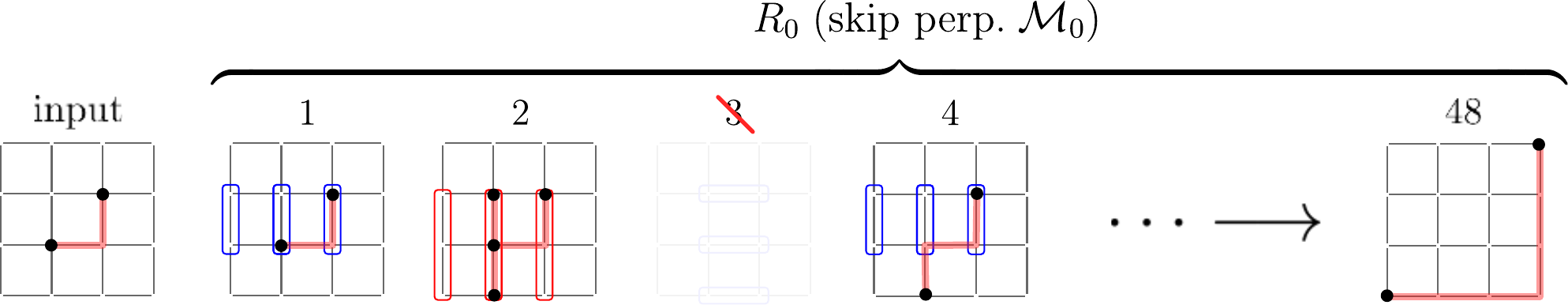}
 \end{equation*}
where we show the syndromes of the output state as well as the total operator acting on the state at the end of each indicated time step (eliminating closed loops as they correspond to stabilizers).  At step 4, the syndromes are already too far separated to be corrected afterwards. Let us consider the action of the actual $R_0$ gadget on the same input state to see the effect of $\cM_0^h$ in step 3:
 \begin{equation*} 
 \includegraphics[width = 0.85 \textwidth]{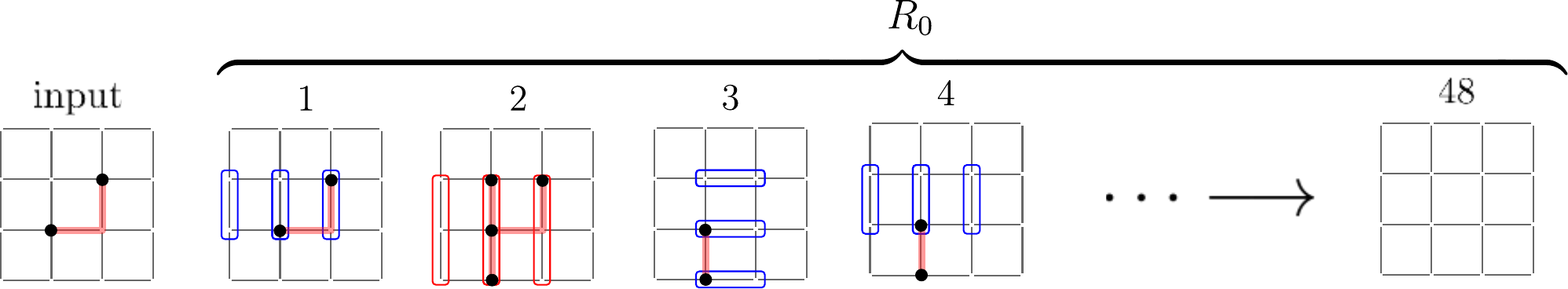}
 \end{equation*}
The perpendicular applications of $\cM$ bring the syndromes close to one another by step 4, allowing them to be eventually annihilated.  One can gain further intuition for the operation of $R_0$ by examining its action on different input syndrome configurations. Two repetitions of $R_0^v$ (and similarly, $R_0^h$) are needed to properly coarse-grain errors in two dimensions. For example, a corner error will undergo some nontrivial change under the first $R_0^v$ but the syndromes will not be fully coarse-grained until the second $R_0^v$ has been applied. 

\begin{remark} \label{R0_search_remark}
An exhaustive numerical search~\cite{code} reveals that $R_0$ acting on a syndrome-free input state with a single link error anywhere in its spacetime support, followed by a noise-free application of $R_0$, yields a syndrome-free state. The same property holds for the shorter version $R_0'$. 
\end{remark}

Finally, let us remark that, unlike in Tsirelson's automaton, the spatial supports of the neighboring  $R_0$ gadgets \hi{covering} the entire lattice overlap on the boundaries of the 1-cells\footnote{In spacetime, the supports of their constituent level-0 gadgets do not overlap -- this is why $R_0$ gadgets can tile the entire space and do not require additional staggering in time.}, which is impossible to avoid. This makes the analysis more difficult than in the one-dimensional case, where neighboring majority votes had disjoint support. As noted in Def.~\ref{def:composition}, we will still use the $\otimes$ symbol to compose overlapping gadgets in space (at a fixed point in time) so long as their operation can be parallelized.

\subsubsection{Fault-tolerant simulation: inner simulation} \label{subsec:inner-simulation}

We will now define the ($X$ part of the) fault-tolerant simulation for the toric code automaton.  It will be based on the error-correcting gadget $R_0$ and fault-tolerant versions of the gate gadgets. However, the fault tolerant prescription will not be as simple as the one discussed in Sec.~\ref{sec:tsirelson}.  Instead, the simulation is divided into two parts.  First, from the level-0 error correction and gate gadgets, we construct error correction and gate gadgets at some fixed level $k$ using a prescription we call the \emph{inner} simulation. The  EC and gate gadgets obtained from the procedure will be shown to satisfy certain EC and gate properties that we will discuss in later sections. We then use these gadgets as primitives for an \emph{outer} simulation (the language is borrowed from \emph{inner} and \emph{outer} codes from code concatenation), which is a conventional fault-tolerant prescription for the circuit, corresponding to a generalization of the procedure in Sec.~\ref{sec:tsirelson}. The inner simulation scheme is used because we need to go $k$ levels up at a time to suppress the error rate from $p$ to $(Ap)^2$, while the outer simulation scheme is used to provide further error suppression similarly to standard circuit concatenation\footnote{It is quite likely that we do not need two separate types of simulations and can get away with just using the outer simulation at all levels.  This would clearly be much more elegant. Then, one would simply perform $k$ iterations of the outer simulation to achieve error suppression $p \to Ap^2$ and continue to perform outer simulation to further amplify the error suppression. Unfortunately, by the time we realized this, circa 50 pages of existing proofs that were tailored to the inner and outer simulation prescription would need to be rewritten. Thus, we decided to leave this modification as an exercise for the interested reader.}.

 For the rest of the paper, we will refer to $\cI$, $\cT$, and $\cM$-type gadgets as gate gadgets or simply gates (not to be confused with unitary gates). We also have conditions similar to the Gate A and Gate B conditions in mind when designing these gadgets: this will be discussed in detail in Sec.~\ref{sec:FT_proof}.

We start by discussing inner simulation. First, we define $\cI_1$, $\cT_1^{v/h}$ and $\cM_1^{v/h}$. 
$\cI_1$ is defined simply as a repetition of $\cI_0$ gadgets layered over a single 1-cell for 5 time steps. $\cT_1^{v/h}$ and $\cM_1^{v/h}$ are defined such that their action restricted to configurations on $k$-frame is analogous to that of  $\cT_0^{v/h}$ and $\cM_0^{v/h}$, respectively (see below).

\begin{definition}\label{def:T1}
     $\cT_1^h$ is defined to operate in the same way as 3 parallel rows of Tsirelson's $\cT_1$ applied horizontally (c.f. Def.~\ref{def:tsirelsongadgets}). 
     Pictorially, it is:
 \begin{equation}
 \includegraphics[width = \textwidth]{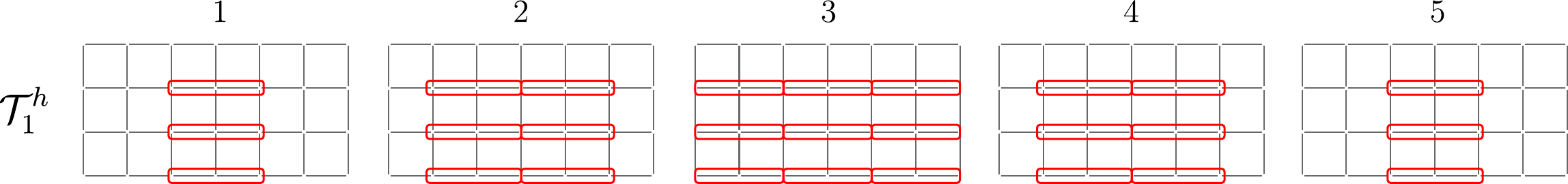}
 \end{equation}
 Similarly, $\cT_1^v$ is defined to be three parallel columns of Tsirelson's $\cT_1$ applied vertically\footnote{One may observe that since coarse-grained syndrome configurations only lie on the corners of the 1-cell, it is unclear whether $\cT_0$'s acting in the bulk of the gadget is necessary. However, this is in fact important because we want the operation of $\mathcal T_1^{h}$ to be unaffected in the presence of single-qubit transient errors. If this error acts perpendicularly to the action of $\cT_0$'s, it has support in the bulk and can result in large output damage due to the operation of $\mathcal T_1^{h}$. To avoid this, it turns out the additional $\cT_0$'s are required in the bulk.}. In the pictures above, the gadget is shown together with its qubit support.
\end{definition}

Next, we define the $\cM_1$ gate. 

\begin{definition} \label{def:M1}
     $\cM_1^h$ is defined to be 3 parallel rows of Tsirelson's $\cM_1$ (Eq.~\eqref{eq:M1-ts}) applied horizontally. Similarly, $\cM_1^v$ is defined to be three parallel columns of Tsirelson's $\cM_1$ applied vertically. The qubit support of this gadget is defined in Def.~\ref{def:stab-quant-gadg}.
\end{definition}

Having introduced $\cI_1$, $\cT_1^{v/h}$ and $\cM_1^{v/h}$, we now define the inner simulation, which shows how to construct higher-level gate gadgets $\cI_n$,  $\cT_n^{v/h}$ and $\cM_n^{v/h}$, as well as error-correcting gadgets $\EC_n$ at different levels for $n \leq k$. We will call $k$ the level of inner simulation. 

\hi{Moving forward,  we will introduce the notation $C[G^{(1)},G^{(2)},...]$ as a shorthand where $C[a,b,c]$ is some fixed circuit composed of basis operations $\{a,b,c..\}$, and $C[G^{(1)},G^{(2)},...]$ is obtained by substituting $\{ a,b,c..\}$ with gadgets $G^{(1)},G^{(2)},...$ \,.     A circuit $C[\cG^{(1)}_k, \cG^{(2)}_k, \cG^{(3)}_k, ...]$ composed only from gadgets at level $k$ will often be shorthanded as $G[\cG_k]$. }

For the definition below, the relevant circuits will be the level-1 gadgets defined in Defs.~\ref{def:T1} and \ref{def:M1} via elementary gadgets, namely $\cI_1 = \cI_1[\cI_0,\cT_0^{h/v},\cM_0^{h/v}]$, $\cT_1^{h/v} = \cT_1^{h/v}[\cI_0,\cT_0^{h/v},\cM_0^{h/v}]$, $\cM_1^{h/v} = \cM_1^{h/v}[\cI_0,\cT_0^{h/v},\cM_0^{h/v}]$, and the error correction gadget $R_0 = R_0[\cI_0,\cT_0^{h/v},\cM_0^{h/v}]$.

\begin{definition}[Inner fault tolerant simulation]\label{def:FT-constr-TC-inner}
    Define the level-1 error correction gadget to be 
    \begin{equation}
        \EC_1 =  R_0[\cI_0,\cT_0^{h/v},\cM_0^{h/v}].
    \end{equation}
    %
    %Denoting gate gadgets as $\cG$, we may write
    %
    %\begin{equation}
    %    \EC_1 = R_0[\cG_0], \quad \cG_1 = \cG_1[\cG_0],
    %\end{equation}
    %
    We recursively define the gate and EC gadgets at level $n \leq k$ to be
    %
    %\begin{equation}
    %    \EC_n =  R_0[\EC_{n-1} \circ \cG_{n-1} \circ  \EC_{n-1}]
    %\end{equation}   
     \begin{equation} \label{eq:inner-3}
        \EC_n =  R_0[\EC_{n-1}^{\otimes K_{\cI}} \circ \cI_{n-1},\EC_{n-1}^{\otimes K_{\cT}} \circ \cT_{n-1},\EC_{n-1}^{\otimes K_{\cM}} \circ \cM_{n-1}] \circ \EC_{n-1}^{\otimes 9} 
    \end{equation}
    where $K_{\cG}$ is the number of $0$-cells in the support of the level-0 version of the associated gadget $\cG_{n-1}$.  We can write this in shorthand as 
     \begin{equation} \label{eq:inner-1}
        \EC_n =  R_0[\EC_{n-1}^{\otimes K_{\cG'}} \circ \cG'_{n-1}] \circ \EC_{n-1}^{\otimes 9}, 
    \end{equation}
    For the simulated gadgets:
    %\begin{equation}
    %    \cG_n =  \cG_1[\EC_{n-1} \circ \cG_{n-1} \circ  \EC_{n-1}].
    %\end{equation} 
    \begin{equation} \label{eq:inner-2}
        \cG_n =  \cG_1[\EC_{n-1}^{\otimes K_{\cG'}} \circ \cG_{n-1}'] \circ  \EC_{n-1}^{\otimes 9 K_{\cG}}.
    \end{equation}  

    One can interpret these definitions as a substitution rule where gadgets at level $n$ are formed by constructing the sequences $R_0[\cdot]$ and $\cG_1[\cdot]$ using level $n-1$ gate gadgets (instead of level $0$ gadgets), and inserting level $n-1$ EC gadgets in between consecutive level $n-1$ gate gadgets.
\end{definition}

We remark that the inner simulation is slightly different from the conventional fault tolerant simulation  $FT(\cdot)$ defined in Sec.~\ref{sec:tsirelson}. 
In particular, the recursive construction of gadgets (Eqs.~\eqref{eq:inner-1} and \eqref{eq:inner-2}) used in the inner simulation is slightly different. The inner simulation and the conventional fault tolerant simulation are identical if the inner simulation is comprised of precisely one level (i.e. $n = 1$); however, when one compares $FT^{n-1}(\cG_1)$ and $\cG_n$ (the latter corresponding to $n$ levels of inner simulation), as well as  $FT^{n-1}(\EC_1)$ and $\EC_n$ for $n > 1$, the gadgets obtained using the $FT(\cdot)$ prescription and the inner simulation differ by the inner simulation's additional repetitions of $\EC_{i}$ ($i \leq n-2$) gadgets in between certain gates.  With these additional repetitions, the proofs of certain properties are slightly simpler than they would be otherwise.

\begin{definition}[Supports of level-$n$ gadgets: toric code]\label{def:suppsTC}
Define the qubit supports of level-0 gadgets $\mathrm{supp}_q(G_0)$ to be the qubits whose relative position is shown in the following figure:
\begin{equation} \label{eq:supp-TC}
\includegraphics[width = 0.65\textwidth]{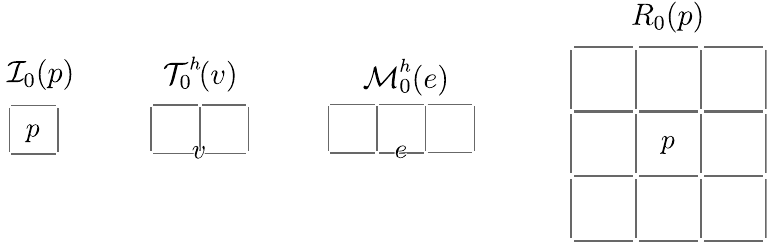}
\end{equation}
The supports of $\cT_0^v(v)$ and $\cM_0^v(e)$ can be obtained by a $90^\circ$ rotation of that of $\cT_0^h(v)$ and $\cM_0^h(e)$, respectively.

We also define the vertex support $\mathrm{supp}_v(G_0)$ of the gadget to be the set of all vertices appearing in the figures above for gadget $G_0$.   

The qubit (vertex) supports of a level-$n$ gate  $\mathrm{supp}_{q(v)}(\cG_n)$ and a level-$n$  error-correction gadget $\mathrm{supp}_{q(v)}(\EC_n)$ are the set of all the level-0 links (vertices) contained in the qubit (vertex) support of all level-0 gadgets that comprise the level-$n$ gadget. 
\end{definition}

Notice that, in this context, when we talk about gadget support we mean spatial support, i.e. set of locations in space at a fixed point in time. 

Our definition of gadget support is counter-intuitive to some extent as it is tailored to correctly capture the effect of noise on the simulated gadgets. To explain this, consider the definition of level-0 gadget support.  First, the qubits that enter the stabilizer measurements required for the gadget do not appear in the support. This is because we assume that, while the measurement outcome can be erroneous, it does not explicitly apply errors on the qubits involved in the measurement. Next, the support includes qubits that the gadget applies the feedback to as well as additional qubits. This is because the action of the gadget at some higher level $k$ involves qubits on more $k$-cells than $0$-cells in the corresponding level-$0$ gadget. Thus, when `renormalized' to level 0, the possible errors that a level-$k$ gadget failure can induce necessarily fit in a larger qubit support.

In the automaton, the elementary gadgets always tile the plane, as we have already discussed in Def.~\ref{def:action-tiling}. The non-elementary gadgets of the inner simulation are defined in a way that the action of neighboring gadgets can be parallelized, and therefore we can also introduce layers of non-elementary gadgets:
\begin{definition}[Layer of gadgets]
We define $\cL(\cG_n)$ to be a \emph{layer} of level-$n$ gadgets in the automaton corresponding to \hi{covering} the space with the level-$n$ gates $\cG_n$ 
% (possibly together with $\cI_n$ identity gates) 
or error correction gadgets $\EC_n$. 
\end{definition}
We note that the actions of neighboring gadgets can be parallelized and do not require staggering in time.  We will often be abusing notation, when, for example, part of the plane is tiled with $\cT^{h/v}_n$ gates and another is tiled with $\cI_n$ gates, in which case the notation for such a layer could be $\cL(\cT_n^{h/v})$ or $\cL([\cT_n^{h/v}, \cI_n])$ (and similarly for $\cT$-gate layers); the former notation will be adopted whenever this will not cause ambiguity.

We also define a restriction of a syndrome configuration to a region, which will be used in later sections:

\begin{definition}[Restriction]\label{def:restriction}
    For $\sigma \in \mathbb{F}_2^{|\cL_0|}$, the restriction of $\sigma$ to the spatial region $V$, denoted $\sigma|_{V} \in \mathbb{F}_2^{|\cL_0|}$ is defined such that
    \begin{equation}
    \sigma|_{V}(\bm{x}) = \begin{cases}
    \sigma(\bm{x}), & \bm{x} \in V, \\
    0, & \bm{x} \notin V.
    \end{cases}
    \end{equation}

    For $\varepsilon \in \mathbb{F}_2^{|\cL_0|} \times [0,T]$, the restriction of $\varepsilon$ to the region $V$ over the course of time $t \in [0,T]$, denoted $\varepsilon|_{V}$ is defined such that
    \begin{equation}
  \varepsilon|_{V}(\bm{x}, t) = \begin{cases}
    \varepsilon(\bm{x}, t), & (\bm{x},t) \in V \times [0,T], \\
    0, & (\bm{x},t) \notin V \times [0,T].
    \end{cases}
    \end{equation}

    For a layer of gadgets $\cL(G)[\sigma, \varepsilon]$ (where each gadget $G$ is a circuit of elementary gadgets of depth $T$)   with input syndrome $\sigma$ and noise realization $\varepsilon$, the restriction of the layer of gadgets to the region $V$ over the course of its operation, denoted $\cL(G; V)[\sigma, \varepsilon]$ is defined to be the output syndrome of a modified operation of gadgets where the input syndrome is restricted to $V$, and at each time step $t \in [0,T]$, all constituent elementary gadgets are applied, the operator to the corresponding noise realization is applied (without restricting to $V$), and all syndromes outside of $V$ are set to $0$. Namely, the entire operation of the layer of gadgets is restricted to the spatial support given by $V$, and not just the output. Note that by definition, $\cL(G; V)[\sigma, \varepsilon]\big|_{V} =\cL(G; V)[\sigma, \varepsilon]$.
    Similarly, we can define a restriction of an operation of a single gadget $G[\sigma,\varepsilon]$ to $V$, which we denote as $(G; V)[\sigma, \varepsilon]$.
\end{definition}

The following Lemma establishes an important property of level-$n$ gadgets of the inner simulation:
\begin{lemma} \label{lemma:EC-coarse-grains}
    Acting on an arbitrary input state, the first half of the noiseless gadget $\EC_n$ (understood in terms of the action of a layer of these gadgets restricted to the support of a single $\EC_n$, see Def.~\ref{def:action-tiling}) will coarse-grain the input so that all non-trivial syndromes have support in $\Sigma_n$. 

    Furthermore, the noiseless gadgets $\cM^{h/v}_n$ and $\cT^{h/v}_n$, restricted to their support, map coarse-grained states at level $n$ to coarse-grained states at level $n$ so that scaled-down (by factor $3^n$) versions of the syndromes effectively undergo the action of $\cM^{h/v}_0$ and $\cT^{h/v}_0$.    
    % Furthermore, the noiseless gadgets $\cM^{h/v}_n$ and $\cT^{h/v}_n$ map coarse-grained states at level $n$ to coarse-grained states at level $n$ in a way whose action corresponds to that of $\cM^{h/v}_0$ and $\cT^{h/v}_0$ (in terms of action of a tiling restricted to a support of a single elementary gadget).
\end{lemma}
\begin{proof}
    We provide the proof of the statement for $n > 1$ in Subsec.~\ref{sec:confinement} after introducing additional tools that will make the proof much simpler.  The full proof will be inductive, but here we will prove the base case, corresponding to $n = 1$.  
    
    For the base case, we study the first half of the EC gadget, i.e. $(\EC_1)_{1/2}= R_0 '= R_0^h \circ R_0^h \circ R_0^v \circ R_0^v$. 
    First, let us show that after two repetitions of $R_0^v$, all the syndromes will be moved to the top and bottom 1-links. 
    We verify that the syndromes on the left and right 1-links will be coarse-grained to their corners already after one application of $R_0^v$ since the operation of the gadget on these two links is identical to Tsirelson's $Z_0$ and the syndromes away from these links will not add or remove any additional syndromes on these links during the operation of $R_0^v$.  In addition, the syndromes that are already on the top and bottom 1-links, will either remain there or be cleaned up and will not affect the rest of the gadget.  A brute-force computer search then shows that input syndromes on the middle 4 vertices of the gadget will be supported on the top and bottom 1-links after two rounds of $R_0^v$. 
    
    Next, two rounds of $R_0^h$ are applied to an input state whose syndromes are all on top and bottom 1-links. However, the operation of $R_0^h$ on those 1-links is identical to Tsirelson's $Z_0$ and thus all syndromes will be coarse-grained to $\Sigma_1$.  Moreover, if the input syndrome configuration has already been coarse-grained to lie in $\Sigma_1$, the output syndrome configuration will be the same.  

    In addition, we verify through a brute-force computer search that $\cT_1^{h/v}$ or $\cM_1^{h/v}$ acting on a syndrome configuration with nontrivial syndromes in $\Sigma_1$ will also output a syndrome configuration with nontrivial syndromes in $\Sigma_1$, with their action corresponding to the respective action of $\cT_0^{h/v}$ and $\cM_0^{h/v}$.  This proves the base case.
\end{proof}

In addition, the first half of the $\EC_n$ gadget self-similarly reproduces the behavior of the $R_0'$ gadget albeit at a proportionally larger scale (\hi{for the definition and discussion of $R_0'$}, see the discussion following Def.~\ref{def:R0}). This means that $\EC_n$ does not simply bring the syndromes to the locations in $\Sigma_n$, but also, for small enough deviations of the input syndrome configuration from $\Sigma_n$, will bring the configuration to the closest coarse-grained one.

\subsubsection{Fault-tolerant simulation: outer simulation}
\label{subsec:outer-simulation}

Having defined the inner simulation and fixing some level $k$ for it, we can now finally define the outer simulation, associated with the fault tolerant prescription $\widetilde{FT}(\cdot)$, where the tilde is used to emphasize that this prescription, unlike the $FT(\cdot)$ prescription from Sec.~\ref{sec:tsirelson}, relies on the gadgets obtained from the inner simulation.

\begin{definition}[Outer fault tolerant simulation] \label{def:newFT_\EC}
    Given a circuit $\cC \equiv \mathcal C[\cI_0, \cT_0^{h/v}, \cM_0^{h/v}]$, we use the shorthand $\cC[\cG_0]$ similarly to Def.~\ref{def:FT-constr-TC-inner}. Assuming the inner simulation level is $k$ and using it to define the corresponding gate and EC gadgets $\cG_k$ and $\EC_k$,  the outer fault-tolerant simulation of $\cC$ is defined as 
    \begin{equation}
        \widetilde{FT}(\mathcal C) = \cC[\widetilde{\EC}^{\otimes K_\cG} \circ \cG_k ] \circ \widetilde{\EC}^{\otimes K}
    \end{equation}
    where $\widetilde{\EC} = \EC_k$ and $K$ is the number of $0$-cells in the support of the circuit $\cC$. This definition can be understood as the original circuit with the level-0 gate gadgets replaced with the level-$k$ gate gadgets from the inner simulation, and $\EC_k$ inserted between consecutive level-$k$ gate gadgets. 

    A repeated fault tolerant simulation of a circuit $\cC$, denoted $\widetilde{FT}^n(\cC)$, is a new circuit formed from the iteration $\widetilde{FT}^n(\cC) = \widetilde{FT}(\widetilde{FT}^{n-1}(\cC))$, with $\widetilde{FT}^0(\cC) = \cC$ and $\widetilde{FT}^1(\cC) = \widetilde{FT}(\cC)$.
\end{definition}

Now, we can finally define the toric code automaton based on the inner and outer fault tolerant simulations:

\begin{definition}[2D toric code automaton]
Given an inner simulation level $k$ and a positive integer $T$, the toric code automaton is a $n$-fold outer simulation $\widetilde{FT}^{n}(\cI_0^T)$ (with the simulated circuit being $\cI_0^T$) that operates on an $L \times L$ torus where $L = 3^{n k}$.
\end{definition}

By construction, the toric code automaton is a local measurement and feedback automaton.  Ignoring the difference in the choice of error correction and gate gadgets, the outer simulation is otherwise analogous to the usual fault tolerant simulation in Sec.~\ref{sec:tsirelson}. We note that if $k = 1$, $\widetilde{FT}$ becomes the usual fault tolerant prescription of Defs.~\ref{def:FT_original} and \ref{def:FT_original_repeated}. 

We also emphasize that the inner and outer fault tolerant simulation procedures are slightly different: the outer simulation is built similarly to the inner one, except that by design, there are no extra repetitions of $\widetilde{\EC}$ gadgets at any level of the simulation. For the inner simulation, we instead assume fixed circuits for gadgets $\cG_k$, $\EC_k$ that are then composed together to construct higher-level circuits, which incurs extra repetition of lower-level $\EC$ gadgets. The fixed-circuit convention is convenient when proving some of the properties of these gadgets.  As an illustration, consider the inner simulation circuit for $\cI_2$ preceded and followed by $\EC_2$ (because $\cI_2$ will always appear flanked by $\EC_2$ for any higher-level inner simulation):
\begin{equation} \label{eq:inner-example}
         \EC_2 \circ \cI _2 \circ \EC_2 
          =   \EC_2 \circ (\EC_1^{\otimes 9} \circ \cI_1^{\otimes 9} )^5  \circ \EC_1^{\otimes 9} \circ \EC_2 
\end{equation}
compared to a $k = 1$ outer simulation (i.e. the usual fault-tolerant simulation) repeated twice (where $\EC \equiv \EC_1$):
\begin{equation}\label{eq:outer-example}
\begin{split}
     FT^2(\cI_0) &= FT \left (  \EC\circ \cI_1 \circ \EC\right) \\
     &=  \EC [EC^{K_{\cG}} \circ \cG_1] \circ (\EC^{\otimes 9} \circ \cI_1^{\otimes 9} )^5 \circ \EC [EC^{K_{\cG}} \circ \cG_1] \circ \EC^{\otimes 9} .
\end{split}
\end{equation}
These circuits are only partially expanded for brevity. However, they already illustrate the following point: the circuit for the inner simulation shown in Eq.~\eqref{eq:inner-example} is identical to the circuit for the usual fault-tolerant simulation in Eq.~\eqref{eq:outer-example} except that some of the layers of $\EC_1$ gadgets are now repeated. This is seen, for example, by recalling that $\EC_2$ must also start and end with $\EC_1$ layers and it is additionally either preceded or followed with $\EC_1$ in Eq.~\eqref{eq:inner-example}. At higher levels of simulation, this repetition occurs at higher levels as well.

As discussed in an earlier footnote, we expect that it should be enough to introduce the outer simulation with $k = 1$ (which is the standard fault tolerant simulation prescription) which is then repeated $nk$ times as opposed to $n$ times to account for needing $k$ iterations to obtain $p \to Ap^2$ error suppression.  However, in this paper, we will be using the inner and outer simulation formalism.

\subsubsection{Ideal decoder}
\label{subsec:ideal-decoder-TC}

In this subsection, we will define the ($X$-type) ideal decoder used for the toric code automaton.  Let us first define an ideal decoder that works for a usual fault tolerant prescription $FT(\cdot)$ for the toric code (i.e. $k = 1$), which already requires nontrivial generalization in comparison to concatenated codes. We then generalize this definition upon incorporating inner simulation with $k > 1$.  First, we need to discuss a map that performs a coarse-graining action, which we will refer to as a \emph{pushforward map}.  We note that the pushforward map and the ideal decoder are not actual operations that appear in the automaton, but rather are tools that will help us carry out the fault tolerance proof.

\begin{definition}[Pushforward map] \label{def:pushforward}
    Assume that we are given a toric code state in $\cH_{3^N \times 3^N}$ that can be written as $\cO_X \ket{\psi}$, where $\ket{\psi}$ is a syndrome-free initial logical state of the toric code on a lattice of size $3^N \times 3^N$, and $\cO_X$ is a tensor product of Pauli $X$ operators such that the syndromes $\sigma$ of  $\cO_X \ket{\psi}$ are supported on $\Sigma_1^{3^{N} \times 3^{N}}$ (i.e. are coarse-grained to level 1; the superscript ${3^{N} \times 3^{N}}$ indicates the size of the lattice).  A valid pushforward map $\Phi: \cH_{3^N \times 3^N} \rightarrow \cH_{3^{N-1} \times 3^{N-1}}$ is one which outputs a new toric code state $\cO_X' \ket{\psi'}$ such that 
    \begin{enumerate}
        \item $\ket{\psi'} $ is a syndrome-free logical state of the toric code on a lattice of size  $3^{N-1} \times 3^{N-1}$ that is the same logical state as $\ket{\psi}$ (expressed as a superposition in a chosen logical basis, say of logical $\overline Z$ operators, of a toric code on the appropriate-sized lattice);
        \item Syndromes of $\cO_X'$, which we denote by $\sigma' \in \Sigma_0^{3^{N-1}\times 3^{N-1}}$,  correspond in location with $\sigma$;
        \item The operator $\cO_X'$ is chosen such that the operator  $\varphi(\cO_X')  \cO_X$ is an element of the stabilizer group, where $\varphi$ denotes a natural pullback from lattice of size $3^{N-1} \times 3^{N-1}$ to the 1-frame of a  $3^{N} \times 3^{N}$ lattice. The natural pullback, for example, maps the operator $X$ acting on a 0-link of the $3^{N-1} \times 3^{N-1}$ lattice to $X \otimes X \otimes X$ acting along the associated 1-link of the $3^{N} \times 3^{N}$ lattice.  
    \end{enumerate}
\end{definition}
Informally speaking, the pushforward map implements a coarse-graining procedure: assuming a toric code state with nontrivial syndromes on $\Sigma_1$, we shrink the lattice by a factor of 3 but maintain the same pattern of syndromes (now on $\Sigma_0$), while making sure that no non-trivial logical operation was applied during this process (see point 3 above).  An example of a configuration of syndromes and the corresponding action of the pushforward map is shown below:
\begin{equation} 
\includegraphics[width = 0.6\textwidth]{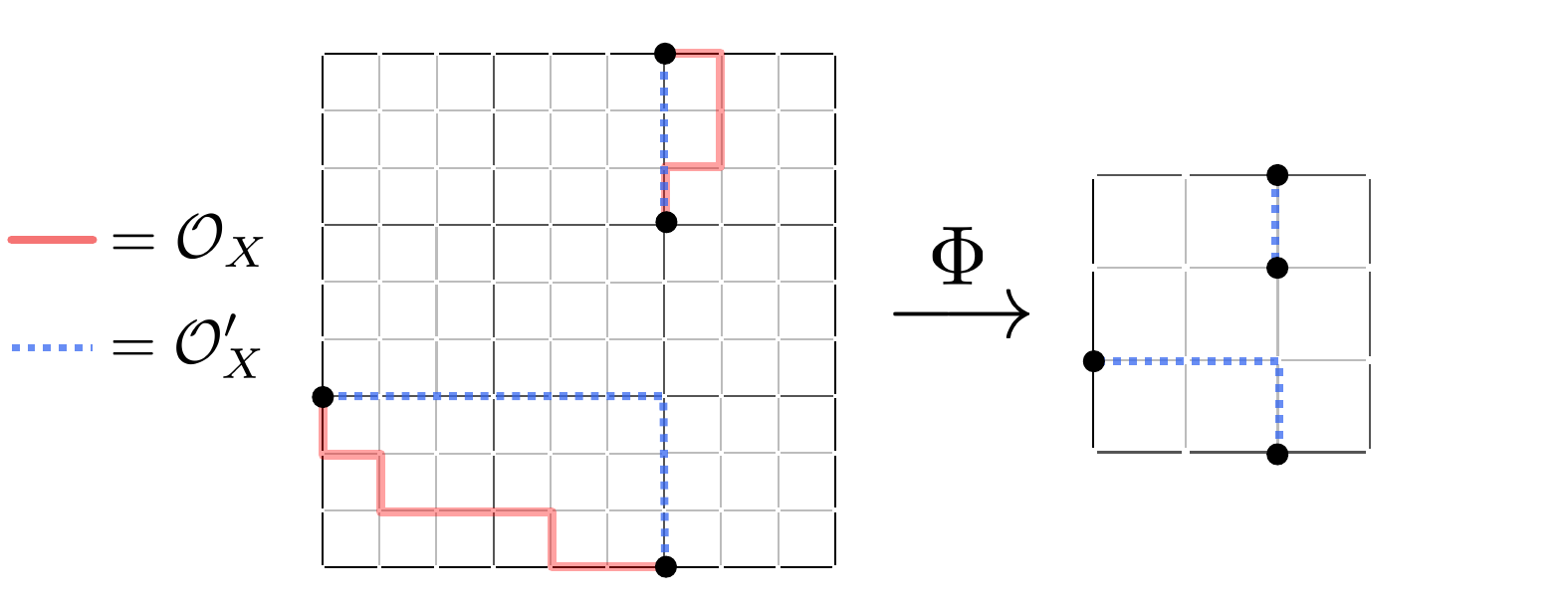}
\end{equation}

We first define an ideal decoder $D_1$ that would work for a conventional  prescription for fault tolerant simulation. From it, we define a level-$k$ decoder, which we will show to be appropriate for our modified prescription based on inner and outer simulations.

\begin{definition}[Ideal decoder for $k = 1$ inner simulation] 
 Assume a noisy toric code state  $\ket{\psi_{3^N \times 3^N}}$ defined on a lattice with linear dimension $3^N$  whose $X$-error syndromes we denote as $\sigma \in \Sigma_0^{3^N \times 3^N}$.  A level-1 ideal decoder $D_1^{(N)}: \cH_{3^N \times 3^N} \rightarrow \cH_{3^{N-1} \times 3^{N-1}}$  is a two-step map consisting of:
 \begin{enumerate}
     \item Application of noiseless $\cL(\EC_1) \equiv \EC_1^{\otimes 3^{N-1}\times 3^{N-1}}$ gadgets to the input toric code state $\ket{\psi_{3^N \times 3^N}}$ producing a new state  $\cL(\EC_1)\ket{\psi_{3^N \times 3^N}} $ with syndromes $\sigma_1 \in \Sigma_1^{3^N \times 3^N}$ (error correction), and
     \item A pushforward map $\Phi$ (Def.~\ref{def:pushforward}) applied to $\cL(\EC_1) \ket{\psi_{3^N \times 3^N}} $ (coarse-graining).
 \end{enumerate} 
  An ideal level-$n$ decoder is defined to be a hierarchical application of $n$ levels of ideal decoder $D_1$, i.e. 
 \begin{equation}
    D_n^{(N)} =  D_1^{(N-n+1)} \circ ... \circ D_1^{(N)}.
 \end{equation}
 
 Moving forward, we will keep the superscript $(N)$ implicit to indicate that the ideal decoder tiles the entire toric code lattice. We will also use the following notation to indicate the two-step action of the ideal decoder:
 \begin{equation}
     D_1 = \Phi \cL(\EC_1)
 \end{equation}
\end{definition}

More informally, the first step of the ideal decoder cleans small enough errors fitting in the support of $R_0$ and otherwise coarse-grains the error configuration to have syndrome support in $\Sigma_1$.  A pushforward map is then applied.  This process is repeated until the lattice is shrunk to some small size.  Therefore, this operation hierarchically eliminates errors at different scales, first eliminating errors at small scales before doing so at larger scales.  As reviewed previously, similar approaches such as renormalization group-based decoders have been proposed for topological codes  (see e.g. Refs.~\cite{bravyi2011analytic,duclos2010fast}), but these constructions assume instantaneous and reliable classical computation, thus suffering from a large amount of non-locality.

Finally, we define the ideal decoder that will be used for the outer simulation.

\begin{definition}[Ideal decoder for the outer simulation for the toric code] \label{def:ideal-decoder-TC}
    Assuming an inner simulation with fixed $k \geq 1$, define the ideal decoder for the outer simulation using a similar procedure to that in Def.~\ref{def:ideal-decoder-TC}.  In particular, the ideal decoder $\widetilde{D}$ is defined to be 
    \begin{equation}
        \widetilde{D} = \Phi^k \cL(\widetilde{EC})
    \end{equation}
    and the associated hierarchical decoder of level $n$ is defined to be a repeated application of the ideal decoder $n$ times:
        \begin{equation}
        \widetilde{D}_n = \widetilde{D}\circ \dots \circ \widetilde{D} = \Phi^k \cL(\widetilde{EC}) \circ \dots \circ \Phi^k \cL(\widetilde{EC})
    \end{equation}
    where the sizes of the tiling and the domain of action of $\Phi$ change appropriately to reflect the change of the size of the lattice as in Def.~\ref{def:ideal-decoder-TC}.
\end{definition}

This is an appropriate decoder for the outer simulation as it is designed to coarse-grain by $k$ levels at a time.
When formulating the generalization of the exRecs method for the toric code (which we do in Sec.~\ref{sec:FT_proof}), we will be pulling levels of hierarchical $\widetilde{D}$ through the system, corresponding to each outer fault tolerant simulation applied to the circuit. 

We once again emphasize that the pushforward is not a physical operation but rather a convenient tool for analyzing the automaton in the presence of noise. In particular, the point of this map is to reduce the size the size of the lattice on which the state is defined  while preserving the logical state and downscaling the operators acting on it. To achieve this for the toric code, the map needs to be global and ``omniscient'': it cannot be expressed through disjoint local operations and requires knowledge of the encoded logical state. Similarly, the ideal decoder should not be viewed as an actual decoder in the canonical sense, but rather as a mathematical tool that first performs an operation of coarse-graining to level $k$ before applying a $k$-fold pushforward map.  These tools are useful for the proofs, as they will allow us to compute the noise model that the automaton experiences as we ``undo'' levels of the simulation.   We note that, in contrast to the toric code, the ideal decoder for concatenated codes can be decomposed into a tensor product of local actions, resulting in a much simpler analysis.

\subsection{Summary of remainder of the paper}

From here on, we will consider a toric code that is initialized in a clean logical state on an $L \times L$ lattice (we will discuss fault-tolerant initialization and readout later in Sec.~\ref{sec:initialization}), and we apply the toric code automaton $T$ times (defined as simultaneous operation of type-$X$ and type-$Z$ circuits $\widetilde{FT}^{n}(\cI_0^T)$ obtained by an outer fault tolerant simulation). The simulation used to construct the circuit makes use of level-$k$ gadgets that are determined by the inner fault tolerant simulation.

In subsequent sections, we will prove the fault tolerance of this automaton in the presence of $p$-bounded gadget noise. In particular, we show that at small enough $p$, the probability of a logical error occurring at at time step $t$ is less than $t$  times an exponentially small quantity in $L^\alpha$ with some constant $\alpha >0$. For this, we will consider the $X$-type automaton separately first and will show that it is fault tolerant in its logical sector (i.e. it protects against $X$-type logical error) under a $p$-bounded decoupled gadget error model. Showing this will constitute the majority of subsequent sections and the majority of the proof as well. Finally, we will show a simple corollary showing that the full quantum toric code automaton is fault tolerant.

To prove the fault tolerance of the $X$-type automaton, we will develop a modification of the exRecs method.  We will then show for some constant $k$ corresponding to the level of the inner simulation, the outer fault tolerant prescription $\widetilde{FT}(\cdot)$ simulates a noisy circuit with $O(p^2)$-bounded ($X$-type) gadget noise, i.e. achieving error suppression. 

The strategy of the proof will be slightly different than the one we used for Tsirelson's automaton. In addition, parts of the proof will be performed by computer-aided search following a very simple algorithm that we explain in detail. The domains of operation of neighboring gadgets overlap, and separate blocks of code do not directly make sense as small codes in the context of the toric code. Therefore, extra care is put into dealing with the effect of overlapping gadget supports.  In addition, the action of each level of the hierarchical ideal decoder is \emph{global} for the toric code (specifically, the pushforward map requires global information about the state to properly operate), which is very dissimilar to Tsirelson's automaton and to concatenated codes. 

% At a conceptual level, 
To determine $k$, we introduce the notion of \emph{gadget nilpotence}, which formalizes the notion that an elementary-level gadget failure inside another gadget $\cG_\ell$ at some level $\ell$ with a coarse-grained input will not disturb the operation of the gadget.  We will provide a computer-aided proof of this property for $\cI, \cT$ and $\EC$ gadgets at level greater or equal to $3$ and for the $\cM$ gadget at level greater or equal to $6$  (which is probably a loose upper bound). Once gadget nilpotence is established, we use this property along with a modification of the exRec method to provide an estimate of $k$ and show that it is $O(1)$ (our very loose estimate gives $k \leq 34$ but in reality the required value is likely much smaller, e.g. 3 or 4 \hi{based on the trends seen in numerics}). Our modification of the exRecs method treats extended rectangles as overlapping in both space and time since neighboring gadgets in the toric code have overlapping support\footnote{This would not be true for concatenated codes}. We also modify filters used in the original exRec method to \emph{$m$-cell filters}, which project onto states with damage operator confined to a $m$-cell,  as opposed to projecting onto states that are simply some distance from a codeword. This bears some resemblance to the sparsity-based proof method of G\'acs~\cite{gacs1983reliable}, which utilizes a decomposition of an error set into burst errors at different scales.  

Moving forward, we first test a simple version of this approach on the measurement and feedback version of Tsirelson's automaton in Sec.~\ref{sec:tsirelson-proof-2}. This introduces some parts of our method on a simpler example before taking on the more complicated task of proving fault tolerance of the toric code automaton, which we do in Sec.~\ref{sec:FT_proof}.  It also proves some important results that will be used in Sec.~\ref{sec:FT_proof}.

%% file: 4.1tsirelson.tex
\section{Fault tolerance of measurement and feedback Tsirelson automaton}
\label{sec:tsirelson-proof-2}

This section serves as a warm-up for the proof of fault tolerance of the toric code automaton. Here, we approach the construction of the automaton for the repetition code in a largely analogous way to that for the toric code. We consider a measurement and feedback automaton for the stabilizer description of the repetition code, and follow an analogous inner and outer fault tolerant simulation prescription as for the toric code.  While the approach is more cumbersome, it will form the necessary skeleton for the proof in the toric code case. In addition, some of the results in this section will also be used in the main proof for the toric code automaton.  

We will use the elementary gadgets and level-1 gadgets presented in Sec.~\ref{modified_tsirelson2}. The basis for the gates is now formed by the elementary gadgets $\cI_0$, $\cT_0$, and $\cM_0$; we also use specifically the composite gadget $Z_0$ shown in Eq.~\eqref{Z0-action}, consisting of 5 steps of application of $\cM_0$ and $\cT_0$-type gadgets.  The inner simulation is defined analogously to Subsec.~\ref{subsec:inner-simulation}, except with $\EC_1 = Z_0[\cI_0, \cT_0, \cM_0]$ and using the gadgets of Tsirelson's automaton as opposed to those of the toric code. The outer simulation $\widetilde{FT}(\cdot)$ is defined as in Subsec.~\ref{subsec:outer-simulation}, with similar changes to adapt it for Tsirelson's automaton, and removing the layer notation in the definition (namely, every occurrence of a layer of gadgets $\cL(G)$ is now simply replaced with a standalone gadget $G$).

We assume some fixed level $k$ of inner simulation, which is used to obtain the gadgets $\cG_n$ and $\cE_n$ with $n \leq k$. We then slightly modify the (standard) Gate and $\EC$ conditions used in the exRecs method to apply it to the outer simulation.
In particular, we will need to introduce a filter that selects out states with damage that is constrained to a bounded spatial region, which we will call a cell filter:
    
\begin{definition}[Cell filter on $k$-cells] \label{def:r_prime}
   An $m$-cell filter is the map that projects onto states \hi{$s \in \mathbb{F}_2^{|\cL|}$} on a $k$-cell that such that there exists a codeword \hi{$c \in \mathbb{F}_2^{|\cL|}$} where non-zero bits of $s \oplus c$ are supported in a single $m$-cell.   
\end{definition}

We also use the same definition (with appropriately removed layer notation) for the ideal decoder as in Def.~\ref{def:ideal-decoder-TC}. In figures, we will use the same symbol for it, except with a label $k$ standing for the level of inner simulation.  

\begin{definition}(Gate and EC conditions with cell filters) \label{def:ec_prime_gate_prime}
    For a fixed level of inner simulation $k$ and fixed integer $m$, we define the following gate conditions for level-$k$ gate gadgets $\cG_k$ (where $\cG_0$ is a $K$-qubit gate for some $K$), and with circuit diagrams being read from left to right): 
    \begin{equation*}
        \centering
        \includegraphics[width=1\linewidth]{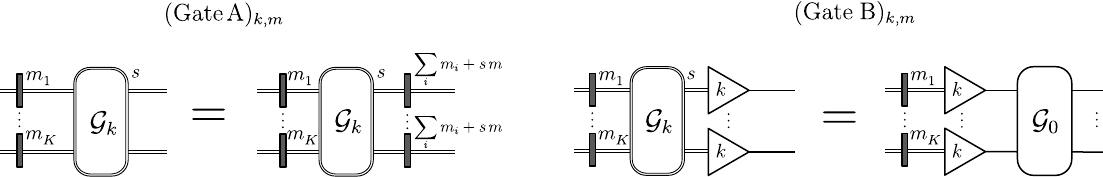}
        \label{fig:gate_prime}
    \end{equation*}
    with $\sum_i m_i + s m \leq m$ as well as $m_i \in \{ 0,m\}$ and $s \in \{ 0,1\}$. The $m$-cell filters are shaded dark to distinguish them from regular filters.  For $\widetilde{\EC} = \EC_k$, we define the ($\EC$ A/B)$_{k,m}$ conditions via
     \begin{equation*}
        \centering
        \includegraphics[width=1\linewidth]{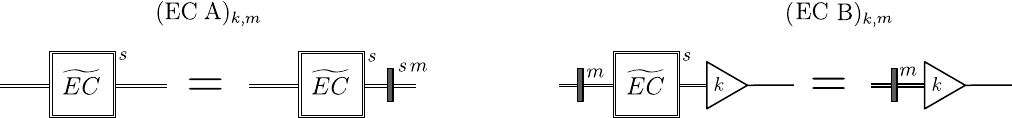}
        \label{fig:\EC_prime}
    \end{equation*}
    with $s \in \{ 0,1\}$. In both conditions, $s$ labels the number of level-0 gadget failures inside a given element. Here and in the rest of the paper, the equality between circuit diagrams holds up to an application of an operator in the stabilizer group of the code to the output state.
\end{definition}

In the definition above, we distinguish these new conditions from the old versions in App.~\ref{app:exRec} by the subscripts in $(\EC A/B)_{k,m}$ and $(\text{Gate } A/B)_{k,m}$.  Throughout the rest of the paper, we will often omit these subscripts, in which case the reader should always assume that we are using the conditions for cell filters.

The following proposition indicates that the new $\EC$ and Gate conditions are still sufficient for showing goodness of exRecs:

\begin{proposition}

The following identity holds for $s + s' + s'' \leq 1$ for any single-qubit gate $\cG_0$:
 \begin{equation*}
        \centering
\includegraphics[width=0.7\linewidth]{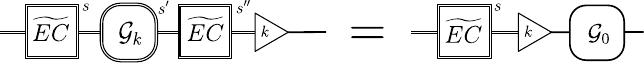}
        \label{fig:goodcorrect1}
    \end{equation*}
and for  $\sum_i s_i + s' + \sum_i s_i'' \leq 1$ for a $K$-qubit gate $\cG_0$:
     \begin{equation*}
        \centering
\includegraphics[width=0.7\linewidth]{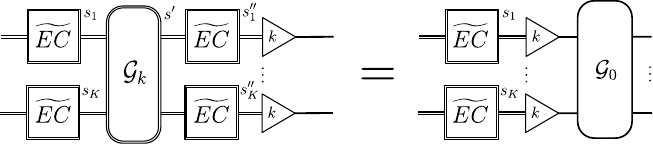}
        \label{fig:goodcorrect2}
    \end{equation*}

Defining a truncated extended rectangle as one without one or more of the trailing $\EC$ gadgets, the following identity holds for $s + s' \leq 1$: 
\begin{equation*}
        \centering
\includegraphics[width=0.6\linewidth]{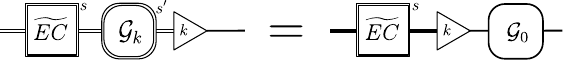}
        \label{fig:goodcorrect3}
    \end{equation*}

There are several scenarios for truncated rectangles for multi-qubit gates, which schematically can be summarized as follows:

    \begin{equation*}
        \centering
\includegraphics[width=0.7\linewidth]{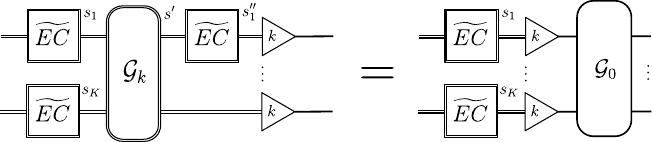}
        \label{fig:goodcorrect4}
    \end{equation*}
and
    \begin{equation*}
        \centering
\includegraphics[width=0.65\linewidth]{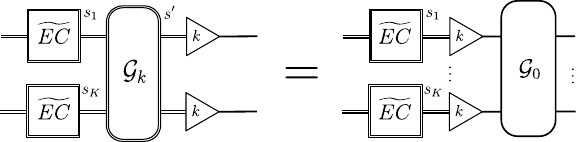}
        \label{fig:goodcorrect5}
    \end{equation*}
    for  $\sum_i s_i + s' + \sum_j s_j'' \leq 1$
\end{proposition}
\begin{proof}
The chain of equalities outlined in App.~\ref{app:exRec} can be used in this case, with the old $\EC$ and Gate properties replaced with the new ones.  For the truncated case, one can replace the missing $\EC$ gadgets with noiseless ones as described in App.~\ref{app:exRec}, and the remainder of the proof is identical.
\end{proof}

The concept of an $\ast$-decoder and its action when an exRec is bad (i.e. when there are two or more gadget failures) is also adopted from the original exRec formulation (see App.~\ref{app:exRec}) with the change $D \rightarrow \widetilde{D}$.

\begin{lemma}[Error suppression]\label{lemma:FT_prime}
Consider an inner simulation for Tsirelson's elementary gadgets for fixed $k$ and associated outer simulation $\widetilde{FT}(\cdot)$, and assume that the gate and EC gadgets  $\cG_k$ and $\widetilde{\EC}$ satisfy $(Gate$ A/B$)_{k,m}$ and $(\EC$ A/B$)_{k,m}$ for some $k$ and $m\leq k$. Then, for a classical circuit $\mathcal{C}$, if the error model in $\widetilde{FT}(\mathcal{C})$ is a $p$-bounded gadget error model, then the resulting circuit simulates $\mathcal{C}$ with an $(Ap)^{2}$-bounded gadget error model for some constant $A$.
\end{lemma}
\begin{proof}
Same as the proof in App.~\ref{app:exRec} if we treat $\widetilde{FT}$ as the fault-tolerant simulation with new fault tolerant gate and EC prescriptions, except using the new Gate and $\EC$ conditions. 
\end{proof}
\begin{definition}\label{def:gadget_linearity}
    Gadget $G$ (which can be either gate $\mathcal G_n$ or $\EC_n$ with $n \leq k$, or a composite gadget comprised of them) is said to act linearly on a subset $\Sigma$ of syndrome locations if, for any input syndrome $\sigma \in \mathbb{F}_2^{|\supp_v(G)|}$ and %auxilliary 
    any other syndrome pattern $\tau \in \mathbb{F}_2^{|\supp_v(G)|}$, the following property holds:  
    \begin{equation}
       G[\sigma, \varepsilon]  = G[\tau|_{\Sigma}] \oplus G[\tau|_{\Sigma} \oplus \sigma, \varepsilon],
    \end{equation}
    where $\oplus$ denotes vector addition in $\mathbb F_2$. If a given vertex $v$ is a location where $G$ acts linearly, $v$ is said to be a linear vertex of $G$. 
\end{definition}

As we will see below, for many of the level-$k$ gadgets appearing in our automaton, all or most of the level-$k$ coarse-grained vertices are linear vertices. The upshot of this is that the linearity property then guarantees that the effect of the noise on the operation of gadgets will not depend on the level-$k$ coarse-grained part of the input state, as long as the syndromes are at linear points. This fact will considerably simplify the proof of fault tolerance for the toric code automaton. 

\begin{proposition} [Linearity properties] \label{prop:linearity_tsirelson}
    The following properties hold:
    \begin{enumerate}
        \item [(1)] $\EC_n$ and $\cI_n$ act linearly on the set of level-$n$ syndrome locations;
        \item [(2)] $\cT_n$ acts linearly on the set of level-$n$ syndrome locations; 
        \item [(3)] $\cM_n$ acts linearly on a subset of level-$n$ syndrome locations, which are the two endpoints of its support.
    \end{enumerate}
\end{proposition}
\begin{proof}
    
     We will first prove (1) and (3) together with a weaker version of (2) where we show that $\cT_n$ acts linearly on a subset of level-$n$ syndrome locations corresponding to the two endpoints. 

     To prove this, we will need the following fact.  Suppose that during the operation of some non-elementary gadget $G$, a particular vertex $v$ is always an endpoint of any constituent elementary gadget appearing during the operation of $G$ whose support includes $v$.  In this case, we claim that $v$ is a linear vertex.  To show this, we define the notion of control vertices of an elementary gadget, which are vertices whose syndromes influence the feedback operator that the elementary gadget applies.  It is clear that the control vertices for the $\cT_0$ and $\cM_0$ gadgets are the middle vertex and the middle two vertices, respectively.  The presence or lack of syndromes at the endpoints of these gadgets does not affect what feedback operator the elementary gadget applies (so that these vertices are linear vertices of $\cT_0$ and $\cM_0$, respectively).  Then, for any such vertex $v$, the total feedback operator applied by $G$ is independent of the syndrome at this vertex.  Denoting the vertex by $v$, 
     \begin{equation}
       G[\sigma, \varepsilon]  = G[\tau|_{v}] \oplus G[\tau|_{v} \oplus \sigma, \varepsilon],
    \end{equation}
    thus proving linearity for this vertex.

    Next, we provide a proof by induction that $\EC_n$, $\cM_n$, $\cT_n$, and $\cI_n$ act linearly on their endpoints; in particular, that the endpoints satisfy the above criterion.  We have already seen that this is true for the endpoints of the level-$0$ gadgets.  To prove the induction step, we assume the statement is true for level-$n$ gadgets and prove it for level-$n+1$ gadgets.  To do so, we write the level-$n+1$ gadgets in terms of the level-$n$ gadgets.  We note that by inspecting the recursive substitution rule, the endpoints of level-$n+1$ gadgets are always endpoints of level-$n$ gadgets.  Therefore, if the endpoints of the level-$n$ gadgets are always non-control vertices for elementary gadgets, the same must hold for the level-$n+1$ gadgets, completing the proof.

    We finally will prove that the midpoint level-$n$ vertices of $\cT_n$ are also linear vertices; this case needs to be analyzed separately because some of the $\cT_{n-1}$ gadgets appearing in $\cT_n$ have these vertices in their interior (in other words, these are their ``control'' vertices, i.e. the feedback of the gadgets can be conditioned on the syndromes in these vertices).  We will first use the fact that, from its definition, $\cT_0$ is a reversible gate, and satisfies the property that $\cT_0[\sigma_1 \oplus \sigma_2, \varepsilon] = \cT_0[\sigma_1, \varnothing] \oplus \cT_0[\sigma_2, \varnothing] \oplus \cT_0[0, \varepsilon]$.  Similarly, for a circuit $\cC$ of purely $\cT_0$ gates, $\cC[\sigma_1 \oplus \sigma_2, \varepsilon] = \cC[\sigma_1, \varnothing] \oplus \cC[\sigma_2, \varnothing] \oplus \cC[0, \varepsilon]$.  Our induction hypothesis will be that $$\cT_n[\tau|_{\Sigma_n} \oplus \sigma, \varepsilon] = \cT_n[\tau|_{\Sigma_n}] \oplus \cT_n[\sigma, \varepsilon],$$ and we will prove an analogous property at level $n+1$.  We have already shown the base case when $n = 0$.  Using the substitution rule for the modified Tsirelson's automaton, we may write $\cT_{n+1}$ in terms of $\cT_{n}$ with layers of $\EC_{n}$ inserted in between.  By the induction hypothesis and due to the fact that $\cI_n$ has linear vertices at its endpoints, for each layer of $\cT_{n}$ and $\cI_{n}$ gates (denote one such layer by $[\cI_{n}, \cdots, \cT_n]$), we can write
    \begin{equation}
    [\cI_{n}, \cdots, \cT_n][\tau|_{\Sigma_n} \oplus \sigma, \varepsilon_1] = [\cI_{n}, \cdots, \cT_n][\tau|_{\Sigma_n}] \oplus [\cI_{n}, \cdots, \cT_n][\sigma, \varepsilon_1],
    \end{equation}
    where $\varepsilon_1$ denotes restriction of noise to that particular time step, and where the first term importantly has support only on the level-$n$ vertices (due to the coarse-graining action of $\cT_n$ and $\cI_n$).  During the layers of $\EC_n$, linearity holds at the level-$n$ vertices, so an identical statement to the above equation can be made about such layers.
    
    Applying layer by layer and repeating this identity, we find that $\cT_{n+1}[\tau|_{\Sigma_n} \oplus \sigma, \varepsilon_1] = \cT_{n+1}[\tau|_{\Sigma_n}] \oplus \cT_{n+1}[\sigma, \varepsilon_1]$, and further restricting $\tau|_{\Sigma_n}$ to $\tau|_{\Sigma_{n+1}}$ proves the induction step.\footnote{In fact, this proves an even stronger property for the $\cT_n$ gadgets, in particular, that all points in $\Sigma_{n-1}$ are linear.}
\end{proof}
In addition to linearity, another important idea that we will utilize in the fault tolerance proof is the idea of gadget nilpotence, which essentially quantifies how good gadgets are at performing error correction:
\begin{definition}[Gate nilpotence] \label{def:tsirelson_nilpotence}
    Let $\cG_0$ be a $K$-qubit gate. A gate gadget $\mathcal G_\ell$ is said to be nilpotent if for $\varepsilon = \text{syn}(\cO)$ where $\cO$ is a Pauli operator supported within a single level-0 gadget located in $\cG_\ell$,  and for $\sigma = \sigma|_{\Sigma_\ell}$, we have $$(\EC_{\ell}^{\otimes K} \circ \mathcal  G_{\ell})[\sigma, \varepsilon] = \mathcal  G_{\ell}[\sigma],$$ where $\EC_{\ell}^{\otimes K}$ are assumed to be noise-free. A gadget $\cG$ is said to be nilpotent at level $\ell$ if $\cG_\ell$ is nilpotent.
\end{definition}

More informally, if a gadget is nilpotent at level $\ell$, a single level-0 gadget failure within a level-$\ell$ version of the gadget (assuming the input syndrome configuration is coarse-grained at level $\ell$) is eliminated, either during the operation of the gadget or by the subsequent noiseless error correction gadgets afterward.  Note that the nilpotence level $\ell$ refers to the number of iterations of the \emph{inner} simulation needed for the desired error suppression.

To prove fault tolerance, we will need to establish the nilpotence (ideally at as low a level as possible) of the distinct gates that make up the error correction gadgets. The following lemma establishes this: 

\begin{lemma} \label{lemma:nilpotent}
     The following gadgets are nilpotent: (a) $\cI$ is nilpotent at level 1, (b) $\cT$ is nilpotent at level 1, (c) $\cI$  is nilpotent at level 2, (d) $\cT$  is nilpotent at level 2, and (e) $\cM$  is nilpotent at level 2.
\end{lemma}
\begin{proof}
    We will be using the terms ``clean input'' or ``damage-free input'' to denote a configuration of syndromes with no nontrivial syndromes.
    
    First, we show the following fact. If a gate $\mathcal{G}$ (where $\cG_0$ is a $K$-bit gate) is nilpotent at level $\ell$ and $\cG_{\ell}$ is a part of a different gadget $\cG'_i$ with $i > \ell$, then a single level-0 gadget failure located within the support of $\cG_\ell$ in an otherwise noise-free gadget $\cG'_i$ cannot change the output of $\cG'_i$. To prove this, we can write $\cG'_i$ in terms of level-$\ell$ $\EC$ and gate gadgets. Since the $\cG_{\ell}$ where the failure occurs is preceded by a noise-free $\EC_{\ell}^{\otimes K}$ gadget, then the input syndromes to this $\mathcal G_\ell $ must be coarse-grained at level $\ell$. Next, because $i > \ell$, $\mathcal G_{\ell}$ must be followed by $\EC_{\ell}^{\otimes K}$ which is noise-free by assumption that there is only one fault in $\cG'_i$. Thus, by nilpotence, we can replace $\cG_\ell[\sigma,\varepsilon] \rightarrow \cG_\ell[\sigma]$ and the gadget failure will not affect the output of $\cG'_i$. 

    Due to the linearity of endpoints of gadgets, the proof for statements (a-d) does not depend on the syndromes at the linear points.  Calling $\sigma$ the input syndromes and $\sigma_{\ell}$ the input syndromes restricted to linear points, we can evolve the gadget $\cG$ assuming the input $\sigma \oplus \sigma_{\ell}$, and $\mathbb F_2$ add the gadget action $\cG[\sigma_{\ell}]$ at the end. 

    (a) Consider $\EC_{1} \circ \cI_1 [\sigma, \varepsilon]$. The only gate failures in the support of $\cI_1$ that are possible are those of $\cI_0$ gadgets; therefore, the output damage is limited to at most one bit (i.e. whose syndromes are contained in a 0 cell). This is cleaned up by the noiseless $\EC_{1}$ at the end. Together with the linearity of $\cI_1$ on the level-1 coarse-grained syndromes, this gives (a).

    (b) Next, consider $\EC_{1}^{\otimes 2} \circ \cT_1 [\sigma, \varepsilon]$.  Similarly, the possible failures in the support of $\cT_1$ are limited to either  $\cI_0$ or $\cT_0$.  In the former case, one finds that the output damage is limited to one bit; in the latter, the output damage has support on two bits, with each bit belonging to a different 1-cell. As a result, the noise-free $\EC_{1}^{\otimes 2}$ at the end restores the correct coarse-grained configuration.  

    (c) We can write 
    \begin{equation}
        \EC_{2}  \circ \cI _2 =  \EC_{2} \circ (\EC_1^{\otimes 3} \circ \cI_1^{\otimes 3} )^5  \circ \EC_1^{\otimes 3}
    \end{equation}
    One possible level-0 gadget failure occurs within the support of an $\cI_1$; from (a) and the fact from beginning of this proof, such failures are inconsequential. Another possible level 0 failure is within $\EC_1$ (excluding the last layer of $\EC_1$), which will produce damage whose support is at most on a single 1-cell; the support of the damage is left unchanged by $\cI_1$s and is coarse-grained to level 1 by the subsequent noise-free $\EC_1$s. Finally, it is eliminated by the noise-free $\EC_{2}$ at the end.  If the level 0 failure occurs within the last $\EC_1$, the damage due to this failure will also be supported in at most a single 1-cell and will be erased by the subsequent noise-free $\EC_{2}$. 

    (d) We can write 
     \begin{equation}
        \EC_{2}^{\otimes 2}  \circ \cT _2 =  \EC_{2}^{\otimes 2} \circ (\EC_1^{\otimes 6} \circ [\cI_1, ... \cT_1])^5  \circ \EC_1^{\otimes 6}
    \end{equation}
    where $ (\EC_1^{\otimes 6} \circ [\cI_1, ... \cT_1])^5$ stands for layers where $I_1$ and $T_1$ gadgets are applied in some combination that we do not specify. If the failure occurs in the support of $\cI_1$ or $\cT_1$, by properties (a) and (b) and the fact from the beginning of the proof, the damage is cleaned by the subsequent layer of $\EC_1^{\otimes 6}$. If one of the gadgets within the support of $\EC_1$ fails, the output damage will still have support within a 1 cell after the subsequent $\cI_1$ or $\cT_1$ gadgets\footnote{To see this, note that by direct verification, a single gadget failure in $\EC_1$ will lead to damage with support on a 1-cell.  A single layer of clean $\cI_1$ and $\cT_1$ following this will at most transport the damage to a different 1-cell.}.  The damage will finally be coarse-grained to level 1 by the subsequent $\EC_1$, and will remain coarse-grained at level 1 and fitting within a 1-cell until it is eliminated by the final noise-free $\EC_2$. Similarly to part (c), if the level 0 failure occurs within the last $\EC_1$, the damage due to this failure will be erased by the subsequent noise-free $\EC_{2}$. 
    
    (e)  We can write
    \begin{equation}
        \EC_{2}^{\otimes 3}  \circ \cM _2 =  \EC_{2}^{\otimes 3} \circ (\EC_1^{\otimes 9} \circ [\cI_1, ... \cT_1])^5 \circ \EC_1^{\otimes 9} \circ \cM_1^{\otimes 3}  \circ (\EC_1^{\otimes 9} \circ [\cI_1, ... \cT_1])^5\circ \EC_1^{\otimes 9}
    \end{equation}
    where, as before, $ (\EC_1^{\otimes 9} \circ [\cI_1, ... \cT_1])^5$ denotes (now different) layers of level-1 identity and swap gates.  We are not able to use the linearity property anymore to remove non-trivial input syndromes in $\Sigma_2$ because these syndromes can be located where $\mathcal \cM_2$ acts non-linearly. 
    
    If a level-0 gadget failure occurs at any point after the layer of $\cM_1^{\otimes 3}$, then the damage right after any subsequent layer of $EC_1^{\otimes 9}$ (or the noise-free layer of $EC_2$) will be coarse-grained at level 1 and will fit in a 1-cell.  Then, same considerations as in previous parts apply, since the remaining gates to be applied consist entirely of $\cI_1$, $\cT_1$, and $\EC_1$. 

    If the failure occurs at any point before the layer of $\cM_1^{\otimes 3}$, we note that the damage before the layer of $\cM_1^{\otimes 3}$ will always fit in an $1$-cell. This is because the gadgets $\cI_1$ and $\cT_1$ are nilpotent against level-0 failures within them, and failure of any of the $\EC_1$ gadgets produces damage contained within a single 1-cell (which we call $m=1$ damage). We note that $m=1$ input damage continues to remain $m=1$ damage under any sequence of clean $\cI_1$, $\cT_1$, and $EC_1$ (since the action of a single one of these gadgets does not grow the size of the damage beyond a 1-cell).  Hence 
    % in either of the gadgets remains $(m=1)$ output damage, albeit possibly in a different 1 cell.
    % Therefore, 
    this case -- as well as the final remaining case where the level-0 failure occurs within the $\cM_1^{\otimes 3}$ layer -- reduces to studying the $\cM_1^{\otimes 3}$ layer with either input damage fitting in a 1-cell or a clean input with a single level-0 gadget failure inside $\cM_1^{\otimes 3}$. By explicitly checking all possible scenarios of this, we find that in all of them, the output damage will have support at most the size of a single $\cM_1$ gadget (which is the size of a 2-cell).  After the following layer of $\EC_1$ gadgets, the resulting damage will be coarse-grained to level 1 and will fit in a single 2-cell.  Thus, we can reduce the level by 1, mapping the problem of considering $\cM_2$ with a level-2 input and a single failed $\cM_1$ in the middle layer to considering an $\cM_1$ gadget with a level-1 input and a single failed $\cM_0$ in the middle layer. Such a failure is cleaned up after running the rest of the gadget and following with a noise-free $\EC$ and we thus find that $\cM_2$ is nilpotent.
\end{proof}

We emphasize that while $\cM_2$ is nilpotent, $\cM_1$ is {\it not}; this is the reason why our proof of fault tolerance in this setting must go up by several levels. As an explicit example demonstrating this, consider $EC_1^{\otimes 3}\cM_1$ acting on the input state $111\, 000\, 000$, 
with a fault occurring in the first $\cT_0$ gate to be applied (c.f. the circuit for $\cM_1$ given in \eqref{eq:M1-ts}). 
This produces the state $111 \, 001 \, 100$ after the first layer of the circuit constituting $\cM_1$. The action of the rest of $EC_1^{\otimes 3} \circ \cM_1$ on this configuration produces the output $111\, 111\, 000$, which contains a logical error (a similar problem does not afflict $\cM_2$ due to the nilpotence of $\cT_1$).

We are now ready to prove the Gate and $\EC$ conditions for the present set of gadgets. 
 \begin{proposition} [] \label{lemma:gate-ec-prime-TS}
 Assume an inner and outer fault tolerant prescription where all elementary gadgets undergo $p$-bounded gadget error model (where each of the gadgets $\cI_0$, $\cT_0$, and $\cM_0$ can fail).   
Then:
\begin{enumerate}
\item The gate gadgets $\cI_\numT$, $\cT_\numT$ and $\cM_\numT$ satisfy $(Gate$ A$)_{k,m}$ and $(Gate$ B$)_{k,m}$  properties and
\item  The $\widetilde{\EC} = \EC_k$ gadget satisfies $(\EC$ A$)_{k,m}$  and $(\EC$ B$)_{k,m}$  properties 
\end{enumerate}
for $k \geq 3$ and $m=1$ (c.f. Def.~\ref{def:ec_prime_gate_prime}).
 \end{proposition}
In the following, we will split the proof into separate parts, one for each property. In several places we will prove slightly stronger results with $k=2$. 

\begin{proof}[Proof ($\EC$ B)]
    We will show $\EC$ B for $k=2$, and $m=1$. 
    % We note that $s$ is either $0$ or $1$. 
    The definition of the $\EC$ B property allows for either a single or no level-0 gadget failures to occur within the support of the $\EC$ gadget; in the former case, the input can be taken to be properly coarse-grained at level $k$. 
    % Consider first the case $s = 1$ and $m = 0$, 
    
    Consider first the case where the failure occurs inside the $\EC$ gadget.
    Since the input syndromes are supported on $\Sigma_k$ and $\EC_k$ acts linearly on these locations, and since $\EC_k[\sigma|_{\Sigma_k}] = \sigma|_{\Sigma_k}$, for the purposes of demonstrating $\EC$ B we can assume that we are given a clean codeword as input.
     We can write $\EC_2$ as
    \begin{equation} \label{eq:EC2}
         \EC_2 =   \EC_{1}^{\otimes 3} \circ \cM_1
        \circ \EC_{1}^{\otimes 3} \circ \cT_1 \cI_1        
        \circ  \EC_{1}^{\otimes 3} 
        \circ \cM_1      
        \circ \EC_{1}^{\otimes 3} \circ  \cI_1 \cT_1   
         \circ \EC_{1}^{\otimes 3}  \circ \cM_1 \circ \EC_{1}^{\otimes 3}  
    \end{equation}
    If the level-0 failure occurs in the support of one of the  $\cT_1$ or $\cI_1$ gate gadgets, it is inconsequential because of their nilpotence (Lemma~\ref{lemma:nilpotent}). If it is located inside any of the $EC_1$ or $\mathcal M_1$ gadgets, one can directly verify that the output damage of $\EC_2$  is contained within a single 1 cell (the fact that the input is coarse-grained at level 2 is necessary for this to be true). This damage will be cleaned up by an ideal level-2 decoder.  

    Now let us consider the case when $ s = 0$ and $m = 1$. The first layer of $\EC_1^{\otimes 3}$ inside of $\EC_2$ will coarse-grain this damage to level 1. Then, treating this as an input, the operation of the remainder of the gadget is equivalent to level reducing by 1, in which case the damage fits in a 0-cell and the remaining operation of the gadget is equivalent to $\EC_1$ which cleans up such an input.  This is again consistent with a 1-cell filter acting on the output, which proves $\EC$ B for $k = 2$, $m=1$.

    Finally, let us show $(\EC\text{ }B)_{3,1}$ is true given that  $(\EC\text{ }B)_{2,1}$ is true. First,  $\EC_3[\sigma,\varepsilon]$ can be expressed through level 2 gates and $\EC$ gadgets. When $s = 1$, $m=0$ and the failure occurs inside one of the level 2 gates, it is inconsequential because of the nilpotence property (Lemma~\ref{lemma:nilpotent}). If the error occurs within any of the $\EC_2$ gadgets, we use the fact that $(\EC\text{ }B)_{2,1}$ holds. Therefore, the failure either has no effect or it causes $1$-cell damage at the output of the $\EC_2$ gadget. If this is the last $\EC_2$ in the circuit for $\EC_3$, this proves the claim. Otherwise, the damage is input into the subsequent level-2 gate, wherein it is coarse-grained to level 1 by the first layer of clean $\EC_1$ in the level-2 gate. Therefore, we can formally reduce the level by 1 and consider this damage as level-0 input damage to a level-1 gate. If the gate is $\cI$ or $\cT$, the output damage still fits in a 0-cell, which means (going back up one level) it is cleaned up by the subsequent $\EC_2$. If the gate is $\cM$, one can explicitly check that the damage is either cleaned up or remains unchanged, in which case it will be cleaned up by the subsequent $\EC_2$. This concludes the proof of $(\EC\text{ }B)_{3,1}$. 
\end{proof}

\begin{proof}[Proof (Gate A)]
    We first will show a stronger property, namely $(\text{Gate A})_{2,1}$.  We will assume the gate gadget acts on $K$ bits.  Consider first the case $s = 1$, and $m_i = 0$ for $i = 0,\cdots,K$.  By Lemma~\ref{lemma:nilpotent}, $\cI_2$, $\mathcal{T}_2$ and $\mathcal{M}_2$ are nilpotent,  i.e. $\EC_{2} \circ \mathcal{I}_{2}[\sigma, \varepsilon] = \mathcal{I}_{2}[\sigma]$, $\EC_{2}^{\otimes 2} \circ \mathcal{T}_{2}[\sigma, \varepsilon] = \mathcal{T}_{2}[\sigma]$ and $\EC_{2}^{\otimes 3} \circ \mathcal{M}_{2}[\sigma, \varepsilon] = \mathcal{M}_{2}[\sigma]$.  According to the same logic as in the proof of $\EC$ B, if $\varepsilon$ occurs before the last layer of $\EC_1$ in either of the gate gadgets, then $\mathcal{T}_{2}[\sigma, \varepsilon] = \mathcal{T}_{2}[\sigma]$ and $\mathcal{M}_{2}[\sigma, \varepsilon] = \mathcal{M}_{2}[\sigma]$.  If $\varepsilon$ occurs during the last layer of $\EC_1$ in either of the gate gadgets, then it causes damage confined to at most a single 1-cell, consistent with an 1-cell filter acting on the output.

    We now consider the case $s = 0$, and $m_i = 1$ for some $i$, i.e. when the input damage (relative to the nearest level-2 coarse-grained configuration) is contained within a single 1-cell.  After the first layer of $\EC_1$, the syndromes lie entirely in $\Sigma_1$.  The operation of the remainder of both gadgets can be then understood by reducing the level by 1, i.e. by considering $\mathcal{T}_1$ and $\mathcal{M}_1$ with input damage contained in a 0-cell.  The output will have damage contained at most in a 0-cell per each 1-cell, which when going back up 1 level is consistent with $m = 1$.  

    Finally,  $(\text{Gate A})_{3,1}$ property follows from $(\text{Gate A})_{2,1}$ by similar considerations to those in the proof of $\EC$ B. 
\end{proof}

\begin{proof}[Proof (Gate B)]

    Assume the gate gadget acts on $K$ bits.  For the case $s = 1$, and $m_i = 0$ for $i = 0,\cdots,K$, this condition follows from nilpotence of all gates at level 2 together with $(\EC\text{ }B)_{2,1}$ and considerations from the proof of $\text{Gate}$ A. The output state is thus equivalent to having undergone a noiseless gate with possible $1$-cell damage at the end. The ideal decoder can eliminate this damage. 
    The case $s = 0$, and $m_i = 1$ can be similarly analyzed. 
\end{proof}

\begin{proof}[Proof ($\EC$ A)]

    First, we express $\EC_3$ through level-2 gadgets only, assuming a single level-0 gadget failure and now an \emph{arbitrary} input state. The circuit starts with a layer of $EC_2^{\otimes 9}$, and we denote $t_0$ to be the point in time right after this layer. 
    
    Consider first the situation when the level-0 gadget failure occurs after $t_0$. In that case, the first layer of $EC_2^{\otimes 3}$ is error-free, and the state at $t_0$ is coarse-grained at level 2. A level-0 gate failure then must occur in the support of one of the subsequent level-2 gates $\mathcal G_2$ or inside one of the subsequent layers of $EC_2^{\otimes 3}$. From Lemma~\ref{lemma:nilpotent}, we know that in the former case, the state is not changed. 
    If the failure occurred in the support of one of the $\EC_2^{\otimes 3}$ layers, we combine considerations for $\numT = 2$ from the proofs of $(\EC\text{ }B)_{2,1}$ and $(\text{Gate }A)_{2,1}$ to find that the damage is limited to the support of a 1-cell after each gate or $\EC_2$ layer. This is consistent with adding a 1-cell filter at the end.

    Consider now the case when the level-0 gadget failure occurred before $t_0$ (i.e. in the support of one of the $\EC_2$ gadgets). For reference, this gadget is written out in Eq.~\eqref{eq:EC2}.     
    If the level-0 gadget failure occurred
    within the first layer of $\EC_{1}^{\otimes 9} $, the output of one of these $\EC_1$ gadgets contains arbitrary damage in its support and the damage is coarsed-grained to level 1 for the other $\EC_1$ gadgets. A brute-force computation (by a computer algebra software; see App.~\ref{app:reversibility}) shows that the output of the subsequent $\EC_1 \circ \mathcal M_1$ gadget has the property that there exists (a non-unique) input configuration which is coarse-grained to level 1 that would cause the same output. Let us refer to this property as producing a \emph{reversible output}. Choosing a level-1 coarse-grained input with such an output, the action of the gadget with the level-0 failure can be replaced by the action of the gadget on this input, which will be coarse grained at level 3, fulfilling the $\EC\text{ }A$ condition. 
    
    Next, consider the case when the gadget failure occurs within the first layer of $\EC_2^{\otimes 3}$, in particular after the first layer of $\EC_1^{\otimes 9}$ and before the last layer of $\EC_1^{\otimes 9}$. Because the first layer of $\EC_1$ gadgets is noise-free, if the fault occurs either in the support of subsequent $\cI_1$ or $\cT_1$, due to the nilpotence of those gates, it has no consequence. If it occurs elsewhere, the damage will fit in a 2-cell and after the last layer of $\EC_1^{\otimes 9}$ (which is noise-free) the damage is coarse-grained to level 1. Considering the entire $\EC_3$ gadget, we see that the input to the first $\cM_2$ is now coarse-grained at level 2 apart from one 2-cell where the input is coarse-grained at level 1. This reduces to the consideration in the previous paragraph upon reducing the level by one, i.e. we replace $\cM_2$ with $\cM_1$ and $\EC_3$ with $\EC_2$. The $\cM_1$ in question produces a reversible output, which is equivalent to a level-1 coarse-grained input to $\EC_2$ and a noise-free operation of $\EC_2$. This damaged is clearly coarse-grained at level 2, which upon undoing level reduction, corresponds to being coarse-grained at level 3. 

    Finally, if the error occurred in the last layer of $\EC_1^{\otimes 9}$ (within the first layer of $\EC_2^{\otimes 3}$), the output is unchanged under an $1$-cell filter. After the next gate followed by the next layer of $\EC_2$ the additional damage will be cleaned up, as per the $(\text{Gate }A)_{2,1}$ and $(\EC\text{ }B)_{2,1}$ properties.  Therefore, $(\EC\text{ }A)_{k,m}$ holds for $k = 3$ and $m=1$.

This proves the Proposition.

\end{proof}

Having established the EC and Gate properties, fault tolerance readily follows:
\begin{corollary}\label{cor:Tsirelson_pf2}
Consider a classical circuit $\mathcal{C}$ on 1 bit which consists of $T$ idle operations. Construct a fault-tolerant circuit corresponding to $n$ repetitions of the outer simulation $(\widetilde{FT})^{n}(\mathcal{C})$ with inner simulation level $k = 3$ and $m=1$.   Assuming a $p$-bounded gadget error model in $(\widetilde{FT})^{n}(\mathcal{C})$, $\mathcal{C}$ experiences an error model that is a $(Ap)^{2^{n}}$-bounded wire error model for some constant $A$. The probability that the initial logical state $b(0)$ will get corrupted in time $T$ is
 \begin{equation}
 \mathbb{P}\left (\widetilde{D}_n(s(t = T)) \neq \widetilde{D}_n(s(t = 0)) \right ) \leq T \cdot (Ap)^{2^{n}}
 \end{equation}
 \end{corollary}
 \begin{proof}
     With $k = 3$ and $m=1$, both $\EC$ and Gate properties are satisfied by the previous Lemma.  Consider a circuit $\mathcal{C}$ and its fault tolerant simulation $\widetilde{FT}(\mathcal{C})$.  Proposition~\ref{lemma:FT_prime} indicates that if the gadget error model in $\widetilde{FT}(\mathcal{C})$ is $p$-bounded, adding a layer of ideal decoder and pulling it through the circuit, the gadget error model for $\mathcal{C}$ is $A p^2$-bounded.  Starting from $(\widetilde{FT})^{n}(\mathcal{C})$ and iterating this $n$ times as done in the first proof of fault tolerance for the modified Tsirelson's automaton gives a $(Ap)^{2^{n}}$-bounded gadget failure model.  By a union bound, $\mathbb{P}\left (\widetilde{D}_n(s(t = T)) \neq \widetilde{D}_n(s(t = 0)) \right ) \leq T \cdot (Ap)^{2^{n}}$.  Selecting $n = (\log_3 L)/k =  (\log_3 L)/3$ shows that the memory lifetime is exponential in a power of $L$. 
 \end{proof}

%% file: 5.TCproof.tex
\section{Fault tolerance of the toric code automaton}
\label{sec:FT_proof}

% In this section, as in majority of Sec.~\ref{sec:toriccode}, we will consider the $X$ part of the toric code automaton,

The full toric code automaton is constructed in Sec.~\ref{sec:toriccode}. It is a measurement and feedback circuit consisting of simultaneous action of a layer of $X$- and $Z$-type elementary gadgets which, according to Remark~\ref{remark:full-TC-automaton}, can be defined to operate under a general $p$-bounded gadget error model. 
%From  Prop.~\ref{prop:apprnoise},  we can replace this noise model with a $p$-bounded approximate gadget error model (paying a small price in trace distance), where noise operators are Pauli. 
%Furthermore, by Lemma~\ref{lem:equiv_decoupled_gadget} the $p$-bounded approximate noise model can be equivalently replaced with a $C p^{\gamma}$-bounded \emph{decoupled} gadget error model for constants $C$ and $\gamma$.  
\hi{From the analysis in Appendix}~\ref{app:noisedetails}, \hi{see Sec.}~\ref{subsec:TC} \hi{for a summary of the results, we can reduce this to studying a decoupled gadget error model.  In this error model}, the $X$-gadgets apply $X$-type feedback and noise operators only, and $Z$ gadgets, similarly, apply $Z$-type feedback and noise. Fact~\ref{fact:XZ-decoupling} shows that under this error model the operation of the full toric code automaton can be reduced to considering separate $X$ and $Z$ automata. At the end of this section in Corollary~\ref{corollary:X-Z-FT}, we finally explain how fault tolerance of the full toric code automaton follows from that of the $X$ and $Z$ automata separately. 

Therefore, most of this section will be devoted to showing the fault tolerance of the $X$ part of the toric code automaton, 
which is comprised of $X$-type gadgets that apply Pauli $X$ feedback only and  experience Pauli $X$ noise. To construct this  automaton, we use the inner and outer fault tolerant simulation prescription.  We then develop a modification of the extended rectangles method, prove that it can be used to derive the appropriate error suppression as a function of the number of fault tolerant simulations, and use this show the fault tolerance of the $X$-type automaton.

The rest of this section is structured as follows. The main difficulty in proving the fault tolerance of $X$-type automaton is that the supports of neighboring gadgets must overlap, which can cause interactions between the operations of the gadgets and enable errors to spread (this is not an issue for standard circuit concatenation). In Subsec.~\ref{sec:confinement}, we explore a property we call syndrome confinement, which allows us to show that the interaction between any pair of neighboring gadgets in our construction is limited to the boundary between them, and moreover, this region of interaction cannot influence error syndromes in the rest of space. This property is needed to adapt the extended rectangles method to the setting of the toric code automaton. In addition, the ideal decoder has to operate in a global way for the toric code (due to the global nature of the pushforward map), while the ideal decoder for Tsirelson's automaton is local. To deal with these challenges, we introduce the notion of spatial truncation and appropriately adapt the concepts of $\ast$-decoders and the Gate and EC conditions. This is done in Subsec.~\ref{sec:new-exrecs}. 

Having developed a modified exRecs formalism for the toric code, we show that fault tolerance follows from Gate and EC conditions, and finally, in Subsec.~\ref{sec:pfnilpTC}, we show that these conditions hold for the gadgets comprising the toric code automaton. To show that these properties hold, we use the strategy partially developed in the previous section via the measurement-and-feedback version of Tsirelson's automaton. Namely, we introduce the notion of nilpotence for the toric code gadgets and provide a computer-aided proof of this property. We show that there exists a level $k$ for the inner simulation such that the associated outer simulation achieves error suppression.  Upon a single iteration of the outer simulation, the original circuit is simulated with $\widetilde p$-bounded noise, where $\widetilde{p} \sim (Ap)^{2}$. Repeated outer simulation thus achieves (quasi)exponential suppression of the logical error rate with the system size. 
Finally, we also discuss fault-tolerant initialization and readout in Subsec.~\ref{sec:initialization}, and discuss some preliminary numerical results in Subsec.~\ref{sec:num-sim-and-conj}.

\subsection{Confinement of syndromes} \label{sec:confinement}

For the reasons explained above, we will focus on $X$-type gadgets and drop the $X$ label for the remainder of the text unless otherwise specified. To deal with gadgets that act on overlapping domains, we introduce a property that we call confinement. In our proof, we interchangeably work with physical operators (which are supported on qubit locations, i.e. the links of the lattice) and error syndromes (located on the vertices of the lattice). This is because some of the arguments are easier to formulate using operators, while others are easier to formulate using syndromes. 

It is more natural to think of confinement in the syndrome language. Intuitively speaking, the input syndromes and the syndromes created due to noise in a confinement region cannot affect the syndromes outside of it. If the overlap of the supports between two neighboring gadgets is a confinement region, then the gadgets can ``interact'' in this region, but what happens in this region cannot influence syndromes elsewhere.
Such a property will limit how neighboring gadgets can affect each other's operation and thus will play an important role in proving fault tolerance.
The gadgets that we introduced for the toric code automaton were designed so that any two neighboring gadgets overlap on confinement regions only.

\begin{definition}[Confinement]\label{def:confinement}
Consider a layer of gadgets $\cL(G)$ of some level $n$ within the toric code automaton. A subset $\Lambda \in \cL_0$ of syndrome locations is called a confinement region if for input syndrome $\sigma \in \mathbb{F}_2^{|\cL_0|}$ and Pauli noise realization with syndrome $\varepsilon$, we have
\begin{equation} \label{eq:def-conf}
\cL(G)[\sigma, \varepsilon] = \eta\big|_{\Lambda} \oplus \cL(G; \cL_0 \setminus \Lambda) [\sigma, \varepsilon]
\end{equation}
where $\eta\big|_{\Lambda} = \cL(G)[\sigma, \varepsilon]\big|_{\Lambda}$. Recall that the notation $\cL(G; \cL_0 \setminus \Lambda)$ is defined in Def.~\ref{def:restriction}.
Syndromes in $\Lambda$ are said to be confined.
\end{definition}
If $\Lambda$ is a confinement region for a layer $\cL(G)$ of the toric code automaton, we will often simply say that $\Lambda$ is a confinement region for $G$.

The ``non-communication'' property between the neighboring gadgets that overlap on confinement regions will be made more precise below. Although the confinement region can be arbitrarily shaped, for simplicity of presentation, we restrict ourselves to the case (relevant for the toric code automaton) where any two level-$n$ neighboring gadgets overlap on $n$-links, which are confinement regions for the associated layer of gadgets. 
\begin{remark} \label{remark:overlaps}
    From this point on, we assume that the gadgets of the toric code automaton (i.e. the gate $\cG_n$ and $\EC_n$ gadgets) have support on connected regions and that neighboring gadgets can overlap on a union of $n$-links or on points in $\Sigma_n$. We also assume that that overlap between any two gadgets is a confinement region for the associated layer of gadgets and the corners of any gadgets are linear vertices; this is proved in Prop.~\ref{prop:confinement-for-TC} for the toric code automaton. 
\end{remark}

\begin{fact}[Confinement separates the bulk from the boundary] \label{fact:0}
Consider a layer of level-$n$ gates $\cL(G)$ as well as a certain level-$n$ gadget $G_1$ in this layer with vertex support $\Lambda_1$ and boundary $\partial \Lambda_1$. Assume that overlap between $G_1$ and any other gadget $G_i$ in the same layer overlaps on a set of boundary $n$-links that we denote as $\lambda_{1i}$, i.e. $\mathrm{supp}_v (G_1) \cap \mathrm{supp}_v (G_i) = \lambda_{1i}$. We also assume that the boundary of the gadget $G_1$, denoted $\partial \Lambda_1 = \cup_i \lambda_{1i}$, is a confinement region. Then, we have
\begin{equation}
    \cL(G) [\sigma, \varepsilon] = \eta|_{\partial \Lambda_1 }  
    \oplus (G_1;\Lambda_1\setminus \partial \Lambda_1)[\sigma, \varepsilon] 
    \oplus \cL(G;\cL_0\setminus \Lambda_1) [\sigma, \varepsilon].
\end{equation}
 where $\eta|_{\partial \Lambda_1 } = \cL(G)[\sigma, \varepsilon]|_{\partial \Lambda_1 }$.
\end{fact}
\begin{proof}
Because $\partial \Lambda_1$ is a confinement region, we can write by definition
\begin{equation} 
\cL(G)[\sigma, \varepsilon] = \eta\big|_{\partial \Lambda_1} \oplus \cL(G;\cL_0 \setminus \partial \Lambda_1)[\sigma, \varepsilon] 
\end{equation}
Because of the assumption that all gadgets overlap on boundaries only, for any gadget $G_i \neq G_1$, we can write $\mathrm{supp}_v (G_1) \cap \mathrm{supp}_v (G_i) \cap (\cL_0 \setminus\partial \Lambda_1) = \varnothing$. Thus, the operation of $\cL(G;\cL_0 \setminus \partial \Lambda_1) $ factors as such:
\begin{equation} 
\cL(G;\cL_0 \setminus \partial \Lambda_1)[\sigma, \varepsilon]= (G_1;\Lambda_1\setminus \partial \Lambda_1)[\sigma, \varepsilon]
    \oplus \cL(G;\cL_0\setminus \Lambda_1) [\sigma, \varepsilon].
\end{equation}
This proves the fact.
\end{proof}

\begin{fact}[Confinement implies non-communication]\label{fact:1}
Consider a layer of level-$n$ gadgets $\cL(G)$ and any two neighboring level-$n$ gadgets $G_1$ and $G_2$ in this layer, which have vertex supports $\Lambda_1$ and $\Lambda_2$ respectively. The pair of gadgets overlap on boundary $n$-links only and each such $n$-link is a confinement region with respect to $\cL(G)$. Call $\partial \Lambda_1$ and $\partial \Lambda_2$ the set of boundary vertices of the two gadgets respectively. Then, 
\begin{align}\label{eq:fact2eq}
    \cL(G)[\sigma, \varepsilon] = \eta|_{\partial \Lambda_1 \cup \partial \Lambda_2}  
    \oplus (G_1;\Lambda_1 \setminus \partial \Lambda_1) [\sigma, \varepsilon]  \oplus  (G_2;\Lambda_2 \setminus \partial \Lambda_2) [\sigma, \varepsilon] \oplus  \mu \big|_{\cL_0 \setminus (\Lambda_1 \cup \Lambda_2)}
\end{align}
where $\eta$ and $\mu$ are some syndrome vectors that generally depend on the input syndromes, noise, and the operation of the layer of gadgets. The supports are illustrated below:
\begin{equation*} 
\includegraphics[width = 0.45\textwidth]{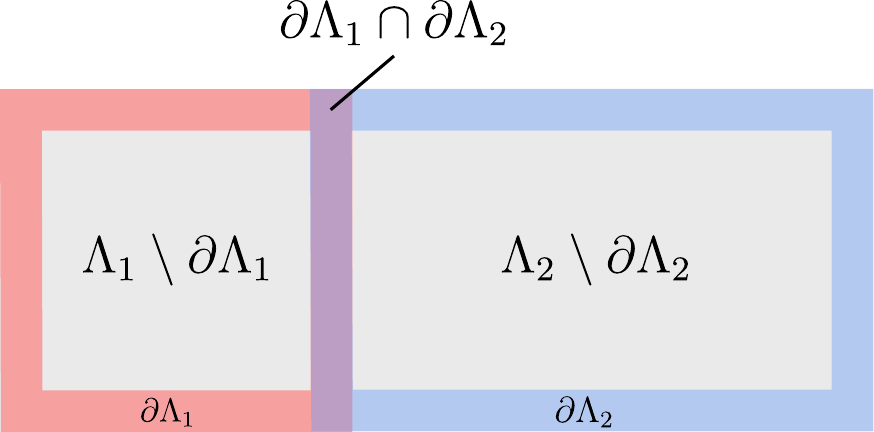}
\end{equation*}

\end{fact}
\begin{proof}
Using the previous fact with respect to the gadget $G_1$, we can write 
\begin{equation}
    \cL(G) [\sigma, \varepsilon] = \eta'|_{\partial \Lambda_1 }  
    \oplus (G_1;\Lambda_1\setminus \partial \Lambda_1)[\sigma, \varepsilon]
    \oplus \cL(G;\cL_0\setminus \Lambda_1) [\sigma, \varepsilon].
\end{equation}
where $\eta'|_{\partial \Lambda_1 } = \cL(G) [\sigma, \varepsilon]|_{\partial \Lambda_1 }$. Therefore,
\begin{equation} \label{eq:aux1}
    \cL(G) [\sigma, \varepsilon]|_{\Lambda_1 \setminus \partial \Lambda_1 }   = 
   (G_1;\Lambda_1\setminus \partial \Lambda_1)[\sigma, \varepsilon].
\end{equation}
Analogously, for the second gadget, using the same logic gives
\begin{equation}\label{eq:aux2}
    \cL(G) [\sigma, \varepsilon]|_{\Lambda_2 \setminus \partial \Lambda_2 }   = 
     (G_2;\Lambda_2\setminus \partial \Lambda_2)[\sigma, \varepsilon].
\end{equation}
Finally, we can write
\begin{align*}
\cL(G) [\sigma, \varepsilon] &= \cL(G) [\sigma, \varepsilon]|_{\partial \Lambda_1 \cup \partial \Lambda_2 } 
\oplus \cL(G) [\sigma, \varepsilon]|_{\Lambda_1 \setminus \partial \Lambda_1 } 
\oplus \cL(G) [\sigma, \varepsilon]|_{\Lambda_2 \setminus \partial \Lambda_2}
\oplus  \cL(G) [\sigma, \varepsilon]|_{\cL_0 \setminus (\Lambda_1\cup\Lambda_2)}\\
 &= \eta|_{\partial \Lambda_1 \cup \partial \Lambda_2 } 
 \oplus (G_1;\Lambda_1\setminus \partial \Lambda_1)[\sigma, \varepsilon]
\oplus  (G_2;\Lambda_2\setminus \partial \Lambda_2)[\sigma, \varepsilon]
\oplus 
\eta|_{\cL_0 \setminus (\Lambda_1\cup\Lambda_2)}
\end{align*}
where in the first equality, we simply decomposed the output of a layer of gadgets over non-overlapping regions. In the second equality, we used Eqs.~\eqref{eq:aux1} and \eqref{eq:aux2} and defined $ \eta|_{\partial \Lambda_1 \cup \partial \Lambda_2 }  =  \cL(G) [\sigma, \varepsilon]|_{\partial \Lambda_1 \cup \partial \Lambda_2 } $ and $\eta|_{\cL_0 \setminus (\Lambda_1\cup\Lambda_2)} = \cL(G) [\sigma, \varepsilon]|_{\cL_0 \setminus (\Lambda_1\cup\Lambda_2)}$. This completes the proof of the fact.
\end{proof}

As stated above, boundary segments of every gadget in the toric code automaton are confinement regions. We will now prove this fact: 
% We will now determine what the confinement regions for the toric code automaton are:
%
\begin{proposition} \label{prop:confinement-for-TC}
Consider the toric code gadgets obtained from inner simulation described in Subsec.~\ref{subsec:inner-simulation}, and recall the definition of gadget linearity in Def.~\ref{def:gadget_linearity}.  The following statements are true: 
\begin{enumerate}
    \item Each $n$-link  $\lambda$ in the support of $\EC_{n}$ is a confinement region with respect to a layer $\cL(\EC_n)$.
    \item The $\EC_n$ gadget is linear on all $n$-frame vertices $\Sigma_{n}$ in its support.
    \item Each $n$-link (or union of links) highlighted in blue in the picture below is a (separate) confinement region for layers of $\cI_n$, $\cT_n^{h}$ and $\cM_n^{h}$ gadgets, respectively: 
    \begin{equation*}
        \centering
        \includegraphics[width=0.8\linewidth]{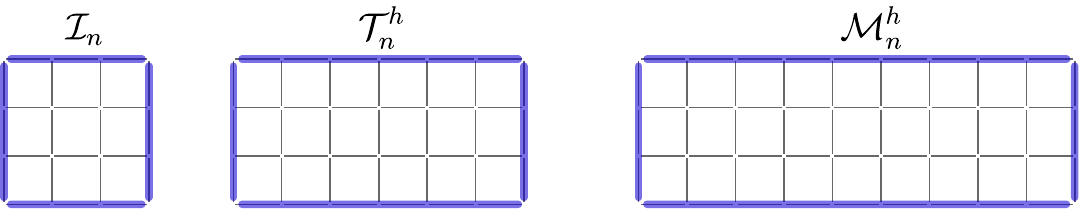}
        \label{fig:confinement}
    \end{equation*}
    where the links that are shown in gray are the $(n-1)$ links of the lattice. Note that confinement regions corresponding to the top and bottom links $\cT_n^{h}$ and $\cM_n^{h}$ are a union of two or three horizontal $n$-links.  For $\cT_n^v$ and $\cM_n^v$, the confinement regions are found similarly by a $90^\circ$ rotation. 
    \item $\cI_n$, $\cT_n^{h}$ and $\cM_n^{h}$ act linearly on the $n$-frame vertices shown below: 
    \begin{equation*}
        \centering
        \includegraphics[width=0.8\linewidth]{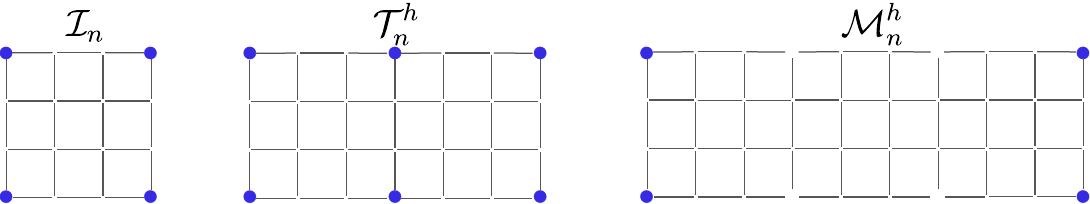}
        \label{fig:linear-TC}
    \end{equation*}
     where the links that are shown in gray are the $(n-1)$ links of the lattice. For $\cT_n^v$ and $\cM_n^v$, the linear points are obtained by a $90^\circ$ rotation.  
\end{enumerate}  
\end{proposition}
\begin{proof}
We first will prove the claim regarding linearity.  In the proof of Prop.~\ref{prop:linearity_tsirelson}, we first proved that if during the operation of some non-elementary gadget a particular vertex is always an endpoint of any constituent elementary gadget, then such a vertex is a linear vertex.  This fact is also true for the toric code for the same reason as given in Prop.~\ref{prop:linearity_tsirelson}, except that endpoints of gadgets become their corners.

First, we show that the corner vertices of the level-$n$ gadgets $\cI_n$, $\cT_n$, and $\cM_n$ are linear vertices.  We show the aforementioned fact using induction on $n$.  When $n = 0$, this is true by definition of elementary gadgets.  We then assume that this is true for $n$.  We write $\cI_{n+1}$, $\cT_{n+1}$, and $\cM_{n+1}$ in terms of $\cI_n$, $\cT_n$, and $\cM_n$ and note that corner vertices of $\cI_{n+1}$, $\cT_{n+1}$, and $\cM_{n+1}$ are always corner vertices of $\cI_n$, $\cT_n$, and $\cM_n$, thus proving the induction step.

Finally, we prove that for $\cT_{n}$, the two remaining level-$n$ vertices are linear vertices.  We follow the proof in Prop.~\ref{prop:linearity_tsirelson}.  However, now one must write $\cT_{n+1}$ in terms of a 2D circuit of $\cT_n$ and $\cI_n$ interspersed with layers of $\EC_{n}$; noting that the corner vertices of $\EC_n$ are linear vertices, the proof follows otherwise identically to that in Prop.~\ref{prop:linearity_tsirelson}.

Next, we prove the claim about confinement regions.  We first show the following fact. Consider a subset of vertices $\lambda$ during the operation of a layer of gadgets $\cL(\cG_n)$.  For an elementary gadget $\cG_0$, define the feedback support $\text{supp}_f(\cG_0)$ to be the three vertices in the bottom border of $\cT_0^h$, the four vertices in the bottom border of $\cM_0^h$ and the two vertices in the bottom border of $\cI_0$ as shown in the figure in Def.~\ref{def:suppsTC}. For vertically oriented gadgets, the feedback support is rotated accordingly.  The use of the term `feedback support' refers to the set of sites where syndromes can be modified due to the action of the gadget.  We claim that $\lambda$ is a confinement region with respect to a layer of $\cG_n$ if any elementary level-0 gadget $\cG_0$ during the operation of a layer of $\cG_n$ gadgets with $\text{supp}_f(\cG_0) \not\subseteq \lambda$ has the property that $\text{supp}_f(\cG_0) \cap \lambda$ is a subset of corners of $\cG_0$. In other words, this means that the feedback support of $\cG_0$ is either contained in $\lambda$ or overlaps with it on linear vertices of $\cG_0$.  Assuming this, calling $\overline{\lambda}$ the complement of $\lambda$, we have $\cL(\cG_0)[\sigma, \varepsilon]|_{\overline{\lambda}} = \cL(\cG_0)[\sigma|_{\overline{\lambda}}, \varepsilon|_{\overline{\lambda}}]|_{\overline{\lambda}}$. Applying the same logic layer by layer, after the entire operation of $\cL(\cG_n)$, we get
\begin{equation}
 \cL(\cG_n)[\sigma, \varepsilon] = (\cL(\cG_n); \overline{\lambda})[\sigma, \varepsilon] \oplus  \cL(\cG_n)[\sigma, \varepsilon]|_{\lambda},
\end{equation}
and setting $\eta|_{\lambda} = \cL(\cG_n)[\sigma, \varepsilon]|_{\lambda}$, this proves that $\lambda$ is a confinement region.

Next, we will show that the above property is true for the boundary $n$-links of each of the gadgets, thereby proving the proposition.  This can be shown by induction.  First, consider $\cL(\cG_1)$ for $\cG_1 = \cI_1, \cT_1, \cM_1$. By direct inspection, the boundary 1-links of these gadgets have the property that any elementary gate that applies feedback at this boundary is either $\cM_0$ and $\cT_0$ whose feedback is contained entirely in a 1-link; otherwise, it is a level-0 gadget whose corners are contained in the 1-link.  This also holds for $\EC_1$, thereby proving the base case.  Then, suppose the fact is true for $\cL(\cG_n)$ where $\cG_n = \cI_n, \cT_n, \cM_n$.  We may write $\cL(\cG_{n+1})$ in terms of $\cL(\cG_{n})$ interspersed with layers of $\cL(\EC_{n})$.  The boundaries of the level-$n+1$ gadgets are always boundaries of the constituent level-$n$ gadgets, and therefore the fact must be true for the level-$n+1$ gadgets.  This proves the induction step, and thus the claim regarding confinement regions.

\end{proof}
We now have all the tools needed to complete the proof of Lemma~\ref{lemma:EC-coarse-grains}, which states that the first half of a noiseless $\EC_n$ gadget coarse-grains input syndromes to $\Sigma_n$, and that $\cM_n^{h/v}, \cT_n^{h/v}$ act on coarse-grained states in a way analogous to the action of $\cM_0^{h/v},\cT_0^{h/v}$. We complete this proof below, and additionally show the statement of the Lemma in the case when the gadgets tile the entire lattice.

\begin{proof} (\textbf{Proof of Lemma~\ref{lemma:EC-coarse-grains}})

    First, we will show the base case for $n=1$ except now for the case when the lattice (or some connected subregion) is tiled with level-1 gadgets rather than for a single gadget restricted to its support.  Suppose we tile the lattice with $\EC_1$ gadgets.  Then, the set of 1-links are confinement regions for $\EC_1$ (as well as $R_0^v$ and $R_0^h$) and therefore, one can evolve the syndromes in the bulk of these gadgets independently of syndromes on the boundaries of these gadgets and assume an arbitrary syndrome configuration on the boundary.  After the first $R_0^v$ there will be no syndromes remaining in the bulk as per the simpler argument (for part of the base case of Lemma~\ref{lemma:EC-coarse-grains}) from Sec.~\ref{sec:toriccode}.  The remaining syndromes will be confined to the boundaries, 
    % and we assume they are treated arbitrarily.  However, 
    and the next round of $R_0^v$ will coarse-grain these syndromes to level-1 vertices for the vertical boundaries, since the action of $R_0^v$ restricted to the boundary is the same as that of Tsirelson's 1D error correction gadget.  The remaining syndromes on the horizontal boundaries will be coarse-grained at level 1 by the next two rounds of $R_0^h$.  The base case for the gadgets $\cT_1$ and $\cM_1$ is much simpler: these gadgets are either stacked horizontally or vertically and correspond to stacks of gadgets that apply feedback (similar to that of Tsirelson's gadgets) on 1D lines each, which will send a level-1 syndrome input to a level-1 syndrome output corresponding to the action of a stack of appropriate level-0 gadgets.
    
    Next, we show the induction step. 
    We can write $\EC_n$ in terms of a sequence of gate gadgets $\cG_{n-1}$ and a \hi{layer of $\EC_{n-1}$ gadgets covering the lattice} between each gate gadget.  By the induction step, after the first layer of $\EC_{n-1}$, the syndromes are coarse-grained to level-$(n-1)$.  Thus, we can study the action of the rest of the gadget by ``reducing the level'' by $n-1$.  By the induction hypothesis, $\cG_{n-1}$ on such a coarse-grained input will keep the syndromes coarse-grained at level $n-1$ and the action will correspond to the action of $\cG_{0}$.   Thus, level reduction allows us to formally view $(n-1)$-level frame as an effective level-0 lattice, where each gate $\cG_{n-1}$ is replaced with a corresponding level-0 gate, and each $\EC_{n-1}$ gadget has no effect, and thus can be ignored. Upon performing level reduction on the entire $\EC_n$, we obtain an effective circuit for $\EC_1$ which coarse-grains the syndromes by the analysis in the base case. 

    Finally, we must prove that $\cM^{h/v}_n$ and $\cT^{h/v}_n$ will have the correct operation on coarse-grained syndrome configurations assuming the induction hypothesis.  Once again, we write $\cM^{h/v}_n$ and $\cT^{h/v}_n$ in terms of $\cM^{h/v}_{n-1}$, $\cT^{h/v}_{n-1}$, and $\EC_{n-1}$.  Since the nontrivial syndromes are coarse-grained to live in $\Sigma_n$, we can perform an analogous form of level reduction by $n-1$ levels to show that this analysis is equivalent to that of $\cM^{h/v}_1$ and $\cT^{h/v}_1$, which is dealt with by the base case.
\end{proof}

Finally, we prove the following fact for the toric code automaton which will be very useful later on.
\begin{fact}\label{fact:4}
Consider a \hi{covering} of the lattice with level-$n$ gadgets corresponding to a particular step of the toric code automaton. Call two neighboring gadgets $G_1$ and $G_2$ and $\Lambda_{12}$ the vertices they overlap on.  Call the supports of the two gadgets $\cR_1$ and $\cR_2$.  Then,
\begin{equation}
\cL(G)[\sigma, \varepsilon]|_{\Lambda_{12} \setminus \partial \Lambda_{12}} = (G_1 \otimes G_2; (\cR_1 \cup \cR_2)\setminus \partial (\cR_1 \cup \cR_2))[\sigma, \varepsilon]|_{\Lambda_{12} \setminus \partial \Lambda_{12}}
\end{equation}
\end{fact}
\begin{proof}
        We first note that each boundary of a level-$n$ gadget is a separate confinement region for level-$n$ gadgets by Prop.~\ref{prop:confinement-for-TC}.  Furthermore, if $\Lambda_1$ and $\Lambda_2$ are the union of all boundary segments of gadgets $G_1$ and $G_2$, then $\Lambda_1 \cup \Lambda_2$ is a confinement region.  This follows directly from the arguments in Prop.~\ref{prop:confinement-for-TC}.  In addition, calling $\cR_1 \cup \cR_2 \equiv \cR$, we can write the outer boundary for the pair of gadgets as $(\partial \cR_1 \cup \partial \cR_2) \setminus(\partial \cR_1 \cap \partial \cR_2) \equiv \partial\cR$. It is a confinement region as it is a union of boundary segments of $G_1$ and $G_2$.  

Applying Fact~\ref{fact:1}, with this choice of confinement region, we find
\begin{equation}
\cL(G)[\sigma, \varepsilon] = \eta \oplus (G_1 \otimes G_2; (\cR_1 \cup \cR_2) \setminus \partial \cR)[\sigma, \varepsilon] \oplus \cL(G; \cL_0 \setminus (\cR_1 \cup \cR_2 \setminus \partial R))[\sigma, \varepsilon],
\end{equation}
where $\eta = \cL(G)[\sigma, \varepsilon] |_{\partial \cR}$. Restricting the output to $\Lambda_{12} \setminus \partial \Lambda_{12} \subset \cR \setminus \partial\cR$ proves the fact.
\end{proof}
This implies that the syndromes on a boundary between two gadgets (minus the endpoints of the boundary) can be reconstructed only by knowing the syndromes in the bulks of both of the gadgets.

\subsection{Modification of the exRecs method}
\label{sec:new-exrecs}

\subsubsection{Main definitions}

Now, we will define the technology needed to extend the exRecs method to the ($X$-part of the) toric code automaton. We will assume a fixed level $k$ for inner simulation throughout. Let us first establish some definitions that we will use to count the number of faults due to noise:
\begin{definition}[Failure in the support of a gadget]
    We say that a level-0 gadget failed in the support of a $\EC_{n}$ or $\cG_{n}$ gadget if a failure occurred on this level-0 (elementary) gadget and the gadget's support lies entirely within the support of the $\EC_{n}$ or $\cG_{n}$ gadget.
\end{definition}

\begin{definition}[exRec]
Assuming some fixed level $k$ for inner simulation, an exRec for a gate  $\cG_k$ consists of a spacetime region corresponding to the operation $\widetilde{\EC}^{\otimes K} \circ \cG_k \circ \widetilde{\EC}^{\otimes K}$ where $K$ is the number of $k$-cells in the spatial support of $\cG_k$. We say that such an exRec is centered at gadget $\cG_k$. We call $\supp_q(R)$ the spatial qubit support of an exRec $R$.
\end{definition}
\begin{definition}[Good/bad exRecs]
An exRec is good if the number of level-0 gadget failures in its spacetime support is $\leq 1$.  Otherwise, it is bad.
\end{definition}
We will utilize exRecs to prove the error suppression property for the fault-tolerant prescription $\widetilde{FT}(\cdot)$ corresponding to inner simulation. Notice that the exRecs for the toric code gadgets overlap not only in time but also in space, specifically on $k$-links. Thus, in addition to dealing with spatial truncation, we will also need to utilize a modified procedure for dealing with temporal truncation. Temporal truncation as well as the prescription for ``pulling'' the $\ast$-decoder through the layers of simulation will be explained in Subsec.~\ref{sec-proof-exrec-method}. 

We first introduce the concept of spatial truncation:

\begin{definition}[Spatial truncation]
A spatially truncated exRec for gadget $\cG_k$ is one where one or more $k$-links or sites in $\Sigma_k$ are removed from the exRec's spatial support throughout the temporal duration of the exRec.  

Denoting by $\Lambda$ the spatial region that is removed and the set  of times corresponding to the operation of the exRec by $\tau$, the truncation region of an exRec is  a spacetime region of removed links and sites $\Lambda \times \tau$.
\end{definition}

\begin{definition}[Good/bad truncated exRecs]
Consider a spatially truncated exRec $R$ centered around gadget $\cG_k$. Count the faults belonging to this exRec in the following way. For each failure of an elementary gadget $\cG_0^{(i)}$ fully contained in $\widetilde{EC}^{\otimes K} \circ \cG_k \circ \widetilde{EC}^{\otimes K}$, consider the restriction of the associated operator $\cO$ applied by these gadgets (including noise) to the truncated boundary $\lambda$ with the bad exRec that is causing truncation. If $\supp (\cO) \subseteq \lambda$, we formally count the fault of this gadget towards the neighboring bad exRec that caused the truncation. Otherwise, we count this fault towards the truncated exRec under consideration. 
\end{definition}

We will use the following algorithm for exRec assignment:
\begin{algorithm}[Assigning support to exRecs] \label{alg:goodbad_assignment}
Assume a fixed level $k$ for the inner simulation and let the associated outer simulation applied to circuit $\cC$ yield the simulated circuit $\widetilde{FT}(\cC)$.   Call $T$ the depth of $\cC$. We use the time index $t$ to enumerate the layers of the simulated level-$k$ gates $\cG_k$ associated to the gates in $\cC$. Then:
 \begin{enumerate}
 \item Construct exRecs associated with gates $\cG_k$ at a specific time slice $t \leq T$. 
 \item Pick an arbitrary order to traverse the exRecs in space at this time slice.  If a given exRec is bad, spatially truncate all surrounding exRecs at the same time slice which have overlapping spatial support with the given bad exRec.
 \end{enumerate}
 This assigns support and good and bad labels for the last layer of exRecs in a given simulated circuit (illustrated in Figure~\ref{fig:truncation}).
\end{algorithm}
As we will see shortly, in order to pull layers of the hierarchical ideal decoder  (Def.~\ref{def:ideal-decoder-TC}) through a simulated circuit, we first use the algorithm above to assign supports of exRecs and mark them as good and bad and then pull the ideal decoder through this layer.  We then repeat this process layer by layer.

To explain this in more depth, we need to introduce the $*$-decoder, which will work analogously to its operation in the usual exRecs method (see Def.~\ref{def:conventional-ast-decoder} and discussion afterwards), except that now it is applied to an entire \emph{layer} of exRecs at once: 
\begin{figure}[!htbp]
    \centering
    \includegraphics[width=0.6\textwidth]{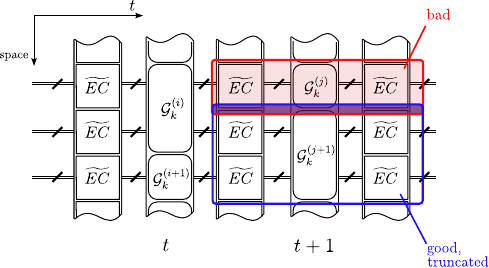}
    \caption{Schematic showing an outer simulation for the toric code automaton and the outcome of one application of Algorithm \ref{alg:goodbad_assignment}. The truncation region is schematically shaded dark.
    In this figure and in the rest of the paper, the crossed wires are used to schematically denote multiple qubits in the support of the operations.  
    }
    \label{fig:truncation}
\end{figure}

\begin{definition}[$*$-decoder for the toric code] \label{def:star-decoder-TC}
Assume a fixed level $k$ for the inner simulation and the associated outer simulation prescription that we use to obtain a simulated circuit $\widetilde{FT}(\cC)$ and run Algorithm~\ref{alg:goodbad_assignment} to label good and bad exRecs and perform truncation.  We append a $*$-decoder $\widetilde{D^*}$ at the end of the circuit and define its operation to be:
\begin{enumerate}
    \item[(a)] For the layer of exRecs that directly precedes the decoder, the decoder stores the syndromes of the noise operators applied in their spacetime support in its classical register, and
    \item[(b)] applies an ideal decoder $\widetilde{D}$. 
\end{enumerate} 
The output of the $\ast$-decoder is the output of the ideal decoder $\widetilde{D}$ together with a classical register containing the data in (a).  Partial tracing out the classical register, this defines a quantum channel supported on qubits of the outputs of bad exRecs.
\end{definition}

Similarly to the ideal decoder, the $\ast$-decoder is not an actual physical operation but rather a mathematical tool that we will use to `undo' levels of simulation in the toric code automaton. Schematically, the $\ast$-decoder's operation can be summarized as follows:
  \begin{equation*}
        \centering
        \includegraphics[width=1\linewidth]{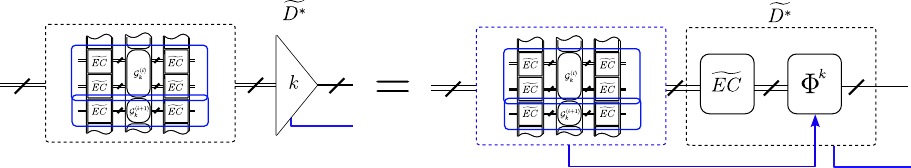}
        \label{fig:confinement1}
    \end{equation*}
where the block on the left stands for the layer of exRecs, and the blue outgoing wire stands for the register with classical information. On the right, the operation of the star decoder can be written as an ideal decoder that consists of applying a noiseless layer of $\widetilde{EC}$ followed by a $k$-fold pushforward map $\Phi^k$ (which is just $k$ applications of $\Phi$), along with the relevant classical register.

We will also need to define a different notion of a filter than the one used for concatenated codes.
Recall that the role of a filter is to project onto a subspace of states that are not too ``far'' from a codestate at the corresponding level of the simulation. A natural way to define this for the toric code is to compare the state with a level-$k$ coarse-grained version obtained by applying $\cL(\EC_k)$. Thus, for each $k$-cell, we define an $m$-cell filter (with $m \leq k-1$) to project on the subspace of states that differ from their coarse-grained version by an operator (call it $\cO'$) whose support fits in an $m$-cell. We also define a $\varnothing$-filter as the filter that forbids the state to differ from a level-$k$ coarse-grained state. 

However, the definition that we present below is somewhat more involved than this description to account for the fact that neighboring $k$-cells share boundaries.  In this case, we first need to specify which filters are assigned to all neighboring cells (if any). Considering boundaries of a specific $k$-cell $\cC_k$ with an $m$-cell filter, there are several scenarios. For a given boundary, if there is no filter on the other $k$-cell that shares this boundary, then we allow anything to happen on the boundary. If there is another $m$-cell filter or $\varnothing$-filter on the other $k$-cell, the $m$-cell within $\cC_k$ containing $\cO'$ is allowed to touch the shared boundary. This applies to every $k$-cell with a filter.
All these scenarios are summarized in the definition below and in Fig.~\ref{fig:m-cellfilter-TC}.

\begin{figure} [!htbp]
    \centering
    \includegraphics[width=0.8\textwidth]{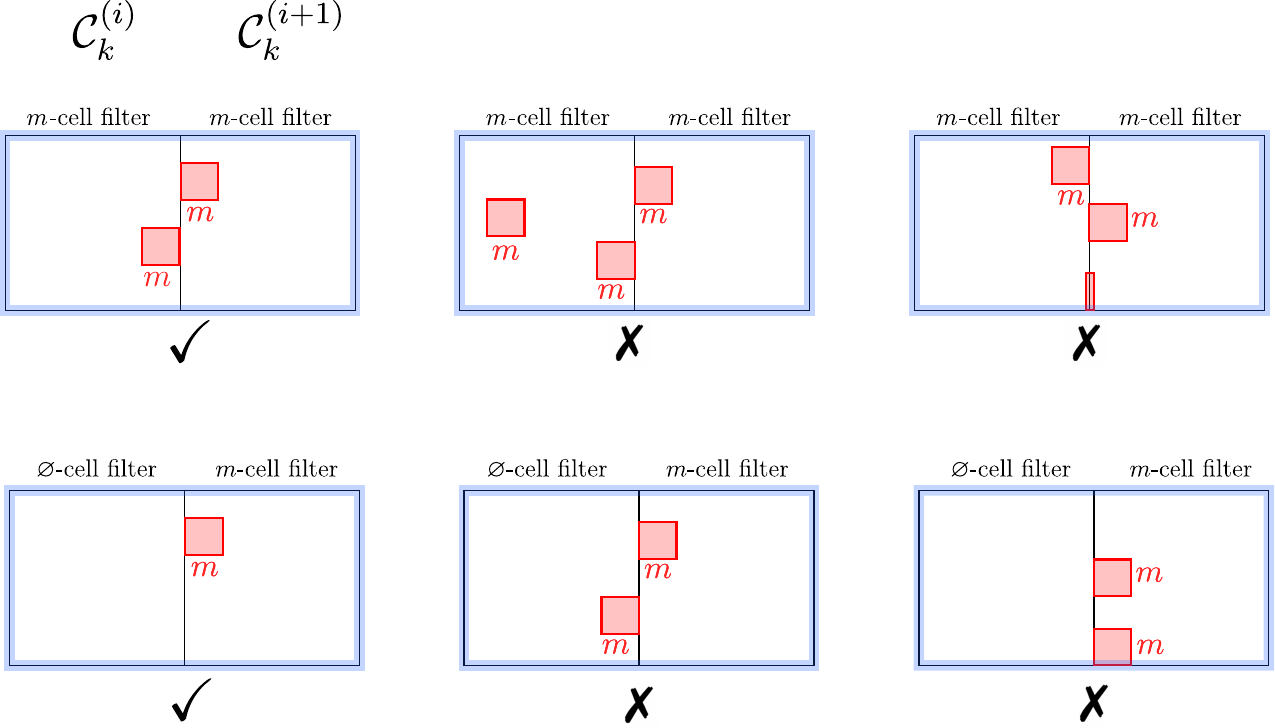}
\caption{
A pair of $k$-cells $\cC_k^{(i)}$ and $\cC_k^{(i+1)}$ with various types of cell filters. The type of the cell filter imposed on $\cC_k^{(i)}$ and $\cC_k^{(i+1)}$ cell is indicated at the top.
All other $k$-cells (i.e. outside of the ones depicted) have no filter assigned to them. 
The check marks indicate that the state passes the appropriate filter, and the cross marks indicate that it does not.
For the examples that pass the filter configuration, for each cell with an $m$-cell filter, there exists a relative damage operator $\cO'$ at level $k$ (i.e. an operator relating the state to its  level-$k$ coarse-grained version) that fits inside an $m$-cell, shaded in red. 
For the examples that do not pass the filter configuration, \emph{any} relative damage operator violates the filter conditions. The red boxes show $m$-cells containing the support of a relative damage operator. }
    \label{fig:m-cellfilter-TC}
\end{figure}
\begin{definition}[$m$-cell filter at level $k$  for the toric code] \label{def:mprime-TC}
We say that a layer of filters has been added or assigned to a toric code state $\ket{\psi}$ with syndrome $\sigma$ if each $k$-cell with $k > m$\footnote{Unless otherwise specified, $k$ coincides with the level of inner simulation for the toric code automaton.} has been assigned an $m$-cell filter, $\varnothing$-cell filter, or no filter.  Each filter associated with a given $k$-cell $\cC_k$ is a projector onto the subspace of states that satisfy the conditions below. We will say a state ``passes'' the filter if it is in the image of the associated projector.  

Suppose  the state $\ket{\psi}$ is related to the initial logical state $\ket \phi$ of the toric code by $\ket \psi = \cO \ket \phi$, with syndrome configuration $\sigma$. Assuming a filter assignment, we obtain a level-$k$ coarse-grained syndrome configuration by applying a noiseless error correction layer $\cL(\EC_k)[\sigma]$. We denote by $\cO_{cg}$ the product of all feedback operators applied by $\cL(\EC_k)$. Define an equivalence class of operators $[\cO_{cg}] = \{\cO: \cO \cO_{cg} \in \cS\}$. 
For the filter conditions under a given filter assignment to be met, we require that there exists an operator $\cO' \in  [\cO_{cg}]$ with qubit support $M$,  such that the following is satisfied.

\noindent For each cell $\cC_k$ that has an $m$-cell filter:
\begin{itemize}
    \item   $M \cap  (\cC_k \setminus \partial \cC_k)$ fits within an $m$-cell.  Here, supports are taken to be on qubits rather than vertices. 
    \item For every boundary $\Lambda$ between $\cC_k$ and a neighboring cell $\cC_k'$:
    %additional properties for it to pass the $m$-cell filter on $\cC_k$: 
    \begin{itemize}
        \item [(i)] If the  neighboring cell $\cC_k'$  has an $m$-cell filter, then there exist $m$-cells $\cC_m \in \cC_k$ and $\cC_m' \in \cC_k'$ such that:
    \begin{equation} \label{m-cell-1}
    M \cap \Lambda = ((M \cap \cC_{m}) \cup (M \cap \cC_{m}')) \cap \Lambda.
      \end{equation}  
      \item [(ii)] If the neighboring $\cC_k'$  has a $\varnothing$-cell filter, then
    there exists an $m$-cell $\cC_m \in \cC_k$ such that:
    \begin{equation}
    M \cap \Lambda = (M \cap \cC_{m}) \cap \Lambda.
      \end{equation}  
    \end{itemize}
    \item If there is no filter on the neighboring cell $\cC_k'$, no condition is imposed on the boundary $\Lambda$.
\end{itemize}
For each cell $\cC_k$ that has a $\varnothing$-cell filter:
\begin{itemize}
    \item   $M \cap  (\cC_k \setminus \partial \cC_k) = \varnothing$. 
    \item For every boundary $\Lambda$ between $\cC_k$ and a neighboring cell $\cC_k'$:
    
    %additional properties for it to pass the $\varnothing$-cell filter on $\cC_k$:
    \begin{itemize}
        \item [(i)] If the  neighboring cell $\cC_k'$  has an $m$-cell filter, then there exists an $m$-cell $\cC_m' \in \cC_k'$ such that:
    \begin{equation}
    M \cap \Lambda =  (M \cap \cC_{m}') \cap \Lambda.
      \end{equation} 
        \item[(ii)]  If the neighboring $\cC_k'$  has a $\varnothing$-cell filter, then
    \begin{equation} \label{m-cell-2}
    M \cap \Lambda = \varnothing.
      \end{equation}  
    \end{itemize} 
        \item[(iii)] If there is no filter on $\cC_k'$, no condition is imposed on $\Lambda$.
\end{itemize}
\end{definition}

We will use a convention where $m = 0$ will stand for $\varnothing$-filter, whereas $m \geq 1$ will stand for an $m$-cell filter.  The operator $\cO_{cg}$ in the definition above relates the syndrome configuration of the state $\ket{\psi}$ to the level-$k$ coarse-grained version of this syndrome configuration.

Thus, in the presence of a layer of filters, all $k$-cells can be grouped into a set of contiguous clusters where each $k$-cell in a cluster is covered by either an $m$-cell or $\varnothing$-cell filter, and $k$-cells not in any cluster do not have a filter assignment. The conditions in Def.~\eqref{def:mprime-TC} are imposed on the syndrome configuration in the interior of each contiguous cluster minus its boundary, since no constraint is imposed on boundaries of clusters or $k$-cells not in any cluster.
A cluster formed by two $k$-cells $\cC_k^{(i)}$ and $\cC_k^{(i+1)}$ is shown in Fig.~\ref{fig:m-cellfilter-TC} along with examples of damage configurations that either pass or do not pass a pair of adjacent filters of a given type.

Moving forward, when a state $\ket{\psi} = \cO \ket{\phi}$ from the definition above passes an $m$-cell filter, we will also say that the damage $\cO$ passes the $m$-cell filter. We will refer to the operator $\cO'$ from the definition above as relative damage at level $k$.  When we say that the damage has been ``removed'' or ``cleaned up'' at level $k$, we mean that a correction operation $\cO''$ has been applied by the automaton so that $\cO'\cO'' \in \cS$.

We can now define the EC and gate properties for the toric code automaton:

\begin{definition}[$(\mathrm{Gate})_{k,m}$ and $(\EC)_{k,m}$ conditions for toric code] \label{def:TC_gateECprime}
 
Consider a layer of level-$k$ gadgets. For the  $(\mathrm{Gate} \ A)_{k,m}$ conditions, as well as the $(\mathrm{Gate} \ B)_{k,m}$ and $(\mathrm{EC} \ B)_{k,m}$ conditions, we assume a layer with $\varnothing$ or $m$-cell filters present on each $k$-cell in the input state (see Def.~\ref{def:mprime-TC}). 
 
The $(\mathrm{Gate} \ A)_{k,m}$ and $(\mathrm{EC} \ A)_{k,m}$ conditions append a layer of cell filters after the layer of level $k$ gadgets via
   \begin{equation*}
        \centering
        \includegraphics[width=1\linewidth]{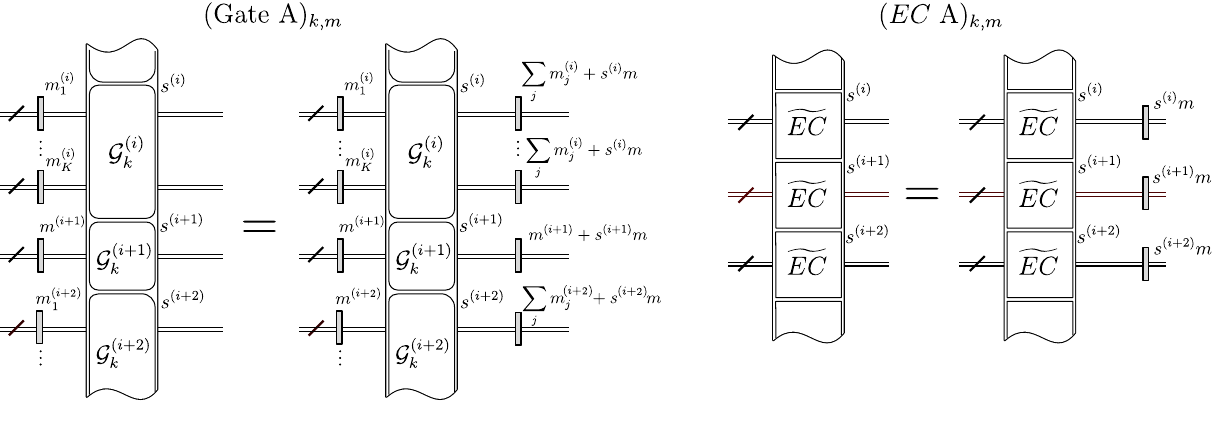}
        \label{fig:TCA}
    \end{equation*}
where we schematically show the input and output qubits as doubled lines (which, in reality, include qubits shared between the supports of neighboring $k$-cells).  $s^{(i)}$ denotes the number of failures of elementary gadgets within the corresponding level-$k$ gadget and $m_j^{(i)} \in \{0,m\}$ determines the size of the cell filter imposed on the $j$-th $k$-cell in the support of $i$-th gadget. In addition, we require that $\sum_{j=1}^{K^{(i)}} \frac{m^{(i)}_j}{m} + s^{(i)} \leq 1$ for each gadget in $(\mathrm{Gate} \ A)_{k,m}$ and $s^{(i)}\leq 1$ in $(\mathrm{EC} \ A)_{k,m}$.
If this holds, a $\left (\sum_{j=1}^{K^{(i)}} m^{(i)}_j + s^{(i)} m\right )$-cell filter is assigned to the output cells of each gate $\cG_k^{(i)}$. 
For the $(\mathrm{EC} \ A)_{k,m}$ condition, a $(s^{(i)} m )$-cell filter is assigned to the output $k$-cell for each $\widetilde{EC}$ gadget.

For the $(\mathrm{Gate} \ B)_{k,m}$ and $(\mathrm{EC} \ B)_{k,m}$ conditions, each gate or $\widetilde{EC}$ gadget in the layer satisfies the following conditions:
  \begin{equation*}
        \centering
        \includegraphics[width=1\linewidth]{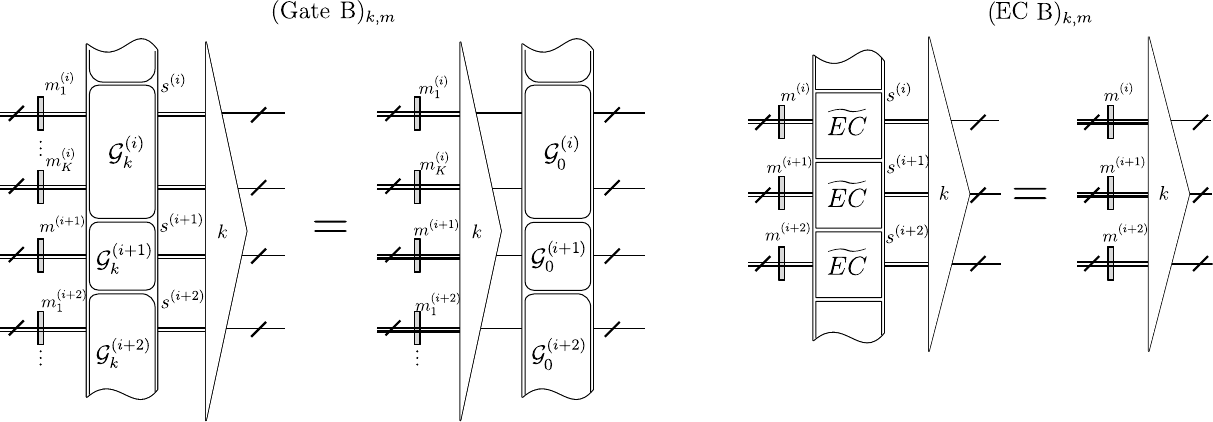}
        \label{fig:TCB}
    \end{equation*}
where the triangle represents the ideal decoder. 
For the $(\mathrm{Gate} \ B)_{k,m}$, we assume that the condition $\sum_{j=1}^{K^{(i)}} \frac{m_j^{(i)}}{m} + s^{(i)} \leq 1$ is fulfilled for each gate gadget. For the $(\mathrm{EC} \ B)_{k,m}$ condition, we similarly assume that $\frac{m^{(i)}}{m} + s^{(i)} \leq 1$ is fulfilled for each gadget. 
\end{definition}
The A-type condition is imposed on a whole layer of gadgets at once because all cell filters are determined simultaneously (since information about which filter to place depends on neighboring filters).  In addition, because the ideal decoder is \emph{global}, its action cannot be separated into disjoint blocks, and, as a consequence, the B-type condition is also imposed on the entire layer.

Note also that for the $\EC$ B and Gate B conditions (we will largely omit the subscript $k,m$ for convenience) one pulls the entire layer of gadgets through the ideal decoder as opposed to pulling one codeblock at a time like for concatenated codes\footnote{An alternative definition of $\EC$ B and Gate B condition where one pulls the action of separate gadgets through the ideal decoder (which would look much more similar to the one in the original exRecs method) could also work. However, it requires a set of careful definitions because the supports of gadgets overlap. As we will shortly see, the layer-wise definition of these conditions adopted here is sufficient and provides a simpler framework for proving fault tolerance. }.

\subsubsection{Confinement in the operator language}

For the proof of our main theorem, we will need to translate ideas like confinement to the operator perspective.  In this subsection, we discuss what the confinement property implies for operators that are applied due to the action of the gadget.

Recall that an elementary gadget acting at time $t$ applies local Pauli feedback $F_t$ conditioned on the syndromes of the input state in its support, followed by the application of the Pauli noise operator $P_t$ restricted to its support. The total operator applied as a consequence of the gadget action is their product $P_t F_t$. In the toric code automaton, elementary gadgets are defined such that a whole layer of them at time $t$ can operate simultaneously. As a consequence, to determine the result of their operation, we can consider each of them independently and then add their outputs.  This cannot be done for composite gadgets because they share boundaries.  As a result, we can still define the operator applied due to the action of a composite gadget but now we have to consider an entire layer of gadgets at once:

\begin{definition} [Operator applied by composite gadget]
    Consider a layer of gadgets, denoted $\cL(\cG)$.  Define $t=t_0$ to be the start of the layer of gadgets and $t=t_1$ to be the end.  For a noise realization $\bm H$ with operator history $\vec P = \{ P_1, ... P_T \}$ (see Def.~\ref{def:gadget_error_model}) and an input state with syndrome $\sigma$, call $\{F_{t}(\sigma_t)\}_{t = t_0}^{t_1}$ to be the Pauli feedback operators applied by the layer of gadget at each timestep (which depend on the gadgets, input syndrome and the noise).  
    
    Define the operator $\cO[\sigma, \vec P; \cL(\cG)]$ to be $\prod_{t=t_0}^{t_1}  P_t F_t(\sigma_t)$, which is the cumulative operator applied by the operator history and the gadget feedback.    
    We will often omit the argument $\cL(\cG)$ since it can usually be deduced from the context.
    
    We will also denote $\cO_{\cR}[\sigma,\vec P; \cL(\cG)]$ to be the restriction of the operator $\cO[\sigma,\vec P; \cL(\cG)]$  to the qubits in the region $\cR$. In particular, we can write it as $\prod_{t=t_0}^{t_1} P_t\big|_{\cR} F_t(\sigma_t)\big|_{\cR} $.
\end{definition}

Recall that, unless otherwise specified, we assume all operators applied due to noise and feedback are of the Pauli-$X$ type.

\begin{definition} [Restriction of operator arguments]
    Consider an input state with syndromes $\sigma$ and noise realization $\bm H$ with noise operator history $\vec P = \{P_1, \cdots, P_T\}$. Consider also a spatial region $\cQ$. Define $\cO\left [\sigma\big|_{\cQ}, \vec P\big|_{\cQ}; \cL(\cG) \right ] $ to be the operator applied by the layer of gadgets $\cL(\cG)$ assuming that input syndromes are restricted to the vertices in $\cQ$ and the noise is restricted to the qubits in $\cQ$. The notation $\vec P\big|_\cQ$ stands for $\{ P_1\big|_\cQ, \cdots, P_T\big|_\cQ \}$. 
\end{definition}

Note that restricting operator arguments to a given region is not the same as performing the restriction at the end.

\begin{lemma}[Operator confinement part 1] \label{lemma:op-conf-1}
    Consider a layer of level-$k$ gadgets $\cL(\cG_k)$ of the toric code automaton acting on an input state with syndromes $\sigma$ along with a noise realization $\bm H$ with operator history $\vec P$. Call $\cR_j$ the qubit support of a gadget $\cG_k^{(j)}$. Fixing a particular gadget $\cG_k^{(i)}$, let $\Lambda$ be the set of $k$-links in the boundary of $\cR_i$, define the extended operator support $\cQ_i$  of $\cG_k^{(i)}$ to be
    \begin{equation}
    \cQ_i = \cR_i \cup \left(\bigcup_{j: \, \cR_j \cap \cR_i \subseteq \Lambda} \overline{\cR}_j\right).
    \end{equation}
    where we introduce the notation $\overline{\cR}_j  = \cR_j \setminus \partial \cR_j $. Then 
    \begin{equation}
        \cO_{\cR_i}[\sigma, \vec P; \cL(\cG_k)] = \cO_{\cR_i}[\sigma \big|_{\cQ_i}, \vec P\big|_{\cQ_i}; \cL(\cG_k)]. 
    \end{equation}
\end{lemma}
\begin{proof}
    The proof is split into two steps.  First, we show that the operator $\cO_{\overline{\cR}_i}[\sigma, \vec P; \cL(\cG_k)]$ only depends on syndromes and noise realization restricted to $\overline{\cR}_i$.  Then we show that the operator $\cO_{\lambda_{ij}}[\sigma, \vec P; \cL(\cG_k)]$ supported on the boundary $\lambda_{ij}$ (which is a set of $k$-links) between $\cR_i$ and $\cR_j$ only depends on syndromes and the noise in the region $\overline{\cR}_i \cup \overline{\cR}_j \cup \lambda_{ij}$.  The Lemma then immediately follows because $$\cO_{\cR_i} = \cO_{\overline{\cR}_i}[\sigma, \vec P; \cL(\cG_k)] \otimes \left( \bigotimes_j \cO_{\lambda_{ij}}[\sigma, \vec P; \cL(\cG_k)]\right).$$

    To show the first step, we write $\cO_{\overline{\cR}_i}[\sigma, \vec P; \cL(\cG_k)] = \prod_{t=t_0}^{t_1}  P_t\big|_{\overline{\cR}_i} F_t\big|_{\overline{\cR}_i}$. We first argue that the feedback operators in the qubit support of $\overline{\cR}_i$ only depend on syndromes in this region. Any level-0 constituent gadget which is a part of higher-level gadget of the toric code automaton either (a) is conditioned on vertices in region $\overline{\cR}_i$ and applies feedback with qubit support in this region or (b) is conditioned on vertices in one of the boundary $\lambda_{ij}$ and applies feedback with qubit support on $\lambda_{ij}$ (an induction-based proof of this is essentially given in the proof of Proposition~\ref{prop:confinement-for-TC}). Restricting to $\overline{\cR}_i$, we see that the total feedback operator must be conditioned on syndromes supported in $\overline{\cR}_i$. The syndromes in this region at each time step will depend on the input syndromes in $\overline{\cR}_i$ and the syndromes of the noise operators $P_t$ for $t = [t_0, t_1]$ restricted to $\overline{\cR}_i$, which we denote $\varepsilon_{\overline{\cR}_i}$.  Because $\varepsilon_{\overline{\cR}_i}$ can be uniquely determined by restricting the operator itself to $\overline{\cR}_i$ (i.e.  $P_t\big|_{\overline{\cR}_i}$), we have
    \begin{equation}
        \cO_{\overline{\cR}_i}[\sigma, \vec P; \cL(\cG_k)] = \cO_{\overline{\cR}_i}[\sigma\big|_{\overline{\cR}_i}, \vec P \big|_{\overline{\cR}_i}; \cL(\cG_k)]. 
    \end{equation}

    Next, we argue that $\cO_{\lambda_{ij}}[\sigma, \vec P; \cL(\cG_k)] $ depends only on syndromes in $\overline{\cR}_i\cup \overline{\cR}_j \cup \lambda_{ij}$. First, we note that endpoints of $\lambda_{ij}$ are linear vertices of all gadgets of level $s \leq k$ by Proposition~\ref{prop:confinement-for-TC}, and thus the operator that is applied along the link does not depend on syndromes at these points.  Using Proposition~\ref{prop:confinement-for-TC}, we note that the region $\partial \cR_i \cup \partial \cR_j \setminus \lambda_{ij}$ is a confinement region.  Therefore, using Fact~\ref{fact:0} for this region as well as the observation that the level-0 gadgets that apply feedback supported on qubits in $\lambda_{ij}$ are only conditioned on syndromes in $\lambda_{ij}$, we write
    \begin{equation}
        \cO_{\lambda_{ij}}[\sigma, \vec P; \cL(\cG_k)] = \cO_{\lambda_{ij}}[\sigma\big|_{\overline{\cR}_i\cup \overline{\cR}_j \cup \lambda_{ij}}, \vec P \big|_{\overline{\cR}_i\cup \overline{\cR}_j \cup \lambda_{ij}}; \cL(\cG_k)]. 
    \end{equation}
    Finally, using that the union of the regions $\overline{\cR}_i\cup \overline{\cR}_j \cup \lambda_{ij}$ for all $j$ (i.e. corresponding to all gadgets $\cG_k^{(j)}$ neighboring $\cG_k^{(i)}$) is $\cQ_i$, we obtain the statement of the Lemma. 
\end{proof}

\begin{lemma}[Operator confinement part 2]  \label{lemma:op-conf-2}
    Under the assumptions of the previous Lemma, consider instead the truncated support of the gadget $\cG^{(i)}_k$ defined as $\widetilde{\cR}_i = \cR_i \setminus \widetilde{\Lambda}$ where $\widetilde{\Lambda} = \{\lambda_j\}$ is the subset of the set of four boundaries that are truncated. Then, the following holds:
    \begin{equation}
        \cO_{\widetilde{\cR}_i}[\sigma, \vec P; \cL(\cG_k)] = \cO_{\widetilde{\cR}_i}[\sigma \big|_{\widetilde{\cQ}_i}, \vec P\big|_{\widetilde{\cQ}_i};\cL(\cG_k)]
    \end{equation}
    where the ``truncated'' extended operator support is defined as 
    \begin{equation}
    \widetilde{\cQ}_i = \cQ_i\setminus \bigcup_{j:\, \cR_i \cap \cR_j \subseteq \widetilde{\Lambda}}  \cR_j.
    \end{equation}
    Namely, we remove the supports of gadgets that neighbor with $\cG^{(i)}_k$ over the truncated boundaries.
\end{lemma}
\begin{proof}
    First decompose $$\cO_{\widetilde{\cR}_i} = \cO_{\cR_i \setminus \partial \cR_i} \otimes \left ( \bigotimes_{j : \lambda_{ij} \notin \widetilde{\Lambda}} \cO_{\lambda_{ij}}\right)$$
    and apply considerations from the proof of previous Lemma to each of the operators in the tensor product.
\end{proof}

\subsubsection{Proof of the modified exRecs method} \label{sec-proof-exrec-method}

In this subsection, we will finally prove that, given a circuit $\cC$ for the $X$-type automaton, one iteration of $X$-type outer fault tolerant simulation $\widetilde{FT}(\cC)$ which experiences a $p$-bounded decoupled gadget error model will result in an $Ap^2$-bounded decoupled gadget error model for $\cC$.  This will be shown assuming that the automaton gadgets satisfy the Gate and EC conditions from Def.~\ref{def:TC_gateECprime} for some $k,m$.  The proof will require us to heavily use the operator language rather than the syndrome language, and the results of the previous subsection will be useful. In general, iterating this fault tolerant simulation will ensure exponential logical error suppression and a threshold. In the next subsection, we will prove that the particular gadgets that we use satisfy the gate and EC properties, and after combining the $X$ and $Z$ sectors of error correction together, will then prove a threshold for the full quantum automaton. 

 \begin{proposition} \label{prop:goodcorrect_TC}
Consider an inner simulation and an associated outer simulation  $\widetilde{FT}(\cdot)$. Assume that the gate and error correction gadgets satisfy $(\mathrm{Gate})_{k,m}$ and $(\EC)_{k,m}$ for some $k$ corresponding to the inner simulation level and some $m$. 

Consider a layer of level-$k$ gates along with leading and trailing layers of the level-$k$ EC gadget, which we will denote $\cL(\widetilde{\EC})\circ \cL ( \cG_k) \circ \cL (\widetilde{\EC})$.  Follow this layer by an ideal hierarchical level-$k$ decoder which we denote $\widetilde{D}$. If all exRecs covering $\cL(\widetilde{\EC})\circ \cL ( \cG_k) \circ \cL (\widetilde{\EC})$ are good, then
   \begin{equation}
       \widetilde{D} \circ  \cL(\widetilde{\EC})\circ \cL ( \cG_k) \circ \cL (\widetilde{\EC}) = \cL ( \cG_0) \circ \widetilde{D}  \circ \cL (\widetilde{\EC})
   \end{equation}
   where each level-0 gate gadget in $\cL(\cG_0)$ on the right hand side is the level-reduced and noiseless version of the associated level-$k$ gate gadget. 
\end{proposition}
 \begin{proof}
    As reviewed in  Prop.~\ref{prop:goodcorrect}, the proof from \cite{gottesman2009introductionquantumerrorcorrection} can be adapted almost verbatim to prove this proposition: specifically, $\EC$ and Gate conditions are applied in the same way as in the proof of Prop.~\ref{prop:goodcorrect}.  The only difference is that an entire layer of gadgets is pulled through the ideal decoder simultaneously whenever $\EC$ B or Gate B conditions are applied. 
 \end{proof}

Before proving a version of the above result that allows for some bad exRecs to be present, we need an additional definition that will suitably modify the operation of a layer of gadgets whenever the layer intersects with the support of bad exRecs. This procedure replaces the usual truncation procedure in the conventional exRecs formalism, which is discussed in App.~\ref{app:exRec}.

\begin{definition}[Temporal truncation] \label{def:temporal-truncation}
    Assume a fault-tolerant simulation of circuit $\mathcal{C}$ of elementary gates using the inner and outer simulation prescription of Secs.~\ref{subsec:inner-simulation} and \ref{subsec:outer-simulation}.

    Consider a layer of exRecs corresponding to $\cL_{tr}(\widetilde{\EC}) \circ \cL(\cG_k) \circ  \cL_{le}(\widetilde{\EC})$, where we label the leading and trailing layers of error correction gadgets as `$le$' and `$tr$', respectively. Assume that the $\cL(\cG_k)$, $\cL_{tr}(\widetilde{\EC})$, and the succeeding layer of level-$k$ gates (not contained in the layer of exRecs under consideration) are start at times $t_0$, $t_1$ and $t_2$, see Fig.~\ref{fig:example-temporal}. The following procedure proceeds backward in time.

    Suppose that $t_0$ labels the latest time slice for which exRecs centered at this timeslice have not been temporally truncated. We run Algorithm~\ref{alg:goodbad_assignment} to label and spatially truncate the exRecs in this layer.  We then define temporal truncation as a procedure that removes all level-0 gadget failures in all bad exRecs centered at time $t_0$.  We denote this by replacing $\cL_{tr}(\widetilde{\EC}) \circ \cL(\cG_k) \circ  \cL_{le}(\widetilde{\EC})$ with $\widetilde{\cL^{t_0}_{tr}}(\widetilde{\EC}) \circ \widetilde{\cL^{t_0}}(\cG_k) \circ \widetilde{\cL^{t_0}_{le}}(\widetilde{\EC})$.  If, upon performing this procedure from later times to earlier ones, the last $\cL_{le}(\widetilde{\EC})$ was instead $\widetilde{\cL^{t_2}_{le}}(\widetilde{\EC})$ (i.e. having been modified by temporal truncation at $t_2$\footnote{Note also that temporal truncation will lead to reassignment of good/bad exRecs during the next iteration of the algorithm. }), the result after temporal truncation at $t_0$ is denoted as $\widetilde{\cL^{t_0,t_2}_{tr}}(\widetilde{\EC}) \circ \widetilde{\cL^{t_0}}(\cG_k) \circ \widetilde{\cL^{t_0}_{le}}(\widetilde{\EC})$.
    %We note that the gadgets in the layer $\cL_{tr}(\widetilde{\EC})$ are contained in overlapping exRecs centered at time slices $t_0$ and $t_2$, while $\cL_{tr}(\widetilde{\EC}) \circ \cL(\cG_k)$ are contained in $t_0$-centered exRecs only. 
\end{definition}

Let us provide some intuition behind the temporal truncation procedure. To determine the error model in the simulated circuit and show the error rate suppression, we will ``pull'' the $\ast$-decoder backward in time through the circuit. Similarly to concatenated codes, this first requires us to introduce some notion of truncation such that the spacetime supports of bad exRecs do not overlap with any other exRecs (we have already introduced the notion of spatial truncation, which does not fully solve this issue). In addition, unlike for concatenated codes, ``pulling'' the $\ast$-decoder through a single exRec is not well-defined. Instead, we pull the $\ast$-decoder through an entire layer of exRecs at a time. 

For the toric code, we will need to use temporal truncation  to appropriately pull the $\ast$-decoder through the circuit; see Theorem~\ref{thm:TCerrsupp} for a more precise description of this.  For all bad exRecs in a given layer, we will determine the effect of faults in them as compared to the output of the temporally truncated circuit. The difference is captured by some operator $\cO$. Thus, the output after this layer of exRecs is equivalent to the output of a temporally truncated layer of exRecs and subsequent application of  $\cO$. We then move the operator $\cO$ past the $\ast$-decoder to obtain a coarse-grained operator $\cO'$. We are left with a temporally truncated layer of exRecs (which are now all good) followed by the $\ast$-decoder, and we can pull the $\ast$-decoder past the layer of gadgets inside these exRecs by Prop.~\ref{prop:goodcorrect_TC}.  In the remaining circuit, the last layer of $\EC$ gadgets (i.e. immediately followed by the $\ast$-decoder) is now temporally truncated. However, the $\EC$ gadgets that experience the effect of truncation have the same effect as the subsequent noiseless $\EC$ within the $\ast$-decoder. Thus, the outcome is analogous to what one would obtain using conventional truncation for concatenated codes. A more precise understanding of why this method works is provided in Theorem~\ref{thm:TCerrsupp} and its proof.

Fig.~\ref{fig:example-temporal} provides an illustration of temporal truncation being implemented twice, once at timestep $t_2$ followed by another time at timestep $t_0$.

\begin{figure}[!htbp] 
\begin{center}
     \includegraphics[width=0.65\textwidth]{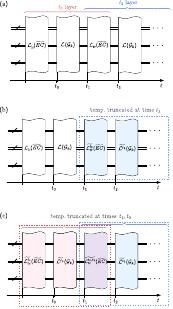}
\end{center}
\caption{ Illustration to Def.~\ref{def:temporal-truncation} of temporal truncation. We start with (a) a simulation where timestamps $t_0, t_1$ and $t_2$ are indicated and the leading and trailing $\widetilde{EC}$ layers are shown for one of the gate layers. Starting temporal truncation from the latest time in the circuit as described in the definition, we eventually arrive at (b), showing the layer of gates at time $t_2$ and the preceding layer of error correction truncated according to exRecs at the associated time. (c) Next, the truncation is performed for  time $t_0$, which involves leading and trailing error correction layers. Since $\cL(\widetilde{EC})$ participates in exRecs centered at both times $t_0$ and $t_2$, it is truncated at each of these times accordingly.   }
    \label{fig:example-temporal}
\end{figure}

Now, we are finally prepared to prove the main result of this subsection.  Recall that we are restricting our attention to either an $X$-type or $Z$-type automaton.  Consider a circuit $\cC$ for one of these automaton and its fault-tolerant simulation $\widetilde{FT}(\cC)$, the latter of which is subjected to a $p$-bounded gadget error model.  We will show that after appending a $*$-decoder to the end of $\widetilde{FT}(\cC)$, and pulling the $\ast$-decoder through exRecs layer by layer, the resulting circuit simulates $\cC$ with an  $Ap^2$-bounded gadget error model, where $A$ is a constant. For this, let us first show the result for a single iteration of pulling the decoder.

\begin{theorem}[Pulling $\widetilde{D^*}$ through one layer of outer simulation]\label{thm:TCerrsupp}
Assume a fault-tolerant simulation of circuit $\mathcal{C}$ of elementary gates using the inner and outer simulation of Secs.~\ref{subsec:inner-simulation} and \ref{subsec:outer-simulation}. Assume $\widetilde{FT}(\cC)$ operates with noise realization $\bm H$ (and associated noise operator history $\vec P$) sampled from a a $p$-bounded gadget error model. Assume also that the assumption in Remark~\ref{remark:overlaps} holds and $(\mathrm{Gate \  A, B} )_{k,m}$  and $(\EC \ \mathrm{ A, B})_{k,m}$  conditions from Def.~\ref{def:TC_gateECprime} are satisfied for some $k$ corresponding to the inner simulation level and some $m$. 

Consider the last layer of simulation centered on gates at time slice $t_0$ (where the succeeding layer of gates occurs at time $t_2$) consisting of $\widetilde{\cL^{t_2}_{tr}}(\widetilde{\EC}) \circ \cL(\cG_k) \circ  \cL_{le}(\widetilde{\EC})$ with the last layer of $\EC$ gadgets temporally truncated. If $t_0$ is the last layer of the entire simulation, no $t_2$ truncation is assumed. 

Run the Algorithm~\ref{alg:goodbad_assignment} to label exRecs centered at gates at time $t_0$ as good or bad, and perform spatial truncation.
Add a $\ast$-decoder after this layer of exRecs, resulting in $\widetilde{D^*} \circ \widetilde{\cL^{t_{2}}_{tr}}(\widetilde{\EC})\circ \cL ( \cG_k) \circ \cL_{le} (\widetilde{\EC})$.
 Pulling the layer of $\ast$-decoder past $\widetilde{\cL^{t_{2}}_{tr}}(\widetilde{\EC})\circ \cL ( \cG_k)$ results in:
\begin{itemize}
    \item[(1)] The layer of gates $\cL ( \cG_0)$ following the $\ast$-decoder, followed by application of some noise operator supported on each of the spatial regions corresponding to bad exRecs.
    \item [(2)] The layer of leading gadgets truncated due to bad exRecs centered at time slice $t_0$; i.e., this layer of gadgets is replaced by $\widetilde{\cL^{t_0}_{le}}(\widetilde{\EC})$; 
    \item[(3)] If the original circuit experiences a $p$-bounded gadget error model, the resulting gadget error model acting after $\cL ( \cG_0)$ is an $Ap^2$-bounded gadget error model for some constant $A$. 
\end{itemize} 

Therefore:
\begin{equation}
    \widetilde{D^*} \circ \widetilde{\cL^{t_{2}}_{tr}}(\widetilde{\EC})\circ \cL ( \cG_k) \circ \cL_{le} (\widetilde{\EC}) =  \cL(\cE(\cG_0)) \circ \cL(\cG_0)  \circ    \widetilde{D^*}  \circ  \widetilde{\cL^{t_0}_{le}}(\widetilde{\EC})
\end{equation}
where $\cE$ is an $Ap^2$-bounded gadget error model (discussed in point (3) above).

\end{theorem}

 Before proceeding with the proof, let us make a remark regarding why we replace the layer of EC gadgets before the ideal decoder with a temporally truncated one, i.e. $\widetilde{\cL^{t_2}_{tr}}(\widetilde{\EC}) \circ \cL(\cG_k) \circ  \cL_{le}(\widetilde{\EC})$. Consider the full circuit $\cC$ and add a $\ast$-decoder to this circuit at the end, in which case the trailing layer of $\widetilde{EC}$ gadgets will not be truncated. Pulling the $\ast$-decoder according to Thm.~\ref{thm:TCerrsupp} once will result in a new and shortened circuit where the EC layer before the $\ast$-decoder is temporally truncated.  In addition, the $\ast$-decoder is now followed by a layer of reduced-level gates experiencing a noise model with suppressed error rate. We then repeat this process and pull the $\ast$-decoder layer by layer through the entire circuit. From Thm.~\ref{thm:TCerrsupp}, we thus obtain the simulated circuit with an $A p^2$-bounded error model. This last conclusion will be formalized in a Corollary after the proof of the Theorem.
 
\begin{proof}
\begin{figure}[!htbp]
    \centering
    \includegraphics[width=0.95\textwidth]{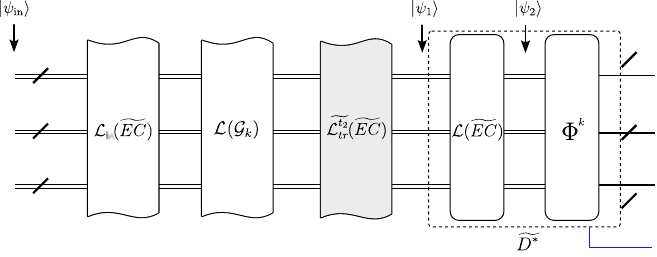}
\caption{ Layers of gadgets $\widetilde{\cL^{t_{2}}_{tr}}(\widetilde{\EC})\circ \cL ( \cG_k) \circ \cL_{le} (\widetilde{\EC})$ followed by a $\ast$-decoder $\cL(D_k^*)$. Also shown are the state of the toric code at the various timesteps that are utilized in the proof of the Theorem. The temporally truncated layer of error correction gadgets is shaded in gray.  }
    \label{fig:wavefunctionillustrations}
\end{figure}
We will show that we can pull all bad exRecs in a given layer past the $\ast$-decoder simultaneously, inducing an error model for the simulated level-0 gadgets in these regions.  Pulling the remaining (good) exRecs past the $\ast$-decoder produces noiseless level-0 gadgets. In this proof, where it does not cause ambiguity, we will use notation for regions to denote their qubit support.

Denoting an input stabilizer state to the entire layer $\widetilde{\cL^{t_{2}}_{tr}}(\widetilde{\EC})\circ \cL ( \cG_k) \circ \cL_{le} (\widetilde{\EC})$ by $\ket{\psi_{\text{in}}}$ (see Fig. \ref{fig:wavefunctionillustrations} for reference),  we write:
\begin{equation} \label{eq:psi1}
    \ket{{\psi}_1} \equiv \widetilde{\cL^{t_{2}}_{tr}}(\widetilde{\EC})\circ \cL ( \cG_k) \circ \cL_{le} (\widetilde{\EC}) \ket{\psi_{\text{in}}} =
     \cO_{\rm trun}  \widetilde{\cL_{tr}^{t_{0},t_2}} (\widetilde{\EC})\circ \widetilde{\cL^{t_0}} ( \cG_k) \circ \widetilde{\cL_{le}^{t_0}} (\widetilde{\EC})  \ket{\psi_{\text{in}}},
\end{equation}
where $\cO_{\rm trun} $ is some operator whose properties we will determine momentarily \footnote{As a notational aside, we do not use $\circ$ for multiplication of operators, as $\circ$ is specifically used for composition of \emph{gadgets}.}. 

Comparing the exRecs in Eq.~\ref{eq:psi1}, we see that they differ only by removing faults in the spacetime supports of bad exRecs centered at time $t_0$.
Consider connected regions of good exRecs and note that each connected region must have its boundary truncated. Using Lemma~\ref{lemma:op-conf-2}, we then conclude that the restriction of the total operator applied by the circuit $\widetilde{\cL^{t_{2}}_{tr}}(\widetilde{\EC})\circ \cL ( \cG_k) \circ \cL_{le} (\widetilde{\EC})$ to any such connected region of good exRecs can only depend on the syndrome and noise in its interior. However, the syndrome and noise in its interior is the same in the circuit $\widetilde{\cL_{tr}^{t_{0},t_2}} (\widetilde{\EC})\circ \widetilde{\cL^{t_0}} ( \cG_k) \circ \widetilde{\cL_{le}^{t_0}} (\widetilde{\EC})$, and thus, we conclude that $\cO_{\rm trun}$ is supported in the union of supports of all bad exRecs. We can thus decompose this operator as:
\begin{equation}
    \cO_{\rm trun}  = \bigotimes_i \cO_{{\mathcal B}^{(i)}_{\mathrm{bad}}} ,
\end{equation}
where ${\mathcal B}^{(i)}_{\mathrm{bad}}$ is the spatial support of $i$-th bad exRec and we assume that some arbitrary resolution has been picked for every pair of operators that share support, and each operator $\cO_{{\mathcal B}^{(i)}_{\mathrm{bad}}}$ is supported in ${\mathcal B}^{(i)}_{\mathrm{bad}}$. Operators  $\cO_{{\mathcal B}^{(i)}_{\mathrm{bad}}}$ are, by definition, any operators that produce the decomposition in equation above. According to Lemma~\ref{lemma:op-conf-1}, each of these operators must exhibit the following dependence:
\begin{equation}
    \cO_{{\mathcal B}^{(i)}_{\mathrm{bad}}} =  \cO_{\cB_{\mathrm{bad}}^{(i)}}[\sigma \big|_{\cQ_{\mathrm{bad}}^{(i)}}, \vec P^{t_{2}}\big|_{\cQ_{\mathrm{bad}}^{(i)}}]
\end{equation}
where, assuming that the boundary of ${\mathcal B}^{(i)}_{\mathrm{bad}}$ is $\Lambda_i$, $\cQ_{\mathrm{bad}}^{(i)}$ is the extended spacetime support associated with the $i$-th bad exRec, defined in the equation below:
\begin{equation} 
\cQ_{\mathrm{bad}}^{(i)} ={\mathcal B}^{(i)}_{\mathrm{bad}} \cup \left(\bigcup_{j: \, \cR_j \cap {\mathcal B}^{(i)}_{\mathrm{bad}} \subseteq \Lambda_i} \overline{\cR}_j\right),
\end{equation} 
where as before $\overline{\cR}_j  = \cR_j \setminus \partial \cR_j $, and $\cR_j$ is the spatial support of the $j$-th exRec (good or bad).  The notation $\vec P^{t_{2}}$ denotes an operator history where the level-0 gadget errors associated with temporal layer truncation at $t_{2}$ have been removed.

We will now consider the action of the $\ast$-decoder on the state $\ket{\psi_1}$. We start by decomposing the $\ast$-decoder into an action of a layer of $\cL(\widetilde{\EC})$, followed by a $k$-fold pushforward map $\Phi^k$ (Def.~\ref{def:pushforward}), i.e. $\widetilde{D^*}  = \Phi^k \cL(\widetilde{\EC}) $ (which is schematic notation since we have not included components like the classical register). 
The action of the noiseless $\cL(\widetilde{\EC})$ layer which is the first part of $\ast$-decoder yields (see Fig.~\ref{fig:wavefunctionillustrations} for reference):

\begin{equation}
    \ket{\psi_2} \equiv \cL(\widetilde{\EC}) \ket{\psi_1} = \cO_{\rm trun}' \cL(\widetilde{\EC}) \circ \widetilde{\cL_{tr}^{t_{0},t_2}} (\widetilde{\EC})\circ \widetilde{\cL^{t_0}} ( \cG_k) \circ \widetilde{\cL_{le}^{t_0}} (\widetilde{\EC})  \ket{\psi_{\text{in}}} ,
\end{equation}
where $\cO'_{\rm trun}$ is an operator whose syndromes are coarse-grained at level $k$ and which has the same support as $\cO_{\rm trun}$.  We now will argue that $\cO'_{\rm trun}$ in fact has support in bad exRecs only.  To see this, note that $\ket{\psi_1}$ and the reference state $\cO_{\rm trun} \ket{\psi_1}$ (which cancels the $\cO_{\rm trun}$ appearing in Eq.~\ref{eq:psi1}) have syndromes differing only in the support of bad exRecs.  Upon applying $\cL(\widetilde{\EC})$, we know from Lemmas~\ref{lemma:op-conf-1} and ~\ref{lemma:op-conf-2} that the operator applied in the support of good exRecs (whose boundaries shared with bad exRecs have been spatially truncated) only depends on the input syndromes restricted to the (truncated) good exRec and bulks of neighboring good exRecs. Therefore, comparing the operator applied by $\cL(\widetilde{\EC})$ to $\ket{\psi_1}$ and the one applied by $\cL(\widetilde{\EC})$ to $\cO_{\rm trun} \ket{\psi_1}$, we see that they  must differ in the support of bad exRecs, implying the same for $\cO'_{\rm trun}$.

We consider the action of $\Phi^k$ on  $\ket{\psi_2}$ and pull the operator $\cO'_{\rm trun}$ through the pushforward map.  The effect of this is 
\begin{equation} \label{eq:oprime}
    \Phi^k \cO'_{\rm trun} = \mathcal{O}'_{\rm red} \Phi^k
\end{equation}
where, using the definition of the pushforward map, we can write
$
\varphi^k(\mathcal{O}'_{\rm red}) \cO'_{\rm trun} \in \cS
$
with $\cS$ the toric code stabilizer group and $\varphi^k$ a $k$-fold pullback map.  From the arguments above, $\mathcal{O}'_{\rm red}$ can be chosen so that
\begin{equation}
\mathrm{supp}\left(\mathcal{O}'_{\rm red}\right) \subseteq \bigcup_i {\cB}_{0,\text{bad}}^{(i)}
\end{equation}
where ${\cB}_{0,\text{bad}}^{(i)}$ is a region on the reduced lattice associated with ${\cB}_{\text{bad}}^{(i)}$ on the original lattice.  

We can now write \begin{equation}
    \Phi^k \ket{\psi_2} \equiv  \Phi^k  \cL(\widetilde{\EC}) \ket{\psi_1} = \cO_{\rm red}'  \widetilde{D^*} \circ \widetilde{\cL_{tr}^{t_{0},t_2}} (\widetilde{\EC})\circ \widetilde{\cL^{t_0}} ( \cG_k) \circ \widetilde{\cL_{le}^{t_0}} (\widetilde{\EC})  \ket{\psi_{\text{in}}} ,
\end{equation}
Using the fact that since all exRecs on the right hand side are now good, Prop.~\ref{prop:goodcorrect_TC} can be used to pull the layer of level-$k$ gates and $\EC$ gadgets through the ideal decoder, to obtain
\begin{equation}
    \widetilde{D^*} \circ \widetilde{\cL_{tr}^{t_{0},t_2}} (\widetilde{\EC})\circ \widetilde{\cL^{t_0}} ( \cG_k) \circ \widetilde{\cL_{le}^{t_0}} (\widetilde{\EC})  \ket{\psi_{\text{in}}}  = \cL(\cG_0)\circ   \widetilde{D^*} \circ \widetilde{\cL_{le}^{t_0}}(\widetilde{\EC}) \ket{\psi_{\text{in}}}.
\end{equation} 
Combining everything together, we obtain:
\begin{equation}
\begin{split}
     \Phi^k \ket{\psi_2} &= \mathcal{O}'_{\rm red} \cL(\cG_0) \circ\widetilde{D^*} \circ \widetilde{\cL_{le}^{t_0}}(\widetilde{\EC}) \ket{\psi_{\text{in}}}
     \\
     &= \cL \left (\cE(\cG_0);\bigcup_i \cB_{0,\text{bad}}^{(i)} \right ) \circ \cL \left (\cG_0;\cL_0^{\rm red} \setminus \bigcup_i \cB_{0,\text{bad}}^{(i)} \right )   \circ \widetilde{D^*}  \circ  \widetilde{\cL_{le}^{t_0}}(\widetilde{\EC}) \ket{\psi_{\text{in}}}
\end{split}
\end{equation}
where we absorbed the effect of $\cO'_{\rm red}$ into the error channel $\cE(\cG_0)$, and the reduced lattice is denoted by $\cL_0^{\rm red}$.  The error $\mathcal{E}$ denotes the effect of the new simulated gadget error model (see Def.~\ref{def:approximate-decoupled-gadget-error-model2}). In the simulated error model, the failed simulated gadgets $F_{E,{t_0}}$ are in one-to-one correspondence with bad exRecs. The operator $\mathcal{O}'_{\rm red}$ is the total operator applied due to the operator history as well as the action of the gadget at time $t_0$ and we can assign a set of failed gadgets corresponding to this operator by partitioning its support over those of the failed gadgets. 

What remains now is to determine whether the simulated noise model is $Ap^2$-bounded. The probability of any extended rectangle being bad is $\leq A p^2$ where $A$ is a constant determined by the maximum number of sites within any type of exRec and the number of possible truncations of an exRec. The probability that all gadgets in a given set $B_0 = \{ \cG_0^{(i)}\}$ are faulty is the probability that all exRecs corresponding to these gadgets are bad, which is $\leq (A p^2)^{|B_0|}$. This proves the theorem. 

\end{proof}

\begin{corollary}\label{cor:toriccodesuppre}
    Consider a circuit $\mathcal{C}[\cI_0, \cT_0, \cM_0]$, and construct the fault tolerant simulation $\widetilde{FT}(\mathcal{C})$ following the prescription outlined in Secs.~\ref{subsec:inner-simulation} and \ref{subsec:outer-simulation}. Assume that the assumptions of Theorem~\ref{thm:TCerrsupp} are satisfied.  Insert a layer of $\ast$-decoder at the end of the circuit.  If the gadget error model for $\widetilde{FT}(\mathcal{C})$ is $p$-bounded, then pulling the circuit through the $\ast$-decoder produces the simulated circuit $\mathcal{C}$ experiencing a gadget error model that is $Ap^2$-bounded for some constant $A$. 
\end{corollary}
\begin{proof}
    We push the circuit through the decoder layer by layer.  Note that due to the truncation procedure, if an elementary gadget failure is assigned to a bad exRec, it cannot be assigned to any good exRecs.  Using the results from the last theorem, a noise model will be induced on gates in $\mathcal{C}$ corresponding to bad exRecs in $\widetilde{FT}(\mathcal{C})$.  The probability that a given exRec is bad is $\leq A p^2$ for some constant $A$, and utilizing a similar argument to that of the previous theorem as well as the fact that a failure cannot be assigned to multiple bad exRecs, the gadget error model in $\mathcal{C}$ is $Ap^2$-bounded.
\end{proof}

Thus, an automaton made out of gadgets satisfying the Gate and $\EC$ conditions simulates a circuit whose noise rate is reduced doubly-exponentially quickly as one goes to higher and higher levels of simulation. To establish fault tolerance of our toric code automaton, it remains to show that the ${\EC}$ and Gate conditions are satisfied for the toric code gadgets constructed in Sec.~\ref{ss:construction_of_tc_gadgets}.

\subsection{Proof of EC and Gate conditions for toric code automaton}\label{sec:pfnilpTC}

Similarly to the analysis of Tsirelson's automaton in Section~\ref{sec:tsirelson-proof-2}, in order to prove the ${\EC}$ and Gate conditions, we will need a suitable generalization of nilpotence for the gadgets in the toric code automaton.  First, we will define the notion of an isolated fault:
\begin{definition}[Isolated faults]\label{def:isolated}
    An elementary gadget fault of gadget $\cG_0$ at spacetime location $\bm x$ is isolated at level $\ell$ if one of the following is true:
    
    \begin{itemize}
        \item There is a composite gadget $ \EC_\ell^{\otimes K} \circ \cG_\ell \circ  \EC_\ell^{\otimes K}$ such that the fault is located either within the leading $\EC_\ell$ gadgets or within $\cG_\ell$, and there are no other faults in the composite gadget as well as any of the neighboring composite gadgets. In addition, the input state must pass through a $\varnothing$-cell filter acting on $\ell$-cells of the composite gadget and its neighboring composite gadgets. We call the isolation region of this fault the union of spacetime supports of the composite gadget $\EC_\ell^{\otimes K} \circ \cG_\ell \circ  \EC_\ell^{\otimes K}$ containing the fault and the supports of the neighboring composite gadgets.
        \item The fault is located within an $\EC_\ell$ gadget that belongs to the very last layer of level-$\ell$ gadgets of the toric code automaton, and there are no other faults in this $\EC_\ell$ gadget and in neighboring $\EC_\ell$ gadgets. The input state to this layer of gadgets is assumed to pass through a $\varnothing$-cell filter in the support of the $\EC_{\ell}$ gadget with the fault and its neighboring $\EC_\ell$ gadgets. We call the isolation region of this fault the union of spacetime supports of the composite gadget $\EC_\ell$ containing the fault and the supports of the neighboring composite $\EC_\ell$ gadgets.
    \end{itemize}
    
    An $m$-cell damage input to a gate or EC gadget is said to be isolated at level $\ell>m$ if there are $\varnothing$-cell filters in the supports of its neighbors, and there are no faults in this gadget and its neighboring gadgets. 
\end{definition}

Next, we define the nilpotence conditions for the toric code automaton.

\begin{definition}[Gate and EC nilpotence for the toric code] \label{def:nilpotenceTC}
    A gate gadget $\mathcal G_{\ell}$ is said to be nilpotent if for coarse-grained input syndrome $\sigma = \sigma|_{\Sigma_\ell}$,  we have 
    \begin{equation} \label{eq:nilp-cond-1}
    \mathcal{L}'(\EC_{\ell}) \circ  \left (   \mathcal{L}(\mathcal G_{\ell}) \circ \mathcal{L}(\EC_{\ell})  \right ) [\sigma, \varepsilon] = \mathcal{L}'(\EC_{\ell}) \circ \mathcal{L}(\mathcal  G_{\ell})[\sigma, \varepsilon \oplus \varepsilon_0],
    \end{equation}
    where we assume the noise restricted to the spacetime support of the above circuit to be an isolated elementary gadget failure in $\mathcal G_{\ell}$ that is associated with syndrome $\varepsilon_0$, and the prime on the last layer of $\EC_\ell$ gadgets denotes that the layer is noise-free in the extended support of $\mathcal G_{\ell}$. The argument $\varepsilon \oplus \varepsilon_0$ denotes the noise syndromes where the isolated gadget failure has been removed.

    An $\EC$ gadget $\EC_{\ell}$ is said to be nilpotent if for coarse-grained input syndrome $\sigma = \sigma|_{\Sigma_\ell}$,  we have 
    \begin{equation} \label{eq:nilp-cond-2}
    \mathcal{L}'(\EC_{\ell}) \circ  \mathcal{L}(\EC_{\ell})[\sigma, \varepsilon] =\mathcal{L}'(\EC_{\ell})[\sigma, \varepsilon \oplus \varepsilon_0],
    \end{equation}
    under the same assumptions made above.
\end{definition}

Note that there is a difference between this definition and the one used for Tsirelson's automaton in Def.~\ref{def:tsirelson_nilpotence} (which did not have the leading layer of $\EC_\ell$s on the left hand side of \eqref{eq:nilp-cond-1}), and that we have additionally separately defined the notion of nilpotence for the $\EC$ gadget.  Below, we will provide a much simpler and equivalent definition that we will be utilizing in the proofs of the main statements.

\begin{lemma}[Simplified nilpotence conditions]
    Define a simplified gate nilpotence condition for $\cG_{\ell}$ to be identical to Def.~\ref{def:nilpotenceTC} except the first two layers $\cL(\cG_{\ell}) \circ \cL(\EC_{\ell})$ are assumed to have precisely one fault which is located in $\cG_{\ell}$ and has syndrome $\varepsilon_0$, with the final layer of $\mathcal{L}'(\EC_{\ell})$ is noiseless. In particular, the simplified gate nilpotence condition is
    \begin{equation} \label{eq:nilp-cond-3}
    \mathcal{L}'(\EC_{\ell}) \circ  \left (   \mathcal{L}(\mathcal G_{\ell}) \circ \mathcal{L}(\EC_{\ell})  \right ) [\sigma, \varepsilon_0] =  \mathcal{L}'(\EC_{\ell}) \circ \mathcal{L}(\mathcal  G_{\ell})[\sigma, \varnothing]
    \end{equation}
    An analogous simplified $\EC$ nilpotence condition is 
    \begin{equation}\label{eq:nilp-cond-4}
    \mathcal{L}(\EC_{\ell}) \circ  \mathcal{L}(\EC_{\ell})[\sigma, \varepsilon_0] =\mathcal{L}'(\EC_{\ell}) [\sigma, \varnothing]
    \end{equation}
    For the toric code automaton, where the assumption in Remark~\ref{remark:overlaps} holds, the simplified nilpotent conditions are the same as the original nilpotence conditions (Def.~\ref{def:nilpotenceTC}) for both gate and EC gadgets.
\end{lemma}
\begin{proof}
    From Fact \ref{fact:1}, the output of the gates in the spatial support of the isolation region is independent of the damage and noise syndromes outside of this region and vice versa.  Thus, the original definition of nilpotence splits into separate statements of Eqs.~\eqref{eq:nilp-cond-1} and \eqref{eq:nilp-cond-2} restricted to be inside and outside the isolation region.
    Outside the isolation region, the RHS and LHS of each of these equations are identical. Restricted to the isolation region, the statement is equivalent to that of simplified Eqs.~\eqref{eq:nilp-cond-3} and \eqref{eq:nilp-cond-4}. Thus, the original definition of nilpotence follows from the simplified one.
    \end{proof}

Moving forward, because the conditions in  Remark~\ref{remark:overlaps} are indeed satisfied for the toric code automaton, we will refer to the simplified nilpotence condition as simply the nilpotence condition moving forward. While we can verify nilpotence by explicitly constructing the circuit of level-0 gadgets for $\cG_{\ell}$, brute-force verification is computationally expensive even when going a modest number (say, three) levels up. Instead, we will use analytical insight to reduce the search space while still being able to provide a modest upper bound on the nilpotence level.  This will be discussed in the proof of the gadget nilpotence levels.  In addition, some of these ideas are rather general and can be used in future constructions of fault tolerant automata.  Before this, we will need a useful fact:

\begin{lemma} \label{lemma:higher-nilpotence-TC}
    If there exists $\ell$ such that all gate and $\EC$ gadgets of the toric code automaton are nilpotent at level $\ell$, then all gate and $\EC$ gadgets are nilpotent at all levels $\ell' > \ell$.
\end{lemma}
\begin{proof}
The proof follows from writing the gadgets $G_{\ell'}$ in terms of gadgets at level $\ell$.  Nilpotence at level $\ell$ shows that any single level-0 gadget failure in some $\cG^{(i)}_{\ell} \circ \EC^{(i)}_{\ell}$ will be cleaned up by applying a noiseless layer of $\EC_{\ell}$, which is the subsequent layer of gadgets.  In the case where the fault resides in the last layer of level-$\ell$ gadgets in $G_{\ell'}$, the subsequent clean layer of $\EC_{\ell'}$ (as required from the definition of nilpotence) will clean the output damage.
\end{proof}

In addition, we will also need a useful definition that keeps track of the location of a level-0 gadget within a higher level gadget:

\begin{definition}[Address]  \label{def:error_location}
    For a circuit $\widetilde{FT}(\mathcal C)$, define the address of an elementary gadget fault to be a sequence of gadgets (either EC or gate) of ascending level, written as
    \begin{equation}
        [\cG_k(\bm{x}_k), \cG_{k-1}(\bm{x}_{k-1}),..., \cG_{0}(\bm{x}_0)],
    \end{equation}
    where $\mathrm{supp}(\cG_{i-1}) \subset \mathrm{supp}(\cG_{i})$ and the argument of $\cG_k(\bm{x}_k)$ denotes a canonical spacetime coordinate for the location of the level-$k$ gadget. The last element $\cG_0$ specifies the faulty elementary gadget and thus, can only be $\cI_0, \cT^{h/v}_0$ or $\cM^{h/v}_0$.  We will usually suppress the argument providing the spacetime coordinate for simplicity. 
\end{definition}

Before proceeding, we will note a few properties of our construction that will simplify the analysis.  First, we will be proving nilpotence for an isolated fault within a level-$\ell$ gadget.  In particular, this means that we tile the lattice with level-$\ell$ gadgets and restrict our consideration to the support of a single level-$\ell$ gadget, \emph{including} level-0 gadgets from the neighboring gadgets applied at its boundary. 
Since the surrounding gadgets are noise-free because of fault isolation and the input is assumed to be level-$\ell$ coarse-grained, then the neighboring level-$\ell$ gadgets can at most add syndromes at the corners of the level-$\ell$ gadget under consideration.  However, we can neglect syndromes on the linear vertices (because of linearity property).  Finally, it suffices to only consider failures of level-0 gadgets whose support is fully contained within a given level-$\ell$ gadget.

\begin{lemma} \label{lemma:nilpotentnumerics-TC} 
    Assuming a gadget error model, the gate gadgets $\cI$,  $\cT^{h/v}$, and the $\EC$ gadget are nilpotent at level 3 and above.  The gate gadget $\cM^{h/v}$ is nilpotent at level 6 and above.
\end{lemma}

\begin{proof}
The proof is partially due to a simple computer-aided search that we designed, which we expect to be relatively general purpose.  We present the details of algorithms used for this search in App.~\ref{app:TC-nilpotence}. In general, we refer the reader to the App.~\ref{app:TC-nilpotence} for a more technical description of the approach; we will present a `big picture' discussion here.

First, we show the nilpotence of the $\EC$ gadget.  For this, we consider an $\EC_\ell$ gadget (with $\ell \geq 1$) containing a level-0 gadget failure isolated at level $\ell$, and assume that the input is coarse-grained at level $\ell$, i.e. non-trivial input syndromes lie in $\Sigma_\ell$.
The address of the fault can be written as:
\begin{equation} \label{eq:coord-1}
    [\EC_{\ell}, \cG_{\ell-1}^{(a_{\ell -1})},..., \cG_1^{(a_1)}, \cG_0^{(a_0)}]
\end{equation} 
where $\cG_i^{(a_i)}$ with $\ell \geq i \geq 1$ stands for either gate or EC gadgets, and we omit the spacetime coordinates of each gadget. Additionally, the superscript $(a_i)$ keeps track of the associated gadget type (namely, whether it is $\cI$, $\cT^{h/v}$, $\cM^{h/v}$ or $\EC$-type gadget). Due to linearity of the $\EC$ gadget on the corners of its support, we can assume that the input to $\EC_{\ell}$ is syndrome-free and $\mathbb{F}_2$ add the input syndromes to the corners of the output of the gadget at the end.

Assuming $\ell > 1$, we can express $\EC_\ell$ by a sequence of gadgets of the form  $\mathcal{G}_1^{(a_1)} \circ \EC_{1}^{\otimes K_{a_1}}$ (where $\mathcal{G}_1^{(a_1)}$ is a gate supported on $K_{a_1}$ 1-cells) with a final application of a layer of $\EC_1$ gadgets\footnote{A similar decomposition exists of $\EC_\ell$ into $\mathcal{G}_m^{(a_m)} \circ \EC_{m}^{\otimes K_{a_m}}$ and a final layer $\EC_m$, with $m < \ell$.}.
Suppose that the level-0 failure occurs during the operation of a given $\mathcal{G}^{({a_1})}_1\circ   \EC_{1}^{\otimes K_{a_1}}$, which we will denote as $\left ( \mathcal{G}_1^{(a_1)}  \circ \EC_{1}^{\otimes K_{a_1}} \right) [\emptyset,\varepsilon_0]$, where $\emptyset$ denotes a clean input and $\varepsilon_0$ corresponds to the level-0 gadget failure. This combination of gadgets will always be immediately followed by a noiseless layer of $\EC_1$ gadgets. 
Using a computer program, we then explicitly compute the action of the circuit $\EC_{1}^{\otimes K_{a_1}} \circ \left ( \mathcal{G}^{(a_1)}_1 \circ  \EC_{1}^{\otimes K_{a_1}} \right) [\emptyset,\varepsilon_0]$, with the last $\EC_{1}^{\otimes K_{a_1}}$ being noiseless, to determine the level-1 coarse-grained output in the spatial support of the $\mathcal{G}^{(a_1)}_1$, which we define as $\varepsilon_1$.

Next, we treat this coarse-grained damage with syndrome $\varepsilon_1$ as an effective transient error inside a level-2 gadget $\mathcal{G}^{(a_2)}_2 \circ  \EC_{2}^{\otimes K_{a_2}}$ at a spacetime location indicated in the address.  Because it is coarse-grained at level 1, we can perform level reduction by 1 level, which maps the problem onto studying a level-0 gadget failure $\varepsilon_0^{(1)} = \Phi (\varepsilon_1)$\footnote{The pushforward map was originally defined to act on states, but this action can also be thought of as inducing a map on syndromes of a state. Thus when we write $\Phi$ as acting on a pattern of syndromes $\sigma = \sigma|_{\Sigma_1}$, the result of this map is the syndromes of the coarse-grained output state $\Phi\ket{\psi}$, where the input state $\ket{\psi}$ has syndrome $\sigma$. } within an effective level-1 circuit $\mathcal{G}^{(a_2)}_1 \circ 
 \EC_{1}^{\otimes K_{a_2}} = \Phi(\mathcal{G}^{(a_2)}_2 \circ  \EC_{2}^{\otimes K_{a_2}})$.  We subsequently iterate this process for each $k$ and corresponding $\cG^{(a_k)}$ (which due to the level reduction is always analyzed as a level-0 gadget failure in a level-1 gadget) to understand how the possible outcomes of gadget failures change as one goes up a level at a time. Nilpotence is achieved if any input error $\varepsilon_0$ 
 % being followed level by level, 
 eventually shrinks to nothing (i.e. $\varepsilon_0^{(\ell)} = \varnothing$) after one passes through a finite number of levels $\ell$. In our computer-aided search, we loop through all possible level-0 gadget faults leading to $\varepsilon_0$ as well as all possible addresses of the faults; let us now present a formalism for handling all of these possibilities.

We formalize this procedure by defining a map $\Gamma[\mathcal{G}_1, \varepsilon_0]$ which is a function of the gadget $\mathcal{G}_1$ and single failure $\varepsilon_0$ of one of its constituent level-0 gadgets. The map outputs the gadget failure $\varepsilon_0^{(1)} = \Phi(\varepsilon_1)$ whose action is determined by evolving $ \EC_{1}^{\otimes K} \circ \left (\mathcal{G}_1 \circ \EC_{1}^{\otimes K}\right ) [\varnothing, \varepsilon_0 ]$ and subsequently applying the pushforward map $\Phi$.  
Then, for all types of gadget $\mathcal{G}_0$, we define the sets $D_{\mathcal{G}_0}$ to be syndromes of all all possible damage outputs caused by the failing the gadget $\cG_0$ (which can be either $\cI_0$, $\cT_0^{h/v}$ or $\cM_0^{h/v}$). 

We then consider the address of a given failure (Eq.~\eqref{eq:coord-1}) and start with the gadget failure $\cG_0^{(0)}$ located within level-1 gadget  $\cG_1^{(a_1)} \circ {\EC}^{\otimes K_{a_1}}$. All possible ways to fail the given gadget $\cG_0^{(a_0)}$ lead to the following set of effective failures of the gadget containing it:
\begin{equation}
D_{\mathcal{G}_1^{(a_1)}} = \{\varepsilon: \exists \varepsilon_0 \in D_{\mathcal{G}_0^{(a_0)}}, \, \Gamma[\mathcal{G}_1, \varepsilon_0] = \varepsilon\}.
\end{equation}
This set collects all possibilities for the output damage syndromes due to any possible failure occurring in a specific gadget $\cG^{(a_0)}_0$ located within $\mathcal{G}_1^{(a_1)} \circ \EC_{1}^{\otimes K_{a_1}}$, where the gate types and their locations are specified by the fault address. 

We then iterate this procedure by replacing $D_{\mathcal{G}_1^{(a_s)}}$ with $ D_{\mathcal{G}_0^{(a_{s-1})}}$ (which we will henceforth denote as $D_{\mathcal{G}_1^{(a_s)}} \to D_{\mathcal{G}_0^{(a_{s-1})}}$) and performing level reduction by one level each time. Upon iterating this procedure, we find the smallest integer $\ell$ such that $D_{\cG_1^{(a_\ell)}} = \varnothing$ for all possible addresses of the fault.   The result of the computer-aided search is detailed in App.~\ref{app:TC-nilpotence}, which indicates that for $\ell = 3$, any level-0 gadget failure is always cleaned up. We thus can conclude that $\EC$ is nilpotent at level $\ell = 3$.

Next, let us determine the level of nilpotence for the gadgets $\mathcal{I}$, $\mathcal{T}^{h/v}$, and $\mathcal{M}^{h/v}$.  Note that if a gadget failure occurs inside $\mathcal{I}_s$ and $\mathcal{T}_s$, then because all of the points in $\Sigma_s$ are linear for both gadgets, we can assume a clean input and afterwards $\mathbb{F}_2$-add the syndromes obtained from the clean versions of the same gadgets acting on the true input.  Therefore, a completely analogous approach to the one used for $\EC_{\ell}$  can also be used to show nilpotence for both $\mathcal{T}_{\ell}^{h/v}$ and $\mathcal{I}_{\ell}$. We obtain that both these gate gadgets are nilpotent at $\ell = 3$.

Finally, we need to check the nilpotence of $\mathcal{M}^{h/v}$ for some level $r \geq \ell$.   
If the input syndromes are restricted to lie on the linear points of $\cM_{\ell}^{h/v}$ (which are the corners of the gadget), a computer-aided search reveals that this gadget is also nilpotent at level $\ell = 3$. 
However, there are 4 nonlinear syndrome locations on the coarse-grained lattice, for which we cannot use linearity properties to reduce the input to a syndrome-free one. Therefore, we cannot use the method described above for the purpose of checking the nilpotence of $\mathcal{M}^{h/v}$.

To establish the nilpotence of $\mathcal{M}^{h/v}$, we start with the combination $\mathcal{M}_{r}^{h/v} \circ \EC_r^{\otimes 3}$ with the assumption that $r > \ell=3$, which we can use to simplify the problem. This is justified as we will see shortly because the $\cM$ gate gadget is not nilpotent at level 3 when there is an arbitrary encoded input on the nonlinear points. 

Consider first the scenario when the level-0 fault occurred in the spacetime support of one of the preceding $ \EC_r$ gates. We note that this gadget is followed by a clean layer of $\EC_{r-1}$ gadgets that start the $\cM_r^{h/v}$ gate afterward. Because the input damage is coarse-grained at level $r$ by assumption, we can use linearity of the corner points of the $\EC_r$ gadget to move the action of the input damage to after the first $\EC_r$ layer has been applied. Then, the nilpotence of all gadgets inside $\EC_r$ at level $r-1 \geq \ell$ (including $\cM^{h/v}_{r-1}$ because the input has now been rendered syndrome-free) implies that the effect of the level-0 fault is cleaned up. 

What remains is to consider the case when the fault occurs during the operation of $\mathcal{M}_{r}^{h/v}$ gadget.  
 The address of such fault can be written as
\begin{equation}\label{eq:coord-3}
    [\mathcal{M}^{h/v}_{r}, \cG_{r-1}^{(a_{r -1})},...,\cG_{q+1}^{(a_{q+1})}, \cG_{q}^{(a_q)},\cG_{q-1}^{(a_{q -1})},...,  \cG_1^{(1)}, \cG_0^{(a_0)}]
\end{equation} 
If one of the gadgets $\cG_{q}^{(a_q)}$ in the address is $\EC_q$ and $q \geq 4$, then we can use nilpotence and linearity of $\EC_q$ on points in $\Sigma_q$, from which it immediately follows that the fault will be cleaned up\footnote{In fact, this reasoning also works for when $\cG_{q}^{(a_q)}$ is $\cI_q$ or $\cT_{q}^{h/v}$; however, our computer code treats all gate gadgets on the same footing and the EC gadget on a different footing (by expressing it through gate gadgets at lower levels), and thus, we keep our explanation consistent to the operation of the code.}. If there is no such $\EC$ gadget in the address, we use the following strategy. First, if there is an $\EC_q$ gadget present in the fault's address with $q\leq 3$, we consider the gadget $\cG_q$ that follows it in the circuit as well as the subsequent layer of $\EC_q$ gadgets that coarse-grain the output of $\cG_q$. Due to the coarse-graining, we can reduce the level by $q$, upon which the failure is that of a $\cG_0$ gadget located within level-reduced $\cM_{r - q}$. By assumption, all the gadgets in the level-reduced address must be gate gadgets only (and not EC). Thus, we reduced the problem to considering a level-0 failure of a gadget within $\cM_{r-q}$ that at each level is located within a gate (and not EC) gadget. We then determine the value of $r-q$ under which the nilpotence is achieved in this problem. Let us call this value $\alpha$. Since $q \leq 3$, we have that $r = 3+ \alpha$ lower bounds the level of nilpotence of the gadget $\cM$. 

To formalize this search, we define a new map $\Gamma[\mathcal{G}_1, \sigma, \varepsilon]$ which is a function of the gadget $\mathcal{G}_1$, the single level-0 gadget failure $\varepsilon_0$, and input damage with syndrome $\sigma \in \Sigma_1$. This map outputs the damage with syndrome $\varepsilon^{(1)}_0$ which is determined by performing a two-step procedure. First, we compute $\mathcal{G}_1[\sigma, \varepsilon]$ and apply $\Phi \circ \EC^K$ (where $K$ is the number of 1-cells in the support of $\cG_1$ and $\EC^K$ is noiseless) which yields $\varepsilon'$.  Then, we separately compute $\mathcal{G}_1[\sigma, \varnothing]$ and apply $\Phi \circ \EC^K$ which yields $\varepsilon''$. Finally, $\varepsilon^{(1)}_0$ is determined by $\mathbb{F}_2$-adding both outcomes: $\varepsilon^{(1)}_0 = \varepsilon' \oplus \varepsilon''$.  
We then consider the address of a given failure (Eq.~\eqref{eq:coord-3}) and analyze the failure of $\cG_0^{(a_0)}$ located within $\cG_1^{(a_1)}$. We define $D_{\mathcal{G}_0^{(a_0)}}$ to be all possible failures of the specific gadget $\cG_0^{(a_0)}$ (which can be either $\cI_0$, $\cT_0^{h/v}$ or $\cM_0^{h/v}$) and then proceed to define
\begin{equation}
D_{\mathcal{G}_1^{(a_1)}} = \{\varepsilon: \exists \varepsilon_0 \in D_{\mathcal{G}_0^{(a_0)}}, \, \exists \sigma \in \widetilde \Sigma_1[\mathcal{G}] :  \Gamma[\mathcal{G}_1, \sigma, \varepsilon_0] = \varepsilon\}.
\end{equation}
where $\widetilde \Sigma_1[\mathcal{G}]$ are the set of non-linear points of $\mathcal{G}$. We then repeatedly iterate this procedure by replacing $D_{\mathcal{G}_1^{(a_k)}} \rightarrow D_{\mathcal{G}_0^{(a_{k-1})}}$ and performing level reduction by one level each time. Upon iterating this procedure, we want to find the smallest integer $r$ such that for every possible address for the level-0 gadget failure in Eq.~\eqref{eq:coord-3},  and every possible coarse-grained input syndrome $\sigma^{(r-q)}$ of the overall gadget, we obtain $D_{\cG_1^{(a_{r-q})}} = \varnothing$.    Once again, this can be computed explicitly with a computer program; the result in App.~\ref{app:TC-nilpotence}, indicates that for $r-q \geq 3$, such a level-0 gadget failure is always cleaned up.  Thus, so long as $r \geq 3+3 = 6$, any level-0 gadget failure will be cleaned up, therefore proving nilpotence of $\mathcal{M}$ at level $6$.

\end{proof}

\begin{lemma} [Effective nilpotence level of Tsirelson's automaton]
    Under the modified definition of nilpotence in Def.~\ref{def:nilpotenceTC}, the gadgets of Tsirelson's automaton obtained using inner fault tolerant simulation prescription are nilpotent at level $\ell_T = 3$.
\end{lemma}
\begin{proof}
Follows directly from using $\EC$ and Gate conditions for Tsirelson's automaton shown in Prop.~\ref{lemma:gate-ec-prime-TS} at $k = \ell_T = 3$ and $m = 1$. 
\end{proof}

The following lemma describes how syndromes, if they are restricted to the boundaries of level-$k$ gadgets are treated by the automaton, and will be used in our proofs of the modified EC and gate conditions. 
\begin{lemma} \label{lemma:Tsirelson-at-boundaries}
    Define an additional modification to the Tsirelson's automaton (constructed using inner and outer fault tolerant simulation prescription), which we will refer to as the boundary Tsirelson automaton (which is the 1D automaton occurring along the boundaries of level-$k$ gadgets):
    \begin{itemize}
        \item $\EC'_{n; 1,2}$ to be the circuit $\cL(\EC_n;\lambda_n)$ restricted to any vertical (1) or horizontal (2) $n$-link boundary $\lambda_n$;
        \item $\cI'_{n; \cT}$ to be the circuit $\cL(\cT_n^h;\lambda_n^v)$ restricted to any $n$-link vertical boundary $\lambda_n^v$, and
        \item  $\cI'_{n; \cM}$ to be the circuit $\cL(\cM_n^h;\lambda_n^v)$ restricted to any $n$-link vertical boundary $\lambda_n^v$;
        \item $\cT'_n$ to be the circuit $\cL(\cT_n^h;\lambda_n^h)$ restricted to any $n$-link horizontal boundary $\lambda_n^h$ in the support of any of the $\cT_n^h$ gadgets, and
        \item $\cM'_n$ to be the circuit $\cL(\cM_n^h;\lambda_n^h)$ restricted to any $n$-link horizontal boundary $\lambda_n^h$  in the support of any of the $\cM_n^h$ gadgets.
    \end{itemize}
    The gadgets of this automaton have the following properties. Restricting the supports of the noise operators to $\lambda_n$ only, $\EC'_{n; 1}$ and $\EC'_{n; 2}$ have the same coarse-graining property at level $n$ as $\EC_n$, $EC$ A,B properties at $k = 3$ and $m = 1$, as well as nilpotence at level $\ell_T = 3$. 
    
    The gate gadgets $\cI'_{n; \cT}$ and $\cI'_{n; \cM}$ have the same coarse-graining action as $\cI_n$; $\cT'_n$ and $\cM'_n$ have the same coarse-graining action as $\cT_n$ and $\cM_n$.  Note that $\cI_n$, $\cT_n$, and $\cM_n$ denote Tsirelson's 1D gadgets and not the toric code automaton gadgets.  Additionally, they satisfy Gate A,B properties at $k=3$ and $m = 1$ (see Prop.~\ref{lemma:gate-ec-prime-TS}), and nilpotence at level $\ell_T = 3$.
\end{lemma}
\begin{proof}
Considering the circuits in the definition, we see that they are identical to those of our modification of Tsirelson's automaton with an addition of idle steps and repeated application of some of the $\cM$ gadgets. We note that some of these idle steps appear due to the restriction of the gadgets $\cM^{h/v}$  and $\cT^{h/v}$ of the toric code automaton to the boundaries that are perpendicular to gadget's orientation; upon restricting to the boundary, the action of the gadget is an idle. In addition, $\EC'$ gadgets are doubled in comparison to the $\EC$ used in the modification of Tsirelson's automaton. 

The proofs of the aforementioned properties of the gadgets of Tsirelson's automaton do not change upon an addition of extra idle steps and repetition of $\cM$ gadgets.
\end{proof}

In what follows, recall that all gadgets of the toric code automaton are nilpotent at level $\ell = 6$, using the results of the previous Lemmas.

 \begin{proposition} [Level of fault tolerant inner simulation] \label{proposition:gate-ec-prime-TC}
Consider the fault tolerant inner and outer simulations introduced in Subsecs.~\ref{subsec:inner-simulation} and \ref{subsec:outer-simulation}. %Assume a $p$-bounded approximate gadget error model. 
Then, there exists a level $k$ of inner simulation and an integer $0<m<k$ such that: 
% Then:
\begin{enumerate}
\item The gate gadgets $\cI_k$, $\cT_k^{h/v}$ and $\cM_k^{h/v}$ satisfy the $(Gate$ A$)_{k,m}$, and $(Gate$ B$)_{k,m}$  properties, and
\item  The $\widetilde{\EC} = \EC_{k}$  gadget satisfies the $({\EC}$ A$)_{k,m}$  and $({\EC}$ B$)_{k,m}$  properties 
\end{enumerate}
defined in Def.~\ref{def:TC_gateECprime}. In particular, setting  $m=\ell + 2 \ell_T + 5 = 17$ ($\ell$ denotes the level of nilpotence of the toric code gadgets determined in Lemma~\ref{lemma:nilpotentnumerics-TC} and $\ell_T$ denotes the level of nilpotence of the gadgets of the modified Tsirelson automaton determined in Sec~\ref{sec:tsirelson-proof-2}) and $k = 2m=34$ satisfies the conditions above.
 \end{proposition}

We will split the proof into separate parts, one for each property.  Our proof will share a lot of similarities with the proof of Lemma~\ref{lemma:gate-ec-prime-TS} but will be more complicated.  Instead of working with the actual numerical values for $m$, we will instead use $m =\ell + 2 \ell_T + 5$, and $k =  2m$, assuming that $\ell$ is the level at which all the gate and EC gadgets of the toric code are nilpotent and $\ell_T$ is the nilpotence level for the Tsirelson's automaton gadgets. The choice of $k  = 2m$ will be made to ensure that failures in neighboring $\widetilde{\EC}$ gadgets can also be properly cleaned up, as we will see in the proof. Keeping $k$, $\ell$, and $\ell_T$ as variables rather than specifying their numerical values will keep the proof applicable in case a tighter upper bound on $\ell$ is found. Otherwise, one can set $\ell = 6$, $\ell_T = 3$ in accordance with Lemma~\ref{lemma:nilpotentnumerics-TC} and Sec.~\ref{sec:tsirelson-proof-2} to obtain $m = 17$ and $k = 34$.
\begin{proof}[Proof ($\EC$ B)]

We begin by showing $\EC$ B, i.e. that an ideal decoder can ``absorb'' a layer of $\widetilde{\EC}$ gadgets if input damage and noise in the layer satisfy the conditions specified in the definition. We note that nowhere in the proof of $\EC$ B will doubling of the operation of the $\EC$ gadget (see Remark~\ref{remark:doubling}) be needed; therefore, this property would still hold if we assumed the ``undoubled'' definition for the gadget. This property will be used to formulate a corollary appearing after this proof.

We split the action of the ideal decoder into  $\Phi^k \circ \cL(\widetilde{\EC})$, and denote the noisy layer of error correcting gadgets in question as $\cL_\varepsilon(\widetilde{\EC})$. Assuming that the input state $\ket{\psi_{\mathrm{in}}}$ and the layer of noisy EC gadgets satisfy the conditions of the ($\EC$ B)$_{k,m}$ property, we will show that
\begin{equation}
    \cL(\widetilde{\EC})\circ \cL_\varepsilon(\widetilde{\EC})\ket{\psi_{\mathrm{in}}} =  \cL(\widetilde{\EC})\circ \cL(\widetilde{\EC})\ket{\psi_{\mathrm{in}}} = \cL(\widetilde{\EC})\ket{\psi_{\mathrm{in}}}. 
\end{equation}
Here, the first equality directly states that the noise is inconsequential, and the second equality simply uses the fact that the output of the first layer of noiseless $\widetilde{\EC}$ gadgets already coarse-grains the input at level $k$, to which a second repetition of this layer would do nothing.

Using the linearity of vertices in $\Sigma_k$ with respect to the $\widetilde{\EC}$ gadget (Prop.~\ref{prop:confinement-for-TC}), we can formally remove the syndromes in $\Sigma_k$, and later $\mathbb{F}_2$ add them to the output.  Thus we will be only concerned with the case of a syndrome-free configuration, with no non-trivial syndromes in $\Sigma_k$ 

Within a layer of $\widetilde{EC}$ gadgets, the filters in the definition of the $\EC$ B condition ensure that each gadget contains either at most one level-0 gadget failure or its input contains damage which fits in an $m$-cell. First consider the case when there is one  level-0 gadget failure inside some gadget $\widetilde{EC}^{(i)}$ isolated at level $k$ and the input is clean in the spatial support of the isolation region, i.e. $s_i = 1$, $m_i = 0$, and $s_j= 0$, $m_j = 0$ for $j$ neighboring $i$. By the confinement property (Fact~\ref{fact:4}), this summarizes the situation of a configuration of faults in $\cL(\widetilde{\EC})$ (acting on a syndrome-free input) which are isolated at level $k$. We will then show what happens when there are faults in neighboring $\widetilde{\EC}$ gadgets (which violates the isolation assumption) and also address the effect of the input damage. Together, this analysis will prove ($\EC$ B)$_{k,m}$.

We can schematically write:
\begin{equation}
\begin{split}
    \widetilde{\EC} \equiv \EC_k = (\EC_{k-1} )^{\otimes 9} &\circ [\cM_{k-1}^v,... \cI_{k-1}] \circ (\EC_{k-1} )^{\otimes 9} \circ [\cT_{k-1}^h,... \cI_{k-1}]\circ (\EC_{k-1} )^{\otimes 9} 
    \\
    &\circ... \circ (\EC_{k-1} )^{\otimes 9} \circ [\cM_{k-1}^v,... \cI_{k-1}] \circ (\EC_{k-1} )^{\otimes 9}
\end{split}
\end{equation}
where $[\cG_{k-1}^{h/v},..., \cI_{k-1}] $ denotes spatially stacked layers of either $\cM$ or $\cT$ gate gadgets (along with $\cI$ gadgets) which depend on the particular timestep (see \eqref{R0-definition} and \eqref{R0-actual-definition} for the precise sequence of gates). This circuit of gadgets is then followed by a noiseless layer of $\widetilde{\EC}$, whose application first starts with a noiseless layer of $\EC_{k-1}$. 
Considering a fault that is isolated at level $k$ and using the assumption that $k > \ell$,  by nilpotence and Lemma~\ref{lemma:higher-nilpotence-TC} (i.e. nilpotence at level $\ell$ implies nilpotence at all higher levels), it follows that the effect of such a fault is completely removed by the end of the noiseless layer of $\EC_{k-1}^{\otimes 9}$ gadgets at the latest.

Now we generalize by considering the case when there are level-0 gadget failures in neighboring $\widetilde{EC}$ gadgets, with the input state still assumed to be syndrome-free. Later, we will argue how the analysis of input damage can be reduced to this case. Suppose a level-0 gadget failure occurred in each of four interacting $\widetilde{\EC}$ gadgets meeting at a corner as shown in Fig.~\ref{fig:communication-TC}. The four failures are schematically labeled $f_{1}, f_2, f_3, f_4$ and they need not occur during the same time step.  We will now use the confinement properties (Prop.~\ref{prop:confinement-for-TC} and Facts~\ref{fact:1} and ~\ref{fact:4}) to understand the nature of the interactions between neighboring gadgets in this case. The gadget failure within a given $\widetilde{\EC}$ can affect the state restricted to the support of this $\widetilde{\EC}$ and its boundary which is a set of four $k$-links.  However, because of confinement (Fact~\ref{fact:1}), the syndromes at any of the $k$-links shared by neighboring $\widetilde{\EC}$ gadgets cannot affect syndromes in the interior of these gadgets.  This means that when faults occur in neighboring $\widetilde{\EC}$ gadgets, the output syndrome configuration can only differ from the syndrome configuration in the isolated case at the boundary between the gadgets.  Furthermore, due to Fact~\ref{fact:4}, the syndromes on a boundary $k$-link can be determined only from the damage and noise in two gadgets it is a boundary of.

Consider the effect of failure $f_1$ within the gadget $\widetilde{\EC}^{(1)}$. We can write the $\widetilde{\EC}$ gadgets, along with the noiseless $\widetilde{\EC}$ gadgets that appear in the following ideal decoder, in terms of gadgets of a lower level $\nil$. Let us assume that the fault $f_1$ occurred within some  composite gadget $\cG^{(1)}_\nil \circ \EC_\ell^{\otimes K}$  in $\widetilde{\EC}^{(1)}$. If the fault $f_1$ is isolated at level $\ell$, its effect would be cleaned up by a similar argument to that above, which requires a number of steps $t_1\leq 3 t_{\nil}$ (where $t_{\nil}$ is $\max (   t(\EC_{\ell}), t(\cG_\ell) )$, and $t(\cG_\ell)$ denotes the maximum number of timesteps in any gate gadget $\cG_\ell$ at level $\ell$\footnote{The factor of 3 in $t_1\leq 3 t_{\nil}$ comes from the fact that an error in a level-$\ell$ gadget, by nilpotence at level $\ell$, can only be cleaned up after acting with a clean layer of $\EC_{\ell}$.  Thus, a loose upper bound on the longest amount of time an error will be cleaned up in is the depth of $\EC^{\otimes K} \circ \mathcal{G} \circ \EC^{\otimes K}$, which is $3 t_{\nil}$.}, and an amount of space that is at most the spatial support of $\cG^{(1)}_\nil \circ \EC_\ell^{\otimes K}$.  Call this spatial region $\Lambda_1$. We note that $\Lambda_1 \cap \partial \, \text{supp}_v(\widetilde{\EC}^{(1)})$ is at most a segment of a $k$-link of length $3^{\nil+1}$. This comes from the fact that the longest side length of a level-$\nil$ gadget is $ 3^{\nil+1}$, coming from $\mathcal{M}^\ell$ gadget.
By confinement (Fact \ref{fact:0}) and nilpotence at level $\nil$, nontrivial syndromes away from the boundary $\text{supp}_v(\widetilde{\EC}^{(1)}) \setminus \partial \, \text{supp}_v(\widetilde{\EC}^{(1)})$ will be cleaned up in time $t_1$. However, there still may be syndromes remaining on the boundary after this time. Thus, the effect of $f_1$ on the boundary can be summarized as creating additional syndromes in $\Lambda_1 \cap \partial (\widetilde{\EC}^{(1)})$ for a duration of time $t_1$.

From confinement Fact~\ref{fact:4}, the syndromes on each $k$-link depend only on the syndromes and gadget operation on this link and in the bulks of the gadgets sharing this link (without the rest of their boundaries). Therefore, each $k$-link can be considered independently and we focus on the $k$-link between gadgets $\widetilde{\EC}^{(1)}$ and $\widetilde{\EC}^{(4)}$, namely $\lambda_{14} = \text{supp} (\widetilde{\EC}^{(1)}) \cap \text{supp}(\widetilde{\EC}^{(4)})$ in Fig.~\ref{fig:communication-TC}. Faults $f_1$ and $f_4$ can both affect the syndrome configuration on the boundary. The effect of each of these faults deposits syndromes in the regions $\Lambda_1 \cap \lambda_{14}$ and  $\Lambda_4 \cap \lambda_{14}$ within temporal intervals of length $t_1$ each (where $\Lambda_4$ is defined analogously to $\Lambda_1$ except for $\widetilde{\EC}^{(4)}$).  For $\Lambda_1 \cap \lambda_{14}$ and $\Lambda_4\cap \lambda_{14}$ to ``interact'' with one another, they cannot be spatially further than a distance $3 \times 3^\nil$ from each other\footnote{Here, `distance' is defined as the minimum distance between a point in $\Lambda_1 \cap \lambda_{14}$ and a point in $\Lambda_4\cap \lambda_{14}$.}, and the temporal duration of their effects must overlap.  
This is because if the damaged regions are further apart than this, the faults that caused them must have been isolated at level $\ell$ and thus would then have been cleaned up by nilpotence.  Under the operation of the two neighboring level-$k$ gadgets $\widetilde{\EC}^{(1)}$ and $\widetilde{\EC}^{(4)}$, call $U$ the set of spacetime events corresponding to when nontrivial syndromes are added to $\lambda_{14}$ due to the fault as well as the action of gadgets oriented perpendicularly to $\lambda_{14}$.  
By the arguments above, $U$ fits in a spacetime box of dimension at most $(6 t_\nil)_t \times (9 \times 3^\nil)_x$. 
Additionally, while these syndromes are added to $\lambda_{14}$, operations along $\lambda_{14}$ can additionally modify the syndrome configuration, and we must quantify how much the spacetime support containing nontrivial syndromes will additionally increase by.

\begin{figure}[!thbp]
    \centering
    \includegraphics[width=0.65\textwidth]{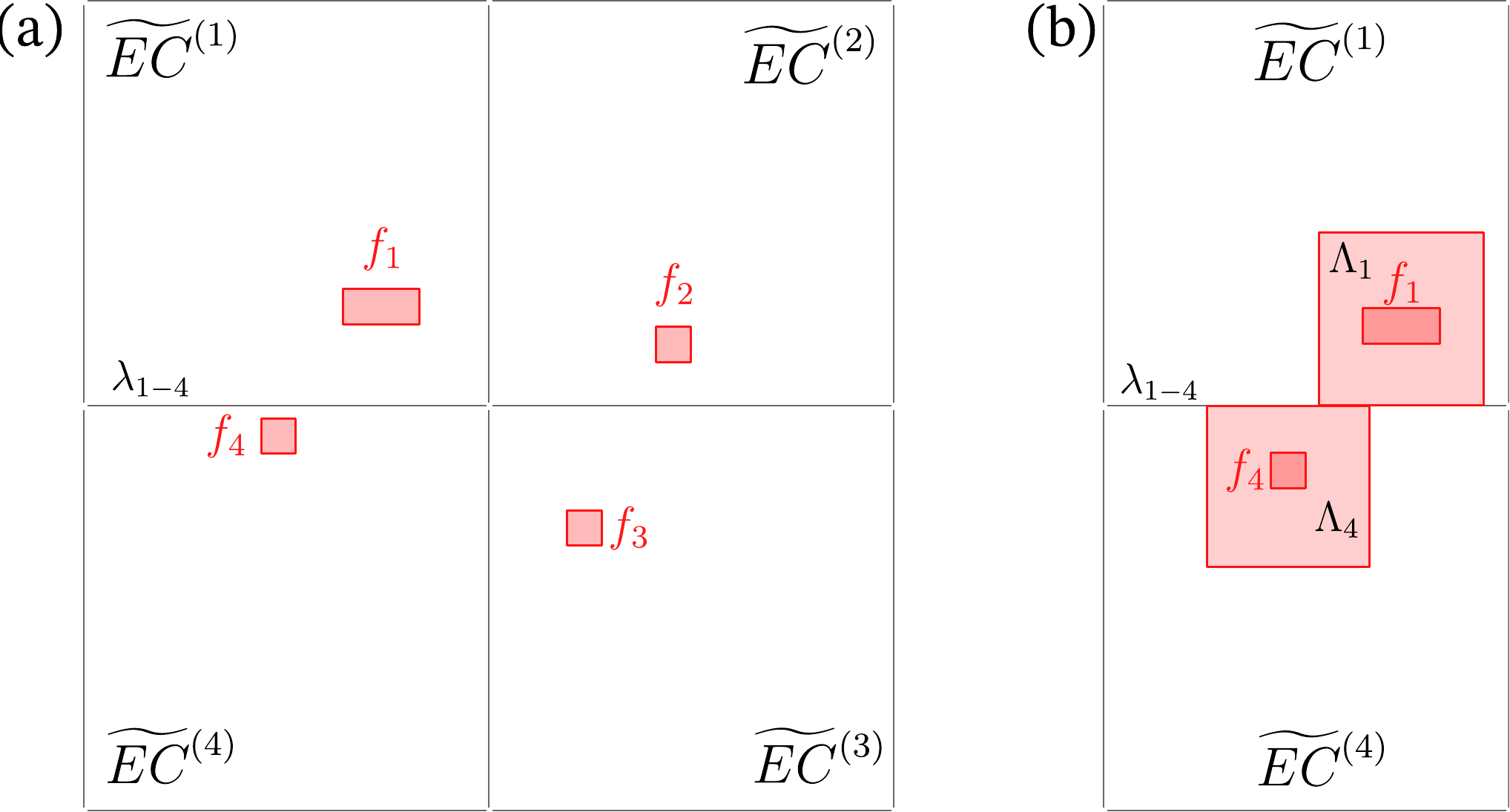}
    \caption{(a) Four neighboring $\widetilde{\EC}$ gadgets and level-0 gadget failures $f_{1-4}$ in their support. (b) Schematic showing the effect of faults $f_1$ and $f_4$ on the link $\lambda_{14}$ between two $\widetilde{\EC}$ gadgets. }
    \label{fig:communication-TC}
\end{figure}
We first observe that the action of the $\widetilde{\EC}$ layer restricted to the $k$-link $\lambda_{14} $ is equivalent to that of the ``boundary'' Tsirelson automaton of Lemma~\ref{lemma:Tsirelson-at-boundaries} at level $k$,  as discussed in Lemma~\ref{lemma:Tsirelson-at-boundaries}. An important modification are the additional ``idle'' steps when no feedback operations occur along the boundary but new syndromes can still be added to the boundary due to gadgets in the bulk.  Call $U$ the set of spacetime events where syndromes are added to the $\lambda_{14}$ boundary due to gadgets supported in the bulk.   After the process of adding syndromes from these gadgets onto the boundary has ended (which, as we discussed above, takes at most $6 t_\ell$ timesteps), we consider the total operator acting on the state and further evolve it under the 1D automaton to determine how much additional spacetime volume is needed to eliminate it. If both faults $f_1$ and $f_4$ occurred in gadgets whose support does not contain the corners of the support of $\widetilde{\EC}$, then the total number of added syndromes must be even.  This follows from the fact that each fault only creates an even number of syndromes, each gadget action preserves the parity of the number of syndromes, and after time $6 t_{\ell}$, syndromes have been removed from $\text{supp}(\widetilde{\EC}^{(1)}) \setminus \partial \, \text{supp}(\widetilde{\EC}^{(1)})$ and $\text{supp}(\widetilde{\EC}^{(4)})\setminus \partial \, \text{supp}(\widetilde{\EC}^{(4)})$.  Else, the number of syndromes on each $k$-link can be odd. In both of these cases, it is possible for us to ``clean'' the resulting damage (via multiplication by stabilizers) so that its support is contained entirely on boundary links.  Call $\cD_{1-4}$ the portion of the damage supported on $\lambda_{14}$.  

Now, we determine how large of a spatial box the damage $\cD_{1-4}$ can fit in. Since $U$ must fit in a spacetime box of the size $(6 t_\nil)_t \times (9 \times 3^\nil)_x$, we will show that the operation of gadgets along the link will spread the $\cD_{1-4}$ operator (we will continue to refer to it as the ``damage operator'' throughout its evolution) to fit in the spatial support of two neighboring level-$(\nil +4)$ gadgets. It will then be coarse-grained to level-$(\nil+4)$ by the soonest layer of level-$(\nil+4)$ $\EC$ gadgets. To show this, let us utilize a simple ``light-cone'' argument. The time $6 t_\ell$ fits at most 6 gate gadgets (i.e. not including EC gadgets). A layer of level-$\ell$ gate gadgets with a layer of preceding level-$\ell$ $\EC$ gadgets spread the damage operator by at most $3 \times 3^\ell$. Therefore, during the time $6 t_\ell$, the damage operator can at most spread over a spatial extent of $2 \times 6 \times 3 \times 3^\ell$. Therefore, the total resulting damage operator will be contained in an interval of size at most $9 \times 3^\ell + 36 \times 3^\ell \leq 3^{\ell + 4}$. Then, we assume that this operator fits in at most a pair of neighboring level-$(\ell+4)$ gate gadgets and is coarse-grained to level-$(\ell+4)$ by the subsequent layer of $\EC_{\ell + 4}$.
We then use level reduction $\nil + 4 \rightarrow 0$ and $k \rightarrow k - (\nil + 4)$, thus resulting in at most two neighboring level-0 gadget failures inside a level $k - (\nil + 4)$ Tsirelson automaton. 

Next, find the smallest value of integer $q$ for which both of the level-0 gadget failures reside in the support of gadget $G_q$ (which could be gate or $\EC$). The value of $q$ is clearly in the range $0 \leq q \leq k - (\nil + 4)$. We split the remaining analysis into two cases. 
First, suppose that $q \geq \nil_T+1$. Since $q$ is the smallest possible level by definition, the faults must belong to supports of different level-$\nil_T$ gadgets. Using nilpotence (this time for the modified Tsirelson automaton, which has nilpotence level $\nil_T$ for all gadgets), we see that the faults are cleaned up separately.  
Next, suppose that $q \leq \nil_T$.  Then we assume that the output configuration of gadget $G_q$ fails arbitrarily, and is subsequently coarse-grained at level $q$.  We perform level reduction by $q$ levels, in which case the fault becomes a single level-0 gadget failure within a level $k - (\nil + 4) - q \geq k - (\nil+ \ell_T + 4)$ gadget.  By nilpotence, this fault can be cleaned up if $k - (\nil+ \ell_T + 4) \geq \ell_T$. Thus, we require $k \geq \ell + 2 \ell_T + 4$ for $\EC$ B to hold (which is true for our choice $m =\ell + 2 \ell_T + 5$, and $k =  2m$).

So far we have shown that $\EC$ B holds in the case when there is strictly no input syndromes. 
Let us now consider the case when the input syndromes are present (subject to the damage passing through an $m$-cell filter for each gadget) but there are no spacetime faults in the (entire) layer. Expressing $\widetilde{\EC}$ through level-$m$ gadgets, we see that the first layer of $\EC_m$ will coarse-grain the damage operator to the corners of a single $m$-cell. Performing level reduction by $m$ levels, each input $m$-cell with non-trivial damage turns into 0-cell input damage within an effective level-$(k - m)$ $\EC$ gadget. If this input damage (which can be viewed as a failure of a level-0 gadget) is isolated at level $\ell$ after the level reduction, from nilpotence, it is cleaned up (assuming $k - m \geq \ell$). If it is not isolated, we use the same argument as in the previous paragraphs. We start with two level-0 gadget faults within a layer of effective level-$(k-m)$ $\EC$ gadgets. In order for a pair of non-isolated faults to be cleaned up, the total level of the gadgets $k-m$ has to be no less than $\ell + 2 \ell_T + 4$ (which is true if we choose $m =\ell + 2 \ell_T + 5$, and $k =  2m$).

Finally, we need to consider the case when there is a combination of faults and input damage, i.e. some $\widetilde{\EC}$ gadgets in $\cL_\varepsilon(\widetilde{\EC})$ are noiseless but there is input damage (i.e. $s_i= 0$ and $m_i  = m$), while others experience no input damage but instead have faults in their spacetime support (i.e. $s_i = 1$, $m_i = 0$).
The analysis can be simplified to analyzing the interaction between faults and input damage for each pair of neighboring gadgets. We have already considered interactions between pairs of faults and between pairs of gadgets with input damage. Thus, what remains is to consider the joint effect of input damage that is not sufficiently isolated from a level-0 gadget failure.
Reducing the level by $m$, the input damage that passed an $m$-cell filter becomes equivalent to failing a level-0 gadget in the reduced-level $\EC_{k-m}$ gadget. In addition, the level-0 gadget failure can at worst be a level-0 gadget failure in the reduced circuit after the soonest layer $\cL(\EC_{k-m})$. This reduces the problem to two level-0 gadget faults in an effective level-$(k-m)$ layer of $\EC$ gadgets. According to the earlier arguments, $k \geq m +  \ell + 2 \ell_T + 4 $ is sufficient for the effect of these two faults will be cleaned up. 

This completes the proof of the $\EC$ B property.
\end{proof}

Before proceeding with proving the rest of the parts in Prop.~\ref{proposition:gate-ec-prime-TC}, let us show a useful corollary of the $\EC$ B property:
\begin{corollary} \label{corollary:remove-input-noise}
    Consider a state $\cO_m \cO^{(k)} \ket{\psi}$, where $\ket{\psi}$ is a logical state, the operator $\cO_m \cO^{(k)} $ passes an $m$-cell filter on all $k$-cells, and the operator $\cO^{(k)}$ has syndromes coarse-grained at level $k$. A noise-free layer $\cL(\widetilde{\EC}^{(1/2)})$ acting on this state yields $\cO^{(k)} \ket{\psi}$. 
\end{corollary}
\begin{proof}
    We use the remark at the beginning of the proof of $\EC$ B saying that the property also holds for undoubled gadgets. In the proof of  $\EC$ B we show a stronger property that the $m$-cell input damage (in the case when $s_i = 0$) is cleared up before the output of the layer of $\widetilde{\EC}^{(1/2)}$, which we can summarize as $ \cL(\widetilde{\EC}^{(1/2)})\cO_m \cO^{(k)} \ket{\psi} = \cO^{(k)} \ket{\psi}$. 
\end{proof}

With this result, we will continue proving Prop.~\ref{proposition:gate-ec-prime-TC}.
\begin{proof}[Proof (Gate B)]

We now prove the Gate B property, which, informally, states that a layer of noisy gates denoted by $\cL_\varepsilon (\cG_k)$ is pulled through an ideal decoder to become noiseless level-reduced gates, provided the noise and the input damage pass through the filters stipulated in Def.~\ref{def:TC_gateECprime}.
Explicitly, for each level-$k$ gate gadget $\cG_k^{(i)}$ in the layer, denoting the number of $m$-cells containing damage operator within a given $k$-cell in its spatial support as $\frac{m_j^{(i)}}{m} \in \{0,1\}$, the total number of these cells in the input as $\sum_{j = 1}^K \frac{m_j^{(i)}}{m}$ and the number of faults inside the gate as $s^{(i)}$,  the condition $\sum_{j = 1}^K \frac{m_j^{(i)}}{m} + s^{(i)}  \leq 1$ must be satisfied.

We again split the action of the ideal decoder into  $ \Phi^k \cL(\widetilde{\EC})$, where the layer of error correction gadgets is noiseless.  Assume an input state $\cO_{\text{in}}\ket{\psi}$ where $\cO_{\text{in}}$ is the operator that, acting on a level-$k$ coarse-grained toric code state $\ket{\psi}$, outputs damage that (together with
% satisfies the Gate B condition.  Also assume that 
the noise realization, whose presence is denoted by the index $\varepsilon$), is consistent with the Gate B condition. Then, we will first show:
\begin{equation} \label{eq:gateb-modification}
 \cL(\widetilde{\EC}) \circ \cL_{\varepsilon}(\cG_k) \circ \cO_{\text{in}} \ket{\psi}=  \cL(\cG_k)\circ  \cL(\widetilde{\EC}) \circ \cO_{\text{in}} \ket{\psi},
\end{equation}
where on the right-hand side, the layer of gates $\cL(\cG_k)$ is noiseless and acts on a level-$k$ coarse-grained state  due to the noiseless $ \cL(\widetilde{\EC})$ preceding it.  Following this state by the pushforward map $\Phi^k$, we can move the layer of noiseless gates $\cL(\cG_k)$ past the map, due to the state at this point being coarse-grained at level $k$. This yields a level-reduced layer of gates $\cL(\cG_0)$, and thus
\begin{equation}
   \Phi^k \cL(\widetilde{\EC}) \circ \cL_{\varepsilon}(\cG_k) \circ \cO_{\text{in}} \ket{\psi} = \cL(\cG_0) \circ \Phi^k \cL(\widetilde{EC}) \circ \cO_{\text{in}} \ket{\psi} 
\end{equation}
which would complete the proof.

Thus, we need to argue that Eq.~\eqref{eq:gateb-modification} holds. 
We will use similar logic to the proof of the $\EC$ B condition. However, one difference consists of having to account for inputs which are coarse-grained at level $k$ (recall that in the analysis of $\EC$ B, we could assume a clean input due to linearity) and also in having a different gadget action at the boundaries. Specifically, the action of gadgets at the boundaries will be a modified Tsirelson automaton according to Lemma~\ref{lemma:Tsirelson-at-boundaries}.

To deal with the coarse-grained input, we can neglect input syndromes at linear points of level-$k$ gadgets (as usual, we can then take these removed input syndromes, act with a noiseless version of the respective level-$k$ gadget, and  $\mathbb{F}_2$ add the result to the output at the end), but we cannot modify syndromes on nonlinear points. 
Despite this, the same arguments as in the proof of $\EC$ B still hold in the presence of these syndromes. This is because the nilpotence property is formulated in the presence of a level-$k$ coarse-grained input and the nilpotence level $\ell_T$ holds for all gadgets.  In particular, we can follow the proof of $\EC$ B -- the analysis is essentially identical but some arguments deserve some mention.  First, the construction of the set $U$ of spacetime events is the same as before.  The level reduction argument by $\ell+4$ is also true even in the presence of input syndromes, and the syndromes from the set $U$ become consecutive level-0 gadget faults after level reduction.  The argument for cleaning up these faults is also identical because definitions like nilpotence work for coarse-grained input damage.  Thus, the faults and the $m$-cell input damage do not have any effect on the output after the application of the noiseless $\cL(\widetilde{\EC})$ layer in the decoder (i.e. the left-hand side of Eq.~\eqref{eq:gateb-modification}). The absence of the effect of the noise is captured by replacing the original layer of gadgets with the noiseless layer of the same gadgets $\cL(\cG_k)$, and the absence of the effect of the $m$-cell input damage can be recast as replacing $\cO_{\text{in}}\ket{\psi}$ with $\cL(\widetilde{EC})\circ \cO_{\text{in}}\ket{\psi}$ (i.e. a layer of $\widetilde{EC}$ gadgets removes damage fitting in an $m$-cell, which is shown in Corollary~\ref{corollary:remove-input-noise}). Upon these replacements, we can additionally remove the last noiseless layer of $\cL(\widetilde{EC})$ as it has no effect on the state that is already coarse-grained to level $k$.  This shows equation Eq.~\eqref{eq:gateb-modification} and thus the Gate B property.
\end{proof}
\begin{proof}[Proof (Gate A)] We now show the Gate A property for the toric code automaton, which states that if $m$-cell filters are present before the action of a layer of gates, cell filters can also be added after the action of the gates if the following condition on noise and input damage is satisfied. Consider gate gadget $\cG_k$, whose spatial support consists of a number of $k$-cells which we denote with $K$. If the number of $m$-cells containing input damage within the $i$-th $k$-cell of $\cG_k$ is $\frac{m_i}{m} \in \{0,1\}$, then the total number of $m$-cells containing input damage in the support of $\cG_k$ (i.e., $\sum_{i = 1}^K\frac{m_i}{m}$) and the number of level-0 faults $s$ in $\cG_k$ must satisfy $\sum_{i = 1}^K \frac{m_i}{m} + s  \leq 1$.

We will use similar logic to our proof of the $\EC$ B and Gate B conditions, this time to show that the damage will undergo limited spread (i.e. will fit within an $m$-cell where $m = \ell + 2 \ell_T + 5 $) before either exiting as an output or being cleaned up completely. There are additional considerations that we will need to address, namely that we cannot assume a clean layer of $\widetilde{\EC}$ gates coming from the ideal decoder at the end (unlike in the $\EC$ B and Gate B conditions). 

First, let us show that when the Gate A conditions are met, the presence of an input $m$-cell damage in the support of a given gate $\cG_k^{(i)}$ cannot lead to an output damage in the support of this gate; this is true even if this $m$-cell damage is not isolated from damage or noise in neighboring level-$k$ gates. As we argued earlier, we can consider each boundary independently of the others.  We then can use a similar argument to that of the proof of $\EC$/Gate B to show that if $k$ is large enough, the effect of the input damage will be cleaned up.
First, we perform reduction by $m$ levels, causing input damage in the support of $\cG_k^{(i)}$, as well as noise or input damage in the supports of neighboring gates, to turn into an effective level-0 gadget failure.  
Thus, we consider a level-0 gadget failure at the very beginning of an effective $\cG_{k-m}^{(i)}$ gate that can possibly interact with neighboring level-0 faults.
First, if $k-m \geq \ell+1$, we note that the damage restricted to regions away from the boundary of the gadget, such that the damage is isolated,  will be cleaned up due to nilpotence (the additional `+1' is because nilpotence can only be applied after expanding a level-($\geq \ell+1$) gate through level-$\ell$ gadgets). Thus, what is left to show is that the damage at the boundary will be cleaned up before the output.

Analogously to the proof of $\EC$ B and Gate B, the set $U$ (corresponding to spacetime events of adding syndromes to the boundary due to level-0 gadgets supported in the bulk of neighboring level-$(k-m)$ gadgets) fits in a spacetime box of size $(6 t_\ell)_t \times (9 \times 3^\ell)_x$.   
A boundary Tsirelson gadget of level $k \geq(\ell + 2 \ell_T + 4)+1$ acting along this boundary will be able to clean up the damage due to events in $U$.
To argue this, we first express this gadget through constituent level-$(\ell + 2 \ell_T + 4)$ Tsirelson $\EC'$ and gate gadgets. The presence of noiseless layers of $\EC'_{\ell + 2 \ell_T + 4}$ gadgets that follow the initial layer of $\EC'$ and gate gadgets allow us to use nilpotence and follow the same reasoning as in the proof of $\EC$/Gate B. This argues that the effect of the $m$-cell input damage in $\cG_{k}^{(i)}$ will be fully cleaned up. 

Thus, only level-0 gadget failures and their interaction with other level-0 gadget failures in neighboring level-$k$ gadgets can lead to output damage. Let us show that the spread of the output damage within a given level-$k$ gate will be contained within an $m$-cell for an appropriately chosen value of $m$.  First, if the level-0 gadget failure is isolated at level $\ell$, then its effect leads to damage spreading so that it still fits in a square with side length at most $3 \times 3^\ell$, and thus, leads to output damage fitting in an $(\ell+1)$-cell. However, if a level-0 gadget failure within a given level-$k$ gate gadget is not sufficiently isolated from a fault within another level-$k$ gadget, the damage along the boundary may grow to have even larger support. As in the proof of $\EC$ B, we wait until the damage from both faults becomes coarse-grained, which, by our loose bound, happens at level $\ell+4$ (thus we already have to assume $m \geq \ell +4$). We then perform level reduction by $\ell + 4$ levels, and consider a level-reduced Tsirelson automaton at the boundary with two neighboring level-0 gadget faults. As before, we assume that there exists the smallest level $q$ such that both faults reside within the same gadget $G_q$. If $q \geq \ell_T +1$, both faults are corrected independently within level-$\ell_T$
 gadgets. Then, the damage from each of the faults spreads such that it fits in at most an $(\ell_T + 1)$-cell. This gives $m \geq (\ell + 4)+ (\ell_T + 1) $. Alternatively, $q \leq \ell_T$, in which case we perform an additional level reduction by $q$ levels, and reduce the problem to a single failure of the $G_q$ gadgets. This single failure may additionally spread by $\ell_T +1$ levels before being cleaned up due to nilpotence. Thus, $m \geq (\ell+ 4) + \ell_T + (\ell_T + 1) =\ell + 2 \ell_T + 5 $. This last lower bound is the largest that we obtained so far, and it determines the acceptable value of $m$.
\end{proof}
\begin{proof}[Proof ($\EC$ A)]
Finally, we show $\EC$ A, which says that if $i$-th $\widetilde{\EC}$ gadget in a layer $\cL(\widetilde{\EC})$ contains $s_i \in \{0,1\}$ faults, its output passes through an $s_i m$-cell filter. 
For this property, we have to assume that the input state is arbitrary. Here is (the only place) where we will use the doubling of the $\widetilde{\EC}$ operation  that naturally divides $\widetilde{\EC}$ into two periods of its operation, as discussed in Remark~\ref{remark:doubling}.   

We will split our analysis into multiple cases.  First, consider the case when faults occur only during the second period of the operation of the entire layer of $\widetilde{\EC}$. Then, the first period coarse-grains the input to level $k> \ell$ by Lemma~\ref{lemma:EC-coarse-grains}. Then, we consider the second period where faults fulfilling $\EC$ A condition can occur. This situation is covered by the proof of Gate A (arguments in which are also applicable to $\EC$ and not only gate gadgets). 
Let us now suppose that the faults only occur during the first period of the entire layer of $\widetilde{\EC}$. In this case, regardless of what the output after the first period was, the second period of operation is noiseless. By the coarse-graining property of a half of the $\widetilde{\EC}$ gadget, it will produce the output that is coarse-grained at level $k$, which is also consistent with the statement of $\EC$ A. 
Finally, we are left to consider a scenario where faults occur during either period in different gadgets. 

The case when the faults can occur at any time does not reduce to the first two cases only if the faults in neighboring $\widetilde{EC}$ gadgets occurred in different periods, are not sufficiently isolated in spacetime, and have occurred near the boundary between the gadgets. Due to Facts~\ref{fact:1} and~\ref{fact:4}, we can restrict our consideration to two neighboring gadgets $\widetilde{EC}^{(1)}$ and $\widetilde{EC}^{(2)}$, and we assume without loss of generality that the fault in gadget $\widetilde{EC}^{(1)}$ occurred during the first period and the fault in neighboring gadget $\widetilde{\EC}^{(2)}$ during the second period.

First, let us consider the effect of the fault inside $\widetilde{EC}^{(1)}$, which we label $f_1$, keeping in mind that it occurred not far from the boundary shared by $\widetilde{\EC}^{(2)}$  and that there is no fault in $\widetilde{\EC}^{(2)}$ until the second period of operation. We label the fault inside $\widetilde{\EC}^{(2)}$ as $f_2$, and also label the boundary between the supports of $\widetilde{EC}^{(1)}$ and $\widetilde{EC}^{(1)}$ gadgets as $\lambda_{12}$. Let us expand the action of the layer of $\widetilde{\EC}$ gadgets through constituent level-$(k-1)$ gadgets. We write:
\begin{equation}
    \cL(\widetilde{\EC}) = \left [\cL'(\EC_{k-1}) \circ \cdots \circ \cL'(\cG_{k-1}) \circ \cL'(\EC_{k-1})\right ] \circ \left [\cL(\cG_{k-1}) \circ \cL(\EC_{k-1}) \circ \cdots \circ \cL(\EC_{k-1}) \right],
\end{equation}
where the first and second periods of operation are grouped into square brackets, and the layers belonging to the second period are additionally marked $\cL'$. 
Consider first the case when $f_1$ occurs before the last layer of $\EC_{k-1}$ gadgets within the first period of operation. In that case, the subsequent $\cL(\EC_{k-1})$ is noiseless and will coarse-grain its input to level $k-1$ in its support, including the boundary $\lambda_{12}$ (because there are no faults in the gadgets in the support of the neighboring $\widetilde{\EC}^{(2)}$ that occur during this coarse-graining). Therefore, at the start of second period, the state is coarse-grained to level-$(k-1)$ in the bulk of the support of $\widetilde{\EC}^{(1)}$ and  $\widetilde{\EC}^{(2)}$, as well as at $\lambda_{12}$. Now, we consider the effect of $f_2$ on top of a level-$(k-1)$ coarse-grained input. If $f_2$ occurred before the last  layer of $\EC_{\ell}$ gadgets within the last $\cL'(\EC_{k-1})$ layer, we can use the fact that this last layer is noiseless, the fact that $k-1 > \ell$,  and that the input is thus coarse-grained at scale $>\ell$ to apply nilpotence and argue that the fault will be cleaned up.  Thus the output syndrome configuration is the same as the one in which the $f_2$ did not occur during the second period, thus being coarse-grained at level $k$ due to the coarse-graining property of $\widetilde{\EC}$.  If $f_2$ occurred in the last layer of $\EC_{\ell}$ gadgets, the damage will simply spread to an $\ell$-cell, which is consistent with $\EC$ A so long as $m > \ell$.

Finally, it is left to consider the case when the fault $f_1$ occurred within the last two layers operating in the first period, i.e. within $\cL(\cG_{k-1}) \circ \cL(\EC_{k-1})$. The input to these layers must be coarse-grained at level $k-1$ (in the region of interest) because the preceding layers must have been noiseless and include $\cL(\EC_{k-1})$. 
We then use the same argument as in the proof of Gate A, taking into account that the fault $f_2$ happens only in the second period which implies that the output after the first period must pass through an $m$-cell filter. Thus, when the second period of operation starts, the input is coarse-grained at least at level $k-1$ up to an $\ell$-cell damage in the support of $\widetilde{\EC}^{(1)}$, after which the fault $f_2$ occurs in the support of $\widetilde{\EC}^{(2)}$. Using the argument from the proof of Gate A again, we conclude that the effect from the input damage and the fault will be cleaned up (including at the boundary $\lambda_{12}$) as long as $(k-1) - \ell \geq \ell + 2 \ell_T + 5$ (which is satisfied if we choose $m =\ell + 2 \ell_T + 5$, and $k =  2m$), with the output configuration passing through the filters in the $\EC$ A condition.

This completes the proof of $\EC$ A and thus, the proof of the Proposition.
\end{proof}

 \begin{corollary} \label{corollary:full-simulaton-TC}
Consider a circuit $\mathcal{C}$ on a 0-cell that is $T$ idle operations, i.e. $\cC = \cI_0^T$.   Construct the fault tolerant simulation of it $\widetilde{FT}(\mathcal{C})$ satisfying the conditions of previous Proposition (\ref{proposition:gate-ec-prime-TC}).  Iterate this construction $n$ times to form ${\widetilde{FT}}^n(\mathcal{C})$ and apply it to a clean toric code state of size\footnote{We assume that $L$ is a power of $3^k$ for simplicity in this proposition, but all the results here and below can be readily generalized to the case when $L$ is any power of $3$ by inserting floor functions in appropriate places.} $L\times L$ with $L = 3^{nk}$.  If the gadget error model for ${\widetilde{FT}}^n(\mathcal{C})$ is $p$-bounded, then pulling the circuit through a hierarchical $\ast$-decoder produces $\mathcal{C}$ with a gadget error model that is $(Ap)^{2^{n}}$-bounded for some constant $A$. 
 \end{corollary}
 \begin{proof}
With the values of $k$ and $m$ in the previous Proposition, both $\EC$ and Gate properties are satisfied.  Consider a circuit $\mathcal{C}$ and its fault tolerant simulation $\widetilde{FT}(\mathcal{C})$.  Corollary~\ref{cor:toriccodesuppre} indicates that if the gadget error model in $\widetilde{FT}(\mathcal{C})$ is $p$-bounded, pulling $k$ levels of hierarchical $\ast$-decoders through the circuit yields a gadget error model in the simulated $\mathcal{C}$ which is $A p^2$-bounded.  Starting from ${\widetilde{FT}}^{n}(\mathcal{C})$ and iterating this $n$ times gives a $(Ap)^{2^{n}}$-bounded gadget failure model.  From $k = 34$,  $n = \frac{1}{34} \log_{3} L$ for an $L\times L$ lattice of qubits, which gives a gadget error model of $(Ap)^{2^{(\log_{3} L)/34}} \approx (Ap)^{L^{0.0186}}$.
 \end{proof}

Before presenting the final theorem, we note that we have thus far been only studying a classical problem corresponding to the operation of $X$-type toric code automaton, which experiences $X$-type errors only.  We first need to combine this $X$-type automaton with the ``dual'' $Z$-type automaton:

 \begin{corollary} \label{corollary:X-Z-FT}
Consider a circuit $\mathcal{C}$ on a 0-cell that is $T$ idle operations, i.e. $\cC = \cI_0^T$.   Construct both the $X$ and $Z$ automata corresponding to the simulation ${\widetilde{FT}}^n(\mathcal{C})$ on a lattice of size $L\times L$ with $L = 3^{nk}$, with the full toric code automaton corresponding to simultaneous operation of the $X$ and $Z$ automata as discussed in Sec.~\ref{subsec:TC}. We run the full toric code automaton on a clean toric code initial state, assuming a $p$-bounded decoupled gadget error model. The full toric code automaton simulates $\mathcal{C}$ with a $(Ap)^{2^{n}}$-bounded decoupled approximate gadget error model (see Def.~\ref{def:approximate-decoupled-gadget-error-model}) for some constant $A$. 
 \end{corollary} 
 \begin{proof}

 Considering the full toric code automaton, we add hierarchical level-$n$ $\ast$-decoders at the end of $X$- and $Z$-type automata. We then simultaneously pull both types of decoders  through the circuit.

 Let us first discuss what happens when we pull one layer of the hierarchical $\ast$-decoders $\widetilde{D}$ of $X$ and $Z$ type through the circuit. For each noise realization this is determined by pulling the $\ast$-decoders for each automata separately and combining the outcomes as in Fact~\ref{fact:XZ-decoupling}. Therefore, pulling the $X$- and $Z$-type $\ast$-decoders through the entire circuit can be done independently and this procedure is then repeated until we pull the hierarchical $\ast$-decoders past the full circuit.

 Consider a noise realization with $E_X, E_Z, F_{E_X}, F_{E_Z}$ sampled from a $p$-bounded decoupled gadget error model. Upon pulling $\ast$-decoders past one simulation layer, we obtain a layer simulated elementary gadgets also of $X$ and $Z$ type. Each failed gadget is in one-to-one correspondence to a bad exRec of $X$ or $Z$ type respectively. Applying this to the outer simulation layer by layer, we obtain a simulated circuit with decoupled gadget error model. The properties of this error model are determined by the distribution of bad exRecs of $X$ and $Z$ type. Any $X$-type exRec is bad if it contains at least two $X$-type gadget failures, and any $Z$-type exRec is then bad if it contains at least two $Z$-type gadget failures.  
 The probability that arbitrary sets of $X$ and $Z$ type exRecs $A_X$ and $A_Z$ are bad is 
 \begin{equation}
 \mathbb{P}[(A_X \subseteq F_{E_X})\wedge (A_Z \subseteq F_{E_Z})] \leq \sum_{a_X, a_Z \text{: min. sets}} \mathbb{P}[a_X \text{ and } a_Z]
 \end{equation}
 where $a_X$ and $a_Z$ are sets of level-0 $X$ and $Z$ gadget failures of minimal size causing all exRecs in $A_X$ and $A_Z$ to be bad.  Since the minimal event for each exRec is two failures, by $p$-boundedness we have
 \begin{equation}
 \mathbb{P}[a_X \text{ and } a_Z] \leq p^{2 |A_X| + 2 |A_Z|}.
 \end{equation}
 There are $\leq A^{|A_X| + |A_Z|}$ such events, with $A$ a combinatorial factor related to the maximal spacetime volume of the exRecs. Thus
 \begin{equation}
 \mathbb{P}[(A_X \subseteq F_X)\wedge (A_Z \subseteq F_Z)] \leq (Ap^2)^{|A_X| + |A_Z|}
 \end{equation}
 thus consistent with an $Ap^2$-bounded decoupled gadget error model.  Iterating this $n$ times to pull the entire hierarchical $\ast$-decoder through proves the Corollary.
 \end{proof}

 The following corollary summarizes the fault tolerance result in a slightly more general way. In particular, it states that if the toric code automaton is run for a superpolynomial amount of time, the logical state can still be ``efficiently recovered''.

\begin{corollary} \label{cor:logical_error_bound}
Assume a syndrome-free initial logical state $\ket{\psi}$ of the toric code on an $L \times L$ lattice of qubits, where $L = 3^{nk}$ for some integer $n$. We apply the full toric code automaton $\widetilde{FT}^n(\cC)$ consisting of a simultaneous operation $X$ and $Z$ circuits as explained in Sec.~\ref{subsec:TC} with $\cC = \cI_0^T$, with the full automaton experiencing a general $p$-bounded gadget error model, and whose respective $X$- and $Z$-type automata satisfy the assumptions of Prop.~\ref{proposition:gate-ec-prime-TC}.

Run the automaton for a time $t$, and if the $t$ did not occur right after a circuit $\widetilde{FT}^n(\cI_0^{T'})$ for some integer $T' \leq T$, append noiseless gates to complete the circuit to $\widetilde{FT}^n(\cI_0^{T'})$. Then, act with a noiseless layer of $\widetilde{FT}^{n-1}(\widetilde{\EC})$ of both $X$ and $Z$ types, producing a state that we call $\rho_t$. Then, $\rho_t$ is a mixture of syndrome-free logical states of the toric code, and
    \begin{equation} \label{main_logical_error_bound}
       \| \rho_t - \ketbra{\psi}{\psi}\|_1 \leq   T'  (A p')^{2^{ (\log_{3} L)/k }}
    \end{equation}
for some constant $A$, where $k = 34$ according to Prop.~\ref{proposition:gate-ec-prime-TC} and $p' = C p^\gamma$ for some positive constants $C, \gamma$.
\end{corollary}
\begin{proof}
By Lemma \ref{lemma:approximate-replacement-tc}, the true state $\rho_t$ is close in trace distance to the state we would obtain under an approximate noise model (see App.~\ref{app:noisedetails} for the definition of the approximate noise model), i.e. $\|\rho_t - \rho_{t, \text{appr}} \|_1 \leq 2 T' \exp(-KL)$. Thus, we will consider $\rho_{t, \text{appr}}$ under a $p$-bounded approximate gadget noise model. Furthermore, by Lemma~\ref{lem:equiv_decoupled_gadget}, the resulting approximate gadget noise model is equivalent to $p'=Cp^\gamma$-bounded decoupled gadget noise model, to which we can now apply Corollary \ref{corollary:X-Z-FT}. 

Due to the application of a noise-free layer of $\widetilde{FT}^{n-1}(\widetilde{\EC})$ gadgets at the end, the state $\rho_{t, \text{appr}}$ is syndrome-free by the coarse-graining property in Lemma \ref{lemma:EC-coarse-grains}. 
This state can be written as a mixture of syndrome-free states from the ensemble $\{\ket{\psi}, L_i \ket{\psi}\}$ where $L_i$ denotes a logical Pauli operator (which can be represented as a product of $X$ and $Z$ logical operators).  %Thus, the state after $\widetilde{FT}^{n-1}(\widetilde{\EC})$ must have no nontrivial syndromes and thus must be a mixture from the aforementioned ensemble. 

We now compute the probability that a state drawn from this ensemble is not $\ket{\psi}$, which can be used to upper bound the trace distance $\|\rho_{t, \text{appr}} - \ketbra{\psi}{\psi} \|_1$.  
As mentioned above, if $t$ is not at the end of a subcircuit $\widetilde{FT}^n(\cI_0^{T'})$ for some $T'$, we complete the circuit by adding noiseless gates.  Since the noise model is $p$-bounded for the noisy part of the circuit, this distribution is stochastically dominated by a $p$-bounded noise model for the \emph{entire} circuit that now includes the completion.

First, we notice that applying a round of error correction $\widetilde{FT}^{n-1}(\widetilde{\EC})$ at the end (which is noiseless, and thus, the $X$ and $Z$ sectors decouple trivially) can be viewed as applying $\ast$-decoders of $X$ and $Z$ types without the pushforward maps. This is because we can alternatively write $\widetilde{D^*}^n = \Phi^k \cL(\widetilde{EC}) \circ \dots \circ  \Phi^k \cL(\widetilde{EC}) = \Phi^{nk} \cL(\widetilde{FT}^{n-1}(\widetilde{\EC}))$. 
According to Corollary \ref{corollary:X-Z-FT}, if we were to apply the $\ast$-decoders to the output of the automaton, the rate of the error model in the simulation is squared for every level of the outer simulation. In total, we obtain the simulated circuit $\cC$  with a $(Ap')^{2^{n}}$-bounded error model effectively acting on it. In the final simulated circuit, a fault of one of the gadgets in $\cC$ is equivalent to an $X$ or $Z$-type logical operator, and its probability must be $\leq C'  T' (A p')^{2^{ (\log_{3} L)/k }}$ by a union bound.
This is true even if we remove the pushforward map at the end. Once we remove the pushforward map, this becomes the probability that a logical operator has been applied to the state, which is $\leq C' T' (A p')^{2^{n}}$. Thus, we obtain that $\|\rho_{t, \text{appr}} - \ketbra{\psi}{\psi} \|_1 \leq  C' T' (A p')^{2^{ (\log_{3} L)/k }}$. 

Finally, we use Prop.~\ref{prop:apprnoise}, Lemma~\ref{lemma:approximate-replacement-tc}, and Lemma~\ref{lem:equiv_decoupled_gadget} to obtain
    \begin{align}
    \|\rho_t - \ketbra{\psi}{\psi} \|_1 &\leq \|\rho_t - \rho_{t, \text{appr}} \|_1 + \|\rho_{t, \text{appr}} - \ketbra{\psi}{\psi} \|_1 \nonumber\\
    &\leq 2 T'\exp(-KL) +  C'T'(A p')^{2^{ (\log_{3} L)/k }} \nonumber \\
    &\leq  T' (A' p')^{2^{(\log_{3} L)/k }}
    \end{align}
for $p' = C p^\gamma$ and suitable positive constants $K, A, A',C,\gamma$. 
\end{proof}

\subsubsection{Fault tolerance for the probabilistic quantum automaton} \label{sec:5-PCA}

Until now, we assumed a \emph{deterministic} quantum automaton model, not the \emph{probabilistic} model introduced in Def.~\ref{def:pca}. Let us briefly discuss what happens in the probabilistic case and also in a more general situation where we allow for measurements of operators other than the 
% n we do not measure the 
correct stabilizer generators.  Our arguments will not be rigorous, but one can readily make them rigorous with more work by incorporating them into the formalism introduced in this paper.

Consider the set $V$ of spacetime locations of vertices associated with the locations of stabilizers whose measurements have been skipped. If this set is drawn from a $p$-bounded distribution, with high probability, one can hierarchically decompose this set into clusters, which, roughly speaking, are constructed in the following way. Call an $(r,R)$-isolated cluster in $V$ a set of points in $V$ that fit in a three-dimensional box with side $r$ and are isolated from any other point in set $V$ by distance at least $R$. 
We can choose positive constants $\alpha, \beta$ (subject to certain constraints) and a hierarchy of scales $ r_1 < r_2 < \dots < r_n <\dots$ (where the scale growth is between exponential and double exponential in the level $n$, i.e. $\gamma < \frac{r_{n+1}}{r_n}< \Theta(2^{\xi^n})$ with some $\gamma>0$, $0<\xi<2$), such that $V = \bigcup_n V_n$ where $V_n$ is the set of all points belonging $(\alpha r_n, \beta r_{n+1})$-isolated clusters in $V \setminus (\bigcup_{i=1}^{n-1}V_i)$ (i.e. where all the smaller scale clusters have been removed).  This is known as sparsity decomposition, and is due to G\'acs~\cite{gacs1983reliable,gacs_slides}. Moreover, the $p$-boundedness of the noise allows us to upper bound the probability of occurrence of clusters inside each set $V_n$\footnote{Probably the most elegant way to explain fault tolerance of a renormalization-group based automaton is through the sparsity decomposition. Consider a decomposition for the set of all error locations $E$, which we decompose $E = \bigcup_n E_n$. For a given automaton, if we can find parameters such that any isolated cluster of size under $\alpha r_n$ is cleaned up in a spacetime volume that is less than $\beta r_{n+1}$ even in the presence of lower-scale clusters, then such an automaton will correct errors in $E_n$ at each level of decomposition. Thus, one needs to show that the automaton at level $n$ can (i) handle isolated error clusters of associated scale and (ii) can in addition tolerate transient smaller-scale errors while doing so.  It would be interesting to prove the fault tolerance of the toric code automaton using this method, which could result in a significant simplification of the current proof. }.  Consider any given cluster $V_n^{(i)}$ in $V_n$. The measurements skipped in this cluster are preceded by the noise channels, assuming there is an additional $p$-bounded gadget noise model that the automaton is experiencing. The noise occurring before the measurements is not projected onto purely Pauli noise due to the skipped measurements. In addition, for each cluster of events where no measurements occur, the effect of this noise associated with a given cluster $V_n^{(i)}$ in $V_n$ can spread due to a light cone argument and will fit in a spacetime box of size $\alpha' r_n$. Let us call the associated spacetime qubit support of this box $Q_n^{(i)}$. We now require a new $(\alpha' r_n, \beta' r_{n+1})$-sparse decomposition $R = \bigcup_n R_n$ of all qubit locations in spacetime (such that the $\beta' r_{n+1} > 2 \alpha' r_n$) where each cluster $R_n^{(j)}$ at level $n$ in this new decomposition covers at least one cluster $Q_n^{(i)}$. Consider a given $R_n^{(j)}$, and a timestep that directly follows it. By definition, the measurements supported on all qubits contained in the spatial extent of $R_n^{(j)}$ must have occurred. This projects the effect of skipped measurements and the noise channels that occurred inside $R_n^{(j)}$ onto an approximate noise model associated with a damage operator $\cO_n^{(j)}$ supported in this spatial region. Then, one can show that there exists a gadget error model such that (i) the set of gadget failures inside $R_n^{(j)}$ results in the same output damage operator $\cO_n^{(j)}$ in the spatial support of $R_n^{(j)}$, (ii) the number of such gadget failures is proportional to the total number of failures in the original error model and finally, that (iii) this gadget error model is $\widetilde{p}$-bounded where $\widetilde{p}$ is polynomially related to $p$. This way, one can reduce the PCA model to a $p$-bounded gadget error model. 

A slightly more involved scenario happens when a gadget measures the wrong operator.  As mentioned in the discussion after Def.~\ref{def:pca}, this assumption is unreasonable if measurements are being made simultaneously, as only commuting operators can be simultaneously measured.  A more realistic scheme is to split each measurement round into several rounds. First, we measure a set of stabilizers  $\cS_1$ (with disjoint supports) and then measure another set of stabilizers $\cS_2$ (with disjoint supports) so that $\cS_1 \cup \cS_2 = \cS'$ (where $\cS'$ now can differ from the toric code stabilizer group $\cS$), and then we finally apply the feedback. It is now reasonable to talk about errors due to measuring an incorrect check.  However, the same considerations as described above (namely, for a given cluster we find the soonest slice in time when correct measurements are performed, and collapse the effect of the noise in the cluster to a damage operator at that time step) should show that this model can similarly be reduced to a $\widetilde{p}$-bounded gadget error model.

\subsection{Fault-tolerant initialization and readout} \label{sec:initialization}

For a fully end-to-end fault tolerant scheme, we need to address both fault-tolerant initialization and readout for the toric code. In fact, both these operations in our construction are closely related to the standard fault-tolerant initialization and readout of the toric code~\cite{Dennis_2002}, except that local initialization takes $O(\text{poly}(L))$ time. It is sufficient to show that we can fault-tolerantly prepare logical states which are eigenstates of the logical $Z$ operators. We show how to prepare the $\ket{\overline {00}}$ logical state (defined after Def.~\ref{def:toric-code}). The other eigenstates are obtained by applying a combination of logical $X$ operators (discussed after Def.~\ref{def:toric-code} and we refer the reader to Ref.~\cite{Dennis_2002} for a more in-depth discussion) which requires at most an additional depth-2 circuit and can thus be fault-tolerantly implemented if $\ket{\overline {00}}$ can.

To prepare the $\ket{\overline {00}}$ logical state, we first start with the product state $\ket{0}^{\otimes N}$ (where $N = 2L^2$ for an $L \times L$ lattice), and then run the 2D toric code automaton, whose operation starts with measuring all stabilizers of the toric code. Upon measuring all stabilizers in the first round, we obtain a state $\ket{\psi}$ which has the property that $X$-type noise (detected by the measurements of $Z$-type star stabilizers) is $p$-bounded. These errors can be effectively `absorbed' into the first layer of noisy $X$-type gadgets with an insignificant change to the noise model; thus, initialization in the Pauli-$X$ error sector can be essentially treated as noiseless. However, upon measuring the $X$-type plaquette stabilizers, the state has an \emph{arbitrary} distribution of $Z$-type noise, which need not be $p$-bounded for small $p$. 

By the results from Sec.~\ref{subsec:TC} and the same arguments as in proofs of Corollaries \ref{corollary:full-simulaton-TC} and \ref{cor:logical_error_bound}, let us assume a $p$-bounded decoupled gadget error model and restrict our attention to the operation of the $Z$-type toric code automaton.  We now show that at times following one full period of the automaton, the state is effectively equivalent to that which would be realized if we instead started with the perfectly clean logical state $\ket{\overline {00}}$\footnote{ That it is this specific logical state with high probability follows from the standard argument from Ref.~\cite{Dennis_2002}. Namely, the starting product state is a $+1$ eigenstate of both logical $Z$ operators. Further operation of the type-$Z$ automaton can only result in logical errors that possibly apply type-$Z$ operators and thus do not change the logical state. The $X$ sector has $p$-bounded noise distribution from the start, and the probability of the logical failure is close to 0. }:
% %
\begin{proposition}[Reliable noisy initialization] \label{prop:noisy-initialization}
    Start with the state $\ket{0}^{\otimes N}$ and run the toric code automaton $\widetilde{FT}^n(\cC)$ where $\cC = \cI_0^T$ under a $p$-bounded decoupled gadget error model, on an $L \times L$ torus with $L = 3^{nk}$. Call $\Delta = \mathrm{poly}(L)$ the depth of circuit $\widetilde{FT}^n(\cI_0)$. Due to the assumption about decoupled gadget error model and Fact~\ref{fact:XZ-decoupling}, we only study the action of the $X$-type automaton under a $p$-bounded gadget error model. 
    
    Call $\rho_\Delta$ the state after time $\Delta$; then, we have 
    \begin{equation}\label{eq:operator-covering1}
    \rho_{\Delta} = \sum_{\cO} p(\cO) \cO \ketbra{\overline{00}}{\overline{00}} \cO
    \end{equation}
    where $\cO$ is a tensor product of Pauli-$Z$ operators.  Call $m_j = m + (j-1) k$ for $j = 1,\cdots, n$, where $m,k$ are constants that are chosen such that Prop.~\ref{proposition:gate-ec-prime-TC} holds.  The probability distribution $p(\cdot)$ has the following property.  Sample $\cO \sim p(\cdot)$.  Then, with high probability, there exists another operator $\cO'$ such that $\cO' \cO \in \cS$ such that the following procedure succeeds:
    \begin{enumerate}
    \item Start with $j = 1$
    \item Find a subset of $jk$-cells such that the operator $\cO'$ restricted to the union of their support passes a layer of $m_j$-cell filters on each of them.  Call $\cC_{m_j}$ the set of  $m_j$-cells containing non-identity Pauli operators in $\cO'$ within these $j k$-cells.
    \item Remove non-identity Pauli operators inside these $j k$-cells, thus updating $\cO'$.
    \item Increment $j \to j+1$, and go to step 2. 
    \item When $j = n$, terminate the algorithm.
    \end{enumerate}
    We say that this procedure succeeds if the remaining operator $\cO' = \mathds{1}$.  Further, denote by $\cC'_{m_j}$ an arbitrary set of $m_j$-cells.  Then,
    \begin{equation} \label{eq:operator-covering2}
    \mathbb{P}(\cC'_{m_j} \subseteq \cC_{m_j}) \leq (Ap)^{2^{j-1} |\cC'_{m_j}|}
    \end{equation}
    for some constant $A$.
\end{proposition}
This proposition shows that the state $\rho_\Delta$ obtained after running one period of the toric code automaton differs from the perfect logical state $\ketbra{\overline{00}}{\overline{00}}$  by an operator that obeys a certain clustering/sparsity property.  In particular, the operator fits in a union of spacetime regions of increasing scale; the probability that a region at scale $k$ is needed to enclose the operator is $\sim (Ap)^{2^k}$.  This property results in $\rho_\Delta$ being essentially indistinguishable from a state one would have obtained by running the noisy automaton on a perfect initial logical state $\ketbra{\overline{00}}{\overline{00}}$ for the same amount of time. 
\begin{proof}
 First, we will show that, upon removing the part of the damage supported in $ \cC_{m_i}$ for $i < j$, the  $m_j$-cell damage in $jk$-cells (corresponding to $ \cC_{m_j}$) occurs according to a $(Ap)^{2^{j-1}}$-bounded distribution. We will do this by a hierarchical application of the $\EC$ A property.

First, we show this for $j=1$ and will then show how to reduce $j > 1$ to this case. The case $j=1$ follows from decomposing the circuit before the time $t$ in terms of  level-$k$ gadgets (which can be done by viewing it as a $\widetilde{FT}\left( \widetilde{FT}^{n-1}(\cI_0) \right)$ where $ \widetilde{FT}^{n-1}(\cI_0)$ is viewed as a circuit of elementary gadgets) and considering the last layer  $\cL(\widetilde{EC})$. We can add an $m$-cell filter or a $\varnothing$-cell filter to the $k$-cells corresponding to $\widetilde{\EC}$ gadgets 
that have at most one level-0 gadget failure.
$\cC_{m_1}$ denotes the set of $m$-cells within which the damage is confined. 
% in these $k$-cells. 
Because the noise is $p$-bounded, $\cC_{m_1}$ must be sampled from an $A p$-bounded distribution for a constant $A$ that depends on the number of elementary gadgets in $\widetilde{\EC}$.

Now, we want to prove the statement for $j > 1$. We show this by `removing' $j-1$ levels of outer simulation with the help of an $\ast$-decoder. For this, consider the last layer of gadgets $\cL(\widetilde{FT}^{j}(\widetilde{EC}))$ in the circuit. Let us append $(j-1)$  levels of a hierarchical application of an $\ast$-decoder to the end of the circuit.  We use the form of the hierarchical decoder $\widetilde{D^*}^{j-1}  = \Phi^{(j-1)k} \cL(\widetilde{FT}^{j-1}(\widetilde{EC}))$. 
Let us first argue that the application of the noiseless $\cL(\widetilde{FT}^{j-1}(\widetilde{EC}))$ contained in the $\ast$-decoder eliminates the same damage as the algorithm in the proposition run up to iteration $j-1$.  This follows from the iterated (and hierarchical) application of Corollary~\ref{corollary:remove-input-noise} with appropriate level reduction.  
Consider the layer of $\widetilde{EC}$ at the beginning of $\cL(\widetilde{FT}^{j-1}(\widetilde{EC}))$. This layer eliminates damage supported in $\cC_{m_1}$ and coarse-grains remaining damage to level $k$.  Additionally, any damage in a $k$-cell remains within this $k$-cell after $\widetilde{EC}$ due to the operator confinement properties. 
Next, consider the layer of $\widetilde{FT}(\widetilde{EC})$ at the beginning of $\cL(\widetilde{FT}^{j-1}(\widetilde{EC}))$, which eliminates damage supported in $\cC_{m_2}$ and coarse-grains all other damage to level-$2k$. %\cmt{Show this by induction.} 
After repeating this $j-1$ times, we obtain that all damage in $\cC_{m_i}$ for $i \leq j-1$ is removed, and the rest of the state will be coarse-grained to level $ik$. 
Because of confinement, the damage contained originally in a cell in $\cC_{m_j}$ will be coarse-grained at level $(j-1)k$ but will still be supported at most within the same $m_j$-cell. In addition, damage in $jk$-cells that do not permit an $m_j$-cell filter will be coarse-grained to level $(j-1)k$ but the supports will remain confined within the same $jk$-cells. 

Upon pulling hierarchical $\ast$-decoder  $(j-1)$ times through the circuit $\cL(\widetilde{FT}(\cI_0))$ and using Corollary \ref{corollary:full-simulaton-TC}, we obtain a simulated circuit which is $\cL(\widetilde{FT}(\cI_0))$ with an  $(Ap)^{2^{j-1}}$-bounded gadget noise model.  For this effective circuit $\cL(\widetilde{FT}(\cI_0)) = \cL(\widetilde{\EC}) \circ \cL(\cI_k) \circ \cL(\widetilde{\EC})$, we consider the last layer of $\cL(\widetilde{\EC})$ and repeat the argument for $j = 1$ to show that the distribution of permissible locations of filters is also $(Ap)^{2^{j-1}}$-bounded.  Finally, the permissible locations of $m$-cell filters in the simulated circuit are in one-to-one correspondence with the set of cells $\cC_{m_j}$ in the original circuit, thus proving the desired property.

Finally, we want to argue that the algorithm in the proposition statement will eliminate the damage with high probability after $j = n$ iterations.  This follows from the fact that upon pulling a level-$nk$ $\ast$-decoder through, the resulting gadget noise model acting on the simulation is $(Ap)^{2^{n-1}}$-bounded, which means that the probability that an error does not fit inside an $m_n$ cell is $(A'p)^{2^n}$. This is the probability of the damage $\cO$ not being the identity after the algorithm terminates.
\end{proof}

Note that the action of a level-$nk$ $\ast$-decoder on the actual state at time $\Delta$ will produce $\ketbra{\overline{00}}{\overline{00}}$ with high probability.  This is because the only errors that are capable of changing this state are Pauli-$X$ type. However, as we discussed, the initialization is effectively noiseless in $X$-noise sector. Therefore, the probability of the noise during the period of time $[0,\Delta]$ applying an $X$-type logical operator is bounded by $ (A''p)^{2^{n}}$ for some constant $A''$.  The  $Z$-type noise can at worst apply logical $Z$ operators during this time, which do not affect the encoded state.  Finally, other logical basis states can be obtained by applying $X$ logical operators after the fault-tolerant initialization.

\begin{remark}
    Under the same assumptions as in Prop.~\ref{prop:noisy-initialization}, assume that the initial state is a perfectly clean logical state of the toric code $\ketbra{\overline{00}}{\overline{00}}$. We can apply the same considerations as in the proof of Prop.~\ref{prop:noisy-initialization}. Thus, upon one full cycle of the toric code automaton, Prop.~\ref{prop:noisy-initialization} guarantees that the distribution of damage for the state with noisy initialization is the same as that for perfect initialization.  For all purposes, the noisy initialization produces a state with the same properties as perfect initialization.
\end{remark}

Note that we assumed the approximate noise model in Prop.~\ref{prop:noisy-initialization} for the brevity of the proof. The result above can be generalized to a more general $p$-bounded gadget noise model using Lemma~\ref{lemma:approximate-replacement-tc}, similarly to the earlier results in the paper.

The implementation of fault-tolerant readout is similar to the standard techniques. It is again sufficient to read out in one logical basis, such as that of the $Z$ logical operator (in which case the outcome is sensitive to the Pauli $X$ component of the noise as well as to measurement errors).  One can prepare an ancilla layer of the toric code in state $\ket{\overline 0\overline 0}$, apply a transversal CNOT operation between the two layers of the toric code, and destructively measure single-qubit $Z$ operators on the ancilla layer, reducing this to a state that can then be ported to a reliable classical computer for reliable logical readout.

\subsection{Numerical simulations and a conjecture} \label{sec:num-sim-and-conj}

In this subsection, we derive an upper bound on the number of level-0 errors required to produce a logical error 
in the toric code automaton (which we call an upper bound on `minimal weight failures').  Motivated by exact results from simulations on small system sizes, we conjecture that this bound is tight. We then show results of Monte-Carlo simulations that support a second conjecture stating that for several examples of relatively natural $p$-bounded noise models, this upper bound appears to control how the memory lifetime scales with $p$.

\subsubsection{Minimal weight failures} \label{ss:memory_upper_bound}

To formulate an upper bound on minimal weight failures, we will first define a specific setting in which errors will be studied. Consider a layer 
 $\cL(\EC_{s}) \circ \cL(\EC_{s})$
 of gadgets applied to a syndrome-free logical state. Assume a noise realization $\bm H$ with operator history
$\vec \cO = \{\cO_1,\cO_2,\dots,\cO_T\}$ applied at each time step satisfying the following properties: (i) each $\cO_t$ is a product of Pauli $X$ operators of weight $w_t$, (ii) the weights satisfy $\sum_{t=1}^T w_t = w$, and (iii) the union of the spacetime support of all operators is contained within the first $\cL(\EC_s)$.
     
The minimal level-$s$ oriented link failure weight $w_{\mathrm{link},o}^{(s)}$ with $o \in \{h,v\}$ is the minimum value of $w$ such that the output state from applying $\cL(\EC_s)$ has two syndromes at the endpoints of a single $s$-link oriented horizontally ($o=h$) or vertically ($o=v$). We will say that this link ``fails''. The minimal level-$s$ failure weight $w_{\mathrm{link}}^{(s)}$ is defined as the maximum of the two oriented failure weights: $w_{\mathrm{link}}^{(s)} = \max_o(w_{\mathrm{link},o}^{(s)})$. 
	
We will also find it useful to define a failure weight for ``corner-shaped errors'': we define the minimal level-$s$ oriented corner failure weight $w_{\mathrm{c},o}^{(s)}$, $o\in \{1,2\}$ as the minimum value of $w$ such that the output state has two syndromes located at diametrically opposite corners of an $s$-cell.  The value of $o$ specifies one of two possible orientations for the corner error and the minimal level-$s$ corner failure weight is $w_{\mathrm{c}}^{(s)} = \max_o(w_{\mathrm{c},o}^{(s)})$.  

Let us now show that the minimal level-$s$ link and corner failure weights satisfy
    \be\label{minimal_weight} 
    w_{\mathrm{link}}^{(s)} \leq 2^{ \lfloor (s+1)/2 \rfloor}, \hspace{0.55cm}
	w_{\mathrm{c}}^{(s)} \leq 2^{\lfloor (s+2)/2 \rfloor }.
	\ee

We will show this by explicitly constructing operator histories that produce the RHS of the bounds in Eq.~\eqref{minimal_weight}.  For an upper bound, it is sufficient to consider an operator history with support within a \emph{single} $\EC_s$ of the layer. 
First consider $s=1$. By Remark~\ref{R0_search_remark}, both $w^{(1)}_{\mathrm{link}}$ and $w_{\mathrm{c}}^{(1)}$ must be at least 2. Furthermore, for any arrangement of two syndromes on the vertices of a 1-cell, we can construct a weight-2 noise realization that contains errors only at time step 1 which when acted on by $\EC_1 \circ \EC_1$ within the layer of gadgets, result in this syndrome configuration. To show $w^{(1)}_{\mathrm{c}} =2$, consider  a particular $\EC_1$ labeled by the coordinate of its bottom left corner $(x,y) = (1,1)$.  An example of a noise realization failing a $1$-corner with output syndromes at $(1,1)$ and $(4,4)$ occurs due to an error supported at time $t=1$ at locations $(1,2,1)$ and $(2,3,1)$(here $(x,y,o)$ indicates the link that connects the vertex at position $(x,y)$ to that at position $(x+1,y)$ if $o=1$ or at position $(x,y+1)$ if $o=2$). Interchanging the $x$ coordinates of these links leads to an output configuration with syndromes at $(4,1)$ and $(1,4)$. Therefore 
\be \label{init_conds} w_{\mathrm{link}}^{(1)} = w_{\mathrm{c}}^{(1)} = 2.\ee 
	
Next, we will show the following recursion relations: 
\be \label{recursive_relns} w_{\mathrm{link}}^{(s)} \leq w_{\mathrm{c}}^{(s-1)}, \hspace{0.55cm} 
w_{\mathrm{c}}^{(s)} \leq 2 w_{\mathrm{link}}^{(s-1)}. 
 \ee 
We show these by explicitly constructing errors of a given weight that are supported in the first half of one of the layers $\cL(\EC_{s-1})$  inside the noisy layer $\cL(\EC_s)$ gadgets (which is then followed by a noiseless layer of $\cL(\EC_s)$ gadgets), and which cause a level-$s$ link or corner failure. We will use this to upper bound $w_{\mathrm{link}}^{(s)}$ and $w_{\mathrm{c}}^{(s)}$ in terms of $w_{\mathrm{link}}^{(s-1)}$ and $w_{\mathrm{c}}^{(s-1)}$. 

For the bound on $w_{\mathrm{link}}^{(s)}$, we first show that it is possible to fail both vertically- and horizontally-oriented $s$-links by causing a certain level-$(s-1)$ corner failure in the very last layer of noisy $\EC_{s-1}$ gadgets; then, after reducing by $s-1$ levels, we obtain a 1-corner failure (see Eq.~\eqref{R0-action}) as input to the second, noiseless, (level-reduced) $\EC_1$ gadget. 
This shows $w_{\mathrm{link},y}^{(s)} \leq w_{\mathrm{c}}^{(s-1)}$. 
We note that the same argument does not go through for $w_{\mathrm{link},x}^{(s)}$ because $\EC_{s}$ is not invariant under $\frac{\pi}{2}$-rotation of the rules. 
Indeed, an explicit computer-aided search shows that {\it no} input state containing a single level-0 corner-shaped error can lead to a failed $x$-oriented 1-link when acted on by a clean $\EC_1$. However, a level-0 corner error occurring after the application of the first pair of $R_v$ layers but in the first half of operation of $\EC_1$ can lead to a coarse-grained level-1 $x$-link failure in the output (this follows by using the above discussion regarding $w_{\mathrm{link},y}^{(s)}$ as well as the fact that first pair of $R_v$ is followed by a pair of $R_h$ in $\EC_1$). This means that there exists a level-$(s-1)$ corner failure occurring in a layer of $\EC_{s-1}$ within $\EC_{s}$, such that by a level reduction argument the output is an $x$-type $s$-link failure. Thus $w_{\mathrm{link},x}^{(s)} \leq w_{\mathrm{c}}^{(s-1)}$ and consequently $w_{\mathrm{link}}^{(s)} \leq w_{\mathrm{c}}^{(s-1)}$.

We adopt a similar approach for the bound on $w_{\mathrm{c}}^{(s)}$ by considering an operator history supported entirely within the first layer $\cL(\EC_{s-1})$  in  $\cL(\EC_s)$. The discussion above Eq.~\eqref{init_conds}---in particular, the examples of errors which act at time $t=1$ and lead to level-1 corner failures under $R_0$---then shows that $w_{\mathrm{c}}^{(s)} \leq2 w_{\mathrm{link},x}^{(s-1)} \leq 2w_{\mathrm{link}}^{(s-1)}$.  As this is true for both orientations of corner errors, this implies Eq.~\ref{recursive_relns}, which then yields the desired Eq.~\ref{minimal_weight}.

We note that the same minimum weight bounds work for $FT^{s-1}(\EC)$\footnote{Note the absence of the tilde: this is the conventional fault tolerant simulation, which is the same as an outer simulation assuming inner simulation with $k =1$.} as for the $\EC_s$ gadgets.

\subsubsection{Numerics and memory lifetime}\label{ss:tc_numerics}

In this subsection we will numerically study the memory lifetime of the toric code automaton subjected to several simple examples of $p$-bounded noise models. Unfortunately, we are unable to definitively extract a threshold noise strength $p_c$ in any of these noise models. This is due to the fact that our automaton is naturally defined only on systems whose linear dimension is a power of 3, and that increasing this power does not always produce a suppression of the logical failure rate (although the suppression appears to be much more rapid than the bound derived in our fault tolerance result). For simple i.i.d Pauli noise, the results below suggest that $p_c$ could be as large as $\sim 10^{-4}$, but more detailed numerical work (as well as optimization of the gadgets) will be required to produce a computable estimate. 

We will simplify the problem slightly by performing $s$ iterations of the outer simulation (i.e. setting the level of inner simulation $k = 1$) and studying how the error suppression depends on $s$ (which establishes strictly stronger properties than proved analytically in the paper). 
For numerical convenience, we will additionally replace $\EC$ gadgets with undoubled gadgets $\EC^{(1/2)}$ (we will drop the superscript henceforth), which reduces the overall circuit depth but does not seem to affect numerical performance. 
Based on the results in the previous subsection, we also provide a conjecture stating that the memory lifetime diverges with $p\rightarrow 0$ much faster than the scaling guaranteed by our main fault tolerance result in Eq.~\eqref{main_logical_error_bound}. 

In analogy with our discussion for Tsirelson's automaton in Sec.~\ref{ss:tsirelson_numerics}, we define the relaxation time for the toric code automaton as follows: 
\begin{definition}[Relaxation time]
	Let $|\psi\rangle$ be a logical state of the toric code on a $3^{s} \times 3^{s}$ torus. Let $\rho_{t,\bm H_t}$ be the state obtained by acting on $|\psi\rangle\langle \psi|$ repeatedly with $FT^{s-1}(\EC)$ for $t$ time steps and noise history $\bm H_t$, followed by a noise-free application of $FT^{s-1}(\EC)$. The relaxation time is defined as: 
	\be t_{\rm rel}^{(s)} = \mathbb{E}_{\bm H_t}[\min_t( t:  \rho_{t,\bm H_t} \neq |\psi\rangle\langle\psi|) ].\ee 
\end{definition}
The definition of $t_{\rm rel}^{(s)}$ is phrased in this way (instead of being phrased in terms of $\|\rho_t - |\psi\rangle\langle \psi |\|_1$) because it provides a simpler metric for numerics. We will be studying the $X$ logical sector only. 

We now define the three error models that will be considered in this subsection. 
\begin{itemize}
\item i.i.d. qubit error model: defined by applying i.i.d. Pauli $X$ (and $Z$) operators with probability $p$ on qubits after each time step.
\item i.i.d. measurement error model: defined by flipping the syndrome measurement outcome at a given time step before applying feedback according to an i.i.d. distrbution with probability $p$.  No other errors occur.
\item i.i.d. gadget error model: defined such that each elementary gadget fails arbitrarily based on an i.i.d. distribution with probability $p$.
\end{itemize}

\begin{figure}[!htbp]
	\centering
	\begin{subfigure}[b]{0.49\textwidth}
		\centering
		\includegraphics[width=\textwidth]{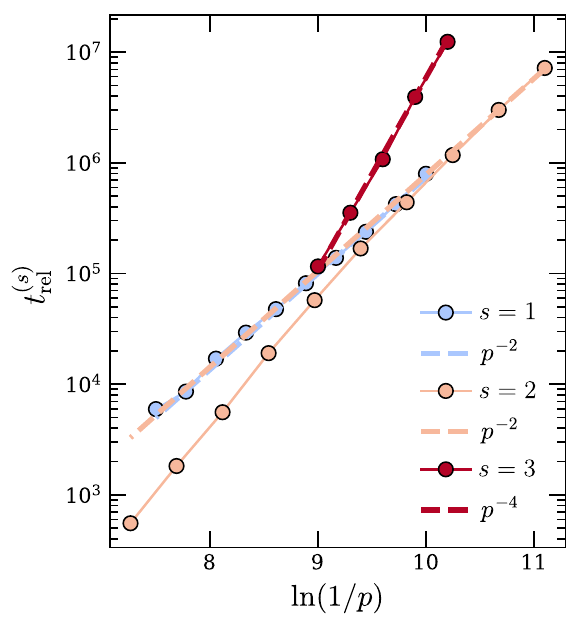} 
	\end{subfigure}
	\hfill
	\begin{subfigure}[b]{0.49\textwidth}
		\centering
		\includegraphics[width=\textwidth]{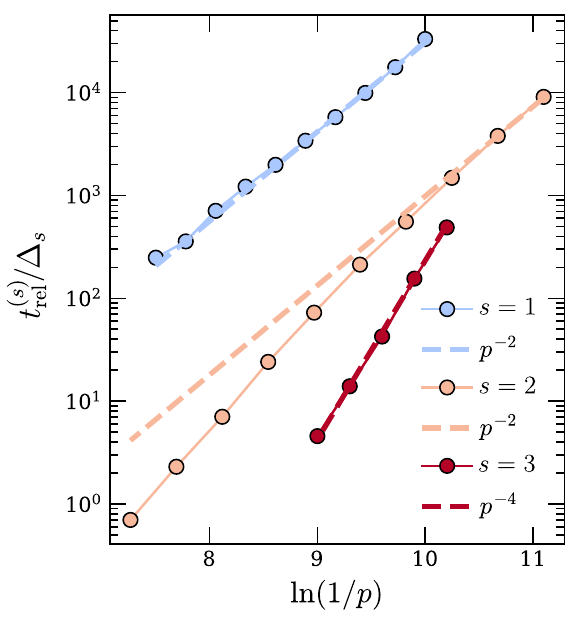} 
	\end{subfigure}
	\caption{ Memory times under repeated application of $FT^{s-1}(\EC)$ gadgets for $s = 1,2,3$ under the i.i.d. qubit noise model with noise strength $p$, showing both the relaxation time $t_{\rm rel}^{(s)}$ (left panel) and the relaxation time normalized by circuit depth of $FT^{s-1}(\EC)$, $t_{\rm rel}^{(s)} / \Delta_s$ (right panel). The dashed lines are fits to power law scalings $t_{\rm rel}^{(s)} \sim p^{-2^{\lfloor (s+1)/2\rfloor}}$ suggested by Eq.~\eqref{conjectured_scaling}. For $s=2$ this fit matches the scaling of $t_{\rm rel}^{(2)}$ only at the smallest values of $p$; at larger values of $p$ the data is better described by a steeper $t_{\rm rel}^{(2)}\sim p^{-3}$ scaling. }\label{fig:tc_monte_carlo}
\end{figure}

In this subsection, we assume that for the gadget error model, elementary gadgets have the supports defined in Def.~\ref{eq:supp-TC}. 

To get an expectation for the scaling of $t_{\rm rel}^{(s)}$ with $p$ in these error models, we first perform a similar analysis to the previous subsection to determine minimal weight failures in the i.i.d. measurement and gadget error models. 
Define a generalized notion of a weight $w$ for an operator history to be the total number of failures (measurement or gadget) involved. 
Using this definition, the link and corner failures $w_{\mathrm{link}}^{(s)}$ and $w_{\mathrm{c}}^{(s)}$ at $s \geq 1$ are otherwise defined in the same way as in the previous subsection. Then, for $s \geq 2$, we can use the same logic as in the previous subsection to show the same bounds 
$ w_{\mathrm{link}}^{(s)} \leq w_{\mathrm{c}}^{(s-1)}, \, w_{\mathrm{c}}^{(s)} \leq 2 w_{\mathrm{link}}^{(s-1)}$.

To produce a more explicit bound, we need to know the values of the minimal failure weights at $s=1$. 
Consider first the measurement error model. It is not hard to argue analytically that $w_{\mathrm{link}}^{(1)} =1$ and $w_{\mathrm{c}}^{(1)} =2$.\footnote{A single measurement error on a clean input can only fail a single $\cT_0^{h/v}$ gadget in a way which creates two anyons along a single row or column of the lattice. This can lead to a link failure, but it is straightforward to see that this can never produce a corner failure. One may also verify that two such $\cT_0^{h/v}$ failures can produce a corner failure (e.g. failures which occur in the second step of the first $R^v_0$ and the fifth step of the first $R^h_0$ in $\EC_1$).} For the gadget error model, failing an $\cM_0$ can fail both a level-1 corner or link, so $w_{\mathrm{link}}^{(1)} = w_{\mathrm{c}}^{(1)} =1$.
Solving the recursion relations generated by 
Eq.~\ref{recursive_relns} with these initial conditions then gives 
	\be \label{general_minimal_weight} 
    w_{\mathrm{link}}^{(s)} \leq 2^{\lfloor (s-s_0)/2 \rfloor}, \hspace{0.55cm} 
	w_{\mathrm{c}}^{(s)} \leq 2^{\lfloor (s+1-s_0)/2 \rfloor}, \ee 
where the offset $s=-1$ for the qubit noise model, $s_0 = 0$ for the measurement error model, and $s_0 = 1$ for the gadget error model. 

\begin{figure}[!htbp]
	\centering
	\begin{subfigure}[b]{0.49\textwidth}
		\centering
		\includegraphics[width=\textwidth]{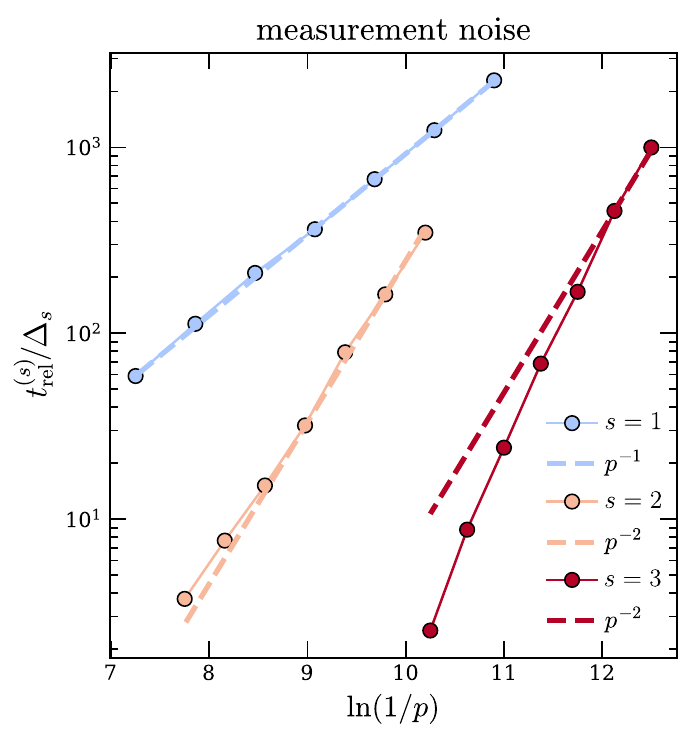} 
	\end{subfigure}
	\hfill
	\begin{subfigure}[b]{0.49\textwidth}
		\centering
		\includegraphics[width=\textwidth]{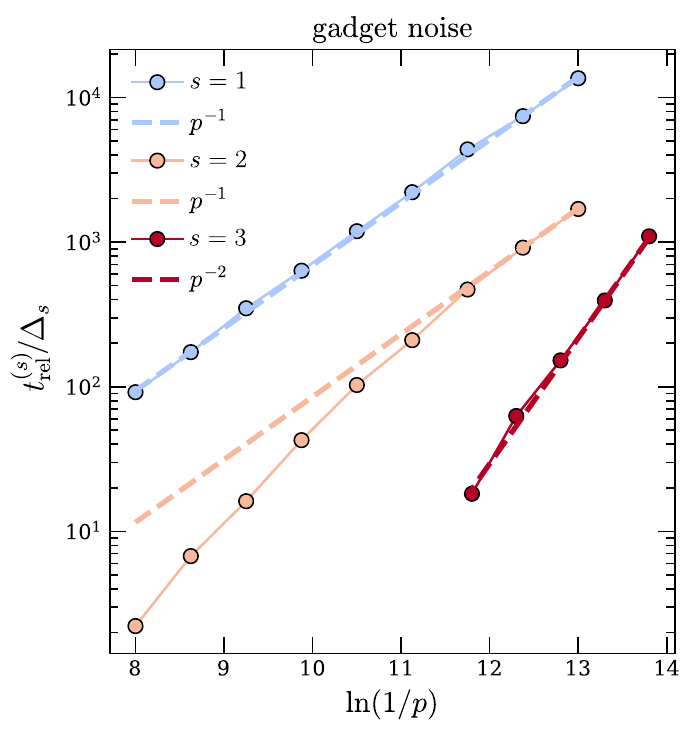} 
	\end{subfigure}
	\caption{Relaxation times in the toric code automaton under strength-$p$ i.i.d. measurement noise (left) and i.i.d gadget noise (right). $t_{\rm rel}^{(s)}$ is shown normalized by $\Delta_s$, the depth of $FT^{s-1}(\EC)$. The dashed lines are fits to power-law scalings suggested by Eq.~\eqref{conjectured_scaling}, with offsets $s_0=0$ for measurement noise and $s_0=1$ for gadget noise. }\label{fig:tc_monte_carlo_other_noise}
\end{figure}

%While we are presently unable to analytically determine how tight these bounds are in general, we can argue that they are tight at least for $s\leq 3$. 
For the gadget error model, these bounds are tight for $s \leq 3$: this follows from $w_{\mathrm{link}}^{(3)}, w_{\mathrm{c}}^{(3)}\geq 2$ by the nilpotence of $FT^{2}(\EC)$ (see Lemma~\ref{lemma:nilpotentnumerics-TC}). 
% Together with the values derived above for $w_{\mathrm{link}}^{(1)}, w_{\mathrm{c}}^{(1)}$, this automatically implies that the bounds in \eqref{general_minimal_weight} are tight for $s \leq 3$ in the gadget error model. 
Since the qubit and measurement error models can be obtained as special cases of the gadget error model, the minimal failure weights in these models must be at least as large as those in the gadget error model; together with the above results regarding  $w_{\mathrm{link/c}}^{(1)}$ in these models, this shows that the bounds are saturated at $s\leq 3$ for all cases except possibly for $w_{\mathrm{link/c}}^{(3)}$ and $w_{\mathrm{c}}^{(2)}$ in the qubit error model, and $w_{\mathrm{link}}^{(2)}, w_{\mathrm{c}}^{(3)}$ in the measurement error model. We have verified that the bounds on the above weights with $s=2$ are saturated by exhaustive numerical search. We were unable to perform a similar search for the above weights with $s=3$, but a Monte Carlo sampling of error configurations suggests that the bounds indeed remain tight. 

%That $w_{\mathrm{link}}^{(2)}=2$ in the measurement error model can be shown by an appropriate special case of the algorithm used for showing nilpotence of the automaton gadgets (Lemma~\ref{lemma:nilpotentnumerics-TC}).\footnote{As mentioned in a previous footnote, with only a single measurement failure and a clean input, the only level-0 gadget errors that can occur are particular failures of $\cT_0^{h/v}$; the code used in the analysis of App.~\ref{app:TC-nilpotence} shows that these errors are cleaned up in a single level.}

% An exhaustive numerical search over error configurations  
% for $s \leq 3$~\cite{code} reveals that the upper bound in Eq.~\ref{general_minimal_weight} is in fact saturated for qubit, measurement and gadget error models\footnote{Due to the large size of the search space, we had to resort to Monte Carlo numerics for $s = 3$.}. 

At small error rates, we expect the scaling of the relaxation time with $p$ to agree with the expression that one obtains by assuming that the relaxation time is dominated by the probability of the smallest-weight configurations, namely $t_{\rm rel}^{(s)} \sim p^{-w^{(s)}_{\rm link}}$. This is a reasonable assumption in the limit of large system sizes when $p/p_c\ll1$ because in that case, minimal-weight events should dominate the probability of logical failure. With this expectation and our (admittedly limited) evidence about the tightness of the bounds on minimal-weight events, we propose the following conjecture:
\begin{conjecture}\label{conj:scaling_guess}
    The following statements hold:
    \begin{enumerate} 
        \item The bound in Eq.~\eqref{general_minimal_weight} is tight for qubit, measurement and gadget error models with model-specific values of $s_0$.
        \item  For a $p$-bounded error model, there exist model-specific positive constants $C$, $\gamma$, $s_0$ such that the relaxation time $t_{\rm rel}^{(s)}$ below threshold in the $p\rightarrow 0$ limit is determined by the minimum-weight logical failure event: 
	\be \label{conjectured_scaling-0} 
    t_{\rm rel}^{(s)} \sim (Cp^\gamma)^{-w^{(s)}_{\rm link}} \sim (Cp^\gamma)^{-2^{\max(0,\lfloor (s-s_0)/2 \rfloor)}}.
    \ee 
    \item For i.i.d. qubit, measurement and gadget error models with error rate $p$, 
    \be \label{conjectured_scaling} 
    t_{\rm rel}^{(s)} \sim p^{-w^{(s)}_{\rm link}} \sim p^{-2^{\max(0,\lfloor (s-s_0)/2 \rfloor)}}
    \ee 
    where $s_0=-1,0,1$ for qubit, measurement, and gadget error models, respectively. 
    \end{enumerate}
\end{conjecture}
The factor 1/2 in the exponent of $\lfloor (s-s_0)/2 \rfloor$ can be understood as having to go two levels up each time in order to achieve logical error probability suppression (as opposed to the $34$ levels used in our main fault tolerance result). 
We now use Monte Carlo simulations to investigate the extent to which this conjecture is true at small system sizes $s \leq 3$. 
% Monte Carlo numerics on the toric code automaton for $s \leq 3$ \cite{code} was performed for each of the three i.i.d. noise models defined above. 
In these numerics, $t_{\rm rel}^{(s)}$ is estimated by sampling random noise realizations for a particular error model and running the automaton until a logical failure (as determined by the result of performing an ideal decoding) occurs. %, in much the same way as the analysis of Tsirelson's automaton in Sec.~\ref{ss:tsirelson_numerics}. 
The results of doing this for the simplest i.i.d. qubit error model are shown in Fig.~\ref{fig:tc_monte_carlo}. At the smallest values of $p$, a good fit to the scaling predicted by Conjecture~\ref{conj:scaling_guess} with $s_0=-1$ is observed. For $s=2$ the relaxation time at larger values of $p$ instead scales as $t_{\rm rel}^{(2)} \sim p^{-3}$---faster than the $\sim p^{-2}$ scaling predicted from Conjecture~\ref{conj:scaling_guess}---with the expected $\sim p^{-2}$ scaling appearing to be recovered at smaller values of $p$. This crossover is likely attributable to an entropic effect: the number of ways to fail $FT(\EC)$ using three errors may be much larger than the number of ways to fail using two errors, so Monte Carlo more frequently probes these configurations (such that an expansion of $t_{\rm rem}^{(2)}$ in $1/p$ reads $t_{\rm rel}^{(2)} \sim Ap^{-2} + Bp^{-3} + \cdots$ with $B \gg A$). 

Fig.~\ref{fig:tc_monte_carlo_other_noise} shows relaxation times for the i.i.d. measurement and gadget error models. In both cases, Conjecture~\ref{conj:scaling_guess} yields good fits to the estimated values of $t_{\rm rel}^{(s)}$ at the smallest values of $p$. Similar to the case of $t_{\rm rel}^{(2)}$ in the i.i.d. qubit error model, crossovers between a faster-diverging scaling at larger $p$ to a slower-diverging scaling form at smaller $p$ (which agrees with Conjecture~\ref{conj:scaling_guess}) are observed for $s=2$ in the gadget-based noise model (where $s_0=1$), and $s=3$ in the measurement error model (where $s_0=0$). %\MD{If you have a strong preference to keep the text below, we should discuss it. I am struggling with some of the logic. Also SB doesn't seem to like the rest of this paragraph as well. }From these (rather limited) numerical results, we conjecture that an extended crossover behavior like this is present for odd $l$ when $o$ is even, and even $l$ when $o$ is odd. In detail: if the scaling in Conjecture~\ref{conj:scaling_guess} holds, the power $\alpha$ appearing in $t_{\rm rel}^{(s)} \sim p^{-\alpha}$ varies with $l$ according to the sequence $\alpha = 1,\dots,1,2,2,4,4,8,8,\dots$, where the number of initial $1$s in this sequence is determined by the offset $o$. Once the level is increased past the initial sequence of $1$s (which happens when nilpotence is achieved), the effective failure rate is squared every two levels. The just-mentioned conjecture states that for the levels $l$ at the ``ends'' of these steps (such that $t_{\rm rel}^{(s)}$ and $t_{\rm rel}^{(s-1)}$ the same asymptotic scaling), there is a crossover regime: at moderate values of $p$ the failure rate appears to decrease when going from $l-1$ to $l$, but at small values of $p$ a crossover to the same scaling as $t^{(s-1)}_{\rm rel}$ is obtained. Evaluating the status of this and the previous conjectures will require a more in-depth numerical study beyond the scope of the present work. 

%% file: 6.3D.tex
\section{Constructing a 3D time translation-invariant memory} \label{sec:3D}

In the previous section, we proved that the toric code automaton is fault tolerant.  Specifically, a fault tolerant encoding $\widetilde{FT}^{n }(\cC)$ with $p$-bounded noise simulates $\cC$ experiencing $A p^{2^{n}}$-bounded noise (for inner simulation level $k\leq 34$).  However, the toric code automaton is hardwired, and thus, the local rules are not specified in a time translation-invariant way.  An unfortunate consequence of this is that in order to implement the automaton, a reliable read-only ``program'' storing $O(\log L)$ bits is needed per each qubit to encode the instruction set for the level-0 gadgets that are applied at different points in time.  This makes the decoder not intrinsically local since one needs to reliably and fault tolerantly store this read-only memory as well, which is not possible if the size of the memory per site diverges in the thermodynamic limit.

Fortuitously, we may get around this issue simply by going one spatial dimension up. This gives us a simple 3D construction that inherits the properties of the 2D automaton but is fault tolerant without the need for a logarithmically-sized read-only memory per qubit.  This is a fully fault tolerant and local scheme in discrete time which on an $L\times L \times  L'$ lattice on a three-dimensional torus will encode $2L'$ logical qubits.  Thus, due to the large encoding rate, it is capable of performing quantum computation as well, which we will not discuss in this paper.  The scheme is still not translation-invariant in space, but this does not contradict locality (in that no read-only memories of divergent size are required to implement the circuit).

First, we will define the notion of a layering (foliation) for the 2D toric code automaton.

\begin{definition}[Layering of the toric code automaton] \label{def:layering}
Consider a fault-tolerant circuit $\widetilde{FT}^{n}(\cI_0)$ corresponding to one period of the toric code automaton operation.  Choosing, for simplicity, $L = 3^{nk}$, call $\Delta = \text{poly}(L)$ the depth of this circuit. 

Consider a stack of $2\Delta$ two-dimensional layers in three dimensions with periodic boundary conditions, forming a lattice of dimensions $(L)_x \times (L)_y \times (2 \Delta)_z$.  A layering of the toric code automaton is an assigned spatial pattern of elementary gadgets to the $2\Delta$ layers. Layers with coordinate $z=2t$ are associated with elementary gates in the circuit $\widetilde{FT}^{n}(\cI_0)$ at timestep $t$, and layers with coordinate $z=2t+1$ are associated with a layer of $\cI_0$ gates.
\end{definition}

With the description of the layering above, we can define the 3D automaton for the toric code:

\begin{definition}[3D automaton for layers of the 2D toric code] \label{def:3D-aut-definition}
Consider a layering of the 2D toric code automaton with $2\Delta$ layers.  Initialize each of the $2\Delta$ layers to each individually host a logical state of the 2D toric code. The 3D automaton for the toric code applies elementary gadgets according to the following prescription:
\begin{itemize}
\item (t = 4s): Apply the gadgets in each layer determined by the layering, Def.~\ref{def:layering};
\item (t = 4s+1): Apply a transversal swap gate between the layers: $\text{SWAP}((\bm{x}, 2t), (\bm{x}, 2t+1))$ where the notation $(\bm x, t)$ specifies the spatial location $\bm x$ and the layer $t$. The two arguments of the swap gate label the spatial locations of the two qubits which are swapped. This step is labeled `SWAPs (even)' in Fig.~\ref{fig:3dfig} because this occurs for even $(t-1)/2$.
\item (t = 4s+2): Another application of the gadgets in each layer determined by the layering, Def.~\ref{def:layering};
\item (t = 4s+3): Apply a transversal swap gate $\text{SWAP}((\bm{x}, 2t+1), (\bm{x}, 2t+2))$. This step is labeled `SWAPs (odd)' in Fig.~\ref{fig:3dfig} because this occurs for odd $(t-1)/2$.
\end{itemize}
This defines a time-periodic operation with period 4. 
\end{definition}
This also resembles the ``Toom layering construction'' introduced in~\cite{gacs1988simple}.

\begin{figure}[!htbp]
    \centering
    \includegraphics[width=0.6\linewidth]{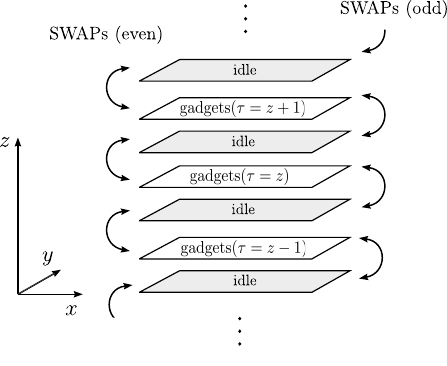}
    \caption{An illustration of the 3D automaton operating on layers of 2D toric code realizing a time translation-invariant memory of $4 \Delta$ logical qubits. This automaton is defined on a three-dimensional torus with dimensions $L \times L \times 2 \Delta$ and operates with a constant period of 4 time steps; $\Delta$ is defined to be the depth of the $\widetilde{FT}^n(\cI_0)$ circuit of the 2D automaton. 
    In each period of its operation, at even times $t$, the 3D automaton first applies elementary gates in individual layers: at spatial position $(x,y,z=2\tau)$ in even-$z$ layers, the rule applies the same elementary gadgets as the 2D automaton executes at location $(x,y)$ and time $\tau$; odd-$z$ layers simply idle.
    At even times $t$, the layers are transversely swapped. When $(t-1)/2$ is even, the layers with coordinates $2z$ with $2z+1$ are swapped (labeled by SWAPS(even)), and when $(t-1)/2$ is odd, the layers with coordinates $2z+1$ with $2z+2$ are swapped (labeled by SWAPS(odd)).  }
    \label{fig:3dfig}
\end{figure}

We now provide some intuition for why these steps are necessary. The 3D automaton works by foliating the toric code automaton in space and moving layers (each corresponding to a toric code) through the 3D space. 
The ``clock'' governing the operation of the 2D automaton is now laid out along the additional spatial dimension. 
Consider a ``moving frame'' that follows any given initial toric code state as it is acted upon by the automata and is swapped from layer to layer. The toric code state will then experience the circuit corresponding to the toric code automaton, interspersed with idle and the timesteps when the entire layer is being transversally swapped.

Let us now explain why $2 \Delta$ (rather than $\Delta$) layers are needed. The operation of simultaneously moving all of the toric code layers by one step cannot be implemented locally (more formally, one would say that such a one-dimension global shift constitutes a nontrivial one-dimensional QCA which cannot be implemented with a constant-depth circuit), and the only way to do so is to add auxiliary layers in between.  With an auxiliary set of layers, the shift of all toric code layers is done in two steps: first, we shift the toric codes from the gadget layer to the auxiliary layer, and second, we shift the toric codes from the auxiliary layer to the next gadget layer.  Half of the toric codes move through space in the one direction while the other half move in the opposite direction. This is fine because the toric code automaton rule is symmetric once laid out in an extra dimension. 
Note that we impose periodic boundaries in the third dimension so that moving a toric code around the 3D substrate emulates the action of repeatedly applying the 2D automaton circuit.

The proof of fault tolerance of this 3D automaton (in the presence of additional errors that are possible on identity gates in auxiliary layers and in SWAP gates) follows analogously to the 2D case:

\begin{theorem} \label{thm:3D}
Consider the 3D automaton defined in Def.~\ref{def:3D-aut-definition}.  Assuming a perfect initialization of each of the $2\Delta$ layers to a logical state of the toric code, so that the layer with coordinate $z$ is in the state $\ketbra{\psi}{\psi}$.  We then run the 3D toric code automaton on this input.  We assume a $p$-bounded gadget noise model, where, in addition to the usual definition, each SWAP gate is also treated as an elementary gadget. Then, the state at time $t$ (which will be located at coordinate $(z + t) \ \mathrm{mod} \ 4 \Delta$, denoted $\rho_t$, fulfills
    \begin{equation}
       \| \rho_{t} - \ketbra{\psi}{\psi}\|_1 \leq C \left \lceil \frac{t}{4\Delta} \right  \rceil \cdot (A p)^{2^{ (\log_{3} L)/k }}
    \end{equation}
where $k = 34$, and $A, C$ are some constants.
\end{theorem}
\begin{proof}
The proof follows from tracking a given copy of the toric code (``going to a moving frame'') as its state is moved along the $z$ direction by swap gates and acted upon by the gadgets of the automaton and idle steps in between swaps. The swap steps viewed in the ``moving frame'' become idle steps.  The noise supported on swap gates, upon measuring the stabilizer in the layer we are tracking and the neighboring layers becomes projected onto correlated Pauli noise on qubits in the support of the faulty swap gate.  For the purposes of analyzing the state in a given layer, the failures of these gates can be simply viewed as $p$-bounded noise on individual qubits. 

First, consider the layer with $z = 0$. The action of the 3D automaton viewed in the moving frame associated with this layer becomes equivalent to repeated application of $\cL(\widetilde{FT}^n(\cI_0))$ with the addition of idle steps and the modification of the noise model. The proofs of Corollary~\ref{cor:logical_error_bound} and the properties leading to it apply to this case without major changes.

For the layers initialized at $z \neq 0$, the action of the 3D automaton viewed in the moving frame is similarly equivalent to repeated application of $\cL(\widetilde{FT}^n(\cI_0))$ (similarly with addition of idle steps and the modification of the noise model), except the circuit of  $\cL(\widetilde{FT}^n(\cI_0))$ starts with its action at time $\lceil z/2 \rceil$, possibly starting with between one and three times steps of idles and swaps. We denote the circuit that produces the state at time $t$ as  $\cL(\widetilde{FT}^n(\cI_0))[z,(z+t)\  \mathrm{mod} \ 4 \Delta]$ (where the arguments correspond to the timesteps when the modified circuit stops and ends. In this case, we follow the proofs of Corollary~\ref{corollary:full-simulaton-TC} and \ref{cor:logical_error_bound}, noting that due to noiseless initialization, we can formally write $\ketbra{\psi}{\psi} = \cL(\widetilde{FT}^n(\cI_0))[0,z] \circ \ketbra{\psi}{\psi}$. This completes $\cL(\widetilde{FT}^n(\cI_0))[z,(t+z)\  \mathrm{mod} \ 4 \Delta]$ to $\cL(\widetilde{FT}^n(\cI_0))[0,(t+z)\  \mathrm{mod} \ 4 \Delta]$, where the first $z$ steps are noiseless. The rest of the arguments in the proofs of Corollary~\ref{corollary:full-simulaton-TC} and \ref{cor:logical_error_bound} apply.
\end{proof}

We assumed noiseless initialization in the theorem above. In the case of noisy initialization, one would similarly start by preparing an appropriate product state on all individual qubits and then running the automaton for the duration of time $t \geq 8 \Delta$ which is needed for each layer to certainly undergo a complete circuit of $\cL(\widetilde{FT}^n(\cI_0))$ at least once. Then, similar considerations to those in the proof of Prop.~\ref{prop:noisy-initialization} can be applied to demonstrate reliable initialization.  

While the 3D automaton is time-independent, it heavily relies on the assumption of operating in discrete time.  In particular, if we assume that each operation is applied according to some continuous-time Poisson process,  
the automaton might no longer have a threshold.  As we discussed in the outlook, this leaves open the important question of whether a local {\it continuous-time} quantum memory exists in a physical number of spatial dimensions.

Finally, let us briefly discuss the PCA model, in which the noise model takes into account gadgets skipping measurements that they should apply at each time step. Similarly to Sec.~\ref{sec:5-PCA}, we will outline an argument but leave a rigorous proof to the future work. In fact, similar arguments as in Sec.~\ref{sec:5-PCA} suffice to show that this effect can be absorbed by a $p$-bounded gadget error model with some effective $p$. However, the supports of these gadgets will more naturally be three-dimensional. Nevertheless, the correlations between faults in different layers does not affect the proof of Thm.~\ref{thm:3D} as a $p$-bounded gadget error model with three-dimensional support is equivalent to a $C p^\gamma$-bounded gadget error model assuming a two dimensional support of gadgets for some positive constants $C,\gamma$.

%% file: AppendixA0.tex
\section{Noise models and $X/Z$ decoupling} \label{app:noisedetails}

We assume that the reader is familiar with the basics of error correction with the toric code (see, e.g., Ref.~\cite{Dennis_2002}), and make only the minimum amount of definitions needed to establish notation. 

\subsection{Quantum automata and noise model}\label{sss:quantum_noise}

Consider a lattice with a qubit on every edge.  The local Hilbert space describing the quantum states of these qubits will be denoted $\cH$.  We will refer to the collection of all edges of the lattice as $\cL_1$. We will be considering stabilizer codes in what follows, which we will refer to using the notation $\mathcal{S} = \langle g_i \rangle$, with $g_i$ denoting the generators of the stabilizer group (which we will sometimes refer to as ``checks'' of the code). We will only be interested in the situation where $\cS$ is a topological stabilizer code, in which case there exists a generating set such that each generator is associated with some location in space (be it a vertex, edge, or a plaquette of the lattice). The set of all such locations will be denoted by $\cL_0$. 

When we use the term ``quantum noise'', we simply mean a quantum channel:

\begin{definition}
Quantum noise $\mathcal{E}_t: L(\cH) \to L(\cH)$ acting at time $t$ is a quantum operation
\begin{equation} \label{eq:channel}
\mathcal{E}_t(\rho) = \sum_{\mu} \cK_{t,\mu} \rho \cK_{t,\mu}^{\dagger}
\end{equation}
where $\sum_{\mu} \cK_{t,\mu}^{\dagger}\cK_{t,\mu} = \mathds{1}$ and $\cK_{t,\mu}$ can depend arbitrarily on $\bigcup_\nu \{\cK_{1,\nu}, \dots, \cK_{t-1,\nu}\}$.
\end{definition}

\begin{definition}[Error syndromes]
A syndrome $\sigma_{v} \in \mathbb{F}_2$ denotes the value of a measurement of the stabilizer generator (check) at location $v \in \cL_0$: $\sigma_{v} = 0$ corresponds to a measurement outcome of $1$ and $\sigma_i = 1$ corresponds to a measurement outcome of $-1$ (``a violated stabilizer'').  We say that there is a non-trivial syndrome at $v$ if $\sigma_v = 1$.  

For a stabilizer state $\ket{\psi}$, the syndromes of the state ${\rm syn}(|\psi\rangle)$ is a vector valued in $\mathbb{F}_2^n$ where $n$ is the number of stabilizer generators (checks), whose $i$th component is $1$ if $g_i \ket{\psi} = - \ket{\psi}$ and $0$ if $g_i \ket{\psi} =  \ket{\psi}$. 

\end{definition}

\begin{definition}[Syndrome of an operator]
Given a stabilizer code $\mathcal{S} = \langle g_i \rangle$, the syndrome of a Pauli operator $\mathrm{syn}(\mathcal{O})$ is a vector valued in $\mathbb{F}_2^{n}$ where $n$ is the number of checks, whose the $i$-th component is $1$ if $[\mathcal{O}, g_i ] = -1$ and $0$ otherwise. Here, $[A,B] = A^{-1}B^{-1}AB$ denotes the group commutator. 

A non-Pauli operator is associated with a definite syndrome $\sigma$ if it is a linear combination of Pauli operators  $\cO_i$ with $\mathrm{syn}(\cO_i) = \sigma$.
\end{definition}

For a CSS code, it is possible to group syndromes into $\mathrm{syn}_X$ and $\mathrm{syn}_Z$ corresponding to the outcomes of measurement of $Z$- and $X$-type Pauli stabilizers, respectively. A discussion of how one can formally decouple the errors into $X$ and $Z$ types when analyzing the automaton can be found in Subsec.~\ref{subsec:TC}.

Our construction of the toric code automaton utilizes the approach of (local) measurements and (local) feedback: the automaton measures the local stabilizer generators of the toric code (which does not modify the logical state), and applies operations conditioned on nearby measurement outcomes. We call this the ``local measurement and feedback model'':

\begin{definition}[Measurement and feedback model]

A measurement and feedback model is a quantum dynamical system that takes a stabilizer group $\mathcal{S}$ as input and operates on a lattice of qubits in the following way. At each time step $t$, it: 
\begin{itemize}
\item [(a)] Measures all of the stabilizer generators  in $\cL_0$, obtaining syndromes $\sigma$,
\item [(b)] Applies the feedback operator $\cF_t(\sigma)$ in the Pauli group which is conditioned on the syndrome measurements, 
\item [(c)] Discards the syndrome measurement outcomes, and
\item [(d)] Applies a noise channel $\mathcal{E}_t$ with Kraus operators $\cK_{t,\mu}$.
\end{itemize}
As a quantum operation, one step of the measurement and feedback model, which we denote by $\cM_t$, acts on a state $\rho$ via 
\begin{equation} \label{eq:meas-feed}
\cM_t(\rho) = \sum_{\sigma, \mu} \cK_{t,\mu} \cF_t(\sigma)\Pi(\sigma) \rho \Pi(\sigma) \cF_t(\sigma ) \cK_{t,\mu}^{\dagger}.
\end{equation}
where $\Pi(\sigma)$ is a projection onto the subspace of states with syndrome $\sigma$.
\end{definition}

We emphasize that according to (c), measured syndromes do not have to be stored. 

A \emph{local} measurement and feedback model is defined completely analogously, except that the feedback operator at a given time $\cF_t(\sigma)$ factors into a product of local feedback operators where each operator depends only on syndromes in a local region. 

Under the premise of the measurement and feedback model, since all the stabilizer generators of the code are measured at each time step, the following simplification of the noise model can be made:

\begin{lemma}[Collapsed noise model] \label{lemma:noise-collapse}
Assume that we start in a syndrome-free logical state $\ketbra{\psi}{\psi}$ (i.e. one with no non-trivial syndromes) of a quantum stabilizer code, run a measurement and feedback automaton up to time $t-1$ and apply the measurement step at time $t$, obtaining the state $\rho_{t-1}$.  Suppose that $\rho_{t-1}$ can be written as
\begin{align} \label{eq:form-of-rho}
    \rho_{t-1} &= \sum_{{\bm H}_{t-1}} p({\bm H}_{t-1}) \left(\prod_{\tau=1}^{t-1} \cO_{t-\tau}\cF_{t-\tau}(\sigma_{t-\tau})\right) \ketbra{\psi}{\psi} \left(\prod_{\tau=1}^{t-1} \cO_{t-\tau}\cF_{t-\tau}(\sigma_{t-\tau})\right)^{\dagger} \nonumber 
    \\
    &\equiv \sum_{{\bm H}_{t-1}} p({\bm H}_{t-1})  \rho_{t-1,{\bm H}_{t-1}}
\end{align}
where  $p(\cdot)$ is some probability distribution related to the noise model and ${\bm H}_{t-1}$ is the noise history, which is a set of pairs $\{ ({\cO}_{\tau}, E_{\tau}) \}_{\tau = 1}^{t-1}$, with $\cO_\tau$ an operator having definite syndrome which is applied by noise at time $\tau$ and $E_{\tau}$ a set of qubit locations at time $\tau$ where noise can act. We will also call $H_\tau = ({\cO}_{\tau}, E_{\tau})$. 
The operators $\{\cF_\tau\}_{\tau = 1}^t$ are feedback operators whose arguments depend on syndromes of operators $\vec \cO = \{ \cO_1, \dots, \cO_\tau\}$, and $\sigma_t$ is the syndrome of the state $\rho_{t-1, \bm{H}_{t-1}}$. Together, they satisfy the consistency condition
% \begin{align}
%     \sigma_{\tau} = \mathrm{syn} \left (  \left(\prod_{\tau'=1}^{\tau-1} \cO_{t-\tau'}\cF_{t-\tau'}(\sigma_{t-\tau'})\right) \ketbra{\psi}{\psi} \left(\prod_{\tau'=1}^{\tau-1} \cO_{t-\tau'}\cF_{t-\tau'}(\sigma_{t-\tau'})\right)^{\dagger} \right ) 
% \end{align} 
\begin{equation}
   \sigma_{\tau} = \mathrm{syn} \left (  \cO_{\tau-1}\cF_{\tau-1}(\sigma_{\tau-1})\right) \oplus \sigma_{\tau-1}
\end{equation}
for $1 \leq \tau \leq t-1$, where $\sigma_0=0$.  Thus, $\rho_{t-1,{\bm H}_{t-1}} = \ketbra{\phi}{\phi}$ with $\ket{\phi}$ having syndrome $\sigma_t$. Because $\rho_{t-1, \bm{H}_{t-1}}$ depends on $\bm{H}_{t-1}$, $\sigma_t$ depends on it as well, which we keep implicit.

We then apply the feedback step and noise step of $\cM_t$ using noise channel $\cE_t$. After the measurements at time $t+1$ have been performed, there exists a probability distribution $p(H_t|{\bm H}_{t-1})$ such that the effect of the feedback and noise at time $t$ can be equivalently written as the following CPTP (completely positive and trace preserving) map:
\begin{align}  \label{eq:effective-automaton}
    \sum _{\sigma_{t+1}}\Pi(\sigma_{t+1} ) \cM_t(\rho_{t-1}) \Pi(\sigma_{t+1}) &=   \sum _{\bm H_t} p(H_t |{\bm H}_{t-1}) p({\bm H}_{t-1}) \cO_t    \cF_t(\sigma_{t}) \rho_{t-1,{\bm H}_{t-1}} \cF_t(\sigma_{t}) ^\dagger \cO_t^\dagger
\end{align}
where $\sigma_{t+1} = \mathrm{syn}( \cO_t    \cF_t(\sigma_{t}) \rho_{t-1,{\bm H}_{t-1}} \cF_t(\sigma_{t}) ^\dagger \cO_t^\dagger)$ are the syndromes measured at time $t+1$, and the operators $\cO_t$  can be written as
\begin{equation}\label{eq:decomnoise}
\cO_{t} = \alpha(t) P_t + \sum_i  \beta^{(i)}(t)  L_i P_t
\end{equation}
with normalization $|\alpha(t)|^2 + \sum_i |\beta^{(i)}(t)|^2 = 1$,  $P_t \in P_N$ ($P_N$ denoting the $N$-qubit Pauli group) satisfies $\mathrm{syn} (P_t) = \mathrm{syn} (\cO_t)$, and $L_i$ is in the set generated by logical $X$ and $Z$  operators of the stabilizer code (excluding the identity operator).  These operators define what we will henceforth call the collapsed noise model. 
\end{lemma}
\begin{proof}
Assuming the form of $\rho_{t-1}$ in the statement, to obtain $\rho_t$ we first apply the feedback operator $\rho_{t-1} \to \sum_{{\bm H}_{t-1}} p({\bm H}_{t-1}) \cF_t(\sigma_t) \rho_{t-1, {\bm H}_{t-1}} \cF_t^{\dagger}(\sigma_t)$ since the feedback operator only depends on syndromes and $\mathrm{syn}(\rho_{t-1, {\bm H}_{t-1}}) = \sigma_t$.  We then define $\widetilde{\rho}_{t-1, {\bm H}_{t-1}} = \cF_t(\sigma_t) \rho_{t-1, {\bm H}_{t-1}} \cF_t^{\dagger}(\sigma_t)$.  Following Eq.~\ref{eq:meas-feed}, we then apply noise $\cE_{t}$, which acts as
\begin{equation}
\cE_{t}\left( \sum_{{\bm H}_{t-1}} \widetilde{\rho}_{t-1, {\bm H}_{t-1}}\right) = \sum_{{\bm H}_{t-1}} \sum_{\mu} p({\bm H}_{t-1}) \cK_{t,\mu}({\bm H}_{t-1}) \widetilde{\rho}_{t-1, {\bm H}_{t-1}} \cK_{t,\mu}^{\dagger}({\bm H}_{t-1}),
\end{equation}
and thus after the subsequent round of stabilizer measurements at time step $t+1$, we can write 
\begin{align}
\sum_{\sigma_{t+1}}\Pi(\sigma_{t+1}) \cM_t(\rho_{t-1}) \Pi(\sigma_{t+1}) = \sum_{\substack{\sigma_{t+1}, \, \mu,\\ {\bm H}_{t-1}}} p({\bm H}_{t-1}) \Pi(\sigma_{t+1}) \cK_{t,\mu}({\bm H}_{t-1}) \widetilde{\rho}_{t-1, {\bm H}_{t-1}} \cK_{t,\mu}^{\dagger}({\bm H}_{t-1}) \Pi(\sigma_{t+1})
\end{align}
We will now argue that $\cK_{t,\mu}$ can be replaced by a simpler operator.  We are free to decompose $\cK_{t,\mu}$ in terms of a basis of operators
\begin{align}
    \cK_{t,\mu}({\bm H}_{t-1}) = \sum_{\varepsilon} c_{t,\mu}(\varepsilon | {\bm H}_{t-1}) \cO_{t,\mu}(\varepsilon)
\end{align}
where $\cO_{t,\mu}(\varepsilon)$ is a linear combination of Pauli operators whose syndrome is $\varepsilon$.  We note the following identity
\begin{align}
\sum_{\sigma_{t+1}}\Pi(\sigma_{t+1}) \cO_{t,\mu}(\varepsilon) \widetilde{\rho}_{t-1, {\bm H}_{t-1}} \cO_{t,\mu}(\varepsilon') \Pi(\sigma_{t+1}) = \delta_{\varepsilon, \varepsilon'} \cO_{t,\mu}(\varepsilon) \widetilde{\rho}_{t-1, {\bm H}_{t-1}} \cO_{t,\mu}(\varepsilon'),
\end{align}
and therefore, we may write
\begin{equation}
\sum_{\sigma_{t+1}}\Pi(\sigma_{t+1}) \cM_t(\rho_{t-1}) \Pi(\sigma_{t+1}) = \sum_{\mu, \varepsilon, {\bm H}_{t-1}} |c_{t,\mu}(\varepsilon| {\bm H}_{t-1})|^2 p({\bm H}_{t-1}) \cO_{t,\mu}(\varepsilon) \widetilde{\rho}_{t-1, {\bm H}_{t-1}} \cO_{t,\mu}^\dagger(\varepsilon).
\end{equation}
We then define\footnote{One can see that $p(H_t | {\bm H}_{t-1})$ is a probability distribution over $H_t$ by taking a trace of both sides of $\sum_\mu\cK_{t,\mu}^\dagger \cK_{t,\mu} = \mathds{1}$.} $p(H_t | {\bm H}_{t-1}) = |c_{t,\mu}(\varepsilon| {\bm H}_{t-1})|^2$.  Finally, we need to prove the decomposition of $\cO_{t,\mu}(\varepsilon)$.  Consider two operators $P_{t}(\varepsilon)$ and $P'_{t}(\varepsilon)$, which have the same syndrome $\varepsilon$ but have the property that $P_{t}(\varepsilon) P'_{t}(\varepsilon)$ is in the stabilizer group.  For any density matrix $\lambda$ with definite syndromes, we have $g_i \lambda g_i = \lambda$ for $g_i$ any element of the stabilizer group, and we can write $P'_t(\varepsilon) \lambda P'_t(\varepsilon) = P_t(\varepsilon) \lambda P_t(\varepsilon)$.
 Then, we can fix some $P_t(\varepsilon)\equiv P_t$, and any operator $\cO_{t,\mu}(\varepsilon) \equiv \cO_{t,\mu}$ can be written as a sum of terms of the form $L_i P_{t}$, where $L_i$ are generated by the logical operators of the code.   This gives the decomposition in Eq.~\ref{eq:decomnoise} for $\cO_{t,\mu}$.
\end{proof}
Note that the form of the density matrix in Eq.~\ref{eq:form-of-rho} follows from the result in Lemma~\ref{lemma:noise-collapse} if we always formally group the next round of measurements with the current step of the operation of the automaton. In that case, the Lemma above proves the inductive step for the claim that Eq.~\ref{eq:form-of-rho} holds at every point in time (the base case of $t = 0$ is trivial). 

The probability distribution $p({\bm H}_{t-1})$ is over the history of noise locations and operators up to time $t-1$.  Since the above statement was formulated for a general channel, $p({\bm H}_{t-1})$ can be arbitrary, but in this paper we will be only considering channels where $p$ has a certain $p$-bounded property, corresponding to the channel having decaying correlations in spacetime, which we will discuss shortly below.  Furthermore, we can write the effect of $t$ steps of the measurement and feedback channel followed by a measurement at time $t+1$ as $\cM_{p,t} \circ \cM_{p,t-1} \circ \cdots \circ \cM_{p,1} (\ketbra{\psi}{\psi})$ (where the subscript `$p$' stands for Pauli), where 

\begin{equation}
    \cM_{p, t}(\rho_{t-1}) = \sum _{{\bm H}_{t}} p({\bm H}_{t-1}) \cO_t    \cF_t(\sigma_{t}) \rho_{t-1,{\bm H}_{t-1}} \cF_t(\sigma_{t}) ^\dagger \cO_t^\dagger.
\end{equation}
Henceforth, we will be working with the effective channels $\{\cM_{p, t}\}$\footnote{This essentially can be understood a version of the path integral over trajectories corresponding to different noise operator histories.}.
\begin{definition}[Noise history/realization]
Consider $T$ rounds of measurement and feedback dynamics under noise model $\cE$, which, due to the discussion above, can be treated as a set of effective channels $\cM_{p,1}, \cM_{p,2}, \cdots, \cM_{p,T}$. Then, the noise history (which we also call noise realization) $\bm H_T$ is the trajectory specifying the set of spacetime qubit locations for the noise (also called set of noise locations) $E = \bigcup_{t= 1}^T E_t$  and the history of operators $\vec{\cO}_T= \{\cO_1, \cO_2, \cO_3,\cdots, \cO_T\}$ (with time steps indicated by the subscripts) applied due to the noise. Using Lemma \ref{lemma:noise-collapse}, we can write each $\cO_{t} = \alpha_t P_t + \sum_i  \beta_t^{(i)}  L_i P_t$.
\end{definition}

Based on the definition above, it follows that $\bigcup_{t = 1}^T \supp (\cO_t) \subseteq E$.\footnote{Note that $\bigcup_{t = 1}^T \supp (\cO_t) \subseteq E$ is not an equality, since noise need not apply an operator at every point in $E$. }  Moving forward, it is more convenient to deal with a purely Pauli error model.  For this, we will need an additional assumption, with which we can prove the Proposition below:

\begin{proposition}[Approximate noise model]\label{prop:apprnoise}
    Consider a noise history specified by $\bm H_T$ under the noise model $\cE$. By Lemma \ref{lemma:noise-collapse}, we write each $\cO_{t} = \alpha(t) P_t + \sum_i  \beta^{(i)}(t)  L_i P_t$ such that among all operators in the sum, $P_t$ is the Pauli operator with the smallest weight.  Writing
    \begin{equation}\label{eq:effective-channel}
        \mathcal{M}_{p, t}\circ \cdots \circ \mathcal{M}_{p,1}(\ketbra{\psi}{\psi}) = \sum_{\bm H_t} p(\bm H_t) \rho_{t, \bm H_t},
    \end{equation} 
    we introduce an approximate noise model that modifies the probability distribution in the effective operation of the measurement and feedback automaton in the following way:
    \begin{equation} \label{eq:approximate-channel}
       \mathcal{M}_{p,\mathrm{appr}, t}\circ \cdots \circ \mathcal{M}_{p,\mathrm{appr},1}(\ketbra{\psi}{\psi}) = \sum_{\bm H_t} q(\bm H_t) \rho_{t, \bm H_t},
    \end{equation}
    where $q$ is defined such that $q(\bm H_t) = 0$ for each $\bm H_t$ where at least one component of $\vec{\cO}_t= \{ \cO_1, \cO_2, \dots, \cO_t\}$ has $\beta^{(i)} \neq 0$ and $q(\bm H_t) = p(\bm H_t)/\eta$ otherwise, where $\eta$ is a global normalization factor.  Further assume that for a given noise history, the probability that $H_\tau = (\cO_\tau, E_\tau)$ has $\beta_\tau^{(i)} \neq 0$ for some $i$ is $\leq \exp(- K L)$ for some constant $K$.  We will say that this automaton experiences an approximate noise model.

    Then, if $\rho(t)$ is the state of the measurement and feedback automaton under the collapsed noise model at time $t$, and $\rho_{\text{appr}}(t)$ is the state of the automaton under the approximate error model at time $t$, the trace distance $T(\cdot)$ between these states can be bounded by
    \begin{equation}
        T(\rho_t, \rho_{\text{appr},t}) = \frac{1}{2}\|\rho_t -  \rho_{\text{appr},t}\|_1 \leq t \cdot \exp(-K L).
    \end{equation}
\end{proposition}
\begin{proof}
The normalization constant $\eta$ is the complement of the probability that at least one of the operators has non-zero $\beta_\tau^{(i)}$; by a union bound and the assumption of the Proposition, we have $\eta \geq 1 - t\cdot \exp(-K L)$. Finally, calling $\cS$ the set of noise histories $\bm H_t$ where there exists at least one $\cO_\tau$ with corresponding $\beta_\tau^{(i)}$ non-zero, we define
    \begin{equation}
    \lambda = \frac{1}{1-\eta} \sum_{\bm H_t \in \cS} p(\bm H_t) \rho_{t,\bm H_t}
    \end{equation}
    which is a normalized density matrix.  Furthermore,
    \begin{equation}
    \rho_t = \eta \rho_{\mathrm{appr}, t} + (1-\eta)  \lambda.
    \end{equation}
    Then,
    \begin{equation}
    T(\rho_t, \rho_{\mathrm{appr}, t}) \leq \eta T(\rho_{\mathrm{appr}, t}, \rho_{\mathrm{appr}, t}) + (1-\eta) T( \lambda, \rho_{\mathrm{appr}, t}) \leq t \cdot \exp(-K L)
    \end{equation}
\end{proof}

As we will soon show, the realistic noise model that we work with indeed satisfies the condition for the approximate noise realization. %Because of this, from now on, we assume the approximate noise model.
Before discussing this, let us finally define a quantum automaton and its gadgets.

\begin{definition}[Elementary stabilizer quantum gadget] \label{def:elementary-gadget}
Given a stabilizer code $\mathcal{S} = \langle g_i \rangle$ and a noise realization (under the approximate noise model) specified by $\bm H_T$ (such that each operator $\cO_t$ is in the $N$-qubit Pauli group), an elementary stabilizer quantum gadget applied after the time step $t-1$ is a quantum operation $G: L(\cH) \to L(\cH)$ with local qubit support $\mathrm{supp}_q(G) \in \cL_1$ and local vertex support $\mathrm{supp}_v (G) \in \cL_0$ that operates between the time steps $t-1$ and $t$ and:
\begin{enumerate}
    \item [(a)] Measures stabilizer generators of the code on the input state in some local region $\cR_0 \subseteq \supp_v(G)$ (usually, $\cR_0 \subset \supp_v(G)$), obtaining syndromes $\sigma \in \mathbb{F}_2^{|\cR_0|}$ (called input syndromes), 
    \item [(b)] Applies an operator $\mathcal{O}(\sigma)$ (feedback)  supported on some local region $\mathcal{R}_q \subseteq \supp_q(G)$ (usually, $\cR_q \subset \supp_q(G)$) that is conditioned on the syndromes obtained in step (a) and is an element of the Pauli group, and 
    \item [(c)] Applies the noise operator in its support that is sampled from the noise model $p(H_t|{\bm H}_{t-1})$ (examples mentioned in Remark~\ref{def:qubit-noise-model-application} and Def.~\ref{def:gadget_error_model}).
 \end{enumerate}
We will also refer to the elementary gadget as operating at time $t$ implying that it operates between time steps $t-1$ and $t$. 
\end{definition}

\begin{definition}[Stabilizer quantum gadget] \label{def:stab-quant-gadg}
    Given a stabilizer code $\mathcal{S} = \langle g_i \rangle$ and a noise model, a stabilizer quantum gadget $G$ is a quantum operation  $G: L(\cH) \to L(\cH)$ composed of 
    % that composes 
    a finite number of elementary stabilizer quantum gadgets.

    Its spatial qubit support $\mathrm{supp}_q(G)$ and vertex support $\mathrm{supp}_v (G)$ are the union of qubit supports and vertex supports of constituent elementary gadgets, respectively. Its spacetime qubit and vertex supports are the unions of qubit supports and vertex supports of constituent elementary gadgets, respectively, together with their associated time steps.
\end{definition}

In our construction, the measurement and feedback model will be implemented by a set of local gadgets that together constitute a \emph{stabilizer quantum automaton}.

\begin{definition}[Stabilizer quantum automaton] 
A stabilizer quantum automaton for a stabilizer code $\mathcal{S} = \langle g_i \rangle$ and a given noise model is a circuit $\cA: L(\cH) \to L(\cH)$ comprised of (stabilizer quantum) gadgets associated with $\mathcal{S}$ that each has local support and domain. A stabilizer quantum automaton is a depth-$t$ circuit of elementary gadgets that evolves an initial density matrix $\rho(0)$ to a density matrix $\rho(t)$.
\end{definition}

We will call the collection of all elementary gadgets acting at time $t$ in the automaton a \emph{layer} of gadgets, which we will denote by $\cL(G)$. We will always assume that the action of neighboring elementary gadgets is parallelizable (i.e. they can be defined to act simultaneously) even if their supports overlap. In particular, by simultaneous we mean that two gadgets with overlapping supports can be applied in either order and the output would not change.  In our construction, we will always have this property for any pair of gadgets with overlapping supports acting at the same time step.  We will also say that the layer is formed by tiling the gadgets. We will assume that the union of supports of all elementary gadgets in a given layer $\cL(G)$ is the lattice of qubits $\cL_1$.

\begin{definition}[Damage]
   Call $\ket{\phi}$ the initial logical state of a stabilizer automaton with stabilizer group $\cS$. Consider a state $\ket {\psi} = \cO \ket {\phi}$  where $\cO$ is a Pauli operator. We call any operator $\cO'$ such that $\cO' \cO \in \cS$ the damage of the state $\ket {\psi}$.
\end{definition}
According to the definition above, we call damage any suitable operator belonging to an equivalence class of operators with a given syndrome whose equivalence relation is set by multiplication by stabilizers.

We will also call $\cL_{g,t}$ the set indexing the spacetime locations of all elementary gadgets that are applied at step $t$. For $\bm{x} \in \cL_{g,t}$, $G(\bm{x}, t)$ is the gadget at a corresponding spacetime location (we assume that we have fixed some convention to uniquely associate a spacetime location to each elementary gadget in the circuit). We also denote the set of all spacetime locations of the elementary gadgets within the automaton as $\cL_g = \bigcup_t \cL_{g,t}$. 

The definitions and properties of the noise shown until now did not depend on the nature of the probability distribution $p(\bm H)$ and are generally applicable. Now, we will finally discuss the constraints on the probability distribution under which the local automaton will be fault tolerant.

\begin{remark} [Applying qubit error model in automaton] \label{def:qubit-noise-model-application}
    A qubit error model is implemented in an automaton using the following scheme. Given noise realization $\bm H_T$, for each time step $t$, we choose a decomposition  $\cO_t = \prod_i \cO_{t,i}$ where $\supp (\cO_{t,i})\cap \supp (\cO_{t,j}) = \varnothing$ for $i \neq j$ so that $\cO_{t,i}$ is contained in the support of a single elementary gadget. The operators $\{\cO_{t,i} \}$ are applied by elementary gadgets at each step according to Def.~\ref{def:elementary-gadget}(c).
\end{remark}

\begin{definition} [Gadget error model] \label{def:gadget_error_model}
The noise model $\cE$ is called a gadget error model if, for any noise realization $\bm H_T$ determining $\vec \cO_T = \{ \cO_1, \dots , \cO_{T-1}, \cO_T\}$ and $E$, we additionally specify sets of elementary gadgets (identified by their coordinates) $F_E = \{ F_1, F_2, \dots, F_T\}$, with $F_t \subseteq \cL_{g,t}$, such that the support of each $f \in F_t$ overlaps with $E$, and $E$ is contained in the union of supports of the gadgets in $F_E$. We call $F_t$ the set of failed elementary gadgets at time $t$.  

In addition, for each time step $t$, we specify the partition $\cO_t = \prod_i \cO_{t,i}$ where each $\cO_{t,i}$ is contained in the support of $f_{t,i} \in F_t$. These operators are applied by elementary gadgets in Def.~\ref{def:elementary-gadget}(c).
\end{definition}

Technically, any gadget error model is implicitly parametrized by the set of supports of elementary gadgets.
Note that we will always assume that the union of supports of all elementary gadgets applied at each time step $t$ covers the entire lattice.

\begin{definition} [$p$-bounded qubit error model]
A noise model $\cE$ is called a $p$-bounded qubit error model, if, for any noise realization $\bm H$ with set of noise locations $E$ and any set of spacetime qubit locations $A$, it satisfies
    \be \label{eq:p-bounded-gadget1}
    \mathbb P [  A  \subseteq E] \leq p^{|A|}.
    \ee
\end{definition}

The class of $p$-bounded noise includes i.i.d qubit noise, but also allows for quite generic forms of correlated noise.  The analysis in this paper should also be readily modified to account for other types of noise like coherent and non-Markovian noise~\cite{Aliferis2005Apr, GottesmanBook}.  

\begin{definition} [$p$-bounded gadget error model]
A gadget error model $\cE_G$ is a $p$-bounded gadget error model if, given noise realization $\bm H$ with a set of noise locations $E$, an associated set of failed gadgets $F_E$, and any set $A_g \subseteq \cL_g$, it satisfies
    \be \label{eq:p-bounded-gadget2}
    \mathbb P [  A_g  \subseteq F_E] \leq p^{|A_g|}.
    \ee
\end{definition}

We note that the gadget error model is implicitly parameterized by the set of supports of elementary gadgets. 
The condition on the $p$-boundedness of a noise model suppresses the probability of large noise events.

Next, we show that different error models corresponding to different choices of gadgets are equivalent (and they are also equivalent to the qubit error model) as long as the support of gadgets contains a constant number of qubits:
\begin{lemma}\label{lem:equiv_gadget}
Any $p$-bounded gadget error model is a $p^{\alpha}$-bounded qubit error model for $0 < \alpha \leq 1$ depending on the maximum number of qubits in the support of any elementary gadget included in the gadget error model.  

Any $p$-bounded qubit error model is a $C p^{\beta}$-bounded gadget error model for constants $C, \beta \geq 1$ depending on the maximum number of qubits in the support of any elementary gadget.
\end{lemma}
\begin{proof}
Suppose we have a $p$-bounded gadget error model.  Call $A$ a set of qubits, and $A_g$ the associated set of gadgets which the qubits in $A$ are supported on.  For a noise realization $E$ and the set of failed elementary gadgets $F_E$ corresponding to $E$, 
\begin{equation}
\mathbb{P}[A \subseteq E] \leq \mathbb{P}[A_g \subseteq F_E] \leq p^{|A_g|}.
\end{equation}
Finally, since $|A_g| \geq \alpha |A|$ for some $0 < \alpha \leq 1$, we have $\mathbb{P}[A \subseteq E] \leq (p^{\alpha})^{|A|}$, and thus, the error model under consideration must be a $p^{\alpha}$-bounded qubit error model.

For the second statement, suppose we have a $p$-bounded qubit error model.  For a noise realization $E$ and the corresponding set of failed elementary gadgets $F_E$ (we can choose any such valid set arbitrarily, i.e. as long as it satisfies Def.~\ref{def:gadget_error_model}),  consider an arbitrary set of gadgets $A_g$. We may then write
\begin{equation}
\mathbb{P}[A_g \subseteq F_E] \leq \mathbb{P}\left[\bigcup_{A: A \sim A_g} \left ( A \subseteq E \right ) \right] \leq \sum_{A:A \sim A_g} \mathbb{P}\left[A \subseteq E\right],
\end{equation}
where, after the first inequality, we take a union over events where $A \sim A_g$ is a subset of qubit locations that are contained in the union of supports of all gadgets in $A_g$ such that support of each gadget in $A_g$ contains an element of $A$.

The number of such sets $A \sim A_g$ is $\leq C^{|A_g|}$ for $C$ a constant related to the maximum number of qubits in the support of a gadget. Therefore 
\begin{equation}
\sum_{A \sim A_g} \mathbb{P}\left[A \subseteq E\right] \leq \sum_{A \sim A_g} p^{|A|} \leq p^{|A|} C^{|A_g|},
\end{equation}
and using $|A| \geq \beta |A_g|$ for $\beta \geq 1$ another $O(1)$ constant related to the gadget support sizes, $\mathbb{P}[A_g \subseteq F_E] \leq (C p^{\beta})^{|A_g|}$ and therefore, the model under consideration must be a $C p^{\beta}$-bounded gadget error model.
\end{proof}

This Lemma proves that any two $p$-bounded gadget error models (even involving a different set of elementary gadgets with different support sizes) are, in fact, equivalent by some rescaling $p \to C p^{\gamma}$ for constants $C$ and $\gamma$, which does not affect the existence of a threshold.  Thus, it suffices to choose a ``canonical'' set of gadget supports for the $p$-bounded gadget error model.

We also need the following fact, as most of our analysis henceforth will apply to approximate noise models rather than exact ones.
\begin{fact}\label{fact:appr_pbounded}
The approximate noise model constructed from a $p$-bounded qubit noise model is a $p (1 + \varepsilon)$-bounded qubit noise model, where $\varepsilon \rightarrow 0$ when $L \rightarrow \infty$. 
\end{fact}
\begin{proof}
Given a $p$-bounded qubit noise model $p(\bm{H}_t)$, the corresponding approximate noise model $q(\bm{H}_t)$ by construction satisfies $q(\bm{H}_t) \leq p(\bm{H}_t)/\eta$ where $\eta = 1-\varepsilon(L)$ with $\varepsilon(L)$ vanishing as $L \to \infty$.  This immediately proves the fact.
\end{proof}

By the previous Lemma, this Fact also holds for the gadget error model.  Therefore, without loss of generality, we can always assume some canonical $p$-bounded gadget error model where all operators that are applied by the noise are Pauli. 

In the classical cellular automaton literature, it is also standard to treat the cellular automaton as probabilistic, wherein there is a small probability that at a given time a particular rule will not be applied. 
In the classical case, the errors that occur from not applying a given rule 
% such an error model 
can be simply absorbed into the gadget error model.  However, in the quantum case, not doing a measurement is distinct from a gadget failure, so this technically comprises a different kind of noise model:

\begin{definition}[Probabilistic quantum automaton]\label{def:pca}
A probabilistic quantum automaton is a stabilizer quantum automaton where the event of not performing a measurement and feedback operation in a set of spacetime points is sampled from a $p$-bounded probability distribution.
\end{definition}

The definition above is not the most general one, as we could technically also include the possibility of performing wrong measurement and feedback operations. Such a wrong measurement does not necessarily commute with the surrounding stabilizers and thus, cannot be performed simultaneously with them. Thus, since we are operating in discrete time and assume that all the measurements are performed simultaneously, we need to slightly modify our model to deal with these kinds of errors. We discuss this setting in more detail in Section~\ref{sec:FT_proof} and provide some arguments for why our construction should still remain fault tolerant in this more general setting.

\subsection{The toric code}

Though we assume the reader is familiar with stabilizer codes, we will present the stabilizer group of the toric code for completeness:

\begin{definition}[Toric code stabilizer group] \label{def:toric-code}
    Consider a square lattice with periodic boundary conditions (i.e., the torus $\mathbb{T}^2$) and associate qubits with the links. The stabilizer group of the toric code is generated by the following operators:
    \begin{equation}
        A_v = \bigotimes_{  e \in \delta v} Z_e, \quad B_p = \bigotimes_{e \in \partial p} X_e
    \end{equation}
    for every vertex $v$ and cell (square plaquette) $p$. Here, $\delta v$ denotes the set of links that touch vertex $v$ and $\partial p$ denotes the set of links in plaquette $p$. 
\end{definition}

This code defines 2 logical qubits on a torus.  Generalizing our decoder for toric codes on other lattices (and to more general cellulations) and manifolds will not be addressed in this work. For the toric code, the logical $X$ operators correspond to products of Pauli $X$ operators on homologically nontrivial loops of the primal lattice (and for logical $Z$, a product of Pauli $Z$ operators around homologically nontrivial loops of the dual lattice). The codestate which is the +1 eigenstate of both logical $Z$ operators (we denote it $\ket{\overline 0 \overline 0}$) is an equal-weight superposition of all possible configurations of $\ket{0},\ket{1}$ states of physical qubits where $\ket{1}$ states form a homologically trivial configuration of loops on the lattice. The remaining logical basis states can be generated by applying logical $X$ operators to this state (which is achieved in a depth-1 circuit). Moving forward, the goal of the toric code automaton will be to preserve an arbitrary logical state, which can be written as a superposition $\ket{\psi} = \alpha_{00} \ket{\overline 0\overline 0}+ \alpha_{01} \ket{\overline 0 \overline 1} + \alpha_{10} \ket{\overline 1\overline 0} + \alpha_{11} \ket{\overline 1\overline 1}$.

Let us show that for a measurement and feedback automaton operating with the toric code stabilizers, we can replace the noise model with an approximate noise model if the former is a $p$-bounded gadget error model. The argument above can be extended to an automaton for any stabilizer code as long as it has a distance scaling polynomially with $L$ (though the combinatorial aspect of the argument would need to be suitably modified).
\begin{lemma} \label{lemma:approximate-replacement-tc}
Consider a stabilizer quantum automaton for the 2D toric code on an $L \times L$ torus.  If the noise model $\cE_p$ is a $p$-bounded gadget error model, then for sufficiently small $p$, it can be replaced with an approximate noise model while obeying $T(\rho_t, \rho_{\text{appr},t}) \leq t \cdot \exp(-K L)$ for some positive constant $K$ and sufficiently large $L$. 
\end{lemma}
\begin{proof}
Following Prop.~\ref{prop:apprnoise}, we need to show the assumption that the marginal distribution $p(H_\tau)$ with associated operator $\cO_\tau$ has the property that $\beta_{\tau}^{(i)} \neq 0$ has probability $\leq \exp(- K L)$ for some constant $K$.  Note that $\cO_\tau$ can only be a Pauli operator $P_\tau$ (and not a linear combination thereof) if the union of supports of failed gadgets $A_t$ does not support a logical representative $L$.  Therefore, $\beta_\tau^{(i)}$ can be non-zero only if the gadgets that failed at timestep $\tau$ contain the support of a logical representative, which in the toric code forms a non-contractible cycle of the torus.  

We now show that this event occurs with negligible probability.  In order to contain a logical representative at time slice $\tau$, a connected path of gadgets of length $\geq L/a$ must fail, where $a$ is an $O(1)$ constant denoting the maximum linear extent of an elementary gadget (the toric code automaton has $a=3$).  Call $B_t$ a path of failed gadgets.  From the definition of $p$-boundedness, $\mathbb{P}[B_\tau \subset E] \leq p^{L/a}$.  The number of such paths $B_\tau$ is the number of self-avoiding walks on the ``gadget lattice'' at fixed time slice $\tau$, where vertices correspond to gadgets, and edges connect two vertices if two gadgets are adjacent or share corners.  The number of such walks is $\leq L^2 (D-1)^{L/a}$, where $D$ is the degree of the gadget lattice (including gadgets touching corners, the degree is $\leq 12$ for the toric code automaton) and the factor of $L^2$ comes from the number of starting points of the walk at fixed time slice $\tau$.  Thus
\begin{equation}
\mathbb{P}[\exists L_i: \text{supp}(L_i) \subset E] \leq \sum_{B_\tau} \mathbb{P}[B_\tau \subset E] \leq L^2 (11 p)^{L/a}. 
\end{equation}
where $L_i$ is some nontrivial logical operator. Therefore, following Prop.~\ref{prop:apprnoise}, we show that the probability that the above holds for all time slices up to time slice $t$ is $\leq L^2 t \cdot (11 p)^{L/a}$, which for some constant $K$ proves the Lemma.
\end{proof}

Thus, due to Prop.~\ref{prop:apprnoise}, Fact~\ref{fact:appr_pbounded} and Lemma \ref{lemma:approximate-replacement-tc}, it is sufficient to consider the approximate Pauli noise model from here on.

Because the toric code is a CSS quantum code (meaning that each stabilizer is either an $X$- or $Z$-type Pauli operator), it is possible to decompose its action into two independent $X$ and $Z$ parts. 

Under an approximate noise model, consider an arbitrary noise operator history $\{\cO_1, \cO_2, \cdots, \cO_T \}$.  We can decompose the noise operator at each time step into Pauli $X$ and $Z$ parts (ignoring inconsequential phase factors $i^n$ if the error is $Y = i X Z$), i.e. $\cO_t = \cO_t^{(X)} \cO_t^{(Z)}$. Using this, we define an $X$-type and $Z$-type automaton and the associated gadgets for the toric code:
\begin{definition}[$X$ and $Z$-type elementary gadgets]\label{def:xzgadgets}
    Assume a gadget noise model that, for each noise operator history, specifies $X$- and $Z$-type noise operators $\cO_{t,i}^{(X)}$ and $\cO_{t,i}^{(Z)}$ which are applied due to failed elementary gadgets. 

    An $X$-type elementary gadget $G^{(X)}$ is defined similarly to Def.~\ref{def:elementary-gadget}, but which involves measurements of $Z$-type stabilizers only and applies $X$-type Pauli feedback only. Then, to an $X$-type gadget at time step $t$ we assign the Pauli noise operator $\cO_{t,i}^{(X)}$ on the support of the gadget $\mathrm{supp}_q(G^{(X)})$.

    Similarly, a $Z$-type elementary gadget $G^{(Z)}$ involves measurements of $X$-type stabilizers and applies $Z$-type Pauli feedback. The Pauli noise operator $\cO_{t,i}^{(Z)}$ is assigned to the support of a $Z$-type gadget at time $t$.
   
\end{definition}

If the toric code automaton is such that its elementary gadgets are of $X$ and $Z$ type only, one can view it as a layer of simultaneously applying type-$X$ elementary gadgets and type-$Z$ elementary gadgets at each time step. We note that, with this definition, the operator $\cO_t$ that is applied can be interpreted as gadget failure of a certain set of $X$ and $Z$ gadgets.

We define an $X$-type toric code automaton associated with the full toric code automaton to be the isolated action of $X$-type gadgets only, with noise restricted to be of $X$-type only. This automaton is built from $X$-type non-elementary gadgets, which are comprised of $X$-type elementary gadgets. The $Z$-type toric code automaton is defined analogously. 

We now show that the simultaneous operation of $X$-type and $Z$-type automata can be equivalently determined by running the two automata separately and combining their actions in a simple way.

\begin{fact} \label{fact:XZ-decoupling-2}
    Assume that we have a toric code automaton that operates with $X$ and $Z$-type gadgets running for $T$ time steps that experiences an approximate gadget noise model.  Sample a noise realization $\bm H_T$ with a history of noise operators $\{\cO_1, \cO_2, \cdots, \cO_T\}$, and assume that the initial state is a toric code state of the form $\cO_{\mathrm{in}}^{(X)} \cO_{\mathrm{in}}^{(Z)} \ket{\psi}$ where $\cO_{in}^{(X/Z)}$ are the input damage operators and $\ket{\psi}$ is a syndrome-free initial logical state of the toric code.
    
    Then, one can determine the total operator applied by the automaton (i.e. due to both its feedback and the noise) under this noise realization by independently running the $X$ and $Z$-type automata whose outputs give the operators $\cO_{\mathrm{tot},T}^{(X)}$ and $\cO_{\mathrm{tot},T}^{(Z)}$, and multiplying these operators together.
\end{fact}
\begin{proof}
Recall that due to Lemma~\ref{lemma:noise-collapse} we can absorb the effect of the arbitrary noise operation followed by stabilizer measurements into some effective probability distribution of Pauli operators applied at each time step, where the distribution is $p$-bounded. At each timestep $t$, the total operator applied by the toric code automaton consists of the feedback applied by all the $X$-type gadgets $\cF_t^{(X)}$, the feedback applied by all the $Z$-type gadgets $\cF_t^{(Z)}$, and the noise operator $\cO_t = \cO_t^{(X)} \cO_t^{(Z)}$.  Assuming that the input state at this timestep has a total operator that has been applied to it due to the noise and operation of the automaton $\cO_{\rm{tot}, t-1} = \cO_{\rm{tot}, t-1}^{(X)} \cO_{\rm{tot}, t-1}^{(Z)}$, the feedback applied by the $X$-type ($Z$-type) gadget depends only on the $X$ ($Z$) part of this operator, i.e. $\cF_{t}^{(X)} \equiv \cF_{t}^{(X)}[\cO_t^{(X)}, \cO_{\rm{tot}, t-1}^{(X)}]$ and $\cF_{t}^{(Z)} \equiv \cF_{t}^{(Z)}[\cO_t^{(Z)}, \cO_{\rm{tot}, t-1}^{(Z)}]$ except for $\cF_{1}^{(X)} \equiv \cF_{1}^{(X)}[\cO_1^{(X)}, \cO_{\rm{in}}^{(X)}]$ and $\cF_{1}^{(Z)} \equiv \cF_{1}^{(Z)}[\cO_{1}^{(Z)}, \cO_{\rm{in}}^{(Z)}]$. The total operator applied to the toric code state relative to the initial state after time step $t$ is $\left ( \cF_t^{(X)} \cO_t^{(X)}\right)\cO_{\rm{tot}, t-1}^{(X)}   \left ( \cF_t^{(Z)}  \cO_t^{(Z)}\right )\cO_{\rm{tot}, t-1}^{(Z)}$.
In addition, we can freely commute $X$ and $Z$-type operators acting on the state because we can ignore any inconsequential negative signs.

Using the given form of the input state and using this decomposition at each time step, we find that the total operator acting on the toric code state during the entire operation decomposes into $\left (\prod_{t = 1}^T \cF_t^{(X)} \cO_t^{(X)} \cO_{\rm{in}}^{(X)} \right) \left (\prod_{t = 1}^T \cF_t^{(Z)}  \cO_t^{(Z)} \cO_{\rm{in}}^{(Z)} \right )$. We recognize each of these operations as the operator applied due to separate action of the $X$ and $Z$-type toric code automaton.
\end{proof}

Due to this decoupling, we will introduce a $p$-bounded ``decoupled'' error model where $Z$-type gadgets experience $Z$ errors and the $X$-type gadgets experience $X$ errors, defined below.

\begin{definition}[$p$-bounded decoupled gadget error model] \label{def:approximate-decoupled-gadget-error-model}

Consider an approximate noise model, where for noise realization $\bm H_T$ with set of noise locations $E$ and noise operator history $\{\cO_1, \cO_2, \cdots, \cO_T \}$ we write  $\cO_t = \cO_t^{(X)} \cO_t^{(Z)}$.  We then partition $E = E_X \cup E_Z$ so that $E_X$ and $E_Z$ denote sets of the spacetime locations of qubits marked for being possibly acted upon by $X$ and $Z$ operators, respectively, and $\bigcup_t \supp (\cO_t^{(X)}) \subseteq E_X$ and $\bigcup_t \supp (\cO_t^{(Z)}) \subseteq E_Z$. 

We additionally specify sets of failed elementary gadgets  $F_{E_X} = \{ F_{1}^{(X)}, F_2^{(X)}, \dots, F_T^{(X)}\}$ and $F_{E_Z} = \{ F_{1}^{(Z)}, F_2^{(Z)}, \dots, F_T^{(Z)}\}$ with $F_t^{(X/Z)} \subseteq \cL_{g,t}^{(X/Z)}$, such that the support of each $f^{(X/Z)} \in F_t^{(X/Z)}$ overlaps with $E_{X/Z}$, and $E_{X/Z}$ is contained in the union of supports of the gadgets in $F_{E_{X/Z}}$.  For each time step $t$, we can choose a decomposition $\cO_t^{(X/Z)} = \prod_i \cO_{t,i}^{(X/Z)}$ where each $\cO_{t,i}^{(X/Z)}$ is contained in the support of $f_{t,i}^{(X/Z)} \in F_t^{(X/Z)}$. 

We call this noise model a $p$-bounded decoupled gadget noise model if, for any subsets of spacetime locations of $X$-type and $Z$-type elementary gadgets $A_X$ and $A_Z$, we have:

\begin{equation}
\mathbb{P}[(A_X \subseteq F_{E_X}) \wedge (A_Z \subseteq F_{E_Z})] \leq p^{|A_X| + |A_Z|}.
\end{equation}
\end{definition}

\begin{lemma}\label{lem:equiv_decoupled_gadget}
A $p$-bounded decoupled gadget error model is a $Cp^{\alpha}$-bounded approximate qubit error model for some $\alpha < 1$.  

A $p$-bounded approximate qubit error model is a $C'p^{\beta}$-bounded decoupled gadget error model for constants $C',\beta > 1$.
\end{lemma}
\begin{proof}
The proof is similar to that of  Lemma~\ref{lem:equiv_gadget}.

To show the first statement, consider an arbitrary set of spacetime locations of qubits $A$.  We can write
\begin{equation}
    \mathbb{P}[A \subseteq E] \leq \sum_{A_g^{(X)}, A_g^{(Z)}: A_g^{(X)}\cup A_g^{(Z)}\sim A}\mathbb{P}[(A_g^{(X)} \subseteq F_{E_X}) \wedge (A_g^{(Z)} \subseteq F_{E_Z})] 
\end{equation}
where we take the sum over all possible sets of $X$- and $Z$-type gadget locations $A_g^{(X)}$ and $A_g^{(Z)}$ such that the support of each gadget in the set contains at least one element in $A$ and the union of supports of all gadgets in $A_g^{(X)}$ and $A_g^{(Z)}$ contains $A$ (denoted by $A_g^{(X)}\cup A_g^{(Z)}\sim A$).

Under a $p$-bounded decoupled gadget error model as well as using the fact that there exists a constant $0 < \alpha \leq 1$ such that 
$|A_g^{(X)}|+|A_g^{(Z)}| \geq \alpha'|A|$ for all $A_g^{(X)}$ and $A_g^{(Z)}$ with $A_g^{(X)}\cup A_g^{(Z)}\sim A$, we have that $\mathbb{P}[(A_g^{(X)} \subseteq F_{E_X}) \wedge (A_g^{(Z)} \subseteq F_{E_Z})] \leq p^{\alpha|A|}$. Using that the total number of such sets is $D^{3|A|}$ (the factor of 3 is because each qubit location is either assigned to $A_g^{(X)}$, $A_g^{(Z)}$, or to both) for some constant $D$, we can write 

\begin{equation}
    \sum_{A_g^{(X)}, A_g^{(Z)}: A_g^{(X)}\cup A_g^{(Z)}\sim A}\mathbb{P}[(A_g^{(X)} \subseteq F_{E_X}) \wedge (A_g^{(Z)} \subseteq F_{E_Z})] \leq (D^3p)^{\alpha |A|}
\end{equation}
choosing $C = D^3$, we obtain the first statement of the Lemma.

For the second statement of the Lemma, we designate a partitioning $E = E_X \cup E_Z$ such that $p(\cO,E)$ can be now represented as $p(\cO,E_X,E_Z)$. $E_{X,Z}$ are in addition chosen such that if a given qubit is in $E_X$ but not in $E_Z$, the noise operator supported on this qubit can only be Pauli $X$ or identity, and similarly for $E_Z$. We then choose any set of $X$-type gadgets $F_{E_X}$ such the support of any gadget $f_X \in F_{E_X}$ contains at least one location in $E_X$ and $E_X$ is contained within the union of supports of all the gadgets in $F_{E_X}$. Similarly, we choose any set of $Z$-type gadgets $F_{E_Z}$ such the support of any gadget $f_Z \in F_{E_Z}$ contains at least one location in $E_Z$, and $E_Z$ is contained within the union of supports of all the gadgets in $F_{E_Z}$.  Then, choose arbitrary sets of $X$-type and $Z$-type gadgets $A_g^{(X)}$ and $A_g^{(Z)}$. We can write 
\begin{equation}
    \mathbb{P}[(A_g^{(X)} \subseteq F_{E_X}) \wedge (A_g^{(Z)} \subseteq F_{E_Z})]  \leq \mathbb{P}\left[\bigcup_{A: A \sim A_g^{(X)} \cup A_g^{(Z)}} \left ( A \subseteq E \right ) \right] \leq \sum_{A: A \sim A_g^{(X)} \cup A_g^{(Z)}} \mathbb{P}\left[A \subseteq E\right],
\end{equation}
where $A \sim A_g^{(X)} \cup A_g^{(Z)}$ denotes that $A$ is contained in the union of all gadget supports in $A_g^{(X)}$  and $A_g^{(Z)}$, and additionally that the support of each gadgets in $A_g^{(X/Z)}$ contains at least one qubit from $A$. 
The number of such sets is $\leq C^{|A_g^{X}| + |A_g^{Z}|}$ for some constant $C$. 
Therefore 
\begin{equation}
\sum_{A: A \sim A_g^{(X)} \cup A_g^{(Z)}} \mathbb{P}\left[A \subseteq E\right] \leq \sum_{A: A \sim A_g^{(X)} \cup A_g^{(Z)}} p^{|A|} \leq p^{|A|} C^{|A_g^{X}| + |A_g^{Z}|},
\end{equation}
We can then write $|A| \geq \beta (|A_g^{X}| + |A_g^{Z}|)$ for some $\beta >0$, and obtain the second statement of the Lemma.
\end{proof}

We remark that  Lemma~\ref{lem:equiv_gadget} will also hold for $p$-bounded \emph{approximate} qubit and gadget noise models. Together with the Lemma that we just proved, this implies that any  $p$-bounded decoupled gadget noise model is equivalent to any $p'$-bounded approximate gadget noise model for some appropriately chosen $p'$, and is equivalent to a $p''$-bounded approximate qubit noise model for some appropriately chosen $p''$. 

When defining the toric code automaton, we will first define its $X$ and $Z$-type elementary gadgets. Due to Fact~\ref{fact:XZ-decoupling}, we can then focus on $X$-type gadgets only, which are used to define the $X$-type toric code automaton. The $Z$-type automaton can be constructed in the same way by working on the dual lattice. The full toric code automaton can be obtained by considering $X$- and $Z$-type automata separately and combining their actions at the end.

We will then show that the $X$-type toric code automaton is fault tolerant (in the sense that it protects the state against $X$-type logical failures) under an approximate $p$-bounded gadget noise model. Because the $Z$-type automaton is analogous, it will have the same properties as the $X$-type one. Then, in Corollary~\ref{corollary:full-simulaton-TC}, we will show that the fault tolerance of the toric code automaton under (a more general) $p$-bounded noise model follows from these properties.

\begin{remark} \label{remark:full-TC-automaton}
    We remark that thus far, we defined the toric code automaton under an approximate decoupled gadget error model only. This is because analyzing an error model with $X$ and $Z$-type elementary gadget failures is much more convenient.

    A more natural error model for the quantum toric code automaton should include the possibility for a failed gadget to apply an \emph{arbitrary} Pauli noise channel. A toric code automaton operating under such a noise model can be defined in the following way. Consider the same automaton of $X$- and $Z$-type elementary gadgets, but modify the way the noise is applied. In particular, ``pair'' $X$ and $Z$ gadgets acting in the same time step, such that an $X$-type gadget at location $(x,y)$ is associated with a $Z$-type gadget on a dual lattice at location $(x+1/2,y+1/2)$. We unite the supports of this pair of gadgets and view them as a single gadget, with the $p$-bounded gadget noise model now defined with respect to these gadgets. 
    
    Due to Lemma~\ref{lem:equiv_gadget} and \ref{lem:equiv_decoupled_gadget}, this more general noise model is equivalent to the decoupled gadget error model upon appropriate rescaling of $p$.
\end{remark}

Due to the results in Fact~\ref{fact:XZ-decoupling} and Lemma~\ref{lem:equiv_decoupled_gadget}, we will consider the $X$-type automaton only in most of what follows. This effectively reduces the problem we are dealing with to a classical one. Most of the paper will be dedicated to constructing the $X$-type automaton and proving its fault tolerance. We will show that fault tolerance of the full automaton in Corollary~\ref{corollary:full-simulaton-TC}.

%% file: AppendixA.tex
\section{Review of the extended rectangles (exRec) formalism}\label{app:exRec}

In this Appendix, we will review extended rectangle method of  Aliferis, Gottesman, and Preskill~\cite{Aliferis2005Apr,gottesman2009introductionquantumerrorcorrection,GottesmanBook} for proving fault tolerance in concatenated code constructions. We want to start with some circuit $\mathcal{C}$, either classical or quantum, and produce a `simulated' version of the circuit $FT(\mathcal{C})$. If each component in the circuit fails with probability $p$, then the probability that the output of the circuit $\mathcal{C}$ is wrong is large. We want $FT(\mathcal{C})$ to have an error suppression property, where the effective error per component $\mathcal{C}$ will scale like $O(p^{\alpha})$ for $\alpha > 1$ -- then, recursively repeating this procedure will yield scalable error suppression.
The way to realize the fault tolerant simulation of circuit $\mathcal{C}$ is by replacing it with a new circuit where each bit is now several bits and corresponds to a logical bit of a small code, and between each gate $G_i$ and $G_{i+1}$ in the circuit (which is now an `encoded' or `logical' gate) one inserts an error correction unit which we will call $\text{EC}$. A layer of $\text{EC}$ units are also added at the beginning and the end of the circuit. As a result, the simulated circuit is slightly larger than the original circuit.

Both the encoded gates and the EC unit need to satisfy certain conditions. Before writing them down, we must define a notion of an $r$-filter and an ideal decoder. In particular, suppose the small code has the set of codewords $S$ of some stabilizer code or classical $\mathbb{F}_2$-linear code on $n$ qubits (or bits). For now, we will use notation consistent with that in Ref.~\cite{gottesman2009introductionquantumerrorcorrection}.  We will also illustrate the concept with classical codes for simplicity, although the generalization to quantum codes is straightforward. 
 
\begin{definition}[$r$-filter]
     Define 
 \begin{equation}
 S_r = \{c: \exists \ s \in S, \,\, \text{wt}(c+s) \leq r\}
 \end{equation}
which is the set of elements that are at most distance $r$ from a codeword. An $r$-filter is an operation that projects onto the subspace $S_r$ of configurations.
\end{definition}
 
An ideal decoder is a map from $\{0,1\}^n$ where $n$ is the number of physical bits in the code to $\{0,1\}^k$ where $k$ is the number of encoded bits and simply performs an error correction and a read-out operation (which are assumed to be done in the absence of errors), extracting the values of the logical bits. Note that for Tsirelson's automaton the decoder simply performs a majority vote, with $n = 3$ and $k = 1$. We emphasize that $r$-filters and ideal decoders are not physical components of the circuit but just auxilliary tools that will be helpful to us.

\begin{definition}[Gate and EC conditions] \label{def:gateEc-OG}
The Gate A and B conditions can be schematically depicted as:
 \begin{equation*}
 \includegraphics[width = \textwidth]{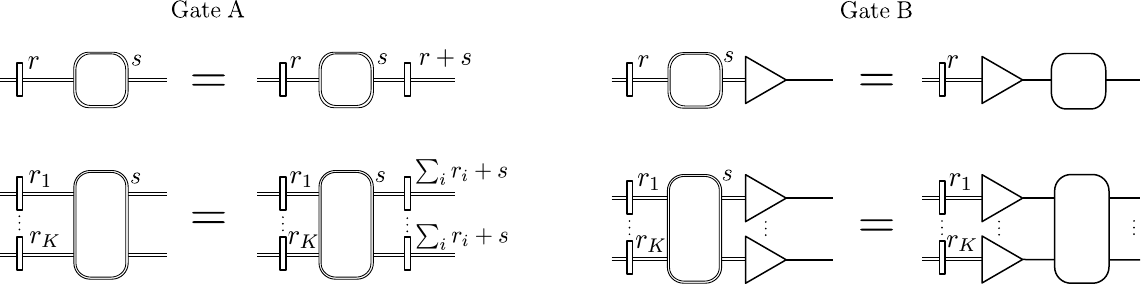}
 \end{equation*}
 and two conditions for EC which we will call EC A and EC B can be depicted as:
 \begin{equation*}
 \includegraphics[width = \textwidth]{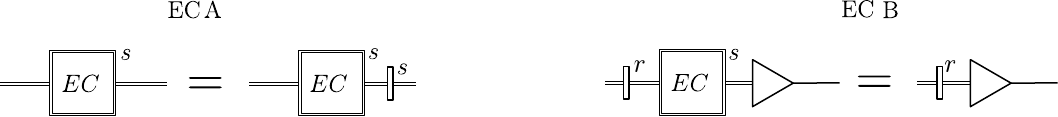}
 \end{equation*}
 Here and in the rest of the paper, the equality between circuit diagrams holds up to an application of an operator in the stabilizer group of the code to the output state. In this notation, the triangle corresponds to an ideal (noiseless) decoder, the double lines correspond to encoded bits (qubits) and the single line corresponds to a single unencoded bit (qubit). EC gadgets are labeled.  The thin rectangle corresponds to the $r$-filter, with the value of $r$ labeled on its top right. The letters $s$ next to the gate/EC gadgets indicate the number of faults that have occurred in the implementation of these operations. The Gate A and Gate B properties each have two associated figures. The top figure shows the condition for a single-bit gate (such as the NOT gate) while the bottom figure shows the condition for a two-bit gate (such as the AND gate). The generalization to $> 2$-bit gates follows similarly. Finally, the values of $r$ and $s$ in these conditions must satisfy an inequality relating to the distance of the small code used for the simulation:
 \begin{align}
 &\text{Gate A: } \sum_i r_i + s \leq t \quad  \quad  \text{Gate B: } \sum_i r_i + s \leq t \nonumber   \nonumber \\
 &\text{EC A: } s \leq t \nonumber  \quad \quad \quad \quad \ \, \quad \quad  \text{EC B: } r + s \leq t \\
 \end{align}
 where $t$ is the maximum number of errors that can be guaranteed to be corrected by a noiseless EC gadget, related to the distance of the small code via $d = 2t + 1$.

\end{definition}

 Let us  provide a bit of intuition for these conditions. For the Gate A property, we require that if an input state is conditioned on being $r$-close to a codeword, then a gate with $s$ faults does not push the input state more than $r+s$ distance from any codeword assuming $r + s \leq t$. This does not necessarily require that the input state remains near the {\it same} codeword; in reality, it can jump to a different codeword so long as it remains close to one. The Gate B property implies that the input state undergoes the correct logical operation so long as it is sufficiently close to a codeword and the gate has few faults inside.  The EC A property is perhaps the trickiest to satisfy: the input state can be arbitrary, but we require that the input state gets sent near a codeword under the EC unit, even if the EC unit has $s \leq t$ faults. However, the input state needs not to be mapped to the closest codeword, any codeword will do.
 Finally, the EC B property implies that if a state is almost a codeword and a faulty error correction operation is performed on it, the state should still remain close to the codeword so that it can still be properly read out by the ideal decoder. Of course, we demand the total faults $r + s \leq t$.

Now we review the method of proving fault tolerance of simulated circuits from ~\cite{gottesman2009introductionquantumerrorcorrection} (the figures in this subsection are adapted from Ref.~\cite{gottesman2009introductionquantumerrorcorrection}). Consider circuit $\mathcal{C}$ and its fault-tolerant simulation $FT(\mathcal{C})$. We are going to define an extended rectangle for each operation to be a region in spacetime containing both the operation as well as the preceding and succeeding error correction modules. For encoding and decoding operations, it is a little different but we do not need to be currently concerned with them. Examples are shown below:
 \begin{equation}
 \includegraphics[width = 0.6\textwidth]{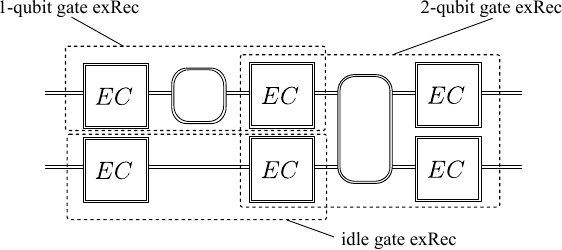}
 \end{equation}
 We now define properties of extended rectangles, which we will call good/bad and correct/incorrect. These are defined with respect to ideal decoders which are placed at the end of the circuit. The logic is that if a rectangle does not contain too many faults, we can pull an ideal decoder through this rectangle, replacing the physical bits of the error correcting code with a single bit and the logical gate with a clean gate acting on the single bit. We then similarly pull decoders through the rest of the simulated circuit to get the original circuit $\mathcal{C}$.
 \begin{definition}[Good/Bad exRecs]
 An extended rectangle is good if the total number of faults in all of the components in its support is $\leq t$. Otherwise it is bad.
 \end{definition}
 \begin{definition}[Correct/Incorrect exRecs]
 An extended rectangle for a single bit gate is correct if:
 \begin{equation}
 \includegraphics[width = 0.65 \textwidth]{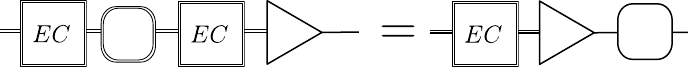}
 \end{equation}
 and an extended rectangle for a multi-bit gate is correct if:
 \begin{equation}
 \includegraphics[width = 0.65 \textwidth]{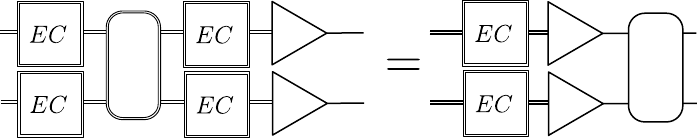}
 \end{equation}
 Otherwise the extended rectangle is said to be incorrect.
 \end{definition}
 An important result can be shown thanks to the EC and gate properties is (we reproduce the proof below):
 \begin{proposition} \label{prop:goodcorrect}
     If an extended rectangle is good, then it is correct.
 \end{proposition}
 \begin{proof}

 The proof follows from the chain of diagrammatics below (assuming $s_1 + s_2 + s_3 \leq t$). 
 \begin{equation}
 \includegraphics[width =  \textwidth]{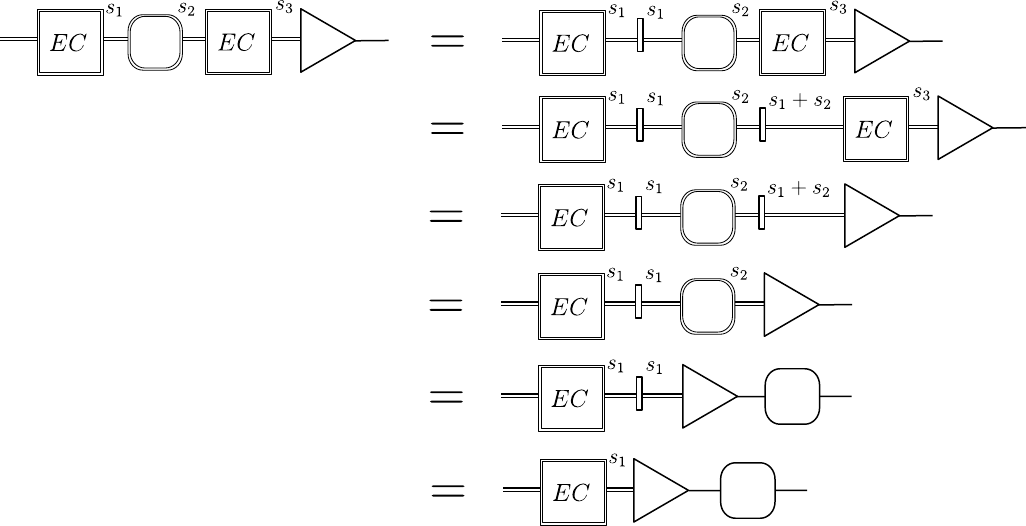}
 \end{equation}
 For completeness, we provide the properties used in each line:
 \begin{itemize}
 \item Line 1: EC A
 \item Line 1 to 2: Gate A
 \item Line 2 to 3: EC B
 \item Line 3 to 4: Gate A (should be understood as equivalence between lines 3 and 4 due to Gate A)
 \item Line 4 to 5: Gate B
 \item Line 5 to 6: EC A  (should be understood as equivalence between lines 5 and 6 due to EC A)
 \end{itemize}
 \end{proof}
 
 Therefore, if all extended rectangles are good, then one can simply pull the decoder through the entire simulated circuit $FT(\mathcal{C})$, producing the original circuit $\mathcal{C}$:
 \begin{proposition}
 If a noisy fault tolerant simulation $FT(\mathcal{C})$ of circuit $\mathcal{C}$ has all of its extended rectangles good, then following $FT(\mathcal{C})$ with a layer of decoders $\mathcal{D}^{\otimes n}$, one per simulated bit, results in the same action as $\mathcal{C}$.
 \end{proposition}

 This means that the logical information processed by the computation was unaffected by errors. However, the likelihood that all exRecs are good is small.  In the situation when an exRec is bad, there are a couple of issues that come to mind. The first is that extended rectangles overlap. If two bad extended rectangles overlap, we cannot properly quantify whether the faults occurred in the support of their intersection or not. This is problematic since it could be the case that with probability $p^2$ two exRecs can be bad, and the effective error probability per exRec is $p$, which indicates no error suppression.   To avoid this overcounting issue, we must ensure that no bad extended rectangles can have overlapping support. This is fixed through a truncation procedure, motivating the definition below:
 \begin{definition}[Truncated exRec] \label{def:truncated-exRecs}
 A truncated extended rectangle is an extended rectangle where one or more of the succeeding EC units are removed.
 
 A truncated extended rectangle is good if the total number of faults in both the operation, preceding EC unit, and any remaining succeeding EC units is $\leq t$. Otherwise, it is bad.

 A truncated extended rectangle is correct if (for a single-bit gate):
 \begin{equation}
 \includegraphics[width = 0.6 \textwidth]{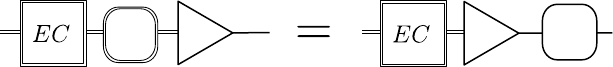}
 \end{equation}
 or for a two-bit gate (generalizable to an $n$-bit gate):
 \begin{equation}
 \includegraphics[width = 0.6 \textwidth]{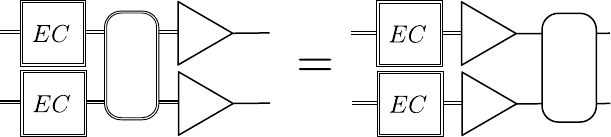}
 \end{equation}
 or
 \begin{equation}
 \includegraphics[width = 0.65 \textwidth]{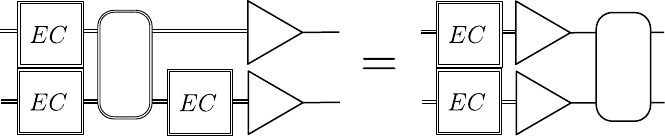}
 \end{equation}
 Else, the truncated extended rectangle is said to be incorrect. 
 \end{definition}
 Truncated extended rectangles have very similar properties to regular extended rectangles:
 \begin{proposition}
 If a truncated extended rectangle is good, then it is correct.
 \end{proposition}
 \begin{proof}
 To show that truncation does not matter we use the identity that an EC unit with zero faults followed by an ideal decoder is equal to an ideal decoder. Then, one can simply insert an EC unit with zero faults before the ideal decoder, and the proof reduces to that for the untruncated case.
 \end{proof}
 Truncated extended rectangles are utilized when an extended rectangle is bad, for reasons discussed previously. However, the way in which we assign extended rectangles is not necessarily unique. The following procedure gives a reasonable way to assign extended rectangles which is sufficient for our purposes. First, divide the circuit into timeslices, where each time slice contains some number of gates. Call the depth of the circuit $T$.
 \begin{algorithm}[Labeling of good and bad exRecs and truncation] \label{alg:OG}
 The following algorithm is used to label and truncate exRecs:
\begin{enumerate}
 \item Start with all gates at time $t = T$. 
 \item Construct extended rectangles for each of these gates at time $t$ and determine whether they are good or bad.
 \item If an extended rectangle is bad, truncate all extended rectangles occurring directly before it and which overlap with it.
 \item Decrement $t \to t-1$ and go to step 2 with the truncated extended rectangles replacing the assignment of extended rectangles at depth $t-1$ if needed.
 \item End when $t = 0$.
 \end{enumerate}
 \end{algorithm}
 
 Finally, we discuss what happens when we try to pull a decoder past a bad extended rectangle, either truncated or untruncated.  For this, we need to define a $\ast$-decoder.
 \begin{definition}[$\ast$-decoder] \label{def:conventional-ast-decoder}
 A $\ast$-decoder is a decoder with an additional data line where the syndromes (values of check operators) are stored. It is denoted as below:
 \begin{equation}
 \includegraphics[width = 0.15 \textwidth]{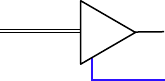}
 \end{equation}
 \end{definition}
 Intuitively, an ideal decoder corresponds to a $\ast$-decoder with the additional data line forgotten. When dealing with bad rectangles, the state to be decoded is often far from a codeword and the ideal decoder produces an incorrect result which depends on the measured sydromes. However, for good rectangles, the state is close enough to a codeword that the correct logical operation is always applied. This motivates the following definitions for correctness in the presence of a $\ast$-decoder.
 \begin{definition}
 An extended rectangle is correct in the presence of a $\ast$-decoder if
 \begin{equation}
 \includegraphics[width = 0.55 \textwidth]{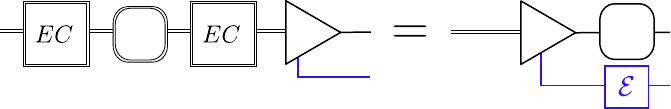}
 \end{equation}
 where $\mathcal{E}$ is an arbitrary operation (corresponding to the change in the values of the syndromes under a logical operation).
 \end{definition}
 \begin{proposition}
 Correctness of any extended rectangle (truncated or otherwise) implies correctness under the $\ast$-decoder.
 \end{proposition}
 Under a bad extended rectangle, the effective action of a gate after pulling a $\ast$-decoder through the exRec obeys 
 \begin{equation}
  \includegraphics[width = 0.75 \textwidth]{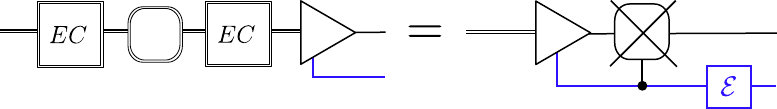}
 \end{equation}
 where there is a conditional action on the values of the syndromes but we do not need to know the particular action. Importantly, this action is localized to the affected qubit and the data line and we can continue to push the decoder past other successive good rectangles without incurring additional failures due to previous bad exRecs.

 We will now prove the main statement of the Appendix, i.e. a $p$-bounded error model in thecircuit $FT(\cC)$ yields an  $O(p^{t+1})$-bounded error model in the simulated circuit $\cC$.
 
 \begin{theorem}\label{thm-1level}
 Consider a circuit $\mathcal{C}$ and construct the fault tolerant simulation of it $FT(\mathcal{C})$. If noise model for gate failure $FT(\mathcal{C})$ is $p$-bounded for some $p$, then by pulling a layer of $\ast$-decoders through the simulated circuit, one recovers the circuit $\mathcal{C}$ simulated with the noise model for gate failure that is is $A p^{t+1}$-bounded with some constant $A$. 
 \end{theorem}
 \begin{proof}
 We start by inserting $\ast$-decoders at the end of $FT(\mathcal{C})$. When we pull the $\ast$-decoders through $FT(\mathcal{C})$, we obtain the circuit $\mathcal{C}$ with some of the components faulty. We want to determine the probability that set $F$ is faulty. We know that each element in $F$ must have been part of a bad extended rectangle, and for each gate $U \in \mathcal{C}$ there is only one extended rectangle containing $U$. Furthermore, due to the truncation procedure, {\it all} bad extended rectangles are disjoint. However, though all rectangles in $F$ are faulty, the exact positions of these extended rectangles depend on the number of faults in other locations in $\mathcal{C}$ since different fault patterns can result in the extended rectangle associated to $f \in F$ being truncated or untruncated.

 We can define a minimal set of errors as the smallest set of possible errors needed to make every component in $F$ fail. In this case, we need at least $t+1$ failures in an exRec corresponding to a component in $F$. For one exRec $R$, there is a constant number of fault locations for $t+1$  errors (call this constant $A'$), we have that (from a union bound along with the $p$-bounded condition):
 \begin{equation}
 \mathbb{P}\left(R \text{ is bad}\right) \leq \sum_{i^{(1)}, i^{(2)},\cdots, i^{(t+1)} \in R} \mathbb{P}\left(i^{(1)}, i^{(2)},\cdots, i^{(t+1)} \in \mathcal{E}\right) \leq A' p^{t+1}.
 \end{equation}
 The true value of $A'$ strongly depends the kind of gate that is contained within the exRec (and can be considerably larger for two-bit/two-qubit gates). However, here we are going to bound everything by some global constant $A'$. Let us assume that the maximum number of bits/qubits involved in a single gate is $k$. Now if we want the probability that all components in $f \in F$ are bad (for simplicity, we will label $f$ an integer from $1$ to $|F|$):
 \begin{align}\label{eq:newerrormodel}
 \mathbb{P}\left(f \text{ is bad }\forall f \in F\right) &= \sum_{\{R_f\}}\mathbb{P}\left(f \text{ is bad }\forall f \in F \text{ and } R_f \text{ is assignment of exRec}\right),\nonumber \\
 &\leq \sum_{\{R_f\}} \sum_{\{i^{(1)}_f, \cdots, i^{(t+1)}_f \in R_f\}}\mathbb{P}\left(i^{(1)}_f, \cdots, i^{(t+1)}_f \in \mathcal{E} \,\,\forall f \in F \text{ and } R_f\right), \nonumber \\
 &\leq \sum_{\{R_f\}} A'^{|F|} p^{(t+1)|F|} \leq (2^k A' p^{t+1})^{|F|}
 \end{align}
 where $\{R_f\}$ is a valid assignment of exRecs to locations in $f \in F$ following the algorithm proposed earlier and $\{i^{(1)}_f, \cdots, i^{(t+1)}_f \in R_f\}$ are sets of $k+1$-tuples of fault locations in each exRec. The third line follows from the logic described previously applied to each bad exRec along with the fact that bad exRecs are disjoint, and the extra factor of $2^k$ comes from the fact that there are at most $2^{k |F|}$ possible extended rectangle configurations one can draw if $F$ is completely faulty. Calling $A = 2^k A'$ proves the claim.
 \end{proof}

%% file: AppendixB.tex
\section{Additional data for some of the proofs}

In this Appendix, we produce the output for some of the computer-assisted proofs that were utilized in the main text, particularly in Sec.~\ref{sec:tsirelson-proof-2} and Sec.~\ref{sec:FT_proof}.

\subsection{Lemma~\ref{lemma:gate-ec-prime-TS}, proof of ($\EC$ A)$_{k,m}$}
\label{app:reversibility}

Here, we display the output of a brute force computation used to prove ($\EC$ A)$_{k,m}$ (which we will simply refer to as $\EC$ A for brevity in the following) for Tsirelson's automaton.  Recall that in the proof, we needed to check that running $\cM_1$ on an arbitrary input state and following with $Z_0^{\otimes 3}$ will not create what we call an irreversible configuration.  This is simply a coarse-grained configuration that does not have an inverse under $\cM_1$.  The only two such irreversible configurations are $(000 \, 111 
, 000)$ and $(111 \, 000 \, 111)$.  To check that neither of these configurations is obtained, we do the computation indicated in Table~\ref{tab:irrev}.  In particular, labels of the columns are the input coarse-grained configurations, where $(abc)_1 = (aaa \, bbb\, ccc)$.  Labels of rows are the damage configurations that we $\mathbb{F}_2$-add on top of the input.  We then run the gadget $Z_0^{\otimes 3} \circ \cM_1$ with this input and show its value in the corresponding cell.  The fact that none of the cells contain either of the irreversible configurations $(010)_1$ or $(101)_1$ means that our gadgets passed the desired test.

Note that there is an alternative proof to $\EC$ A that requires one to double the operation of the error correction gadget and eliminates the need to do this computer search.  This was used for the toric code automaton for the purposes of making the proof shorter.

\begin{table}[t]
\centering
\begin{tabular}{|c|c|c|c|c|c|c|c|c|}
\hline
 & \multicolumn{8}{c|}{\textbf{Input}}   \\ \hline
\textbf{Damage } & {$(\bm{000)_1}$} & {$\bm{(001)_1}$} & {$\bm{(010)_1}$} & {$\bm{(011)_1}$} &{$\bm{(100)_1}$} & {$\bm{(101)_1}$} & {$\bm{(110)_1}$} & {$\bm{(111)_1}$} \\ \hline
$\bm{(000\ 000\ 000)}$ & $(000)_1$ & $(001)_1$ & $(000)_1$ & $(011)_1$ & $(100)_1$ & $(111)_1$ & $(110)_1$ & $(111)_1$ \\ \hline
$\bm{(001\ 000\ 000)}$ & $(000)_1$ & $(001)_1$ & $(000)_1$ & $(011)_1$ & $(100)_1$ & $(111)_1$ & $(110)_1$ & $(111)_1$ \\ \hline
$\bm{(010\ 000\ 000)}$ & $(000)_1$ & $(001)_1$ & $(000)_1$ & $(011)_1$ & $(100)_1$ & $(111)_1$ & $(110)_1$ & $(111)_1$ \\ \hline
$\bm{(011\ 000\ 000)}$ & $(100)_1$ & $(111)_1$ & $(110)_1$ & $(111)_1$ & $(000)_1$ & $(001)_1$ & $(000)_1$ & $(011)_1$ \\ \hline
$\bm{(100\ 000\ 000)}$ & $(000)_1$ & $(001)_1$ & $(000)_1$ & $(011)_1$ & $(100)_1$ & $(111)_1$ & $(110)_1$ & $(111)_1$ \\ \hline
$\bm{(101\ 000\ 000)}$ & $(100)_1$ & $(111)_1$ & $(110)_1$ & $(111)_1$ & $(000)_1$ & $(001)_1$ & $(000)_1$ & $(011)_1$ \\ \hline
$\bm{(110\ 000\ 000)}$ & $(100)_1$ & $(111)_1$ & $(110)_1$ & $(111)_1$ & $(000)_1$ & $(001)_1$ & $(000)_1$ & $(011)_1$ \\ \hline
$\bm{(111\ 000\ 000)}$ & $(100)_1$ & $(111)_1$ & $(110)_1$ & $(111)_1$ & $(000)_1$ & $(001)_1$ & $(000)_1$ & $(011)_1$ \\ \hline
$\bm{(000\ 001\ 000)}$ & $(000)_1$ & $(001)_1$ & $(000)_1$ & $(011)_1$ & $(100)_1$ & $(111)_1$ & $(110)_1$ & $(111)_1$ \\ \hline
$\bm{(000\ 010\ 000)}$ & $(000)_1$ & $(001)_1$ & $(000)_1$ & $(011)_1$ & $(100)_1$ & $(111)_1$ & $(110)_1$ & $(111)_1$ \\ \hline
$\bm{(000\ 011\ 000)}$ & $(000)_1$ & $(011)_1$ & $(000)_1$ & $(001)_1$ & $(110)_1$ & $(111)_1$ & $(100)_1$ & $(111)_1$ \\ \hline
$\bm{(000\ 100\ 000)}$ & $(000)_1$ & $(001)_1$ & $(000)_1$ & $(011)_1$ & $(100)_1$ & $(111)_1$ & $(110)_1$ & $(111)_1$ \\ \hline
$\bm{(000\ 101\ 000)}$ & $(000)_1$ & $(011)_1$ & $(000)_1$ & $(001)_1$ & $(110)_1$ & $(111)_1$ & $(100)_1$ & $(111)_1$ \\ \hline
$\bm{(000\ 110\ 000)}$ & $(000)_1$ & $(011)_1$ & $(000)_1$ & $(001)_1$ & $(110)_1$ & $(111)_1$ & $(100)_1$ & $(111)_1$ \\ \hline
$\bm{(000\ 111\ 000)}$ & $(000)_1$ & $(011)_1$ & $(000)_1$ & $(001)_1$ & $(110)_1$ & $(111)_1$ & $(100)_1$ & $(111)_1$ \\ \hline
\end{tabular}
\caption{The output of a brute-force computation showing that $\cM_1$ is reversible with an arbitrary $m=1$ damage. Only symmetry-inequivalent configurations for damage are listed.  Since there are no outputs in the form $(010)_1$ or $(101)_1$, we conclude that the output of $\cM_1$ is reversible to some clean input.}
\label{tab:irrev}
\end{table}

\subsection{Lemma~\ref{lemma:nilpotentnumerics-TC}, proof of nilpotence of toric code gadgets} \label{app:TC-nilpotence}

In this appendix we discuss the algorithm\footnote{All computer generated results in this Appendix were obtained by the code available at \cite{code}.} that was used in order to verify nilpotence of the gadgets in the toric code automaton.  Recall that when determining nilpotence of $\EC$, we can assume a clean input state thanks to additional assumptions discussed in the Lemma~\ref{lemma:nilpotentnumerics-TC} and the linearity property in Prop.~\ref{prop:confinement-for-TC}. 
Algorithms \ref{alg:clean-address} and \ref{alg:clean-address-actual} below explain how we show nilpotence assuming input state that is coarse-grained to linear points of respective gadget at level $\ell$. Algorithms ~\ref{alg:arb-input} and \ref{alg:arb-input-actual} extend the results to arbitrary level-$\ell$ coarse-grained input for gate gadgets at level $\ell$, which is needed to deal with damage occurring at the nonlinear points of $\cM$.

\begin{algorithm} \label{alg:clean-address}
Computes level-$\ell$ output damage for a specific level-0 $\mathcal{G}^{(0)}_0$ gadget failure.
Input: a sequence of gadgets $(\mathcal{G}^{(a_\ell)}_1, \mathcal{G}^{(a_{\ell-1})}_1, \cdots, \mathcal{G}^{(a_1)}_1)$ obtained from the address  $[\mathcal{G}^{(a_\ell)}_\ell, \mathcal{G}^{(a_{\ell-1})}_{\ell-1}, \cdots, \mathcal{G}^{(a_1)}_1, \mathcal{G}^{(a_0)}_0]$.
\begin{enumerate}
\item  Initialize $D$ as the set of all possible errors (damage outputs) that can be caused by failure of the level-0 gate $\mathcal{G}^{(a_0)}_0$ (which is equal to either $\cI_0$, $\cT_0^{h/v}$ or $\cM_0^{h/v}$) on its support. 
\item Start with $j = 1$ and the level-1 gadget $\mathcal{G}^{(a_j)}_1 \circ \EC^{\otimes K_{a_j}}_1$.
\item Update $D$
\begin{equation}
 D \rightarrow  \{\varepsilon'\big| \  \exists \varepsilon  \in D: \, \varepsilon' = \Gamma[\mathcal{G}^{(a_{j-1})}_1\circ \EC_1^{\otimes K_{a_{j-1}}}, \varepsilon] \}.
\end{equation}
where the error $\varepsilon$ is located at the spacetime location $\bm{x}$ corresponding to the level-reduced (to level-0) address of the $\mathcal{G}^{(a_{j-1})}_{j-1}$ gadget within $\mathcal{G}^{(a_j)}_j \circ \EC^{\otimes K_{a_j}}_j $ (which is the information implicitly included in the error address). 
\item Set $j \to j+1$, go to Step 3.
\item End when $j = \ell$, and output $D$.
\end{enumerate}
\end{algorithm}

To show nilpotence, one would need to apply this algorithm manually for all possible addresses for errors and show that there exists a single $\ell$ for which the output damage set $D_{\cG_\ell}$ will become an empty set. Instead, we will find it convenient to use the following simplified algorithm that includes all possibilities for the noise while going up level by level (which may lead to an overestimate of the nilpotence level): 

\begin{algorithm} \label{alg:clean-address-actual}
Input: integer $\ell$ corresponding to the length of a gadget sequence.
\begin{enumerate}
\item  Initialize $D_{\mathcal{I}}$, $D_{\mathcal{T}^h}$,$D_{\mathcal{T}^v}$,$D_{\mathcal{M}^h}$, $D_{\mathcal{M}^v}$ as the set of all possible errors that can be caused by failure of individual $\mathcal{I}_0$, $\mathcal{T}_0^h$, $\cT^v_0$, $\cM^h_0$, and $\mathcal{M}^v_0$ gates on their support. 
\item  For each $\cG_1^{(a)}$, call $\cE_{\cG_1^{(a)}}$ the  noise effectively experienced by $\mathcal{G}_1 \circ \EC^{\otimes K}_1$ is  a set of $(\bm x, \varepsilon)$ containing  errors  $\varepsilon$ that are produced by failures $D_{\mathcal{I}}$, $D_{\mathcal{T}}$, and $D_{\mathcal{M}}$ at all possible locations $\bm x$ in $\mathcal{G}_1 \circ \EC^{\otimes K}_1$. 
\item  Start with $j = 1$.
\item For each $\cG$, compute (or update) 
\begin{equation}
\cE_{\mathcal{G}} = \left\{(\bm{x}, \varepsilon)\big| \  \varepsilon \in D_{\mathcal{G}'},  \ \left ( \mathcal{G}_1 \circ \EC^{\otimes K}_1 \right)(\bm{x}) = \mathcal{G}' \right\}
\end{equation}
where $\left (\mathcal{G}_1 \circ \EC^{\otimes K}_1 \right)(\bm{x})$ determines the type of the level-0 gadget at spacetime point $\bm{x}$ inside the  level-1 gadget $\mathcal{G}_1 \circ \EC^{\otimes K}_1 $. 

\item For each $\cG$, update $D_{\mathcal{G}}$ via
\begin{equation}
D_{\mathcal{G}} \rightarrow \{\varepsilon'\big| \  \exists \varepsilon \in \cE_{\mathcal{G}}:\, \varepsilon' = \Gamma[\mathcal{G}_1\circ \EC_1^{\otimes K}, \varepsilon] \}.
\end{equation}
% \rightarrow \cE_{\cG}$ for all $\cG$. 
\item Set $j \to j+1$, go to Step 4.
\item End when $j = \ell$, output $D_{\mathcal{G}}$ for all $\cG$.
\end{enumerate}
\end{algorithm}
The need to separately consider vertically and horizontally oriented $\cT_0$ and $\cM_0$ gadgets arises from the lack of $\pi/2$ rotation symmetry of the $R_0$ gadget. 
A diagram describing this computation is illustrated below (for simplicity, we denote $\mathcal{G}'$ to be the  map $\Gamma$ on $\mathcal{G}_1 \circ \EC_1^{\otimes K}$ and $\varepsilon$ for all possible $\varepsilon$ in respective $\cE$, and abuse notation by not distinguishing between vertically- and horizontally-oriented gadgets).

\begin{center}
\begin{tikzpicture}[
  >=stealth,
  ]

  % Define the nodes
\node (13) at (0,1) [below] {$\bm{D_{\cG^{(a_0)}}}$};
  \node (14) at (4,1) [below] {$\bm{D_{\cG^{(a_1)}}}$};
  \node (15) at (8,1) [below] {$\bm{D_{\cG^{(a_2)}}}$};
  \node (16) at (12,1) [below] {$\bm{D_{\cG^{(a_3)}}}$};
  
  \node (1) at (0,0) [below] {$D_{\cI}$};
  \node (2) at (4,0) [below] {$D_{\cI}$};
  \node (3) at (8,0) [below] {$D_{\cI}$};
  \node (4) at (12,0) [below] {$\emptyset$};

  \node (5) at (0,-2) [below] {$D_{\cT}$};
  \node (6) at (4,-2) [below] {$D_{\cT}$};
  \node (7) at (8,-2) [below] {$D_{\cT}$};
  \node (8) at (12,-2) [below] {$\emptyset$};

  \node (9) at (0,-4) [below] {$D_{\cM}$};
  \node (10) at (4,-4) [below] {$D_{\cM}$};
  \node (11) at (8,-4) [below] {$D_{\cM}$};
  \node (12) at (12,-4) [below] {$\emptyset$};

  % Draw horizontal lines with arrows and equation labels for the top row
  \draw[postaction={decorate, decoration={markings, mark=at position 0.3 with {\arrow{>}}}}] 
    (1) -- (2) node[pos = 0.3, above] {\fontsize{4pt}{0pt}\selectfont $\cI'$};
  \draw[postaction={decorate, decoration={markings, mark=at position 0.3 with {\arrow{>}}}}] 
    (2) -- (3) node[pos = 0.3, above] {\fontsize{4pt}{0pt}\selectfont $\cI'$};
  \draw[postaction={decorate, decoration={markings, mark=at position 0.3 with {\arrow{>}}}}] 
    (3) -- (4) node[pos = 0.3, above] {\fontsize{4pt}{0pt}\selectfont $\cI'$};

    \draw[postaction={decorate, decoration={markings, mark=at position 0.3 with {\arrow{>}}}}] 
    (9) -- (10) node[pos=0.3, below] {\fontsize{4pt}{0pt}\selectfont $\cM'$};
  \draw[postaction={decorate, decoration={markings, mark=at position 0.3 with {\arrow{>}}}}] 
    (10) -- (11) node[pos = 0.3, below] {\fontsize{4pt}{0pt}\selectfont $\cM'$};
  \draw[postaction={decorate, decoration={markings, mark=at position 0.3 with {\arrow{>}}}}] 
    (11) -- (12) node[pos = 0.3, below] {\fontsize{4pt}{0pt}\selectfont $\cM'$};

  % FIX THIS
  \draw[postaction={decorate, decoration={markings, mark=at position 0.3 with {\arrow{>}}}}] (5) -- (6) node[pos= 0.3, above] {\fontsize{4pt}{0pt}\selectfont $\cT'$};
  \draw[postaction={decorate, decoration={markings, mark=at position 0.3 with {\arrow{>}}}}] (6) -- (7) node[pos= 0.3, above] {\fontsize{4pt}{0pt}\selectfont $\cT'$};
  \draw[postaction={decorate, decoration={markings, mark=at position 0.3 with {\arrow{>}}}}] (7) -- (8) node[pos= 0.3, above] {\fontsize{4pt}{0pt}\selectfont $\cT'$};

  % Draw diagonal connections with arrows and labels
  \draw[postaction={decorate, decoration={markings, mark=at position 0.2 with {\arrow{>}}}}] 
    (1) -- (6) node[pos = 0.2,sloped, above] {\fontsize{4pt}{0pt}\selectfont $\cT'$};
  \draw[postaction={decorate, decoration={markings, mark=at position 0.2 with {\arrow{>}}}}] 
    (5) -- (2) node[pos = 0.2,sloped, above] {\fontsize{4pt}{0pt}\selectfont $\cI'$};
 \draw[postaction={decorate, decoration={markings, mark=at position 0.2 with {\arrow{>}}}}] 
    (2) -- (7) node[pos = 0.2,sloped, above] {\fontsize{4pt}{0pt}\selectfont $\cT'$};
  \draw[postaction={decorate, decoration={markings, mark=at position 0.2 with {\arrow{>}}}}] 
    (6) -- (3) node[pos = 0.2,sloped, above] {\fontsize{4pt}{0pt}\selectfont $\cI'$};
     \draw[postaction={decorate, decoration={markings, mark=at position 0.2 with {\arrow{>}}}}] 
    (3) -- (8) node[pos = 0.2,sloped, above] {\fontsize{4pt}{0pt}\selectfont $\cT'$};
  \draw[postaction={decorate, decoration={markings, mark=at position 0.2 with {\arrow{>}}}}] 
    (7) -- (4) node[pos = 0.2,sloped, above] {\fontsize{4pt}{0pt}\selectfont $\cI'$};

    \draw[postaction={decorate, decoration={markings, mark=at position 0.2 with {\arrow{>}}}}] 
    (5) -- (10) node[pos = 0.2,sloped, above] {\fontsize{4pt}{0pt}\selectfont $\cM'$};
  \draw[postaction={decorate, decoration={markings, mark=at position 0.2 with {\arrow{>}}}}] 
    (9) -- (6) node[pos = 0.2,sloped, below] {\fontsize{4pt}{0pt}\selectfont $\cT'$};
 \draw[postaction={decorate, decoration={markings, mark=at position 0.2 with {\arrow{>}}}}] 
    (6) -- (11) node[pos = 0.2,sloped, above] {\fontsize{4pt}{0pt}\selectfont $\cM'$};
  \draw[postaction={decorate, decoration={markings, mark=at position 0.2 with {\arrow{>}}}}] 
    (10) -- (7) node[pos = 0.2,sloped, below] {\fontsize{4pt}{0pt}\selectfont $\cT'$};
     \draw[postaction={decorate, decoration={markings, mark=at position 0.2 with {\arrow{>}}}}] 
    (7) -- (12) node[pos = 0.2,sloped, above] {\fontsize{4pt}{0pt}\selectfont $\cM'$};
  \draw[postaction={decorate, decoration={markings, mark=at position 0.2 with {\arrow{>}}}}] 
    (11) -- (8) node[pos = 0.2,sloped, below] {\fontsize{4pt}{0pt}\selectfont $\cT'$};

    %longest arrows
     \draw[postaction={decorate, decoration={markings, mark=at position 0.18 with {\arrow{>}}}}] 
    (1) -- (10) node[pos = 0.1,sloped, below] {\fontsize{4pt}{0pt}\selectfont $\cM'$};
    %longest arrows
     \draw[postaction={decorate, decoration={markings, mark=at position 0.18 with {\arrow{>}}}}] 
    (9) -- (2) node[pos = 0.1,sloped, above] {\fontsize{4pt}{0pt}\selectfont $\cI'$};

     %longest arrows
     \draw[postaction={decorate, decoration={markings, mark=at position 0.18 with {\arrow{>}}}}] 
    (2) -- (11) node[pos = 0.1,sloped, below] {\fontsize{4pt}{0pt}\selectfont $\cM'$};
    %longest arrows
     \draw[postaction={decorate, decoration={markings, mark=at position 0.18 with {\arrow{>}}}}] 
    (10) -- (3) node[pos = 0.1,sloped, above] {\fontsize{4pt}{0pt}\selectfont $\cI'$};

     %longest arrows
     \draw[postaction={decorate, decoration={markings, mark=at position 0.18 with {\arrow{>}}}}] 
    (3) -- (12) node[pos = 0.1,sloped, below] {\fontsize{4pt}{0pt}\selectfont $\cM'$};
    %longest arrows
     \draw[postaction={decorate, decoration={markings, mark=at position 0.18 with {\arrow{>}}}}] 
    (11) -- (4) node[pos = 0.1,sloped, above] {\fontsize{4pt}{0pt}\selectfont $\cI'$};
\end{tikzpicture}
\end{center}

The damage sets $D_{\mathcal{G}}$ for the iterations were computed explicitly via software until they yielded empty sets. For $j=0$, it is all possible errors in the support of the respective level-0 gadget. 
When $j > 0$, we find the damage sets to be (in the following figures, black dots represent syndrome locations corresponding to particular damage realizations; we also denote the number of iterations of the algorithm in the superscript of the gate) 

\begin{equation}\label{eq:damagesets1}
    \includegraphics[width = 0.4\textwidth]{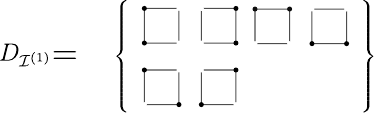}
\end{equation}

\begin{flalign}
    \includegraphics[width = 0.8\textwidth]{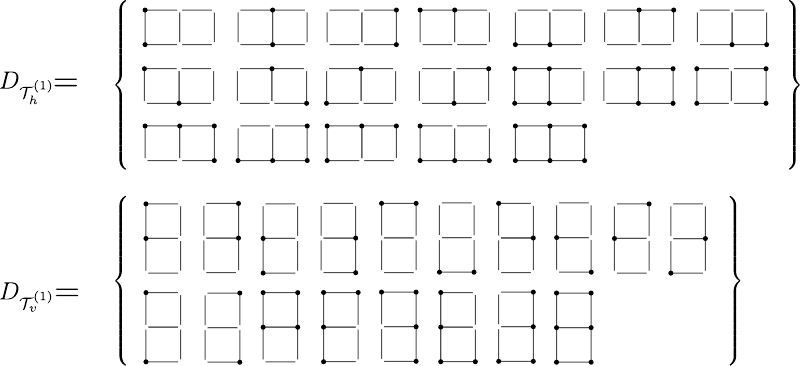}
\end{flalign}

\begin{flalign}
    \includegraphics[width = 0.9\textwidth]{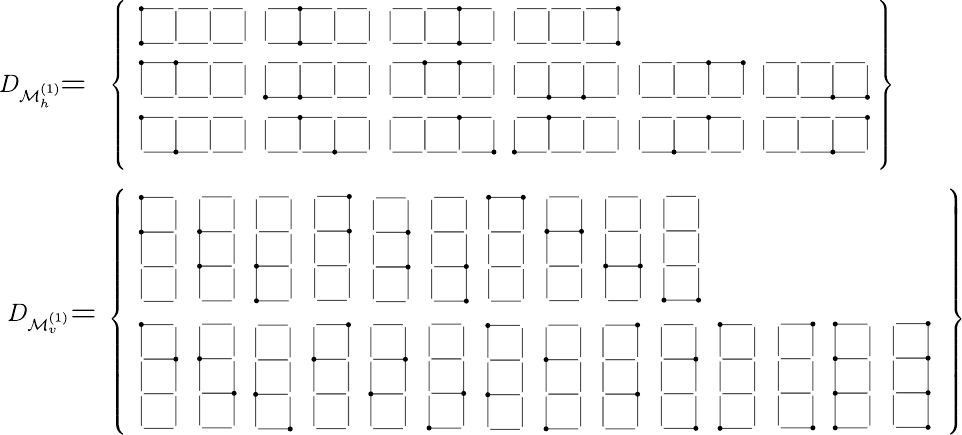}
\end{flalign}

For the next level up, 

\begin{flalign}
    \includegraphics[width = 0.37 \textwidth]{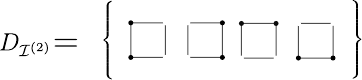}
\end{flalign}

\begin{flalign}
    \includegraphics[width = 0.55\textwidth]{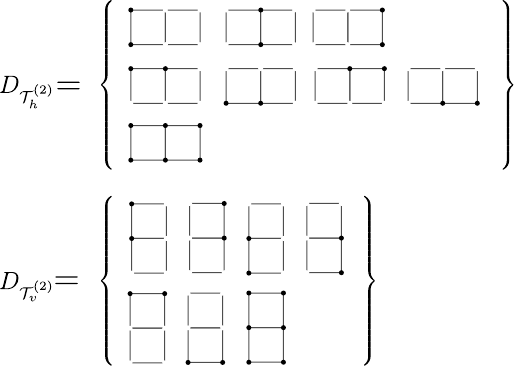}
\end{flalign}

\begin{flalign}
    \includegraphics[width = 0.9\textwidth]{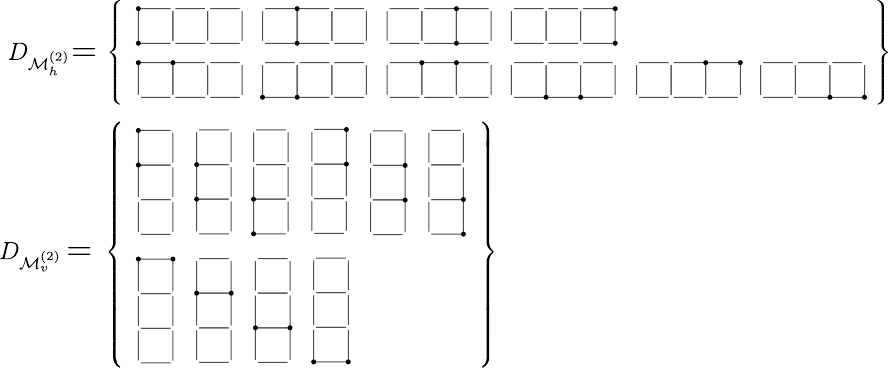}
\end{flalign}

One curious point is that the damage configurations appear to almost always satisfy a `mirror' symmetry (i.e. for every damage with syndrome $\sigma$, the vertically and horizontally reflected version of the damage is present too), even though the leading $\EC$ units to every gadget do not obey this symmetry.  The exception is in $D_{\cT_v^{(1)}}$, where the third element in the second row does not have its horizontal reflection present.  

An additional observation is that upon going up two levels, the gadget error model turns into a link error model (in which the only possible gadgets failures are ones which create syndromes on the endpoints of a single link), similar to the case of wire error model that we originally considered for Tsirelson's automaton (notice that the last error in $D_{\cT^{(2)}}$ looks like three vertical link errors; however, an input wire-model error that looks like a failure of a single vertical link at the middle position produces this damage as an output under an operation of noiseless $\cT_1^{h}$ gadget; thus, this damage is equivalent to a single link failure).

Later in the proof, we also reduce the analysis some sequence of gadgets $(\mathcal{G}^{(a_r)}_1, \mathcal{G}^{(a_{r-1})}_1, \cdots, \mathcal{G}^{(a_r)}_1)$ with input syndromes $(\sigma^{(r)}, \sigma^{(r-1)}, \cdots, \sigma^{(1)})$, assuming that each $\sigma^{(i)}$ is coarse-grained at respective level $i$.  We first defined the map $\Gamma[\mathcal{G}_1, \sigma, \varepsilon]$ which operates as follows:  we discard the syndrome damage $\sigma$ on linear points in $\Sigma_1$ (which we can always add later to the output because of linearity) and thus, assume a syndrome damage $\sigma$ which has support entirely on non-linear points of respective in $\Sigma_1$. We then evolve $\mathcal{G}_1$ under this input with an isolated level-0 fault $\varepsilon$ inside of it before applying a clean $\EC^{\otimes K}_1$, and following the output by the pushforward map $\Phi$.  Then we perform a suitable iteration of this procedure.  The difference with what we did before is that we also need to track the input syndrome at the non-linear locations of the gadgets.  This procedure is summarized by the following algorithm: 

\begin{algorithm} \label{alg:arb-input}
Computes level-$\ell$ output damage for a specific gadget failure $\mathcal{G}^{(a_0)}_0$.
Input: a sequence of gadgets $(\mathcal{G}^{(a_r)}_1, \mathcal{G}^{(a_{r-1})}_1, \cdots, \mathcal{G}^{(a_1)}_1, \mathcal{G}^{(0)}_1)$ obtained from the address  $[\mathcal{G}^{(a_r)}_r, \mathcal{G}^{(a_{r-1})}_{r-1}, \cdots, \mathcal{G}^{(a_1)}_1, \mathcal{G}^{(a_0)}_0]$, and level-reduced input syndromes $(\sigma^{(r)}, \sigma^{(r-1)}, \cdots, \sigma^{(1)})$ to each gadget such that $\sigma^{(i)} \in \widetilde{\Sigma}[\mathcal{G}^{(a_i)}_1]$.
\begin{enumerate}

\item  Initialize $D$ as a set of all possible errors (output damage)  that can be caused by failure of $\mathcal{G}^{(a_0)}_0$ (which can be $\mathcal{I}_0$, $\mathcal{T}_0^{h/v}$, or $\mathcal{M}_0^{h/v}$) in its support. 
\item Start with $j = 1$ and the level-1 gadget $\mathcal{G}^{(a_1)}_1$ with the corresponding input syndrome $\sigma^{(1)}$ supported only on nonlinear points of the gadget in $\widetilde \Sigma_1$. 
\item Update  $D$ via
\begin{equation}
D \rightarrow \{\varepsilon'\big| \  \exists \varepsilon \in D: \, \varepsilon' = \Gamma[\mathcal{G}^{(a_j)}_1, \sigma^{(j)}, \varepsilon] \}.
\end{equation}
where the error $\varepsilon$ is in addition located at the spacetime location $\bm x$ corresponding to the level-reduced (to level 0) address of the $\cG_{j-1}^{(a_{j-1})}$ within $\mathcal{G}^{(a_j)}_j$ (which is the information implicitly included in the error address).
\item Set $j \to j+1$, go to Step 3.
\item End when $j = r$, output $D$.
\end{enumerate}
\end{algorithm}

Once again, to exhaustively inspect all possibilities for errors, we use the following modified algorithm

\begin{algorithm}  \label{alg:arb-input-actual}
Input: integer $r$ corresponding to the length of a gadget sequence.
\begin{enumerate}
\item  Initialize $D_{\mathcal{I}}$, $D_{\mathcal{T}^h}$, $D_{\mathcal{T}^v}$, $D_{\mathcal{M}^h}$, and $D_{\mathcal{M}^v}$ as sets of all possible errors that can be caused by failure of individual level-0 gates on their support. Denote $\widetilde{\Sigma}[\mathcal{G}]$ to be the non-linear points of the gadget $\cG$ (nontrivial only for the $\cM^h,\cM^v$ gadgets).   
\item For each $\cG$, call $\cE_{\mathcal{G}}$ the noise effectively experienced by $\mathcal{G}_1$, which is a set of $(\bm x, \varepsilon)$ containing errors  $\varepsilon$ that are produced by failures $D_{\mathcal{I}}$, $D_{\mathcal{T}}$, and $D_{\mathcal{M}}$ at all possible locations $\bm x$ in $\mathcal{G}_1$. 
\item Start with $j = 1$.
\item For each $\cG$, compute (or update)
\begin{equation}
\cE_{\mathcal{G}} = \left\{(\bm{x}, \varepsilon)\big| \  \varepsilon \in D_{\mathcal{G}'},  \ \mathcal{G}_1 (\bm{x}) = \mathcal{G}' \right\}
\end{equation}
where $\mathcal{G}_1 (\bm{x})$ determines the type of the level-0 gadget at spacetime point $\bm{x}$ inside the  level-1 gadget $\mathcal{G}_1$. 

\item For each $\cG$, update  $D_{\mathcal{G}}$ via
\begin{equation}
D_{\mathcal{G}} \rightarrow \{\varepsilon'\big| \    \exists \varepsilon \in \cE_{\mathcal{G}}, \, \exists \sigma \in \widetilde{\Sigma}[\mathcal{G}_1]: \, \varepsilon' = \Gamma[\mathcal{G}_1, \sigma, \varepsilon] \}.
\end{equation}
\item Set $j \to j+1$, go to Step 4.
\item End when $j = r$, output $D_{\mathcal{G}}$.

\end{enumerate}
\end{algorithm}
A similar diagram can be drawn illustrating the algorithm above, although it has a bit different structure:
A diagram describing this computation is illustrated below (for simplicity, we denote $\mathcal{G}'$ to be the $\Gamma$ map on $\mathcal{G}_1$ and $\varepsilon$ for all possible $\varepsilon$ in respective $\cE$ and all possible coarse-grained syndromes $\sigma$, and as before we abuse notation by not distinguishing between horizontally- and vertically-oriented gates)
\begin{center}
\begin{tikzpicture}[
  >=stealth,
  ]

  % Define the nodes
  \node (13) at (0,1) [below] {$\bm{D_{\cG^{(a_0)}}}$};
  \node (14) at (4,1) [below] {$\bm{D_{\cG^{(a_1)}}}$};
  \node (15) at (8,1) [below] {$\bm{D_{\cG^{(a_2)}}}$};
  \node (16) at (12,1) [below] {$\bm{D_{\cG^{(a_3)}}}$};
  
  \node (1) at (0,0) [below] {$D_{\cI}$};
  \node (2) at (4,0) [below] {$D_{\cI}$};
  \node (3) at (8,0) [below] {$\emptyset$};
  \node (4) at (12,0) [below] {$\emptyset$};

  \node (5) at (0,-2) [below] {$D_{\cT}$};
  \node (6) at (4,-2) [below] {$D_{\cT}$};
  \node (7) at (8,-2) [below] {$\emptyset$};
  \node (8) at (12,-2) [below] {$\emptyset$};

  \node (9) at (0,-4) [below] {$D_{\cM}$};
  \node (10) at (4,-4) [below] {$D_{\cM}$};
  \node (11) at (8,-4) [below] {$D_{\cM}$};
  \node (12) at (12,-4) [below] {$\emptyset$};

  % Draw horizontal lines with arrows and equation labels for the top row
  \draw[postaction={decorate, decoration={markings, mark=at position 0.3 with {\arrow{>}}}}] 
    (1) -- (2) node[pos = 0.3, above] {\fontsize{4pt}{0pt}\selectfont $\cI$};
  \draw[postaction={decorate, decoration={markings, mark=at position 0.3 with {\arrow{>}}}}] 
    (2) -- (3) node[pos = 0.3, above] {\fontsize{4pt}{0pt}\selectfont $\cI$};
  \draw[postaction={decorate, decoration={markings, mark=at position 0.3 with {\arrow{>}}}}] 
    (3) -- (4) node[pos = 0.3, above] {\fontsize{4pt}{0pt}\selectfont $\cI$};

    \draw[postaction={decorate, decoration={markings, mark=at position 0.3 with {\arrow{>}}}}] 
    (9) -- (10) node[pos=0.3, below] {\fontsize{4pt}{0pt}\selectfont $\cM$};
  \draw[postaction={decorate, decoration={markings, mark=at position 0.3 with {\arrow{>}}}}] 
    (10) -- (11) node[pos = 0.3, below] {\fontsize{4pt}{0pt}\selectfont $\cM$};
  \draw[postaction={decorate, decoration={markings, mark=at position 0.3 with {\arrow{>}}}}] 
    (11) -- (12) node[pos = 0.3, below] {\fontsize{4pt}{0pt}\selectfont $\cM$};

  % FIX THIS
  \draw[postaction={decorate, decoration={markings, mark=at position 0.3 with {\arrow{>}}}}] (5) -- (6) node[pos= 0.3, above] {\fontsize{4pt}{0pt}\selectfont $\cT$};
  \draw[postaction={decorate, decoration={markings, mark=at position 0.3 with {\arrow{>}}}}] (6) -- (7) node[pos= 0.3, above] {\fontsize{4pt}{0pt}\selectfont $\cT$};
  \draw[postaction={decorate, decoration={markings, mark=at position 0.3 with {\arrow{>}}}}] (7) -- (8) node[pos= 0.3, above] {\fontsize{4pt}{0pt}\selectfont $\cT$};

  % Draw diagonal connections with arrows and labels
  \draw[postaction={decorate, decoration={markings, mark=at position 0.2 with {\arrow{>}}}}] 
    (1) -- (6) node[pos = 0.2,sloped, above] {\fontsize{4pt}{0pt}\selectfont $\cT$};
  \draw[postaction={decorate, decoration={markings, mark=at position 0.2 with {\arrow{>}}}}] 
    (5) -- (2) node[pos = 0.2,sloped, above] {\fontsize{4pt}{0pt}\selectfont $\cI$};
 \draw[postaction={decorate, decoration={markings, mark=at position 0.2 with {\arrow{>}}}}] 
    (2) -- (7) node[pos = 0.2,sloped, above] {\fontsize{4pt}{0pt}\selectfont $\cT$};
  \draw[postaction={decorate, decoration={markings, mark=at position 0.2 with {\arrow{>}}}}] 
    (6) -- (3) node[pos = 0.2,sloped, above] {\fontsize{4pt}{0pt}\selectfont $\cI$};
     \draw[postaction={decorate, decoration={markings, mark=at position 0.2 with {\arrow{>}}}}] 
    (3) -- (8) node[pos = 0.2,sloped, above] {\fontsize{4pt}{0pt}\selectfont $\cT$};
  \draw[postaction={decorate, decoration={markings, mark=at position 0.2 with {\arrow{>}}}}] 
    (7) -- (4) node[pos = 0.2,sloped, above] {\fontsize{4pt}{0pt}\selectfont $\cI$};

    \draw[postaction={decorate, decoration={markings, mark=at position 0.2 with {\arrow{>}}}}] 
    (5) -- (10) node[pos = 0.2,sloped, above] {\fontsize{4pt}{0pt}\selectfont $\cM$};
  \draw[postaction={decorate, decoration={markings, mark=at position 0.2 with {\arrow{>}}}}] 
    (9) -- (6) node[pos = 0.2,sloped, below] {\fontsize{4pt}{0pt}\selectfont $\cT$};
 \draw[postaction={decorate, decoration={markings, mark=at position 0.2 with {\arrow{>}}}}] 
    (6) -- (11) node[pos = 0.2,sloped, above] {\fontsize{4pt}{0pt}\selectfont $\cM$};
  \draw[postaction={decorate, decoration={markings, mark=at position 0.2 with {\arrow{>}}}}] 
    (10) -- (7) node[pos = 0.2,sloped, below] {\fontsize{4pt}{0pt}\selectfont $\cT$};
     \draw[postaction={decorate, decoration={markings, mark=at position 0.2 with {\arrow{>}}}}] 
    (7) -- (12) node[pos = 0.2,sloped, above] {\fontsize{4pt}{0pt}\selectfont $\cM$};
  \draw[postaction={decorate, decoration={markings, mark=at position 0.2 with {\arrow{>}}}}] 
    (11) -- (8) node[pos = 0.2,sloped, below] {\fontsize{4pt}{0pt}\selectfont $\cT$};

    %longest arrows
     \draw[postaction={decorate, decoration={markings, mark=at position 0.18 with {\arrow{>}}}}] 
    (1) -- (10) node[pos = 0.1,sloped, below] {\fontsize{4pt}{0pt}\selectfont $\cM$};
    %longest arrows
     \draw[postaction={decorate, decoration={markings, mark=at position 0.18 with {\arrow{>}}}}] 
    (9) -- (2) node[pos = 0.1,sloped, above] {\fontsize{4pt}{0pt}\selectfont $\cI$};

     %longest arrows
     \draw[postaction={decorate, decoration={markings, mark=at position 0.18 with {\arrow{>}}}}] 
    (2) -- (11) node[pos = 0.1,sloped, below] {\fontsize{4pt}{0pt}\selectfont $\cM$};
    %longest arrows
     \draw[postaction={decorate, decoration={markings, mark=at position 0.18 with {\arrow{>}}}}] 
    (10) -- (3) node[pos = 0.1,sloped, above] {\fontsize{4pt}{0pt}\selectfont $\cI$};

     %longest arrows
     \draw[postaction={decorate, decoration={markings, mark=at position 0.18 with {\arrow{>}}}}] 
    (3) -- (12) node[pos = 0.1,sloped, below] {\fontsize{4pt}{0pt}\selectfont $\cM$};
    %longest arrows
     \draw[postaction={decorate, decoration={markings, mark=at position 0.18 with {\arrow{>}}}}] 
    (11) -- (4) node[pos = 0.1,sloped, above] {\fontsize{4pt}{0pt}\selectfont $\cI$};
\end{tikzpicture}
\end{center}

The damage sets $D_{\mathcal{G}}$ for the iterations were computed explicitly via software until they yielded empty sets. For $k=0$, it is all possible errors in the support of respective level-0 gadget. For higher $k$, we find the damage sets to be as shown below (here, we show respective syndromes for each element in the list). Note that the damage sets are smaller than in Eq.~\ref{eq:damagesets1} because the definition of $\cG'$ does not contain the $\EC$ gadget at the beginning, for the reasons explained in the proof of Lemma~\ref{lemma:nilpotentnumerics-TC}.

For level 1, we obtain (as before, we denote the number of iterations of the algorithm that have been completed as a superscript to the gate)

\begin{equation}
    \includegraphics[width = 0.28\textwidth]{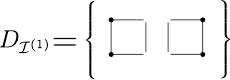}
\end{equation}

\begin{equation}
    \includegraphics[width = 0.7\textwidth]{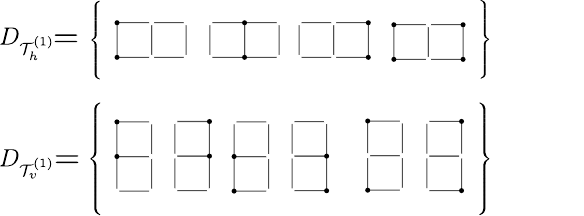}
\end{equation}

\begin{equation}
    \includegraphics[width = 0.75 \textwidth]{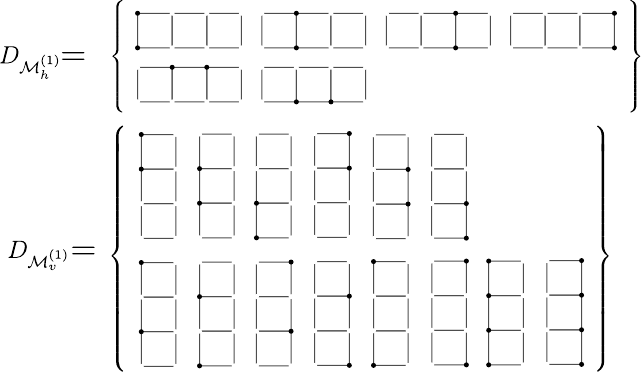}
\end{equation}

For level 2, all damage sets are trivial except those for $\cM^{(2)}_{v/h}$: 

\begin{equation}
    \includegraphics[width = 0.45\textwidth]{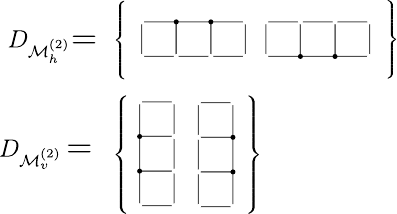}
\end{equation}

Notice that these damage sets also possess mirror symmetry, and that after going up two levels, the error model is reduced to be a single link failure on the $\cM$ gadgets.